\documentclass[12pt, oneside, hidelinks]{article}

\usepackage[margin=1in]{geometry} 
\usepackage{amsmath,amsthm,amssymb, graphicx, multicol, array,hyperref, bbm, comment, caption, subcaption,  float, booktabs, xpatch, color,changepage,mathtools}
\usepackage{diagbox}
\newcolumntype{L}[1]{>{\raggedright\arraybackslash}p{#1}}

\usepackage{enumitem}
\usepackage{booktabs}

\usepackage{setspace} 

\usepackage[margin=10pt,labelfont=bf]{caption}
\newcommand\tcaptab[1]{\captionsetup{position=top, font=normalsize, labelfont=bf, textfont=normalfont, justification=centering, margin=0mm, aboveskip=1mm, belowskip=0mm, labelsep=colon, singlelinecheck=false}\caption{#1}}
\newcommand\bnotetab[1]{\captionsetup{position=bottom, font=scriptsize,  textfont=normalfont, margin=1mm,aboveskip=1mm, belowskip=-1mm, justification=justified, singlelinecheck=false}\caption*{#1}}
\newcommand\tcapfig[1]{\captionsetup{position=top, font=normalsize, labelfont=bf, textfont=normalfont, justification=centering, margin=0mm, aboveskip=1mm, belowskip=0mm, labelsep=colon, singlelinecheck=false}\caption{#1}}
\newcommand\bnotefig[1]{\captionsetup{position=bottom, font=scriptsize,  textfont=normalfont, margin=1mm, aboveskip=1mm, belowskip=-1mm, justification=justified, singlelinecheck=false}\caption*{#1}}

\newcommand\subcaptab[1]{\captionsetup{position=bottom, font=small, labelfont=bf, textfont=normalfont, justification=centering, margin=0mm, aboveskip=4mm, belowskip=0mm, labelsep=space, singlelinecheck=false}\caption{#1}}

\def\*#1{\mathbf{#1}}
\def\+#1{\mathbb{#1}}
\def\wt#1{\widetilde{#1}}

\newcommand{\I}{{-1}}
\newcommand{\T}{\top}

\newcommand{\Lp}{\left(}
\newcommand{\Rp}{\right)}

\newcommand{\weight}{\theta}
\newcommand{\Weight}{\Theta}
\newcommand{\diag}{\textnormal{diag}}
\newcommand{\twt}{\textnormal{wt}}
\newcommand{\tr}{\textnormal{tr}}
\newcommand{\iid}{\textnormal{iid}}
\newcommand{\idx}{\textnormal{idx}}
\newcommand{\fidx}{s}
\newcommand{\tfidx}{\tilde{s}}
\newcommand{\CDF}{\textnormal{CDF}}
\newcommand*{\thisdraft}{This draft: July 2020} 
\newcommand*{\firstdraft}{First draft: May 2018}  
\usepackage[section]{placeins}

\newcommand{\pl}{\overset{p}{\rightarrow}}

\newcommand{\norm}[1]{\left\lVert#1\right\rVert}
\newcommand{\RNum}[1]{\uppercase\expandafter{\romannumeral #1\relax}}

\newtheorem{assumption}{Assumption}
\newtheorem{theorem}{Theorem}
\newtheorem{lemma}{Lemma}
\newtheorem{proposition}{Proposition}

\usepackage{natbib}

\setcitestyle{open={(},close={)}}

\usepackage{etoolbox}
\patchcmd{\chapter}{plain}{empty}{}{}

\begin{document}

	\title{Interpretable Sparse Proximate Factors for Large Dimensions\thanks{\scriptsize We thank Bryan Kelly, Kay Giesecke, Serena Ng, Jos\'e Luis Montiel Olea, Dacheng Xiu and seminar and conference participants at Stanford University, Columbia University, NBER-NSF Time-Series Conference, SoFiE Meeting, California Econometrics Conference, SoFiE Summer School on Machine Learning and Empirical Asset Pricing and INFORMS for helpful comments.}}
	\date{\thisdraft \\ \firstdraft}
	\date{\firstdraft \\ \thisdraft }
	
	\author{Markus Pelger\thanks{\scriptsize Stanford University, Department of Management Science \& Engineering, Email: \url{mpelger@stanford.edu}.}
		\and
		Ruoxuan Xiong\thanks{ \scriptsize Stanford University, Department of Management Science \& Engineering, Email: \url{rxiong@stanford.edu}.}
	}

	\begin{titlepage}
		\maketitle
		\thispagestyle{empty}
		\begin{abstract}
			{\small
			 This paper proposes sparse and easy-to-interpret proximate factors to approximate statistical latent factors. Latent factors in a large-dimensional factor model can be estimated by principal component analysis (PCA), but are usually hard to interpret. We obtain proximate factors that are easier to interpret by shrinking the PCA factor weights and setting them to zero except for the largest absolute ones. We show that proximate factors constructed with only 5-10\% of the data are usually sufficient to almost perfectly replicate the population and PCA factors without actually assuming a sparse structure in the weights or loadings. Using extreme value theory we explain why sparse proximate factors can be substitutes for non-sparse PCA factors. We derive analytical asymptotic bounds for the correlation of appropriately rotated proximate factors with the population factors. These bounds provide guidance on how to construct the proximate factors. In simulations and empirical analyses of financial portfolio and macroeconomic data we illustrate that sparse proximate factors are close substitutes for PCA factors with average correlations of around 97.5\% while being interpretable.

				\vspace{1cm}
				

			}
		\end{abstract}
	\end{titlepage}

	\begin{onehalfspacing}

		\section{Introduction}
		Large dimensional datasets are becoming increasingly relevant to study problems in economics and finance. These datasets contain rich information and open the possibilities of new findings and new efficacious models, which is feasible only if we better understand the datasets. Factor modeling \citep{bai2002determining,bai2003inferential,fan2013large} is a method that summarizes information in large-dimensional panel data and is an active research topic. In a large-dimensional factor model, both the time dimension and the cross-section dimension of the data sets are large and most of the co-movement can be explained by a few factors. Latent factors estimated from the data are particularly appealing as the underlying factor structure is usually not known. These factors are usually estimated by Principal Component Analysis (PCA).
		Latent PCA factors have been successfully used in economics and finance, for example for prediction and forecasting \citep{stock2002forecasting, stock2002macroeconomic}, asset pricing of many assets \citep{lettau2020factors,kelly2018characteristics}, high-frequency asset return modeling \citep{pelger2018,xiu2018}, conditional risk-return and term structure analysis \citep{ludvigson2007,ludvigson2009} and optimal portfolio construction \citep{fan2013large}.\footnote{Of course, the application of latent factor models goes beyond the applications listed above, e.g. inferring missing values in matrices \citep{xiong2020,candes2010power,candes2011robust}.} However, as the latent PCA factors are linear combinations of all cross-sectional units, they are usually hard to interpret. This poses a challenge for modeling and understanding the underlying structure to explain the co-movement in the data. This becomes particularly relevant for applications in economics and finance where the objective is to understand the underlying economic mechanism.       
		
		Practitioners and academics alike have used an intuitive approach to interpret latent statistical factors by focusing on the largest factor weights. A pattern in the largest factor weights suggests an economic interpretation as suggested by \cite{lettau2020factors} and \cite{pelger2019}. In this paper, we formalize this idea and show that the factors that are only based on the largest factor weights provide already an excellent approximation to the population factors. This step further reduces the dimensionality of the problem helping to better understand the economic mechanism in the problem. 
		
		We propose easy-to-interpret proximate factors for latent factors. We exploit the insight that cross-section units with larger factor weights have a larger signal-to-noise ratio, hence providing more information about the underlying factors. We start by estimating the underlying factor structure with conventional PCA, which returns the weights of PCA factors as the eigenvectors of the largest eigenvalues of a sample covariance matrix. Our method consists of three simple steps. First, we set all factor weights to zero except for the largest absolute ones. Second, the proximate factors are obtained from a simple regression on the thresholded factor weights. Third, the loadings are obtained from a regression on the proximate factors. We work under a general scenario where the factor weights and loadings in the true model are not sparse, i.e. the number of nonzero elements is proportional to the size of the cross-section, while the sparse factor weights have only a finite and small number of nonzero entries.
		
		We show that one needs to make a clear distinction between factor weights and loadings. In a conventional PCA analysis the loadings serve two purposes: First, they measure the exposure of the cross-sectional data to the factors. Second, they correspond to weights to construct the latent PCA factors. The same view is for example taken in a sparse PCA setup where loadings and hence also the factor weights are shrunk to a sparse matrix. We show that in an approximate factor model with non-sparse factor weights and loadings, we can construct proximate factors with sparse factor weights that are very close to the population factors. These proximate factors have non-sparse loadings which are consistent estimates of the true population loadings. The proximate factors are considerably easier to interpret as they are only based on a small fraction of the data, while enjoying the same properties as the non-sparse factors.\footnote{A common method to interpret low-dimensional factor models is to find a rotation of the common factors with a meaningful interpretation. This approach uses the insight that factor models are only identified up to an invertible transformation and represent the same model after an appropriate rotation. The $varimax$ criterion proposed by \cite{kaiser1958varimax} is a popular way to select factors whose factor weights have groups of large and negligible coefficients. However, in large-dimensional factor models with non-sparse factor weight structure, finding a ``good'' rotation becomes considerably more challenging. It is generally easier with our sparse proximate factors to find a rotation that has a meaningful interpretation.}

		We develop the statistical arguments that explain why the sparse proximate factors can be used as substitutes for the non-sparse PCA factors in general approximate factor models. The conventional derivations used in large-dimensional factor modeling to prove consistency do not apply to our proximate factors. The conventional non-sparse PCA factors are a weighted average over a large number of cross-sectional units which diversifies away the idiosyncratic noise. In contrast, proximate factors are weighted averages over only a finite and small number of cross-sectional units, which in general do not average out the noise and lead to biased estimates of the population factors. However, we show that, by selecting the most informative units as weights, this bias can become negligible. 
		The closeness between proximate factors and true factors is measured by the generalized correlation.\footnote{Generalized correlation (also called canonical correlation) measures how close two vector spaces are. It has been studied by \cite{anderson1958introduction} and applied in large-dimensional factor models \citep{bai2006confidence,pelger2018,pelger2018state,andreou2019inference}.} We provide an asymptotic probabilistic lower bound for the generalized correlation.\footnote{In simulations, we verify that the lower bound has good finite sample properties.} This lower bound is easy to calculate and depends on the number of nonzero elements in the sparse factor weights and on the tail distribution of population factor weights which we can characterize with Extreme Value Theory (EVT). This lower bound provides guidance on constructing the proximate factors that are guaranteed to have a high correlation with the true factors. One insight of this bound is that in the case of heteroskedastic errors the most informative factor weights are the units with the largest loadings relative to the unit-specific noise variance. Moreover, when factor weights have unbounded support, we show that proximate factors asymptotically span the same space as latent factors.
		Importantly, the estimated loadings of the proximate factors converge to the true population loadings up to the usual rotation. This surprising result is due to the two-stage procedure for estimating the loadings that averages out the idiosyncratic noise in the time dimension. Hence, regressions based on the more interpretable proximate factors will asymptotically yield the same results as using the harder to interpret PCA factors.

		Our results are of practical and theoretical importance. First, we provide a simple and easy-to-implement method to approximate latent factors by a small number of observations. These sparse proximate factors usually provide a much simpler economic interpretation of the model, either by directly analyzing the sparse composition or after rotating them appropriately. Our method has already been successfully used in \cite{lettau2020factors} to interpret asset pricing factors and in \cite{pelger2019} to interpret high-frequency factors. Second, we show in empirical and simulation studies that this approximation works surprisingly well and almost no information is lost by working with the sparse factors. 
		Third, an important contribution of this paper is to characterize the asymptotic properties of a sparse factor estimator if the population model is not sparse. 
		Our asymptotic bounds for the correlations between the proximate factors and the population factors provides the theoretical reasoning why our method works so well. In particular, it clarifies under which assumptions the proximate factors are a good approximation or even converge to the population factors. 
		Fourth, the asymptotic bounds provide guidance on how to select the key tuning parameter for our estimator, that is, the number of nonzero elements.\footnote{The degree of sparsity can also be chosen to obtain a sufficiently high generalized correlation with the estimated PCA factors or by cross-validation arguments to explain a sufficient amount of variation in the data. Our theoretical bound provides an alternative to select the tuning parameter based on arguments that should also hold out-of-sample.}
		
		

		We need to overcome three major challenges in the derivations of the probabilistic lower bound for the generalized correlation. First, we need to show that the estimated loadings converge uniformly to some rotation of the true loadings under the general approximate factor model assumptions.\footnote{We impose assumptions similar to \cite{bai2002determining}.}
		Uniform consistency of the estimated loadings is necessary for our argument that cross-section units with the largest estimated loadings have the largest probabilistic ``signal'' for the underlying factors. Second, the hard-thresholding procedure that sets most weights to zero to get the sparse factor weights has the down-side that we lose some of the large sample properties in the cross-section dimension. We need to take into account the cross-section dependency structure in the errors among a few cross-section units, which directly affects the noise level in calculating the generalized correlation. Third, the sparse factor weights, as well as proximate factors, are in general not orthogonal to one another in contrast to their non-sparse estimated version. A proximate factor for a specific population factor can also be correlated with other latent factors, which is reflected in the generalized correlation. Our assumptions and results need to take these complex inter-factor correlations into account.

		In two empirical applications, we apply our method to a large number of financial portfolios and a large-dimensional macroeconomic dataset. We find that in both datasets, proximate factors with around 5-10\% of the cross-sectional units can very well approximate the non-sparse PCA factors with average correlations of around 97.5\%. The proximate factors explain almost the same amount of variation as the non-sparse PCA factors. The sparse factors have an economic meaningful interpretation which would be hard to obtain from the non-sparse representation. For example, in the portfolio data, the non-sparse PCA factor weights are nonzero for all 370 assets, which makes any interpretation challenging. In contrast, the 30 nonzero weights for the proximate factors are grouped together in portfolios with specific characteristics that allow us to assign an economic label to these factors. For example, we discover the well-known ``momentum'' factor as one of the latent factors.

		Sparse loadings have already been employed to reduce the composition of factors and to make latent factors more interpretable. Most work formulates the estimation of sparse loadings as a regularized optimization problem to estimate principal components with a soft thresholding term, such as an $\ell_1$ penalty term or an elastic net penalty term \citep{zou2006sparse,mairal2010online,bai2008forecasting,bai2016efficient}.\footnote{\cite{choi2010penalized, lan2014sparse} and \cite{kawano2015sparse1} estimate the sparse loadings by minimizing the sum of the negative log-likelihood of the data with a soft thresholding term.} An alternative approach is to take a Bayesian perspective and specify sparse priors for factor loadings and use Bayesian updating to obtain posteriors for sparse loadings \citep{lucas2006sparse,bhattacharya2011sparse,kaufmann2013bayesian,pati2014posterior}. All these approaches assume that loadings are sparse in the population model, which allows these approaches to develop an asymptotic inferential theory. Nevertheless, the assumption of sparse population loadings may not be satisfied in many datasets. For example, the exposure to a market factor is universal and non-sparse in equity data. It is important to understand that this line of work typically does not distinguish between factor weights and loadings, i.e. it uses the shrunk loadings for factors weights and exposure. In contrast, we only estimate sparse factor weights, but non-sparse loadings. Additionally, we do not assume that the population factor weights are sparse. It is possible to adjust the sparse PCA method to make this important distinction between factor weights and loadings, i.e. loadings are obtained from a second stage regression. However, in contrast to our approach, there is no statistical theory explaining and justifying this adjusted sparse PCA approach in the same general setup that we are using. Furthermore, lasso-type estimators have the well-known shortcoming that they create biased estimates by the way how they are shrinking large elements and a similar non-optimal shrinking happens for sparse PCA. It turns out that sparse PCA performs substantially worse than our method in simulations and empirical applications.\footnote{Another method to increase the understandability and interpretability of factors is to associate latent factors or factor loadings with observed variables. Some latent factors can be approximated well by observed economic factors, such as Fama-French factors for equity data \citep{fama1992cross} or level, slope, and curvature factors for bond data \citep{diebold2006forecasting}. \cite{fan2016robust} propose robust factor models to exploit the explanatory power of observed proxies on latent factors. Another approach is to model how the factor loadings relate to observable variables. \cite{connor2007semiparametric}, \cite{connor2012efficient}, and \cite{fan2016projected} at least partially employ subject-specific covariates to explain factor loadings, such as market capitalization, price-earning ratios, and other firm characteristics. However, in order to explain latent factors by observed variables, it is necessary to include all the relevant variables, some of which might not be known. Our sparse proximate factors can provide discipline on which assets and covariates to focus on.    
		}

		The rest of the paper is structured as follows. Section \ref{sec:model-setup} introduces the model and the estimator. Section \ref{sec:theory} shows the consistency result for the estimated loadings and the asymptotic probabilistic lower bound for the generalized correlation. Section \ref{sec:simulation} presents simulation results. In section \ref{sec:empirical} we apply our approach to financial and macroeconomic data. Section \ref{sec:conclusion} concludes the paper. All proofs and additional empirical results are delegated to the Internet Appendix.

		\section{Model Setup}\label{sec:model-setup}

		\subsection{Estimator}\label{subsec:estimator}

		Suppose we observe a large dimensional panel data set $\*X \in \mathbb{R}^{N \times T}$ with $N$ units over $T$ time periods, where both $N$ and $T$ are large. Furthermore, suppose $\*X$ has a latent factor structure with $K$ common factors: 
		\begin{eqnarray}\label{eqn-true-model}
		\*X = \*\Lambda \*F^\T  + \*e, 
		\end{eqnarray}
		where $\*F \in \mathbb{R}^{T \times K}$,  $\*\Lambda \in \mathbb{R}^{N \times K}$, and $\*e \in \mathbb{R}^{N \times T}$ are common factors, factor loadings, and idiosyncratic components which are all unobserved.\footnote{We assume that we have consistently estimated the number of factors $K$.} A $K$-factor model can be estimated from the panel data by PCA. The PCA factors $\hat{\*F} \in \mathbb{R}^{T \times K}$ and loadings $\hat{\*\Lambda} \in \mathbb{R}^{N \times K}$ minimize a quadratic loss function 
		\begin{eqnarray}\label{eqn-quad-loss}
		\{\hat{\*F}, \hat{\*\Lambda} \} = \arg \min_{\*F, \*\Lambda} \frac{1}{NT} \sum_{i = 1}^N \sum_{t = 1}^T (X_{it} - \Lambda_i^\T F_t)^2, 
		\end{eqnarray}
		where $X_{it}$ is the observation of the $i$-th cross-section unit at time $t$, $F_t = [F_{t,1}, \cdots, F_{t,K}]^\T \in \mathbb{R}^{K \times 1} $ is the factor value at time $t$, and $\Lambda_i = [\Lambda_{i,1}, \cdots, \Lambda_{i, K}]^\T \in \mathbb{R}^{K \times 1}$ is the factor loading of the $i$-th cross-section unit. Latent factors and loadings are only identified up to an invertible transformation: For any invertible matrix $H$, $H {\Lambda}_i$ and $(H^\T)^{-1}{F}_t$ specify the same factor model as in \eqref{eqn-true-model}. Under the standard identification assumption that $\frac{1}{N} \hat{\*\Lambda}^{\top}\hat{\*\Lambda}$ is an identity matrix and $\frac{1}{T} \hat{\*F}^{\top} \hat{\*F}$ is a diagonal matrix, 
		the solution to the optimization problem \eqref{eqn-quad-loss} is the PCA estimate\footnote{We refer to $\frac{1}{NT} \*X\*X^\T$ as a sample covariance matrix and as in our applications we first demean the data. However our theoretical and empirical results are only based on the spectral decomposition of a second moment matrix and hold for general nonzero means. \cite{lettau2020theory} provide an answer when to demean the data.}
		\begin{eqnarray}
		&& \frac{1}{NT} \*X \*X^\T  = \hat{\*\Lambda} \hat{D}\hat{\*\Lambda}^{\top} ,  \label{eqn-estimated-loading} \\
		&& \hat{\*F} = \frac{1}{N} \*X^\T \hat{\*\Lambda}.\label{eqn-estimated-factor}
		\end{eqnarray}
		The estimated loadings $\hat{\*\Lambda}$ are the eigenvectors of the $K$-largest eigenvalues of $\frac{1}{NT} \*X\*X^\T$ multiplied by $\sqrt{N}$. The estimated factors $\hat{\*F}$ are the coefficients from regressing $\*X$ on $\hat{\*\Lambda}$. $\hat{D} = diag(\hat{D}_1, \hat{D}_2, \cdots, \hat{D}_K)$ denotes a diagonal matrix with the $K$ largest eigenvalues in descending order. Note, that $\hat{\*\Lambda}$ serves two purposes here: First, the estimated loadings are the cross-sectional weights to construct the estimated factors as shown in equation \eqref{eqn-estimated-factor}. Second, $\hat{\*\Lambda}$ is also the estimated exposure to the factors.
		\cite{bai2003inferential} shows that in an approximate factor model, $\hat{\*\Lambda}$ and $\hat{\*F}$ are consistent estimators of $\*\Lambda$ and $\*F$ up to an invertible transformation.

		The general approximate factor model framework assumes that the loadings $\*\Lambda$ and thus their consistent estimator $\hat{\*\Lambda}$ are not sparse, which is a necessary assumption for a consistent estimator of the factors in the presence of noise $\*e$. Here sparse means that only a finite number of elements in the loadings are nonzero that does not grow with $N$ and $T$. 
		As the estimated factors are linear combinations of the cross-section units $\*X$ weighted by a non-sparse $\hat{\*\Lambda}$, they are composed of almost all cross-section units, which are hard to interpret.

		We propose a method to estimate proximate factors that are sparse and hence more interpretable. The method is based on the following steps: 
		
		\begin{enumerate}
			\item \textit{Sparse factor weights:}
			$\hat{\*\Lambda}$ are the standard PCA estimates of the loadings. The proximate factor weights $\wt{\*W}$ are the largest elements in $\hat{\*\Lambda}$ obtained as follows:
			We shrink the $N-m$ \textit{smallest} entries in \textit{absolute values} in each estimated loading vector $\hat{\*\Lambda}_k$ to 0 and only keep the largest $m$ elements to get the sparse weight vector $\wt{\*W}_k$ for each $k$, where $m$ is finite. We normalize each sparse weight vector to have length one, that is, $\lVert\wt{\*W}_k\lVert_2 = 1$. \footnote{Formally, denote $\*M = [\*M_1, \*M_2, \cdots, \*M_K] \in \mathbb{R}^{N \times K}$ as a mask matrix indicating which factor weights are set to zero. $\*M_j$ has $m$ 1s and $N-m$ 0s. The sparse factor weights $\wt{\*W}$ can be written as 
				\begin{eqnarray} \label{eqn-tilde-lambda-def}
				\wt{\*W} =
				\begin{bmatrix}
				\frac{ \hat{\*\Lambda}_1 \odot \*M_1}{\norm{\hat{\*\Lambda}_1 \odot \*M_1}} & \frac{\hat{\*\Lambda}_2 \odot \*M_2}{\norm{\hat{\*\Lambda}_2 \odot \*M_2}}  & \cdots &
				\frac{\hat{\*\Lambda}_K \odot \*M_K}{\norm{\hat{\*\Lambda}_K \odot \*M_K}}
				\end{bmatrix}, 
				\end{eqnarray}
				where $\hat{\*\Lambda}_j$ is $j$-th estimated loading. The vector $\*M_j$ has the element 1 at the position of the $m$ largest loadings of $\hat{\*\Lambda}_j$ and zero otherwise. $\odot$ denotes the Hadamard product for element by element multiplication of matrices.} 
			\item \textit{Proximate factors:} We regress $\*X$ on $\wt{\*W}$ to obtain the proximate factors $\tilde{\*F}$: 
			\begin{eqnarray}\label{eqn-tilde-f-def}
			\tilde{\*F} = \*X^\T \wt{\*W} (\wt{\*W}^\T  \wt{\*W})^{-1}.
			\end{eqnarray}
			\item \textit{Loadings of proximate factors:} We regress $\*X$ on $ \tilde{\*F}$ to obtain the loadings $\tilde{\*\Lambda}$ of the proximate factors:
			\begin{eqnarray}\label{eqn-tilde-lam-def}
			\tilde{\*\Lambda} = \*X^\T  \tilde{\*F} ( \tilde{\*F}^\T   \tilde{\*F})^{-1}.
			\end{eqnarray}
		\end{enumerate}
		
		
		Proximate factors $\tilde{\*F}$ approximate latent factors $\*F$ well as measured by the generalized correlation. The generalized correlation equals the correlation between appropriately rotated proximate and population factors as defined in the subsequent sections. It is natural to measure the distance between two factors by their correlation. If two factors are perfectly correlated they explain the same variation in the data and provide the same results in linear regressions.
		
		\subsection{Illustrative Example}\label{sec:toy}
		
		We illustrate the intuition in a one-factor model. In this case, the generalized correlation is equal to the squared correlation between $\tilde{\*F}$ and $\*F$. For simplicity, assume that the factors and idiosyncratic component are i.i.d. over time, that is $F_t \stackrel{\iid}{\sim} (0, \sigma_{\*F_1}^2)$ and $e_{it} \stackrel{\iid}{\sim} (0, \sigma_e^2) $.
		Furthermore, our proximate factor consists of only one cross-sectional observation, i.e. $m=1$. Without loss of generality, the nonzero entry in $\wt{\*W}_1 \in \mathbb{R}^{N \times 1}$ is $\tilde w_{1,1}$ which is normalized to $\tilde w_{1,1}=1$. The squared correlation between $\tilde{\*F}$ and $\*F$ equals 
		\begin{eqnarray*}
			\widehat{\textnormal{corr}}(\tilde{\*F}, \*F)^2 &=& \frac{(\*F^\T  \tilde{\*F}/T)^2}{(\*F^\T  \*F/T) (\tilde{\*F}^\T \tilde{\*F}/T)} =\left( \frac{\*F^\T (\*F \Lambda_{1,1} + \*e_1)/T}{(\*F^\T \*F/T)^{1/2} ((\*F \Lambda_{1,1} + \*e_1)^\T (\*F \Lambda_{1,1} + \*e_1)/T)^{1/2}} \right)^2   \\
			&\pl& \frac{1  }{ 1   + \frac{1}{\left(\sigma_{\*F_1}/\sigma_e \right)^2 \Lambda_{1,1}^2}},
		\end{eqnarray*}
		where $\sigma_{\*F_1}/\sigma_e$ denotes the signal-to-noise ratio.\footnote{The second equation follows from $\tilde{\*F} = \*X^\T \wt{\*W} (\wt{\*W}^\T  \wt{\*W})^{-1} =  \*X^\T \wt{\*W} = (\*F_1 \Lambda_{1,1} + \*e_1) \tilde w_{1,1} =  (\*F_1 \Lambda_{1,1} + \*e_1)$.} The correlation increases with the size of $|\Lambda_{1,1}|$ and the signal-to-noise ratio $\sigma_{\*F_1}/\sigma_e$. If the largest population loading is sufficiently large, the correlation will be close to one. In the rest of the paper, we formalize this idea under a general setup. In the first part we discuss the proximate factors based on the largest loadings. Then, we generalize the idea to select the weights with the largest information content relative to the noise.
		
		\subsection{Assumptions}
		We impose several assumptions that are close to, but slightly stronger than, those in \cite{bai2002determining}. In order to show that $\tilde{\*F}$ is ``close'' to $\*F$ as measured by the generalized correlation, $\hat{\*\Lambda}$ needs to be a uniformly consistent estimator for $\*\Lambda$ up to some invertible transformation. 
		The following assumptions are necessary for all theorems in this paper and are standard assumptions for the general approximate factor model.\footnote{Assumption \ref{ass_factor} about the population factors is the same as \cite{bai2002determining}. Assumption \ref{ass_loading} allows loadings to be random. Since the loadings are independent of factors and errors, all results in \cite{bai2002determining} hold. Assumption \ref{ass_error}.1 imposes moment conditions for errors, which is the same as Assumption C.1 in \cite{bai2002determining}. This assumption implies that $\sigma_e^2$ is bounded. Assumption \ref{ass_error}.2 is close to Assumption 3.2 (i) and (ii) in \cite{fan2013large}. This assumption restricts the cross-section dependence of errors and is standard in the literature on approximate factor models. Since $\Sigma_e$ is symmetric, $\norm{\Sigma_e}_{1} \leq M$  is equivalent to $\norm{\Sigma_e}_{\infty} \leq M$. $\norm{\Sigma_e}_1 \leq M$ implies that $\sum_{j=1}^N |\tau_{ij,tt}| \leq M$. Together with $\tau_{ij,tt}$ being the same for all $t$ from the stationarity of $e_t$, this assumption implies Assumption C.3 in \cite{bai2002determining}. Assumption \ref{ass_error}.3 allows for weak time-series dependence of errors, which is slighter stronger than  Assumption C.2 in \cite{bai2002determining}. Assumption \ref{ass_error}.6 is the time average counterpart of Assumption C.5 in \cite{bai2002determining}. Assumption \ref{ass_f_e} implies Assumption D in \cite{bai2002determining}. The fourth moment conditions in Assumptions \ref{ass_error}.6 and \ref{ass_f_e}, together with Boole's inequality or the union bound, are used to show the uniform convergence of loadings, without assuming the boundedness of loadings.  Since Assumptions \ref{ass_factor}-\ref{ass_f_e} imply Assumption A-D in \cite{bai2002determining}, all results in \cite{bai2002determining} hold.}. 
		We assume that there exists a positive constant $M < \infty$ that can be used in all the assumptions. For a matrix $A = [a_{it}] \in \+R^{N \times T}$ we define the Frobenius norm as  $\norm{A}_F = [\tr(A^\T A)]^{1/2}$, the $L_1$-norm as $\norm{A}_1= \max_{1 \leq i \leq N} \sum_{t = 1}^{T} |a_{it}|$ and spectral norm as $\norm{A}_2 = \sigma_{\max}(A)$ where $\sigma_{\max}(A)$ is the largest singular value of $A$. 
		

		\begin{assumption}[{Factors}] \label{ass_factor}
			$\+E \norm{F_t}_F^4 \leq M < \infty$ and $\frac{1}{T} \sum_{t=1}^T F_t F_t^\T \xrightarrow{P} \Sigma_F$ for some $K \times K$ positive definite matrix $\Sigma_F$.
		\end{assumption}
		
		\begin{assumption}[{Factor Loadings}] \label{ass_loading}
			$\+E  \norm{\Lambda_i}_F^4 \leq M < \infty$ and $\frac{1}{N} \sum_{i=1}^N \Lambda_i \Lambda_i^\T \xrightarrow{P} \Sigma_{\Lambda}$ for some $K \times K$ positive definite matrix $\Sigma_{\Lambda}$. Loadings are independent of factors and errors.
		\end{assumption}
		
		\begin{assumption}[{Time and Cross-Section Dependence and Heteroskedasticity}] \label{ass_error}
			Denote \\ $\+E[e_{it}e_{js}]=\tau_{ij,ts}$, $\sigma_e^2 = \max_{i,j,t,s} |\tau_{ij,ts}|$. Then for all $N$ and $T$,
			\begin{enumerate}
				\item $\+E[e_{it}]=0$, $\+E[|e_{it}|^8] \leq M$;
				\item $e_t$ is stationary, $\Sigma_e = [\sigma^2_{e,ij}]_{N \times N}$ is the covariance matrix of $e_t$ with $\norm{\Sigma_e}_1 \leq M$;
				\item $\forall i$, $|\tau_{ii,ts}| \leq |\tau_{ts}|$ for some $\tau_{ts}$, $\frac{1}{T}\sum_{s=1}^T \sum_{t=1}^T |\tau_{ts}| \leq M $;
				\item $\+E[T^{-1/2}(\*e_i^\T \*e_j - \+E [\*e_i^\T \*e_j] )]^4 < M$.
				
			\end{enumerate}
		\end{assumption}
		\begin{assumption}[{Moments}] \label{ass_f_e}
			For all $i \leq N, t \leq T$, $\+E  \norm{T^{-1/2}\sum_{t=1}^T F_t e_{it}}_F^4 < M$.
		\end{assumption}
		
		We use  $h(m) = \max_{j_1, j_2, \cdots, j_m \in \{1, \cdots, N\}, t \in \{1, \cdots, T\}} \left( \sum_{k \neq l}  |\tau_{ j_k j_l,tt}|^2/\sigma_e^4 \right)^{1/2}$ to denote the maximum of the square root of total pairwise squared correlations among $m$ cross-section units.  Since each proximate factor is a linear combination of $m$ cross-section units, the estimation error of the proximate factor is determined by the errors of these $m$ cross-section units.  Note that if $e_{it}$ is i.i.d, then $h(m)=0$. If the errors of the $m$ cross-section units are perfectly dependent, $h(m)=\sqrt{m(m-1)}$. In general it holds that $0 \leq h(m) \leq \sqrt{m(m-1)}$.

		\section{Theoretical Results}\label{sec:theory}
		
		\subsection{Consistency}

		\cite{bai2002determining} and \cite{bai2003inferential} show that factors and loadings can be estimated consistently with PCA under Assumptions \ref{ass_factor}-\ref{ass_f_e} when $N, T \rightarrow \infty$. 
		A key element is that the loadings are pervasive which implies that $\hat {F_t} = \frac{1}{N} \*X_t^{\top} \hat{\*\Lambda} =  F_t \cdot  \frac{\*\Lambda^{\top} \hat{\*\Lambda}}{N} + \frac{(\*e_t^{\top})^{\top} \hat{\*\Lambda}}{N} = F_t H^{-1} + o_p(1)$, because the errors can be cross-sectionally diversified away. 
		Here $H^{-1}$ is the conventional invertible rotation matrix. 
		
		However, consistency does not hold in general for proximate factors $\tilde{\*F}$ for the following reason: 
		\[\tilde{\*F} =  \*X^\T \wt{\*W} (\wt{\*W}^\T  \wt{\*W})^{-1}= \*F \*\Lambda^{\top}  \wt{\*W} (\wt{\*W}^\T  \wt{\*W})^{-1} + \*e^\T\wt{\*W} (\wt{\*W}^\T  \wt{\*W})^{-1}\]
		and $\*e_t^\T\wt{\*W}_k = \sum_{i=1}^N e_{it} \wt w_{i,k} = \sum_{i=1}^m e_{\fidx_k(i) ,t}  \tilde{w}_{\fidx_k(i),k}$ where $\fidx_k(i)$ is the index of the $i$-th largest element in $\wt{\*W}_k$
		, that is, the sums are only taken over the finite number of nonzero weight entries.
		Even in the special case when the idiosyncratic components are i.i.d with variance $\sigma_e^2$, the variance of the $(t,k)$-th entry in $\*e^\T\wt{\*W}$ is $\mathrm{Var} \Big(\ \sum_{i=1}^m e_{\fidx_k(i),t}  \tilde{w}_{\fidx_k(i),k} \Big) = \sigma_{e}^2$, which does not vanish as $N, T \rightarrow \infty$. In summary, PCA factors average the idiosyncratic component over infinitely many non-sparse loadings leading to a law of large numbers, while for proximate factors this average is taken over only a finite number of elements without diversifying away the noise. 
		Although proximate factors are in general inconsistent, they can still be very close to the population factors. We will show under which conditions the correlations between the population and proximate factors are very close to one.

		\subsection{Loadings of Proximate Factors}

		The estimated loadings of the proximate factors asymptotically span the same space as the population loadings and hence will lead to the same regression or projection results. Hence, up to an invertible transformation, the loadings of the proximate factors are consistent.\footnote{Note, that our notion of consistency does not imply point-wise consistency, i.e. for a finite number of cross-section units the loadings can be distinct.} This result actually does not depend on the degree of sparsity $m$. The key element is that the loadings are obtained in a second stage regression that diversifies away idiosyncratic noise in the time dimension. Hence, it is important to distinguish between loadings and factor weights as sparse loadings are in general not consistent estimator of the population loadings.

		One of the major problems when comparing two different sets of loadings is that a factor model is only identified up to invertible linear transformations. Two sets of loadings represent the same model if the loadings span the same vector space. As proposed by Bai and Ng (2006) the generalized correlation is a natural candidate measure to describe how close two vector spaces are to each other. Intuitively, we calculate the correlation between the loadings of the proximate and the population factors after rotating them appropriately. The generalized correlation measures of how many loading vectors two sets have in common.
		The generalized correlation between the loadings of proximate and population factors is defined as
		\[\rho_{\tilde{\*\Lambda}, \*\Lambda} \coloneqq \tr \Lp ( \*\Lambda^\T \*\Lambda/N )^{-1} ( \*\Lambda^\T \tilde{\*\Lambda}/N ) (\tilde{\*\Lambda}^\T \tilde{\*\Lambda}/N )^{-1} ( \tilde{\*\Lambda}^\T \*\Lambda/N ) \Rp.\]
		Here, the generalized correlation $\rho_{\tilde{\*\Lambda}, \*\Lambda}$, ranging from 0 to $K$ (number of factors), measures how close $\tilde{\*\Lambda}$ and $\*\Lambda$ are. If $\tilde{\*\Lambda}$ lies in the space spanned by $\*\Lambda$, then $\rho=K$. Otherwise, if the space spanned by $\tilde{\*\Lambda}$ is orthogonal to the space spanned by $\*\Lambda$, then $\rho = 0$.\footnote{The generalized correlation is also known as canonical correlation. Our measure is based on the Euclidian inner product without demeaning the loadings which is appropriate for measuring the span of the loading matrix. Our generalized correlation measure is the sum of the squared individual generalized correlations. The first individual generalized correlation is the highest correlation that can be achieved through a linear combination of the proximate loadings $\tilde{\*\Lambda}$ and the population loadings $\*\Lambda$. For the second generalized correlation, we first project out the subspace that spans the linear combination for the first generalized correlation and then determine the highest possible correlation that can be achieved through linear combinations of the remaining $K-1$ dimensional subspaces. This procedure continues until we have calculated the $K$ individual generalized correlations. Mathematically the individual generalized correlations are the square root of the eigenvalues of the matrix $(\*\Lambda^\T\*\Lambda/N)^{-1} (\*\Lambda^\T \tilde{\*\Lambda}/N) (\tilde{\*\Lambda}^\T \tilde{\*\Lambda}/N)^{-1} (\tilde{\*\Lambda}^\T \*\Lambda/N)$.}

		\cite{bai2002determining} and \cite{bai2003inferential} show that the PCA loadings $\hat {\*\Lambda}$ are consistent estimators of the rotated loadings $ \*\Lambda H$. The rotation matrix $H$ is defined as $H= D^{-1/2} V_H^{\top}  \Sigma_F^{1/2}$ based on the spectral decomposition  $\Sigma_F^{1/2} \Sigma_{\Lambda} \Sigma_F^{1/2} = V_H D V_H^{\top}$, that is,  $D = \diag(D_1, D_2, \cdots, D_K)$ are the eigenvalues of $\Sigma_F^{1/2} \Sigma_{\Lambda} \Sigma_F^{1/2}$ in decreasing order and $V_H$ are the corresponding eigenvectors. If the population factor model has the same rotation as the PCA estimates, that is, $\Sigma_F$ is diagonal and $\Sigma_{\Lambda}=I_K$ then $H$ simplifies to an identity matrix. In this case, the generalized correlation $\rho_{\tilde{\*\Lambda}, \*\Lambda}$ simplifies to the sum of squared correlations between the vectors of estimated and population loadings.  
		
		The rotation is relevant as our factor weights are only close to the largest rotated population loadings $\*\Lambda H$. Our results depend on Proposition \ref{thm:uniform-consistency-loading} in the Appendix that shows the uniform consistency of the estimated loadings. Hence, weights based on the largest estimated loadings $\hat {\* \Lambda}$ are asymptotically very close to weights derived from the largest rotated population loadings $\* \Lambda H$. We denote by $\*W$ the matrix of the $m$ largest factor weights based on the rotated population loadings, i.e. it follows the same definition as $\wt{\*W}$ but applied to the population loadings $\*\Lambda H$.

		We need to impose one additional weak assumption on the residuals:
		\begin{assumption}\label{ass:max-eigen-err-cov}
			The largest eigenvalue of $\frac{1}{NT} \*e^{\top} \*e$ is $o_p(1)$.
		\end{assumption}
		Assumption \ref{ass:max-eigen-err-cov} is very weak and essentially satisfied in any sensible approximate factor model. It is standard and has been imposed in many related papers, e.g., \cite{fan2013large}. It is slightly stronger than Assumption \ref{ass_error}.3, that implies that the largest eigenvalue of the population residual auto-covariance matrix is bounded. Under suitable additional assumptions on the tail behavior of the residuals, Assumption \ref{ass_error}.3 implies that $\|\frac{1}{N} \*e^{\top} \*e\|_2$ is $O_p(1)$ which would be sufficient. 
		\begin{theorem}\label{thm:gen-lam}
			Suppose Assumptions \ref{ass_factor}-\ref{ass:max-eigen-err-cov} hold. If $\|\left(\*W^{\top} \*\Lambda \right)^{-1}\|_2 =O_p(1)$, then we have 
			\begin{align*}
			\rho_{\tilde{\*\Lambda}, \*\Lambda} \xrightarrow{P} K.  
			\end{align*}
		\end{theorem}
		The assumption that $\|\left(\*W^{\top} \*\Lambda \right)^{-1}\|_2 =O_p(1)$ is essentially a full rank assumption on $\*W^{\top} \*\Lambda$. It requires that the sparse set of cross-section units, which we use to construct the proximate factors, is affected by all factors in a non-redundant way. In the case of only one factor, i.e. $K=1$, it is trivially satisfied. In the case of multiple factors, we need to rule out that the largest elements of two loading vectors are identical. This is an assumption on the tail dependency of the loadings. If for example the loading vectors are independent and $m \geq K$, then this condition is satisfied.
		
		Our notion of consistency does not imply point-wise consistency of the loadings, i.e., for a finite number of loading elements $\tilde{\*\Lambda}$ can be different from $\*\Lambda H$. Our notion of consistency measures the asymptotic difference between vectors whose length goes to infinity and hence a finite number of elements have a negligible effect. The generalized correlation measure is the appropriate measure if we intend to use loadings for projections or in a cross-sectional regression. The strong result in Theorem \ref{thm:gen-lam} states that cross-sectional regressions with loadings of proximate factors yield the same results as using the population loadings up to an invertible matrix $H$. Note that this theorem has broader implications that go beyond proximate factors. For example, it justifies why an iterative procedure to estimate latent factors leads to a consistent estimator after a few iterations.\footnote{Instead of applying PCA, latent factors can also be estimated by an iterative procedure where for a set of candidate factors a first stage set of loadings is estimated with a time-series regression, which is then used in a second stage to obtain factors in a cross-sectional regression. This procedure is iterated until convergence. For example, \cite{bai2017} use a variation of this approach. Our result shows that in an approximate factor model one step is sufficient to obtain a consistent estimator of the loadings.}
		Next, we will show that the proximate factors themselves are also very close to the population factors.

	\subsection{One-Factor Case}\label{subsec:proxy-one-factor}
		We start with the one-factor model and characterize the correlation between population and proximate factors with a lower bound based on order statistics. Under additional assumptions, we can give explicit expressions for these order statistics with counting statistics and extreme value theory. In both cases, we derive analytical solutions for the lower bound. We use the results for comparative statics and preparing for the more general case.   
		
		The properties of the proximate factors depend on the probabilistic properties of the largest loadings that are used as weights. We denote by $\Lambda_{(m),1}$ the $m$-th largest element in $|\*\Lambda_1|$. Theorem \ref{thm-evt-one-factor} provides a closed-form lower bound on the squared correlation between proximate factors and population factors, $\rho \coloneqq \frac{(\*F^{\top} \tilde{\*F} )^2}{\*F^{\top} \*F \tilde{\*F}^{\top} \tilde{\*F}}$, in terms of the probability distribution of the $m$-th largest loading $P(|\Lambda_{(m),1}| \geq y)$, the signal strength $\frac{1}{T} \*F^{\top} \*F \pl \sigma_{\*F_1}^2$ and the noise level $\sigma_e^2$.

		\begin{theorem}\label{thm-evt-one-factor}
            Suppose Assumptions \ref{ass_factor}-\ref{ass_f_e} hold and there exists a continuous function $\bar{G}_{1,m}(\cdot)$ such that  the complementary cumulative distribution function of the $m$-th order statistic $|\Lambda_{(m),1}|$ converges to $\bar{G}_{1,m}(\cdot)$, that is, $\lim_{N \rightarrow \infty} P(|\Lambda_{(m),1}| \geq y) = \bar{G}_{1,m}(y)$. 
			For any given finite $m$ and any threshold $0 < \rho_0 < 1$, it holds that
			\begin{eqnarray}\label{eqn-evt-one-factor}
			\lim_{N, T\rightarrow \infty} P\left( \rho > \rho_0 \right) \geq \bar{G}_{1,m}(y_m) \qquad \text{with   $y_m = \left(\frac{1+h(m)}{m} \frac{\sigma_{e}^2}{\sigma_{\*F_1}^2} \frac{\rho_0}{1 - \rho_0} \right)^{1/2}$.}
			\end{eqnarray}
		\end{theorem}

		Theorem \ref{thm-evt-one-factor} guarantees that the squared correlation $\rho$ is above the target threshold $\rho_0$ with at least the probability $\bar{G}_{1,m}(y)$. In practice we might set the probability bound to 95\% and find the lower bound for the correlation $\sqrt{\rho_0}$. This theorem allows us to derive comparative statics about when proximate factors work well. Denote the right-hand side of Inequality \eqref{eqn-evt-one-factor} as $\underline{p}\coloneqq \bar{G}_{1,m}(y_m)$. Given a limiting order statistic $\bar{G}_{1,m}(.)$ for the loadings, the exceedance probability $\underline{p}$ depends only on $y_m$:
		\begin{enumerate}[noitemsep]
			\item The larger $\rho_0$, the larger $y_m$ and the smaller the exceedance probability $\underline{p}$; 
			\item The larger the signal-to-noise ratio $\sigma_{\*F_1}/\sigma_e$, the smaller $y_m$ and the larger $\underline{p}$; 
			\item The larger the cross-section dependence of errors $h(m)$, the smaller $\underline{p}$; 
			\item The number of nonzero elements $m$ affects $\underline{p}$ in two ways: First, in most cases, $\frac{1+h(m)}{m}$ decreases with $m$, which raises $\underline{p}$. Second, larger $m$ results in more subtraction terms in $\underline{p}$ with the opposite effect leading to a trade-off. 
		\end{enumerate}

		If the population loadings $\Lambda_{i,1}$ are i.i.d., then we can provide an explicit expression for $\bar{G}_{1,m}(y_m)$ in \eqref{eqn-evt-one-factor}. Furthermore, if the loadings have unbounded support, the correlation between the proximate and population factor converges to one. 
		
		\begin{proposition}\label{prop-one-factor}
			Suppose Assumptions \ref{ass_factor}-\ref{ass_f_e} hold and the loadings $\Lambda_{i,1}$ are i.i.d.  with continuous cumulative density function $\CDF_{|\Lambda_{i,1}|}(y) = P(|\Lambda_{i,1}| \leq y)$, then the result in \eqref{eqn-evt-one-factor} holds with $\bar{G}_{1,m}(y_m)$ equal to
			\begin{eqnarray} \label{thm1_eqn1}
			\bar{G}_{1,m}(y_m) = 1 - \lim_{N \rightarrow \infty} \sum_{j = 0}^{m - 1} {{N}\choose{j}} (1 - \CDF_{|\Lambda_{i,1}|}(y_{m}))^j  \CDF_{|\Lambda_{i,1}|}(y_{m})^{N-j}.
			\end{eqnarray}
			If $\CDF_{|\Lambda_{i,1}|}(y_{m}) < 1$, then $\lim_{N, T\rightarrow \infty} P(\rho > \rho_0) \rightarrow 1$.
		\end{proposition}
		
		Proposition \ref{prop-one-factor} states that if the loadings have unbounded support, then the correlation between the proximate and population factor converges to one. At a first glance, it seems to be at odds with our previous observations that the idiosyncratic component in a proximate factor cannot be diversified away. The intuition behind this consistency result is based on the growing signal-to-noise ratio. If loadings are sampled with an unbounded support independently of the idiosyncratic component, then for growing $N$ the largest loadings are unbounded and their signal-to-noise ratio explodes. Hence with high probability, the largest absolute loadings do not coincide with large idiosyncratic movements and the variation of these cross-section units is essentially only explained by the factor. Hence, selecting the cross-section unit with the largest loading is close to picking the factor itself. Proposition \ref{prop-one-factor} shows that a more dispersed distribution of $|\Lambda_{i,1}|$ leads to a larger $\underline{p}$, where the unbounded support is just an extreme case of this observation.
		
		After proper rescaling, loadings with unbounded support can also be interpreted as approximately sparse population loadings. If some loading elements diverge and we normalize $| \*\Lambda_{1}|=1$, a large number of elements in $\* \Lambda$ need to converge to 0. Hence, a sparse estimator is consistent if the true population model is sparse itself. The important contribution of this paper is that we can also characterize the asymptotic properties of the sparse estimator if the population model is not sparse. 

		
		A more general perspective for providing a lower bound for $\rho$ employs EVT, which relaxes the i.i.d. assumption of $\Lambda_{i,1}$ and which we will also pursue for the multi-factor case. The distribution of the largest loadings can be modeled by EVT under general conditions. 
		We denote the loading with the largest absolute value by $|\Lambda_{(1),1}|$. Under the assumptions of EVT, there exist sequences of constants $\{a_{1,N} > 0\}$, $\{b_{1,N}\}$ such that $P((|\Lambda_{(1),1}|-b_{1,N})/a_{1,N} \leq z) \rightarrow G^\ast_1(z) $ with the generalized extreme value (GEV) distribution $G^\ast_1(z)=\exp\left ( -\left(1+ \xi \left(\frac{z-\mu}{\sigma} \right) \right)^{-1/\xi} \right )$. The three parameters location $\mu$, scale $\sigma$ and shape $\xi$ completely characterize the GEV. The results for the $m$-th largest loading in the case of dependent loadings is more complex and summarized in Lemma \ref{lemma:order-stats-dist} in the Appendix.\footnote{Lemma \ref{lemma:order-stats-dist} is adapted from Theorem 3.3 in \cite{hsing1988extreme}. $|\Lambda_{i,1}|$ is indexed by the cross-section units. We assume that $|\Lambda_{i,1}|$ are exchangeable and can be properly reshuffled to satisfy the strong mixing condition.} Lemma \ref{lemma:order-stats-dist} provides the limiting distribution for the extreme order statistic of a strictly stationary sequence $|\Lambda_{i,1}|$ satisfying a strong mixing condition. 
		In more details, it provides the necessary and sufficient condition such that there exists a sequence $u_{1,N}(\tau)$ and function $G_{1,m}(\tau)$ with $\lim_{N \rightarrow \infty} P(|\Lambda_{(m),1}| \geq  u_{1,N}(\tau)) = G_{1,m}(\tau)$. 

		\begin{proposition}\label{prop:GEV}
			Suppose Assumptions \ref{ass_factor}-\ref{ass_f_e} hold and the population loadings $\Lambda_{i,1}$ and sequence $u_{1,N}(\tau)$ satisfy the assumptions in Lemma \ref{lemma:order-stats-dist}. Then, the $m$-th largest loading satisfies $\lim_{N \rightarrow \infty} P(|\Lambda_{(m),1}| \geq  u_{1,N}(\tau)) = \bar{G}_{1,m}(\tau)$ and the result in \eqref{eqn-evt-one-factor} holds with $\bar{G}_{1,m}(y_m)$ equal to
			\begin{align}
			\bar{G}_{1,m}(y_m) = 1 - e^{-\tau} \left[1 + \sum_{j = 1}^{m-1} c_{1,j} \cdot \frac{\tau^j}{j!}  \right], \label{eqn:GEV}
			\end{align}
			for $u_{1,N}(\tau) = y_m$. Both $u_{1,N}(\cdot) $ and the constants $c_{1,j}$ are determined by $\Lambda_{i,1}$'s distribution and its dependence structure. 
			
			The result simplifies for i.i.d. loadings. Suppose Assumptions \ref{ass_factor}-\ref{ass_f_e} hold, $\Lambda_{i,1}$ is i.i.d. and there exist sequences $\{a_{1,N} > 0\}$, $\{b_{1,N}\}$ such that $P((|\Lambda_{(1),1}|-b_{1,N})/a_{1,N} \leq z)$ converges to a non-degenerative distribution function. Then, $u_{1,N}(\tau) = a_{1,N} z_\tau + b_{1,N}$, $c_{1,j} = 1$ in Equation \eqref{eqn:GEV} and Theorem \ref{thm-evt-one-factor} holds with
			\begin{align}
			\bar{G}_{1,m}(y_m) = 1 - e^{-\tau} \cdot  \sum_{j = 0}^{m-1} \frac{\tau^j}{j!} \qquad \text{with  } \tau= \left( 1+ \xi \left( \frac{(y_m-b_{1,N} )/a_{1,N} -\mu}{\sigma}\right) \right)^{-1/\xi}. \label{eqn:evt-one-factor-iid}
			\end{align}	
			The extreme value parameters $\mu, \sigma, \xi, a_{1,N}$ and $b_{1,N}$ are determined by the distribution of $|\Lambda_{(1),1}|$.
		\end{proposition}

		The sequences $a_{1,N}$ and $b_{1,N}$ determine to which of the three extreme values distributions, Gumbel, Frechet and Weibull, the tail distribution of $|\Lambda_{i,1}|$ belongs. Table \ref{tab:extreme-value-example} lists examples for common distributions and the corresponding extreme value parameters. If the distribution of the loadings is known, its parameters can be estimated with maximum likelihood and for many distributions, the bounds in Proposition \ref{prop:GEV} can be directly inferred. More generally, the parameters of the GEV distribution can be estimated with block maxima methods as described in \cite{ancona2000comparison}, \cite{leadbetter1982extremes} and \cite{hsing1988extreme}.

		\begin{table}[h!]
			\begin{center}
				\tcaptab{Examples of Extreme Value Distributions}
				{\small
					\begin{tabular}{cccccccc}
						\toprule
						Family & Distribution & $G^\ast_1(z)$ & $ a_{1,N}$  & $b_{1,N}$ & $\mu$ & $\sigma$ & $\xi$  \\
						\midrule
						Gumbel & $\Lambda_i \sim \mathrm{Exp}(1)$ & $\exp(-e^{-z})$ & 1 & $N$ & 0 & 1 & 0 \\
						& $ \Lambda_i \sim \mathcal{N}(0,\sigma_\Lambda^2)$ & $\exp(-e^{-z})$ & $\frac{\sigma_\Lambda}{2N \phi(b_{1,N})}$ & $\sigma_\Lambda \Phi^{-1}(1-\frac{1}{2N})$ & 0 & 1 & 0  \\
						& $ \Lambda_i \sim \mathcal{N}(\mu_\Lambda,\sigma_\Lambda^2)$ & $\exp(-e^{-z})$ &  $\frac{1}{N\psi(b_{1,N}| \mu_\Lambda, \sigma_\Lambda)}$ & $\Psi^{-1}\Big(1 - \frac{1}{N}|\mu_\Lambda, \sigma_\Lambda\Big)$ & 0 & 1 & 0  \\
						Frechet & $\CDF_{\Lambda_i}(x) = \exp(-\frac{1}{x})$ & $\exp(-1/z)$ & $N$ & 0 & 1 &  1 & 1 \\
						Weibull & $\Lambda_i \sim \mathrm{U}(0,1)$ & $e^z$ & $1/N$ & 1 & 1 & 1 & $-1$ \\
						\bottomrule
				\end{tabular}}
				\bnotetab{This table shows examples of extreme value distributions. They characterize the tail distribution for i.i.d. loadings: $P((|\Lambda_{(1),1}|-b_{1,N})/a_{1,N} \leq z) \rightarrow G^\ast_1(z) $ with the GEV distribution $G^\ast_1(z)=\exp\left( -\left(1+ \xi \left((z-\mu)/\sigma \right) \right)^{-1/\xi} \right)$ where $\mu$ is the location parameter, $\sigma$ the scale parameter and $\xi$ the shape parameter.
					Hence, for $m=1$ and $y_m = a_{1,N} z_\tau + b_{1,N} $, we have $\bar{G}_{1,m}(y_m) = 1 - G^\ast_1(z_\tau)$ in \eqref{eqn:evt-one-factor-iid}. Here, $\psi(\cdot | \mu_\Lambda, \sigma_\Lambda)$ and $\Psi(\cdot| \mu_\Lambda, \sigma_\Lambda)$ denote the PDF and CDF of the absolute values of the loadings $|\Lambda_i|$ with distribution $ \Lambda_i \sim \mathcal{N}(\mu_\Lambda,\sigma_\Lambda^2)$. $\Phi$ and $\psi$ denote the CDF and PDF of a standard normal distribution.}
				\label{tab:extreme-value-example}
			\end{center}
		\end{table}

		\subsection{Multi-Factor Case}\label{subsec:proxy-multi-factor}

		The arguments of the one-factor model extend to a model with multiple factors. As our simulations and empirical results illustrate, the simple thresholding method provides proximate factors that explain very well the non-sparse PCA factors. However, formalizing the properties of the lower bound is more challenging. First, we have to work with generalized correlations instead of simple correlations to take into account the possible rotations of the latent factors. 
		Second, the sparse factor weight vectors are in general not orthogonal to each other in contrast to the PCA-loadings. Third, we need to study the joint distribution of extreme values for multiple loading vectors. In order to derive analytical bounds, we need to include additional correction terms. Under additional assumptions, we can sharpen the theoretical bounds and neglect the correction terms.
		
		One of the major problems when comparing two different sets of factors is that a factor model is only identified up to invertible linear transformations. We will again use the generalized correlation to measure how many factors two sets have in common. The generalized correlation between proximate factors and population factors is defined as
		$$\rho \coloneqq \tr \left( (\*F^\T  \*F/T)^{-1} (\*F^\T  \tilde{\*F}/T) (\tilde{\*F}^\T \tilde{\*F}/T)^{-1} (\tilde{\*F}^\T \*F/T) \right).$$ 
		Here, the generalized correlation $\rho$, ranging from 0 to the number of factors $K$, measures how close $\*F$ and $\tilde{\*F}$ are. If $\tilde{\*F}$ lies in the space spanned by $\*F$, then $\rho=K$.\footnote{The individual generalized correlations are the square root of the eigenvalues of the matrix $(\*F^\T \*F/T)^{-1} (\*F^\T  \tilde{\*F}/T) (\tilde{\*F}^\T \tilde{\*F}/T)^{-1} (\tilde{\*F}^\T \*F/T)$.}

		
		The multivariate bounds in Theorem \ref{thm-evt-multi-factor} have three adjustments relative to the one-factor model. The first adjustment arises from using the generalized correlation measure. Hence, we derive bounds on the sum of eigenvalues instead of a simple correlation. These bounds use eigenvalue inequalities which essentially change the value of $y_m$ in the one-factor case in Theorem \ref{thm-evt-one-factor}. The second complication arises because the sparse weights $ \wt{\* W}$ are in general not orthogonal to each other. We will subtract a correction term from the lower bound that provides a bound on this effect. If the proximate weights are ``non-overlapping'', that is, do not share nonzero entries for the same indices, the correction term vanishes. The third complication in the multi-factor case is the relationship between the sparse and non-sparse eigenvectors. As $ {\*W}^{\top} \*\Lambda H$ is in general not a diagonal matrix we need to take the off-diagonal elements into account which adds a second correction term to the lower bound. Under additional assumptions on the population loadings and non-overlapping sparse weights, we can also neglect this correction term. Below we illustrate the intuition behind these correction terms with an example.

		\begin{theorem}\label{thm-evt-multi-factor}
		Suppose Assumptions \ref{ass_factor}-\ref{ass_f_e} hold and there exists a continuous function $\bar{G}_{m}(\cdot)$ such that the joint $m$-th order statistic for the columns $\*V_k$ in $\*V = [v_{i,k}] := \*\Lambda H $, denoted as $(| v_{(m),1}|, \cdots, |v_{(m),K}| )$, follows $\lim_{N \rightarrow \infty} P(|v_{(m),1}| \geq y_1, \cdots, |v_{(m),K}| \geq y_K) = \bar{G}_{m}(y_1, \cdots, y_K)$. Then for any given finite $m$, $(y_{1}, \cdots, y_{K})$, $0< \underline{\gamma} \leq 1$ and $\bar \gamma_{\*V,e} \geq 1$ we have  
			\begin{align}
			\nonumber & \lim_{N, T\rightarrow \infty} P\left( \rho \geq K - \frac{\bar \gamma_{\*V,e} \sigma_{e}^2}{m\underline{\gamma}^2 } \sum_{k=1}^K \frac{1}{D_k y_{k}^2} \right)   \\
			\geq&  \bar{G}_{m}(y_{1}, \cdots, y_{K})  - \lim_{N \rightarrow \infty} P\Big(\sigma_{\min}(B) < \underline{\gamma}\Big) - \lim_{N \rightarrow\infty} P\Big(\norm{{\*W}^\T \Sigma_e {\*W} }_2 > \bar \gamma_{\*V,e} \sigma_{e}^2 \Big) \label{eqn-evt-multi-factor} 
			\end{align}
			where $D = \diag(D_1, D_2, \cdots, D_K)$ are the eigenvalues of $\Sigma_{F}^{1/2}\Sigma_{\Lambda}\Sigma_{F}^{1/2}$ in decreasing order and 
			\begin{align*}
			B =[b_{kl}] = \Big[ \frac{D_k^{1/2} \sum_{i = 1}^{m}v_{ \fidx_l(i),k} v_{\fidx_l(i),l}  }{ D_l^{1/2} \sum_{i = 1}^{m} v_{\fidx_l(i),l}^2  } \Big], \quad {\*W}     = \begin{bmatrix}
			\frac{ \*V_1  \odot \*M_1 }{\norm{\*V_1  \odot \*M_1 }} & \frac{\*V_2  \odot \*M_2 }{\norm{\*V_2  \odot \*M_2 }} & \cdots & \frac{ \*V_K  \odot \*M_K }{\norm{ \*V_K  \odot \*M_K }}
			\end{bmatrix}.
			\end{align*}
			$\sigma_{\min}(B)$ is the minimum singular value of $B\in \+R^{K \times K}$ and $\fidx_k(l)$ is the index of the $l$-th largest element in $\*V_k$. ${\*W}\in \+R^{N \times K}$ are the sparse weights based on the rotated population loadings $\*V$ and $\*M_k = [M_{k,i}] \in \{0,1\}^{N}$ where $M_{k,i} = 1$ if there exists $1 \leq j \leq m$ such that $i = \fidx_k(j)$ and 0 otherwise. 
		\end{theorem}
		
		In the proof of Theorem \ref{thm-evt-multi-factor} we show that the generalized correlation can be bounded from below by
		$\rho \geq K - \norm{{\*W}^\T \Sigma_e {\*W} }_2 \tr \Big( \Big( \frac{1}{T} {\*W}^\T  \*V D \*V^\T {\*W} \Big)^{-1}  \Big) + o_p(1)$. If the two matrices that appear in this product converge to diagonal matrices, then both correction terms disappear.   
		
		The additional parameter $\underline{\gamma}$ accounts for the impact of the off-diagonal terms to the trace of the matrix $\left( \frac{1}{T} {\*W}^\T  \*V D \*V^\T {\*W}\right)^{-1}$. In the special case when these off-diagonal terms converge to zero, it holds that $\sigma_{\min}(B) = 1$ and $P(\sigma_{\min}(B) < \underline{\gamma}) = 0$ for $\underline{\gamma}=1$.
		In order to illustrate this point we will consider the simple example of $m=1$ and $K=2$, i.e. a two factor model where the proximate factors only take the largest loading values. Without loss of generality, we assume the first element $v_{1,1}$ is the largest element for the first vector $\*V_1$ and the second element $v_{2,2}$ is the largest element for $\*V_2$. Then, we have
		\begin{align*}
		B=  \begin{pmatrix} D_1^{1/2} & 0 \\ 0 &  D_2^{1/2} \end{pmatrix} \*V^\T \*W \begin{pmatrix} \frac{1}{D_1^{1/2} v_{1,1}} & 0 \\ 0 & \frac{1}{D_2^{1/2} v_{2,2}}\end{pmatrix}
		= \begin{pmatrix} 1 & \frac{D_1^{1/2} v_{2,1}}{D_2^{1/2} v_{2,2}}  \\  \frac{D_2^{1/2} v_{1,2}}{D_1^{1/2} v_{1,1}}  & 1 \end{pmatrix}. 
		\end{align*}		
		The smallest singular value $\sigma_{\min}(B)$ of the matrix $B$ is a measure for how large the off-diagonal elements are. In the special case of a diagonal matrix $B$, the smallest singular value equals 1. Our lower bound will depend on an asymptotic probabilistic bound for $\sigma_{\min}(B)$, which depends on the distribution of the loading vectors. For example in the special case of independent vectors $\*V_{k}$, that have i.i.d normally distributed elements, a random element of the rotated loading vector $\*V_{k}$ is $O_p(1)$, while the largest element is unbounded in the limit. Hence, the ratio is $\frac{v_{1,2}}{v_{1,1}}=o_p(1)$. For this special case the multi-factor case is a direct extension of the one-factor model, i.e. we can apply the one-factor result to each factor in the multi-factor model individually. Another special case is when the population loadings have sparse structures themselves. More specifically, it is sufficient to have non-overlapping weight vectors $\*W_{k}$ and ``locally sparse'' $\*V$, that is, the vector $\*V_{k}$ has elements that are sufficiently small for the indices of nonzero elements of $\*W_{l \neq k}$. In both special cases we have $ \lim_{N \rightarrow \infty} \sigma_{\min}(B)=1$ while our theorem allows us to also take into account a more complex structure.  

		The additional parameter $\bar \gamma_{\*V,e}$ accounts for the impact of the off-diagonal terms in ${\*W}^\T \Sigma_e {\*W}$. In the special case of non-overlapping entries in $\wt{\*W}$, it holds for $\bar \gamma_{\*V,e}=1+h(m)$ that $\lim_{N \rightarrow\infty} P\Big(\norm{{\*W}^\T \Sigma_e {\*W} }_2 > \bar \gamma_{\*V,e} \sigma_{e}^2 \Big)=0$ and hence we can neglect this term. This is implied by independent loading vectors. In that case the joint $m$-th order statistic for the columns $\*V_k$ also simplifies to the product of the marginals as summarized by the next proposition. 
		
		\begin{proposition}\label{prop:evt-multi-factor-simplified}
			Suppose Assumptions \ref{ass_factor}-\ref{ass_f_e} hold,$\tilde{\*V}_k$ and $\tilde{\*V}_l$ are asymptotically independent for $k\neq l$, and there exists a continuous function $\bar{G}_{k,m}(\cdot)$ such that the $m$-th order statistic for each column $\*V_k$ follows $\lim_{N \rightarrow \infty} P(|v_{(m),k}| \geq y_k) = \bar{G}_{k,m}(y_k)$. Then Inequality \eqref{eqn-evt-multi-factor} simplifies to 
			\begin{align}
			\lim_{N, T\rightarrow \infty} P\left( \rho \geq K - \frac{(1+h(m)) \sigma_{e}^2}{m\underline{\gamma}^2 } \sum_{k=1}^K \frac{1}{D_k y_{k}^2 } \right)  
			\geq   \prod_{k=1}^K \bar{G}_{k,m}(y_{k})  - \lim_{N \rightarrow \infty} P\Big(\sigma_{\min}(B) < \underline{\gamma}\Big), \label{eqn:evt-multi-factor-simplified}
			\end{align}
			If each column $\*V_k$ satisfies the assumptions of Proposition \ref{prop:GEV}, then each $\bar{G}_{k,m}$ equals Equation \eqref{eqn:GEV} for dependent $v_{i,k}$ and Equation \eqref{eqn:evt-one-factor-iid} for i.i.d. $v_{i,k}$.
		\end{proposition}
		The same comparative statics holds as in the one-factor case. In addition, the bound is sharper if the rotated loadings for different factors have less dependence and share fewer nonzero entries.

		Our analysis uses the conventional PCA rotation of the loadings that implies $\hat{\*\Lambda}^{\top}\hat{\*\Lambda}/N=I_K$. All our results extend to any arbitrary rotation of the PCA loadings $\hat{\*\Lambda} \tilde H$ where $\tilde H$ is a full-rank matrix. We simply replace the estimated PCA loadings by $\hat{\*\Lambda} \tilde H$, the rotated population loadings by $\*\Lambda H \tilde H=\*V \tilde H$ and rotated population factor by $F ({H^{\top}})^{-1} ({\tilde H}^{\top})^{-1}$ and all theorems and propositions continue to hold. Theorem \ref{thm-evt-multi-factor} provides guidance on how to choose a rotation $\tilde H$ that provides proximate factors with a higher guaranteed correlation. If the estimated weights are non-overlapping we can avoid or at least reduce the correction terms which increases the correlation. However, this can come a cost of taking loading elements with lower absolute value and hence lower signal to noise ratio. Our theorem provides the exact trade-off between these choices.

		\subsection{Weighted Proximate Factors}\label{subsec:weighted-estimator}
		
		Weighted proximate factors are a generalization of our approach to select more informative sparse weights. It is motivated by the observations that the cross-sectional units with the largest loadings relative to the standard deviation of the idiosyncratic noise should provide a better approximation to the factors. We estimate the weighted proximate factors by applying a cross-sectional weight $N \times N$ to the observations before extracting the PCA factors and selecting the largest weights. This is different from selecting a $K \times K$ rotation of the PCA factors as discussed in the previous section. 
		
		We denote by $\Weight \in \mathbb{R}^{N \times N}$ a cross-sectional weight matrix. The weighted proximate factors are obtained from a \textit{weighted PCA}. First, we cross-sectionally re-weight the observations $\*X^{\twt}:= \Weight \*X$ and then apply PCA to $\frac{1}{NT} \*X^{\twt} (\*X^{\twt})^{\top}$ to obtain the weighted PCA loadings $\hat{\*\Lambda}^{\twt}$.\footnote{Note that the weighted PCA loadings $\hat{\*\Lambda}^{\twt}$ need to be multiplied by $\Weight^{-1}$ to be consistent estimates for the population loadings $\* \Lambda$ up to the usual rotation matrix.} Second, we apply the same procedure as in Section \ref{subsec:estimator} to obtain the sparse weighted factors weights $\wt{\*W}^{\twt}$ that have $m$ nonzero elements and length 1. Third, we regress $\*X^{\twt}$ on $\wt{\*W}^{\twt}$ to obtain the proximate factors $\tilde{\*F}^{\twt}$:
		\begin{align}\label{eqn-tilde-f-def-wt}
		\tilde{\*F}^{\twt} = {\*X^{\twt}}^\T \wt{\*W}^{\twt} \big( (\wt{\*W}^{\twt})^\T \wt{\*W}^{\twt}  \big)^{-1} .
		\end{align}
		Finally, the loadings of the proximate factors $\tilde{\*\Lambda}^{\twt}$ are the coefficients of the regression of $\*X$ on $\tilde{\*F}^{\twt}$:
		\begin{align}\label{eqn-tilde-lam-def-wt}
		\tilde{\*\Lambda}^{\twt} = \*X \tilde{\*F}^{\twt} \big(  (\tilde{\*F}^{\twt})^\T \tilde{\*F}^{\twt} \big)^{-1}. 
		\end{align}
		We propose to use the inverse of the standard deviation of the residuals as weights. This choice is motivated by efficiency arguments and can also improve the efficiency for the non-sparse weighted PCA. \cite{bai2003inferential} shows that under the assumptions of an approximate factor model and for $\sqrt{N}/T \rightarrow 0$ the estimated PCA loadings follow the same distribution as an OLS regression of the population loadings on $\*X$:
		\begin{align*}
		\sqrt{N} \left(H^{\top} \hat F_t - F_t  \right) = \left( \frac{1}{N} \* \Lambda^{\top} \*\Lambda \right)^{-1} \frac{1}{\sqrt{N}} \*\Lambda \*e_t + o_p(1).
		\end{align*}
		Let us set $\Weight = \Sigma_e^{-1/2}$. It is straightforward to show that the weighted PCA factors from a regression of $\*X^{\twt}$ on $\hat{\*\Lambda}^{\twt}$ have the same distribution as the following generalized least square regression
		\begin{align*}
		\sqrt{N} \left( (H^{\twt})^{\top}  \hat F_t^{\twt} - F_t  \right) = \left( \frac{1}{N} \* \Lambda^{\top}\Sigma_e^{-1}  \*\Lambda \right)^{-1} \frac{1}{\sqrt{N}} \*\Lambda^{\top} \Sigma_e^{-1}  \*e_t + o_p(1).
		\end{align*}
		Under additional assumptions on the residuals this choice of the weighting matrix will lead to the most efficient pointwise estimator for the factors while the estimator for the loadings is not affected by the weighting.\footnote{The efficiency argument requires that the time and cross-sectional dependency of the residuals can be separated. This would be satisfied if residuals are generated as $\*e=\*A_N \*\epsilon \*A_T$ where the $N \times T$ matrix $\*\epsilon$ has i.i.d. entries and the weak cross-sectional correlation is captured by the $N \times N$ matrix $\*A_N$, while the $T \times T$ matrix $\*A_T$ captures the time-series correlation.} However, estimating the large dimensional residual covariance matrix is challenging and appears to be empirically unstable.\footnote{There are various shrinkage estimator for estimating a large dimensional residual covariance matrix, for example the hard-thresholding estimator proposed in \cite{fan2013large}.} A simpler and more stable weighting matrix is the diagonal matrix $\Weight = \text{diag}(\Sigma_e^{-1/2})$ with the inverse noise standard deviations as diagonal elements as proposed by \cite{jones2001extracting} and \cite{boivin2006more}. A simple consistent estimator is based on the diagonal elements of $\hat{\Sigma}_e = \frac{1}{T} \sum_{t = 1}^{T} \hat{\*e}_t \hat{\*e}_t^\T$. This would result in the most efficient estimator if the residuals are uncorrelated. Under heteroskedastic and correlated residuals it could still be more efficient. The result is analogous to weighted and generalized least square regressions. 
		
		Under heteroskedastic noise the most informative proximate weights should be the largest loadings after reweighing them by the noise standard deviation. Using the same setup as in the toy example in Section \ref{sec:toy} but with $e_{it} \overset{\iid}{\sim}(0,\sigma_{e,i}^2)$ and the nonzero weight for unit $i$, the correlation is maximized by the largest $|\Lambda_{i,1}/\sigma_{e,i}|$:
		\begin{align*}
		\widehat{\textnormal{corr}}(\tilde{\*F}, \*F)^2=\frac{1  }{ 1   + \frac{1}{\sigma_f^2 \left( \Lambda_{i,1}/\sigma_{e,i} \right)^2}}.
		\end{align*}
		
		
		Proposition \ref{prop:weighted} provides a complete characterization for the generalized correlation $\rho^{\twt}$ between the weighted proximate factors $\tilde{\*F}^{\twt}$ and the population factors $\*F$ for diagonal weighting matrices. 
		\begin{proposition}\label{prop:weighted}
			Suppose Assumptions \ref{ass_factor},\ref{ass_error} and \ref{ass_f_e} hold, the weighting matrix $\Weight = \diag(\weight_1, \cdots, $\\$ \weight_N)$ is a diagonal matrix with $0 < \weight_i < \infty$ for all $i$, and Assumption \ref{ass_loading} holds for the weighted loadings $\Weight\*\Lambda$.  
            In the one-factor case, we assume that the $m$-th order statistic of $|\Weight\*\Lambda|$, denoted by $|\Weight\*\Lambda|_{(m),1}$, satisfies $\lim_{N \rightarrow \infty} P(|\Weight\*\Lambda|_{(m),1} \geq y) = \bar{G}_{1,m,\weight}(y)$ for some continuous function $\bar{G}_{1,m,\weight}(\cdot)$. Then Inequality \eqref{eqn-evt-one-factor} in Theorem \ref{thm-evt-one-factor} continues to hold with
			\begin{eqnarray}\label{eqn-evt-one-factor-weighted}
			\lim_{N, T\rightarrow \infty} P\left( \rho^{\twt} > \rho_0 \right) \geq \bar{G}_{1,m,\weight}( y_{m,\weight}^{\twt} ),
			\end{eqnarray}
			where $y_{m,\weight}^{\twt} = \frac{\sigma_{e,\weight}}{\sigma_{e}} y_m $ and $\sigma_{e,\weight}^2 = \max_{i,j,t,s} \weight_i \weight_j \cdot |\+E[e_{it} e_{js}]|$ . If $\Weight=\diag(\Sigma_e^{-1/2})$ then
			\[\bar{G}_{1,m,\weight}\left(\frac{\sigma_{e,\weight}}{\sigma_{e}} y_m \right) \geq \bar{G}_{1,m}(y_{m}). \]
			In the multi factor case, we assume that the $m$-th order statistic for each weighted column $\Weight\*V_k$ denoted as $(|\weight v|_{(m),1}, \cdots, |\weight v|_{(m),K} )$ satisfies $\lim_{N \rightarrow \infty} P(|\weight v|_{(m),1} \geq y_1, \cdots, |\weight v|_{(m),K} \geq y_K) = \bar{G}_{m,\weight}(y_1, \cdots, y_K)$ for some continuous function $\bar{G}_{m,\weight}(\cdot)$. Then Inequality \eqref{eqn-evt-multi-factor} in Theorem \ref{thm-evt-multi-factor} continues to hold with $\bar{G}_{m}(y_{1}, \cdots, y_{K})$ replaced by $\bar{G}_{m,\weight} \left(\frac{ \sigma_{e,\weight}}{ \sigma_{e}} y_{1}, \cdots, \frac{ \sigma_{e,\weight}}{ \sigma_{e}} y_{K} \right)$ and $\*V$ replaced by $\Weight\*V$. If $\Weight=\diag(\Sigma_e^{-1/2})$ then 
			\[\bar{G}_{m,\weight}\left(\frac{ \sigma_{e,\weight}}{ \sigma_{e}} y_{1}, \cdots, \frac{ \sigma_{e,\weight}}{ \sigma_{e}} y_{K} \right) \geq \bar{G}_{m}(y_{1}, \cdots, y_{K}). \]
		\end{proposition}
		The proposition states that in the one-factor case the weighted proximate factor based on the inverse noise standard deviation will always have a higher correlation bound than the unweighted proximate factor. In this sense the weighted proximate factor strictly dominates. The situation is more complex in the multi-factor case. There are three terms in \eqref{eqn-evt-multi-factor} and we can always improve the first term using the weighted approach with the inverse standard errors. However, the comparison is unclear for the second and third terms because these two terms involve maximum and minimum eigenvalues of matrices determined by the distribution of rotated loadings and the weighting itself can change the distribution. Our simulation and empirical analyses indicates that the weighted proximate factors always perform at least as well as the unweighted proximate factors.

		\subsection{Number of Nonzero Elements $m$}
		
		The results for the multi-factor model generalize to different levels of sparsity for each proximate factor. To simplify the presentation of the results, we have used the same number of nonzero elements $m$ for each of the proximate factors. When we use a different number of nonzero elements $m_j$ to construct each proximate factor $\tilde{\*F}_j$, the asymptotic bounds have to be modified as explained by the following proposition.
		\begin{proposition}\label{prop:vary-m}
			Assume the number of nonzeros in $\wt{\*W}_j$ is $m_j$ that varies with $j$ and the assumptions in Theorem \ref{thm-evt-multi-factor} hold. Then Theorem \ref{thm-evt-multi-factor} continues to hold with the generalized correlation threshold on the left-hand side of Inequality \eqref{eqn-evt-multi-factor} replaced by $K - \frac{\bar \gamma_{\*V,e} \sigma_{e}^2}{\underline{\gamma}^2 } \sum_{k=1}^K \frac{1}{m_k D_k y_{k}^2}$ and the right-hand side complementary order statistic replaced by $\bar{G}_{m_1,...,m_K}(y_{1}, \cdots, y_{K})$.
		\end{proposition}

		We will show in our empirical results that stronger PCA factors can usually be approximated well with fewer nonzero units compared to weaker PCA factors. This result directly follows from the theoretical bounds with factor specific sparsity $m_j$.

		In order to select the number of nonzero elements $m$ in the proximate factors, we propose two approaches. First, we use the theoretical results to choose $m$ that achieves a target probabilistic  lower bound for the generalized correlation. In our empirical analysis we set the probability lower bound to 95\% and calculate the generalized correlation $\rho_0(m)$ that is achieved with at least 95\% for different $m$, that is, $P (\rho \geq \rho_0(m)) \geq 0.95$. Then, we select $m$ to achieve a target generalized correlation.  As we show in the simulations and empirical data there is typically a sharp increase in the correlation by setting $m$ to 5-10\% of the data but only minor increases by including more observations. 
		Second, we use a completely data-driven approach. We estimate the average generalized correlation between the estimated PCA factors and our proximate factors and set $m$ to achieve a target correlation. This second approach can be implemented on a training data set to obtain the factor weights and the nonzero elements $m$ and then evaluated out-of-sample on test data. We also pursue this second approach in our empirical analysis. We confirm again that for our applications 5-10\% of the observations are sufficient to replicate the PCA factors very well.

		\subsection{Comparison with Sparse PCA}\label{subsec:comp-with-spca}

		Our method is closely related to using Lasso \citep{tibshirani1996regression} to sparsify loadings estimated from PCA. Most papers about estimating sparse principal components are more or less related to the method proposed by \cite{zou2006sparse}. They estimate sparse loadings by adding an $\ell_1$ penalty term\footnote{A generalization is to add an $\ell_2$ penalty term which leads to an Elastic Net penalty \citep{zou2005regularization}.} to the objective function (\ref{eqn-quad-loss}) yielding the formulation
		\begin{eqnarray}\label{eqn:spca-obj}
		(\bar{\*F}, \bar{\*\Lambda}) &=& \arg\,\min_{\*F, \*\Lambda} \sum_{i=1}^N \norm{\*X_i - \*F \*\Lambda^\T \*X_i}_2^2  + \alpha \sum_{j=1}^K  \norm{\*\Lambda_j}_1.  \\ 
		&& \text{s.t. } \*F^\T  \*F = I_{K} \nonumber
		\end{eqnarray}		
		Both methods, proximate factors and sparse PCA, are easy to implement. However, our method is different from sparse PCA in three aspects. First, our method does not impose the sparsity assumption on the loadings. The sparse factor weights are used to construct the proximate factors. However, the loadings of the proximate factors are in general non-sparse. It is a key insight of our paper that the factor weights and loadings should be treated differently.   
		Second, even though all methods have tuning parameters, these parameters work differently. In our method, $m$ is the number of nonzero entries in each factor weight vector. Although $\alpha$ in Lasso controls the sparsity of loadings, $\alpha$ cannot control the exact number of nonzero entries in individual loadings.\footnote{The number of nonzero elements in all loadings for the Lasso estimator is monotonically decreasing in the parameter $\alpha$. Hence, under certain conditions, there is a one-to-one mapping between the level of sparsity of the full loading matrix and the $\ell_1$ penalty weight. However, except for special cases, $\alpha$ cannot control the sparsity of a specific loading vector. Using different $\ell_1$ penalties $\alpha_j$ for different loading vectors allows for more control for the sparsity level in a specific loading vector, but it is not straightforward and not always possible to select any desired sparsity level.}
		Third, the thresholding in our approach is essentially a variable selection without changing the proportion of the largest weights, while the shrinkage in sparse PCA selects and rescales the largest loadings. It is well-known that a Lasso estimator is biased because it scales down the population parameters. A similar phenomenon occurs with sparse PCA where the largest elements in each loading vector are scaled differently. For a simple case, we can show in Proposition \ref{prop:comp-lasso} below that this downscaling makes sparse PCA factors always worse than our proximate factors that do not suffer from this bias. 
		
		We want to point out that the comparison between sparse PCA and our proximate factors is not the key element of this paper. Our important insight, that factor weights are different from factor loadings, could also be taken into account with sparse PCA methods, that is, the sparse PCA loadings can be used as factor weights and the loadings could be obtained from a second stage regression. This is different from how sparse PCA is considered in the literature. The important contribution of our paper is to provide a statistical framework for a sparse estimator of a non-sparse population model which is missing for sparse PCA. Furthermore, our method gives applied researchers more control over setting a desired sparsity level for individual factors and is very transparent in its application.  
		

		We consider now the special case of a one-factor model with cross-sectionally i.i.d. distributed errors. For sparse PCA with $\ell_1$ penalty similar to \cite{jolliffe2003modified} we can map this estimator into our framework. In more detail, we estimate the loadings $\bar \Lambda$ by minimizing $\norm{\*X-\*\Lambda \hat{\*F}^\T }^2_F + \alpha \norm{\Lambda}_1$ and use them as weights to construct the sparse PCA factors $\bar{\*F} =\*X^\T \bar{\*\Lambda} (\bar{\*\Lambda}^\T  \bar{\*\Lambda})^{-1}$.
		The difference in the generalized correlation between the our method and sparse PCA is
		\begin{align}
		\nonumber \Delta \rho = & \tr \left( (\*F^\T \*F/T)^{-1} (\*F^\T  \tilde{\*F}/T) (\tilde{\*F}^\T \tilde{\*F}/T)^{-1} (\tilde{\*F}^\T \*F/T) \right) \\
		& - \tr \left( (\*F^\T \*F/T)^{-1} (\*F^\T \bar{\*F}/T) (\bar{\*F}^\T  \bar{\*F}/T)^{-1} (\bar{\*F}^\T  \*F/T) \right) \label{deltarhodef}
		\end{align}
		\begin{proposition} \label{prop:comp-lasso}
			Assume a one factor model and $\frac{1}{T} \*e_i^\T \*e_i \rightarrow \sigma_e^2$.
			The two tuning parameters $m$ and $\alpha$ are chosen such that proximate factors and sparse PCA select the same number of nonzero elements. 
			As $N, T \rightarrow \infty$, we have with probability one $\Delta \rho \geq 0$.
		\end{proposition}
		Proposition \ref{prop:comp-lasso} states that our method approximates the population factor at least as well as sparse PCA. Although Proposition \ref{prop:comp-lasso} is stated in a restrictive setting, we show in simulations that $\tilde{\*F}$ has a higher generalized correlation with $\*F$ even when the assumptions in Proposition \ref{prop:comp-lasso} do not hold.

		\section{Simulation}\label{sec:simulation}

		We evaluate proximate factors (PPCA), weighted proximate factors (PPCA (wt)) based on inverse standard errors, and sparse PCA (SPCA). We confirm that proximate factors are a very good approximation of the population factors and that the lower bounds are an accurate description of the exceedance probabilities for the generalized correlations. 
		
		Our baseline model assumes
		\[\Lambda_i \overset{\iid}{\sim} \mathcal{N}(0,I_{K}), \quad F_t \overset{\iid}{\sim} \mathcal{N}(0,\Sigma_F),\quad   e_{it} \overset{\iid}{\sim} \mathcal{N}(0,\sigma_i^2), \quad \sigma_i  \overset{\iid}{\sim} \mathrm{U}(0.5,1). \]
		In this model errors are heterogeneous. \footnote{We also simulate data with time-series dependence or with cross-section dependence in the errors. The results are presented in Section \ref{sec:add-simulation} in the Internet Appendix.} 
		We study the impact of varying $N$, $T$, $\Sigma_F$ and $m$. For each set of parameters and a fixed threshold $\rho_0$, we run 1000 Monte-Carlo simulations to calculate the empirical probability of $P(\rho \geq \rho_0)$ by counting the percentage of times that $\rho \geq \rho_0$.   
		
		\subsection{Probabilistic Bounds}\label{susec:sim-prob}
		
		We start with the case of only one factor and study the probability bound provided in Theorem \ref{thm-evt-one-factor} and Proposition \ref{prop:GEV}. We set $\rho_0=0.95$ for the squared correlation which means that we want to evaluate the probability that the proximate factors have correlation of at least 97.5\% with the population factors. The probabilist lower bound $\underline{p}=\bar{G}_{1,m}(y_m)$ depends on $y_m$ which in our setup simplifies to $y_m=\left(\frac{1}{m  \sigma_{\*F_1}^2} \frac{0.95}{1 - 0.95} \right)^{1/2}$ as $h(m)=0$ and $\sigma_{e}^2 = 1$. Equation \eqref{eqn:evt-one-factor-iid} and the GEV for $\Lambda_{i} \sim \mathcal{N}(0, 1)$ in Table \ref{tab:extreme-value-example} simplify the expression for $\bar{G}_{1,m}(y_m)$: 
		\begin{align}
		\bar{G}_{1,m}(y_m) = 1 - e^{-\tau}  \sum_{j = 0}^{m-1} \frac{\tau^j}{j!} \;\; \text{with  } \tau= \exp \Big(- \frac{y_m - b_{1,N} }{1/(2N \psi(b_{1,N}) ) }   \Big) \text{ and } b_{1,N} = \Psi^{-1}\Big(1-\frac{1}{2N}\Big). \label{eqn:simtoy}
		\end{align}
		
		\begin{figure}[t!]
			\centering
			\tcapfig{Correlations in One-Factor Model as a Function of $m$ }
			\begin{subfigure}{.5\textwidth}
				\centering
				\includegraphics[width=1\linewidth]{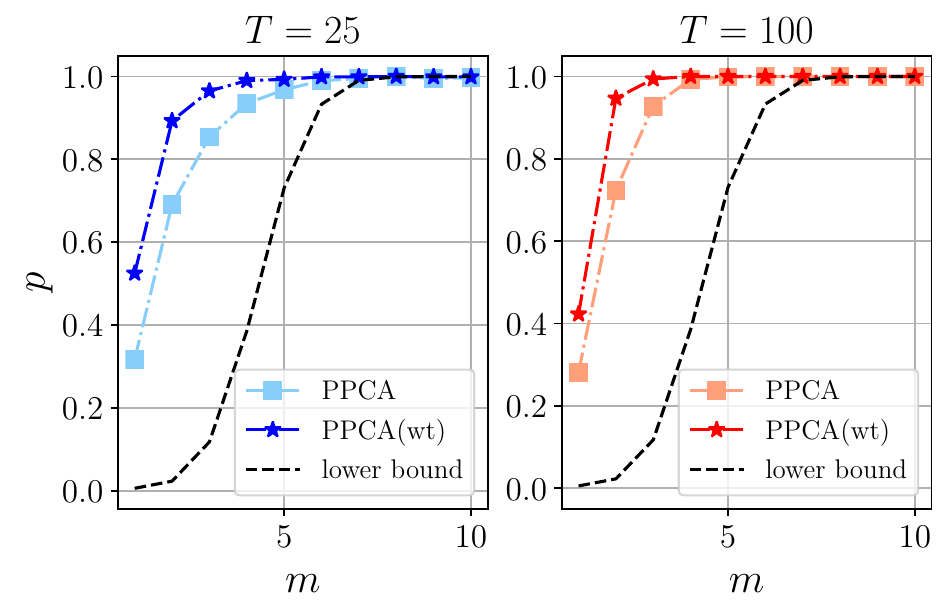}
				\subcaptab{$\sigma_{\*F_1} = 1.0$}
			\end{subfigure}%
			\begin{subfigure}{.5\textwidth}
				\centering
				\includegraphics[width=1\linewidth]{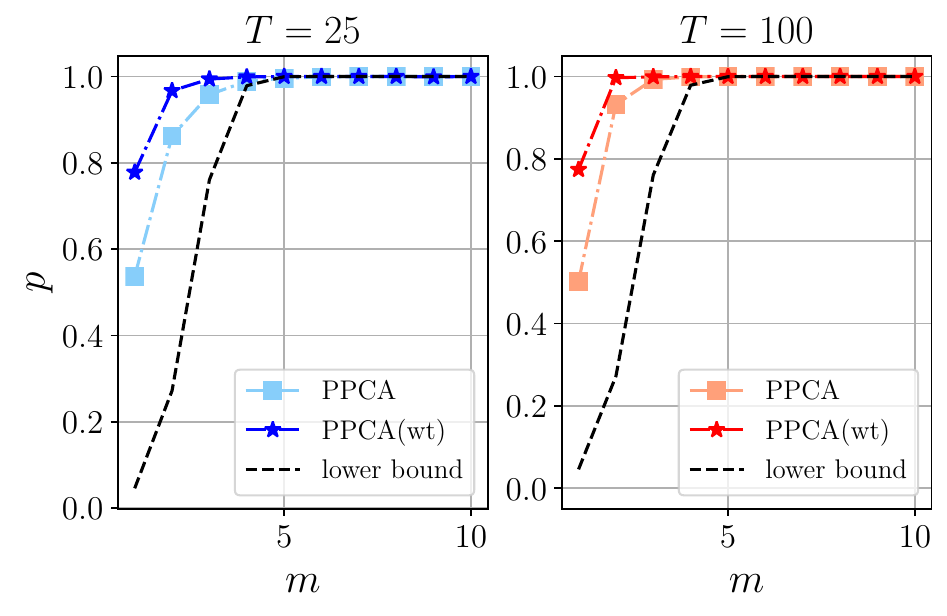}
				\subcaptab{$\sigma_{\*F_1} = 1.2$}
			\end{subfigure}
			\bnotefig{This figure compares $P(\rho \geq \rho_0)$ based on 1,000 Monte Carlo simulations and the probability lower bound $\underline{p} = \bar{G}_{1,m}(y_m)$ as a function of $m$. PPCA are the unweighted proximate factors and PPCA (wt) use the inverse standard errors as weights. We set $N=100$ and $\rho_0=0.95$. The left plots use $\sigma_{\*F_1} = 1.0$ while the right plots have the higher $\sigma_{\*F_1} = 1.2$. We calculate the probabilities for $T=25$ and $T=100$.  Both, $P(\rho \geq \rho_0)$ and $\underline{p}$, are very close to 1 with about 5-10\% of units $m$ to construct the proximate factors.}
			\label{fig:thm-evt-one-factor}
		\end{figure}

		\begin{figure}[t!]
			\centering
			\tcaptab{Generalized Correlations as a Function of $N$}
			\begin{subfigure}{.5\textwidth}
				\centering
				\includegraphics[width=1\linewidth]{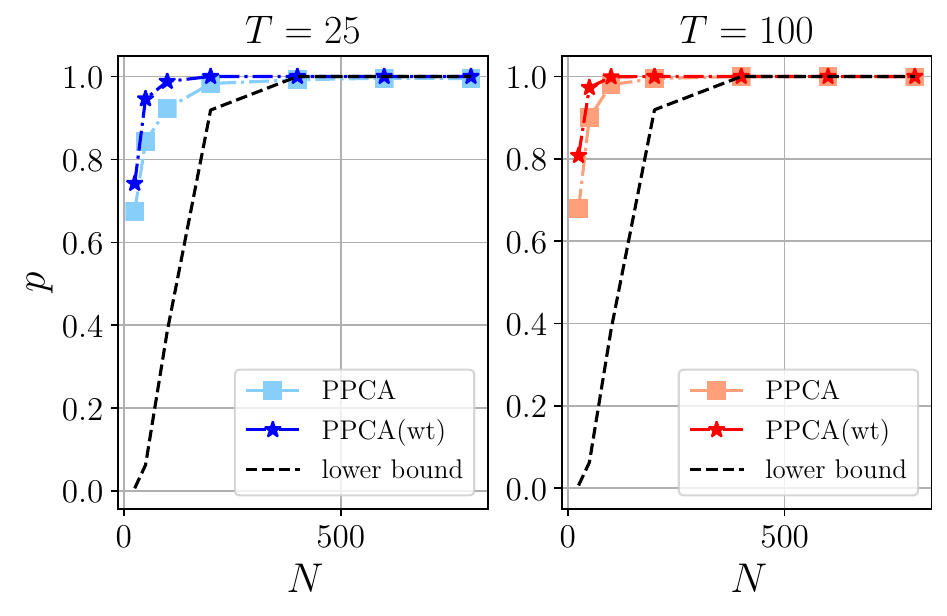}  \subcaptab{One-factor model}
				\label{fig:thm-evt-one-factor-N}
			\end{subfigure}%
			\begin{subfigure}{.5\textwidth}
				\centering
				\includegraphics[width=1\linewidth]{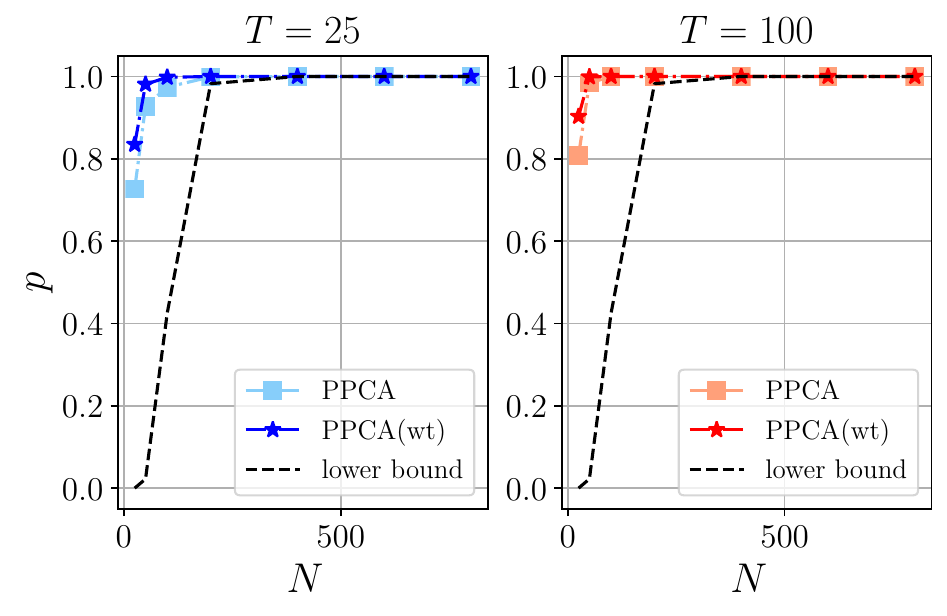}
				\subcaptab{Multi-factor model}
				\label{fig:thm-evt-multi-factor-N}
			\end{subfigure}
			\bnotetab{This figure compares $P(\rho \geq \rho_0)$ based on 1,000 Monte Carlo simulations and the probability lower bound $\underline{p}$ as a function of $N$. The nonzero elements are set to $m=4$ in all cases. PPCA are the unweighted proximate factors and PPCA (wt) use the inverse standard errors as weights. The one-factor model uses $\rho_0 = 0.95$ and $\sigma_{\*F_1}=1.0$, while the multi-factor model sets $\rho_0 = 1.9$ and $[\sigma_{\*F_1}, \sigma_{\*F_2}]=[1.2, 1.0]$. The lower bound is $\underline{p} = \bar{G}_{1,m}(y_m)$ for one factor and is defined in Equation \eqref{eqn:mod-lower-bound-multi-factor} for two factors. We calculate the probabilities for $T=25$ and $T=100$. Both, $P(\rho \geq \rho_0)$ and $\underline{p}$, are very close to 1 for $N > 250$.
			}
		\end{figure}

		Figure \ref{fig:thm-evt-one-factor} shows how $P(\rho \geq \rho_0)$ and $\underline{p}$ change with different $m$, $T$ and signal-to-noise ratio $\sigma_{\*F_1}/\sigma_e$. First, all probabilities are monotonically increasing in $m$, that is, less sparse factors are closer to the population factors. Second, there is a sharp incline in probabilities for $m \leq 10$ but then $P(\rho \geq \rho_0)$ becomes very close to 1. This means that $\tilde{\*F}$ has 97.5\% correlation with $\*F$, but is based on less than 10\% of all observations. Third, in both cases the lower bound $\underline{p}$ closely approximates the probability $P(\rho \geq \rho_0)$ for $m \geq 5$ and hence can be used for guidance on evaluating the proximate factors. Fourth, the larger the signal-to-noise ratio, the larger $P(\rho \geq \rho_0)$ and $\underline{p}$. Fifth, proximate factors are more highly correlated with the population factors for larger $T$. Equation (\ref{eqn-evt-one-factor}) has an $o_p(1)$ term and this term decreases with $T$. Last but not least, the weighted proximate factors perform better in this setup with heteroskedastic errors.

		Figure \ref{fig:thm-evt-one-factor-N} shows that, as $N$ increases, both $P(\rho \geq \rho_0)$ and $\underline{p}$ converge to 1 with only $m=4$ nonzero weight elements. This confirms our result in Proposition \ref{prop-one-factor}. As $\Lambda_i$ is normally distributed with unbounded support, our approach selects the $m$ units that are themselves very close to the population factors.
		
				\begin{figure}[t!]
			\centering
			\tcapfig{Generalized Correlations in Multi-Factor Model as a Function of $m$}
			\begin{subfigure}{.5\textwidth}
				\centering
				\includegraphics[width=1\linewidth]{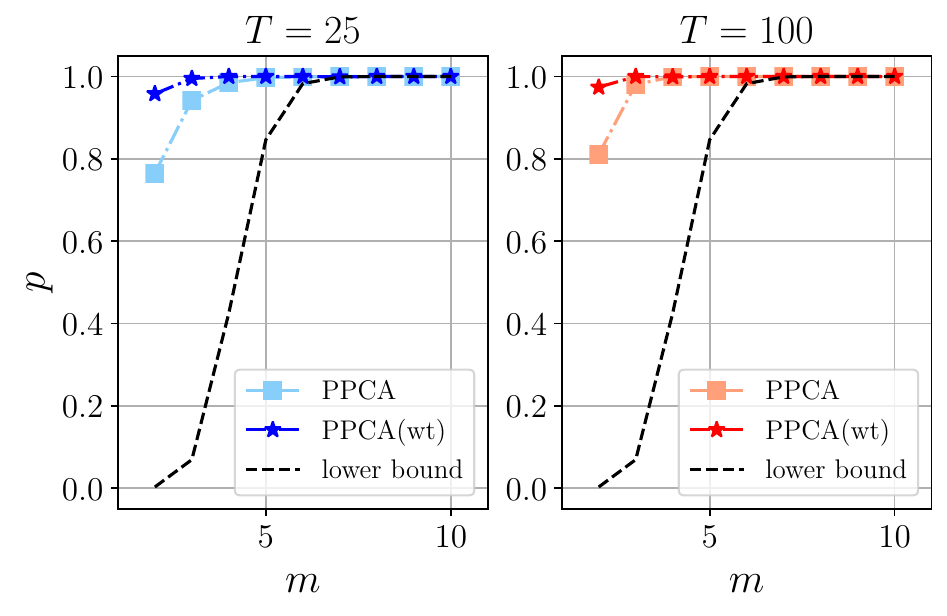}
				\subcaptab{$[\sigma_{\*F_1}, \sigma_{\*F_2}] = [1.2, 1.0]$}
			\end{subfigure}%
			\begin{subfigure}{.5\textwidth}
				\centering
				\includegraphics[width=1\linewidth]{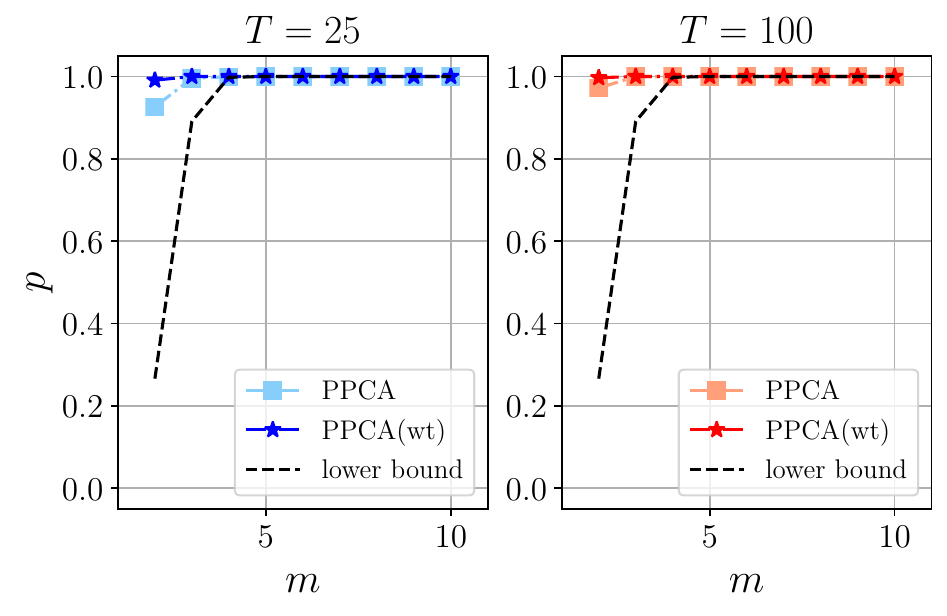}
				\subcaptab{$[\sigma_{\*F_1}, \sigma_{\*F_2}] = [1.5, 1.2]$}
			\end{subfigure}
			\bnotefig{This figure compares $P(\rho \geq \rho_0)$ based on 1,000 Monte Carlo simulations and the probability lower bound defined in Equation \eqref{eqn:mod-lower-bound-multi-factor} as a function of $m$. We have $K=2$ and set $\rho_0=1.9$ and $N=100$. PPCA are the unweighted proximate factors and PPCA (wt) use the inverse standard errors as weights. The left plots use $[\sigma_{\*F_1}, \sigma_{\*F_2}] = [1.2, 1.0]$ while the right plots have the higher $[\sigma_{\*F_1}, \sigma_{\*F_2}] = [1.5, 1.2]$. We calculate the probabilities for $T=25$ and $T=100$.  Both, $P(\rho \geq \rho_0)$ and $\underline{p}$, are very close to 1 with about 5-10\% of units $m$ to construct the proximate factors.}
			\label{fig:thm-evt-multi-factor}
		\end{figure}
		
		Next, we consider the case of two factors, i.e. $K = 2$ and study the probability bound provided in Theorem \ref{thm-evt-multi-factor} and Proposition \ref{prop:evt-multi-factor-simplified}. We set $\rho_0= K \times 0.95 = 1.9$ and assume that the factors are independent with $\Sigma_F = \diag(\sigma_{\*F_1}^2, \sigma_{\*F_2}^2)$. For independent and normally distributed loadings the non-overlapping condition is satisfied and the correction terms in \eqref{eqn:evt-multi-factor-simplified} can be neglected. We set $y_1=y_2$, which simplifies the bound to
		\begin{align}\label{eqn:mod-lower-bound-multi-factor}
		\lim_{N, T\rightarrow \infty}  P(\rho > \rho_0 ) \geq \underline{p} = \bar{G}_{1,m}(y_{1}) \bar{G}_{2,m}(y_2)\;\; \text{with } y_1=y_2=\left(\frac{1}{m}\Big( \frac{1}{\sigma_{\*F_1}^2} + \frac{1}{\sigma_{\*F_2}^2}  \Big)\frac{1}{K-\rho_0} \right)^{1/2}
		\end{align}
		In our special case $ \bar{G}_{1,m}(y_{1})$ and $\bar{G}_{2,m}(y_2)$  take the same form as in the one-factor case in \ref{eqn:simtoy}. Figure \ref{fig:thm-evt-multi-factor} displays the probability $P(\rho \geq \rho_0)$ and $\underline{p}$ as a function of $m$ for different $T$ and factor variances in the multi-factor case. All the results from the one-factor case carry over to multiple factors.

		\subsection{Comparison with Sparse PCA}\label{subsec:sim-comp-wtih-spca}
		
		We compare the estimation of the factors, loadings and the common component of our methods PPCA and PPCA (wt) with SPCA. We have already shown theoretically that in the special case of a one-factor model with homoskedastic errors, the proximate factors provide a better estimator for the population factors than sparse PCA. We consider now the more general setting of multiple factors with dependent and heteroscedastic errors. For the various estimators we calculate the generalized correlation of estimated factors $\hat{\*F}$ and loadings $\hat{\*\Lambda}$ with the population values and also calculate the root-mean-squared error (RMSE), 
		\[ \mathrm{RMSE} = \left(\frac{1}{NT} \sum_{i=1}^N \sum_{t=1}^T \left( X_{it}-  \hat X_{it} \right)^2 \right)^{1/2}  \text{ with }  \hat X_{it} = \hat{\Lambda}_i^\T \hat{F}_t,\]
		for different estimators $\hat{\*F}$ and $\hat{\*\Lambda}$. 
		The loadings for PPCA and PPCA (wt) are obtained from a second stage regression on the factors as specified in \eqref{eqn-tilde-lam-def} and \eqref{eqn-tilde-lam-def-wt}. The SPCA factor $\bar{\*F}$ and loadings $\bar{\*\Lambda}$ are the solution to the optimization problem \eqref{eqn:spca-obj} for different $\ell_1$ penalties $\alpha > 0$.  In order to make various approaches comparable, we choose the number of nonzero elements $m_j$ for each factor weight in PPCA  and PPCA (wt) to match the number of nonzero elements in each factor weight for SPCA for a given penalty $\alpha$. This means that all methods have the same degree of sparsity for each factor but can use different entries and magnitudes for the weights.
		
		We report the in-sample (training) and out-of-sample (test) results. For the first case, we estimate the parameters on the training data set and calculate the goodness-of-fit measures on it. For the second case, we use the factor portfolio weights estimated on the training data set to construct the factors and loadings on the test data set.\footnote{The loadings for SPCA on the test data set are the same as on the training data set. However, the loadings of PPCA and PPCA (wt) are different on the training and test data as they are estimated in a second stage regression.}

		Motivated by our empirical application, we simulate a $K=5$ factor model where $\Sigma_{F} = I_K$ and errors are correlated and heterogeneous:
		\begin{itemize}
			\item Heteroskedasticity: $e_{it} = \sigma_i v_{it}$, $\sigma_i \stackrel{\iid}{\sim} \mathrm{U}(1, 3)$ and $v_{it} \stackrel{\iid}{\sim} \mathcal{N}(0,1)$ 
			\item Cross sectional dependence: $\*e_t \stackrel{\iid}{\sim} \mathcal{N}(0, \Sigma_e)$, where $\Sigma_e = (c_{ij}) \in \+R^{N \times N}$ with $c_{ij} = 0.5^{|i - j|}$
		\end{itemize}

		Figure \ref{fig:sim-rho-rmse} reports the generalized correlations for factors and loadings and RMSE for the different methods. On both training and test data, PPCA and PPCA (wt) have higher correlations with the population factors and smaller RMSE than SPCA.  Interestingly PPCA and PPCA (wt) can actually provide as good a fit as conventional PCA without thresholding. Moreover, PPCA (wt) better approximates the population factors and has smaller RMSE than unweighted PPCA which supports our arguments in Section \ref{subsec:weighted-estimator}. 
		\begin{figure}[t!]
			\tcaptab{Comparison between PPCA, PPCA (wt) and SPCA}
			\begin{adjustwidth}{-0.6cm}{}
				\centering
				\begin{subfigure}{.7\textwidth}
					\centering
					\includegraphics[width=1\linewidth]{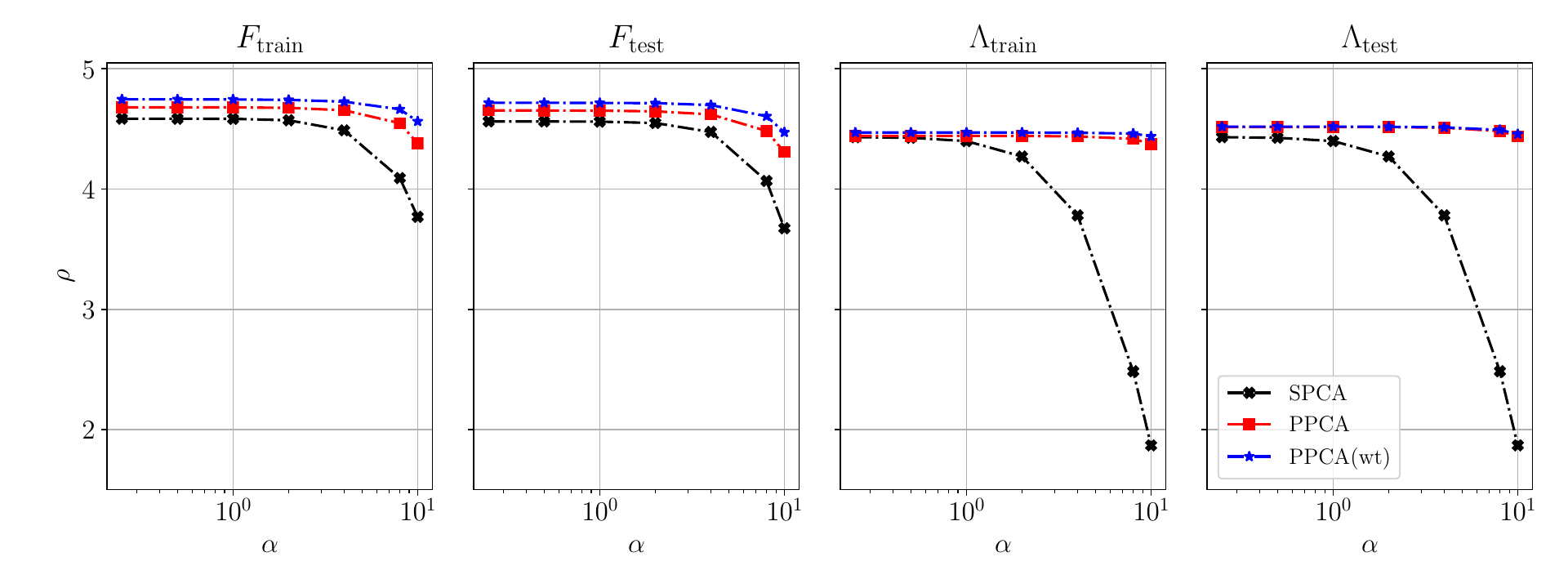}
					\subcaptab{Generalized correlations}
				\end{subfigure}%
				\begin{subfigure}{.35\textwidth}
					\centering
					\includegraphics[width=1\linewidth]{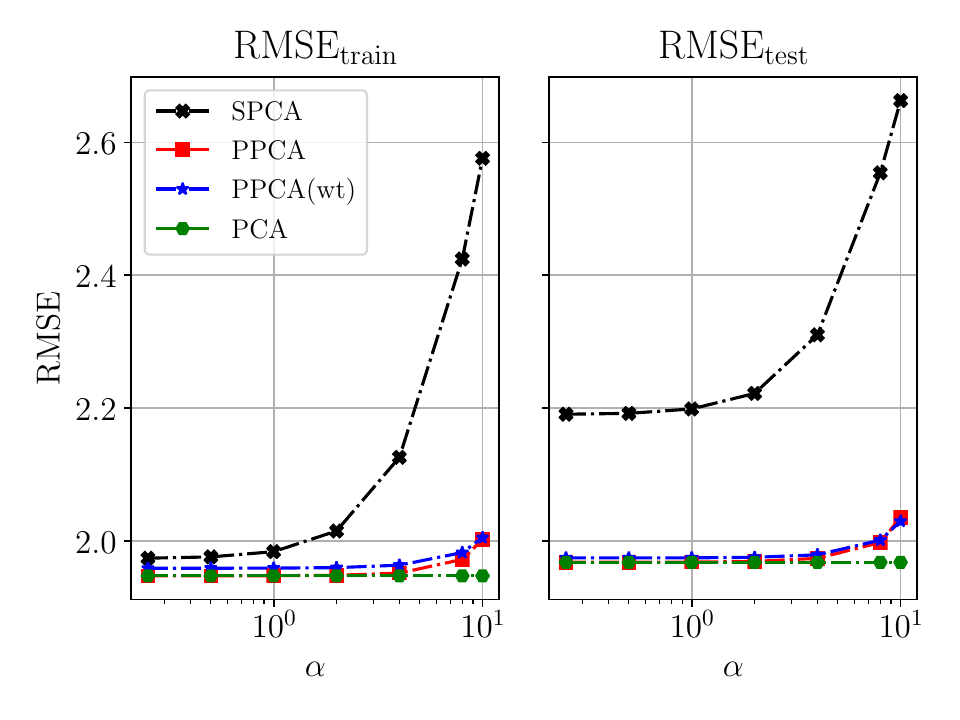}
					\subcaptab{RMSE}
				\end{subfigure}
			\end{adjustwidth}
			\bnotefig{This figure compares the in-sample and out-of-sample generalized correlations for factors and loadings, and in-sample and out-of-sample RMSE in $\*X$ for PPCA, PPCA (wt) and SPCA. We set $N=100$, $T=100$ and $K=5$. In order to achieve the same sparsity level for various methods, we first choose $\alpha$, the $\ell_1$ penalty for SPCA, and set the number of nonzero weights $m_j$ for each factor in PPCA and PPCA (wt) to obtain the same number as SPCA. PPCA and PPCA (wt) have higher in-sample and out-of-sample generalized correlations for both factors and loadings. Moreover, PPCA and PPCA (wt) have lower in-sample and out-of-sample RMSE.}
			\label{fig:sim-rho-rmse}
		\end{figure}
		
		We have already argued in Proposition \ref{prop:comp-lasso} that under certain conditions sparse PCA factors provide a worse approximation to factors than PPCA. This result seems to continue to hold in more general setups. More importantly, estimating sparse loadings in a population model with non-sparse population loadings is simply miss-specified. It is striking how much lower the correlations of the SPCA loadings are with the population loadings. This loading estimation error leads to the substantially larger RMSE. Both, PPCA and PPCA (wt), have correlations and errors that can be almost twice as good as SPCA. 

		\section{Empirical Application}\label{sec:empirical}

		In this section, we apply our method to two relevant datasets, 370 financial portfolios and 128 macroeconomic variables. We compare the proximate factors $\tilde{\*F}$ with the PCA factors $\hat{\*F}$ by two metrics, the generalized correlation $\rho$ between $\tilde{\*F}$ and $\hat{\*F}$ and the proportion of variation explained by $\tilde{\*F}$ and $\hat{\*F}$. In both datasets, proximate factors constructed from our approach with 5-10\% of cross-section observations are very close to the conventional PCA factors. As proximate factors are composed of only a few cross-section observations, we can find economically meaningful labels for each latent factor.
		
		\subsection{Financial Portfolio Data}
		We study the monthly returns of 370 portfolios sorted in deciles based on 37 anomaly characteristics from 07/1963 to 12/2016. For each characteristic, we have ten decile-sorted portfolios that are updated yearly and are linear combinations of returns of U.S. firms in CRSP. These are the same portfolios studied in \cite{lettau2020factors} and are compiled by \cite{kozak2017shrinking}.\footnote{We thank \cite{kozak2017shrinking} for allowing us to use their data.} These 37 anomaly characteristics are listed in Table \ref{tab:list-anomaly} in the Internet Appendix.\footnote{\cite{kozak2017shrinking} use a set of 50 anomaly characteristics. We use 37 of those characteristics with the longest available cross-sections. The same data is studied in \cite{lettau2020factors}.} Each of these portfolios can be interpreted as an actively managed portfolio based on a signal that has been shown to be relevant for the risk-return tradeoff. In contrast to individual stock data these managed portfolios are more stationary and \cite{lettau2020factors} show that they can be modeled by a stable factor structure.

		\begin{figure}[t]
			\centering
			\tcapfig{Financial Portfolio Data: Generalized Correlation and Explained Variance}
			\begin{subfigure}{.5\textwidth}
				\centering
				\includegraphics[width=0.8\linewidth]{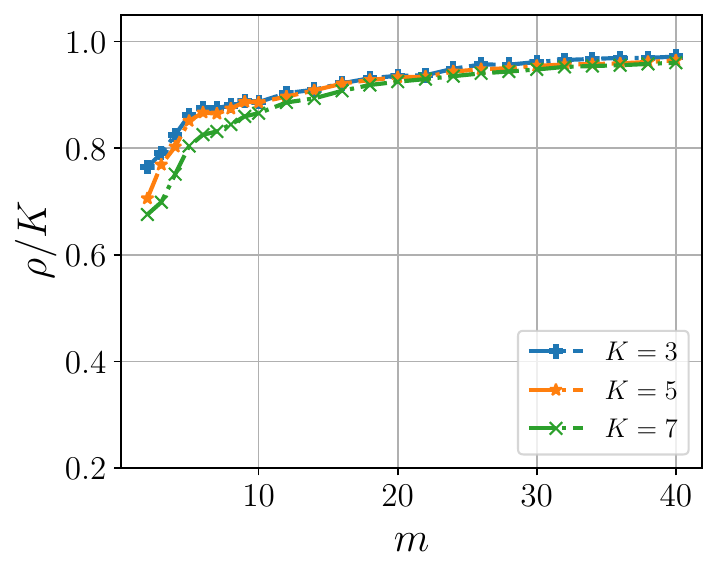}
				\subcaptab{Generalized Correlations}
			\end{subfigure}%
			\begin{subfigure}{.5\textwidth}
				\centering
				\includegraphics[width=0.83\linewidth]{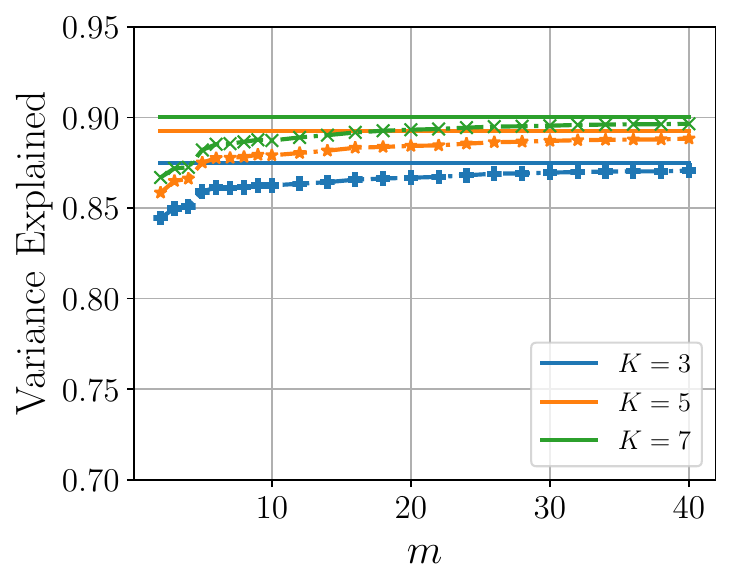}
				\subcaptab{Variance Explained}
			\end{subfigure}
			\bnotefig{The left figure shows the normalized generalized correlation $\rho/K$  between $\tilde{\*F}$ and $\hat{\*F}$ as a function of $m$. The right figure shows the proportion of variation explained by $\tilde{\*F}$ and $\hat{\*F}$ as a function of $m$. Each proximate factor has the same number of nonzero elements $m$. $N$ is 370 and $T$ is 638 in this data. In both figures $K$ varies from 3 to 7. Proximate factors constructed with only 5\% of the portfolios are very close to the non-sparse PCA factors.}
			\label{fig:single-sorted-portfolios}
		\end{figure}

		\begin{table}[t!]
			\centering
			\tcaptab{Financial Portfolio Data: Generalized Correlation Decomposition}
			\begin{tabular}{cccccc}
				\toprule
				\diagbox{$m$}{} & $\hat{\*F}_1$ & $\hat{\*F}_2$ & $\hat{\*F}_3$ & $\hat{\*F}_4$ & $\hat{\*F}_5$  \\
				\midrule
				10 & 0.993 & 0.912 & 0.771 & 0.918 & 0.837  \\
				20 & 0.995 & 0.948 & 0.883 & 0.949 & 0.890 \\
				30 & 0.996 & 0.965 & 0.935 & 0.966 & 0.910 \\
				40 & 0.997 & 0.971 & 0.958 & 0.975 & 0.923 \\
				\bottomrule
			\end{tabular}
			\bnotetab{This table shows the generalized correlation between each $\hat{\*F}_j$ and all $\tilde{\*F}$ in a model with 5 factors. $N$ is 370 and $T$ is 638.  These generalized correlations correspond to the $R^2$ from a regression of each $\hat{\*F}_j$ on all $\tilde{\*F}$. The sum of each row equals $\rho$. In general, stronger factors have larger signal-to-noise ratio and can be better approximated with a smaller $m$.}
			\label{tab:single-sorted-portfolios}
		\end{table}
		
		In Figure \ref{fig:single-sorted-portfolios}, we calculate the generalized correlation $\rho$ between estimated PCA and proximate factors and compare the proportion of variation explained by these factors for different number of factors $K$. We normalize the generalized correlation by the number of factors. The closer $\rho/K$ to 1, the better the sparse factors approximate the non-sparse PCA factors.  The more portfolios we include in the proximate factors, the larger the correlation $\rho$ and the amount of variation explained. When the number of nonzero entries $m$ in each sparse loading is around 20, corresponding to around 5\% of the cross-section units, the ratio is close to 0.9, implying that the average correlation between each proximate factor in $\tilde{\*F}$ and each estimated factor in $\hat{\*F}$ is around 95\%. The explained variation is defined as $\left(\sum_{i=1}^N\sum_{t=1}^T \hat X^2_{it} \right)/ \left(\sum_{i=1}^N\sum_{t=1}^T X^2_{it} \right)$, where $\hat X_{it}=\hat \Lambda_i^{\top} \hat F_t$ is the based on the estimated factors and loadings. As the data is demeaned the explained variation equals the percentage of variance explained in the data. By construction, PCA factors maximize the explained variance in-sample. However, the variance explained by $ \tilde{\*F}$ becomes very close to that by $\hat{\*F}$ when $ \tilde{\*F}$ is constructed with the 10 cross-section units with the strongest signals.

		Table \ref{tab:single-sorted-portfolios} shows the generalized correlations for each of the five factors with their proximate version. As by construction, the PCA factors are uncorrelated, the individual generalized correlations correspond to the $R^2$s from a linear regression of each PCA factor on the five proximate factors. Even with only 10 portfolios, we can capture each of the latent factors very well, in particular the first two factors. 30 portfolios are sufficient to almost perfectly replicate the latent factors.\footnote{As many investors trade on factors, trading on proximate factors can also significantly reduce trading costs.}

		We study each of the five proximate factors in more detail. Figure \ref{fig:single-sorted-portfolios-4th-loading} provides a clear picture of the assets used to construct the fourth factor, labeled as the momentum factor. We use the fourth factor as an example and present the figures for the other factors in the Internet Appendix. First, the fourth proximate factor is only composed of five strategies related to momentum, which are Industry Momentum, 6-month Momentum, 12-month Momentum, Value Momentum and Value Momentum Profitability. Second, this factor is a long-short factor, i.e., it has negative weights for the smallest decile portfolios and positive weights for the largest decile portfolios. Most financial factors are constructed as long-short factors to capture the difference between the extreme characteristics. For comparison, we also show the portfolio weights for the fourth PCA factor in Figure \ref{fig:single-sorted-portfolios-4th-loadingPCA}. The pattern suggests that the fourth factor is related to momentum, but our proximate factor provides the argument that the information for this factor is almost completely captured by momentum portfolios.

		\begin{figure}[t!]
			\centering
			\tcapfig{Financial Portfolio Data: Portfolio Weights of 4th Proximate Factor}
			\includegraphics[width=1\linewidth]{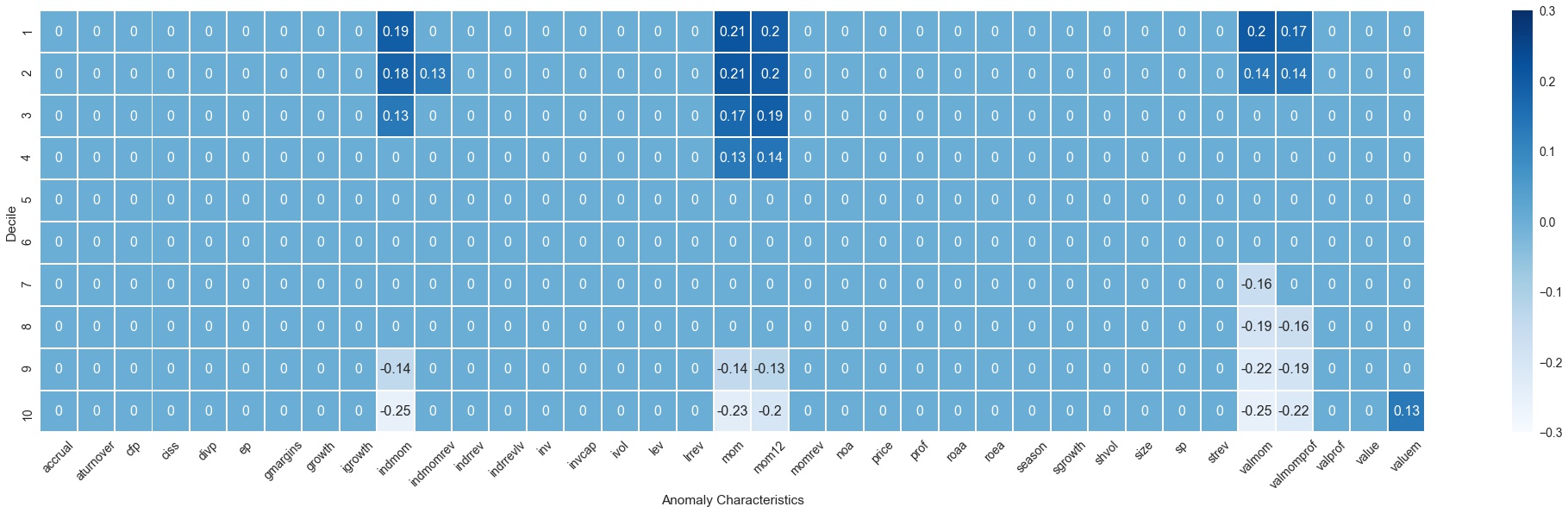}
			\bnotefig{This figure shows the portfolio weights of the 4th proximate factor. The 30 nonzero entries in the portfolio weights are composed of the extreme quantiles of momentum related characteristics. These factor weights are positive for the low quantiles and negative for the high quantiles. This proximate factor can be interpreted as a long-short momentum factor.  The names and description of the 37 anomaly characteristics are listed in Table \ref{tab:list-anomaly} in the Internet Appendix.}
			\label{fig:single-sorted-portfolios-4th-loading}
		\end{figure}
		
		\begin{figure}[t!]
			\centering
			\tcapfig{Financial Portfolio Data: Portfolio Weights of 4th PCA Factor}
			\includegraphics[width=1\linewidth]{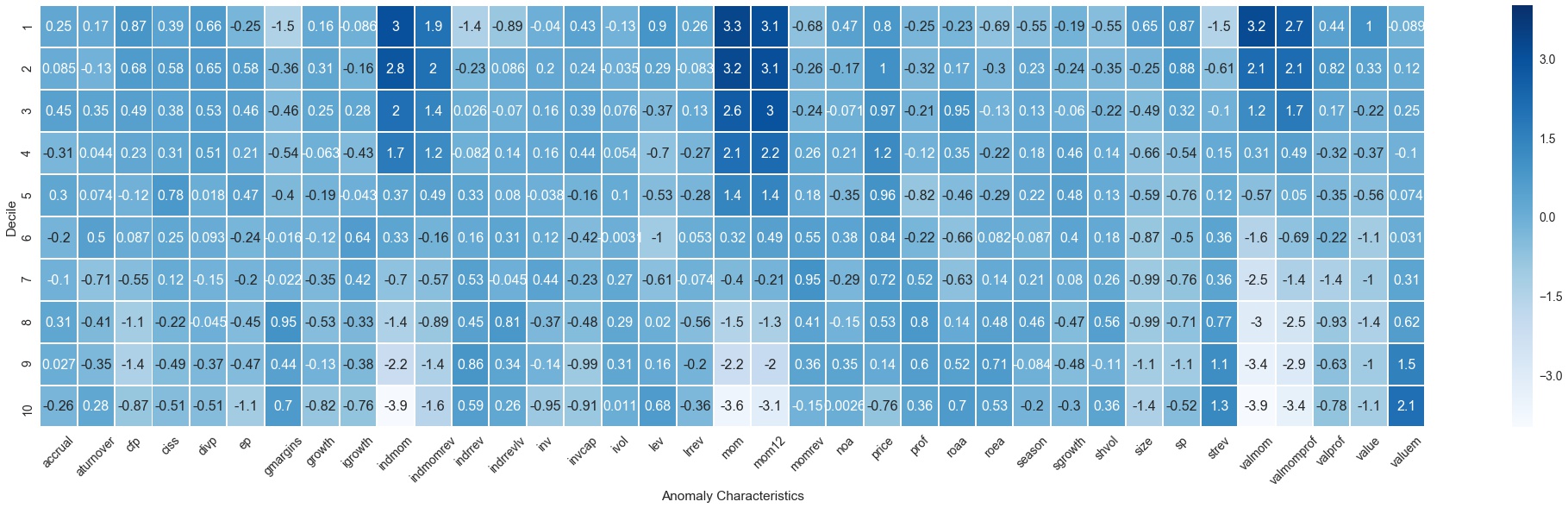}
			\bnotefig{This figure shows the portfolio weights of the 4th PCA factor. In contrast to the 4th proximate factor all factor weights are nonzero which makes the interpretation challenging.}
			\label{fig:single-sorted-portfolios-4th-loadingPCA}
		\end{figure}
		
		\begin{figure}[t!]
			\centering
			\tcapfig{Financial Portfolio Data: Composition of Proximate Factors}
			\includegraphics[width=1\linewidth]{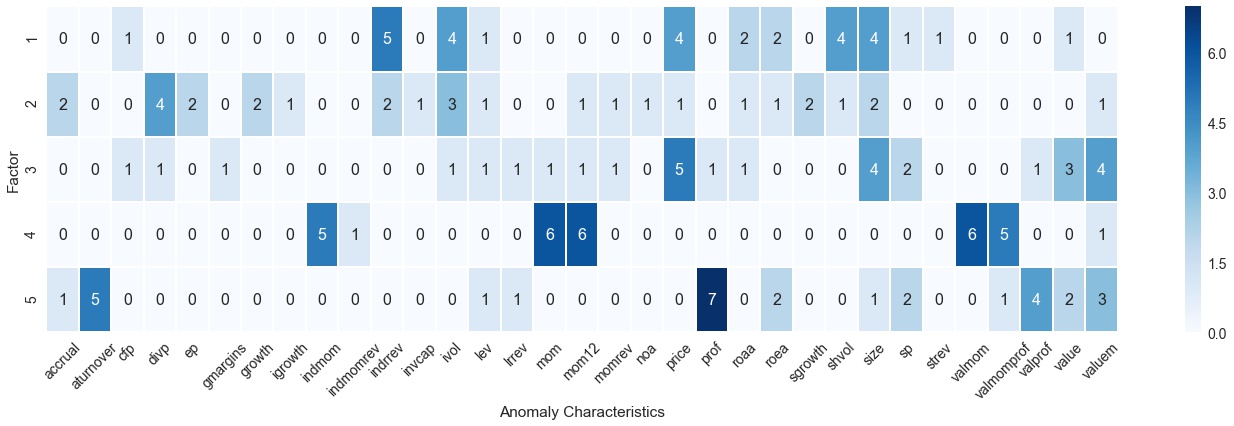}
			\bnotefig{This figure shows the number of nonzero weights of each proximate factor for each characteristic. We study 5 factors and each proximate factor is constructed with 30 portfolios. The names and description of the 37 anomaly characteristics are listed in Table \ref{tab:list-anomaly} in the Internet Appendix.}
			\label{fig:single-sorted-portfolios-loading}
		\end{figure}

		Figure \ref{fig:single-sorted-portfolios-loading} displays the composition of each latent factor based on the characteristic groups. It provides a first intuition for interpreting each latent factor. Based on Figures \ref{fig:single-sorted-portfolios-1st-loading} to \ref{fig:single-sorted-portfolios-5th-loading} we assign the following labels to the other latent factors: The first factor is a ``long''-only market factor. The second factor loads on categories such as dividend/price, earning/price, investment capital, and asset and sales growth suggesting a ``value'' interpretation. The third factor loads strongly on price and large market capitalization and hence we label it as a big-price factor. The last factor is clearly an asset turnover-profitability factor.
		
		\begin{figure}[t!]
			\tcapfig{Financial Portfolio Data: Out-of-Sample Generalized Correlations and RMSE}
			\begin{adjustwidth}{-1cm}{}
				\centering
				\begin{subfigure}{.4\textwidth}
					\centering
					\includegraphics[width=1\linewidth]{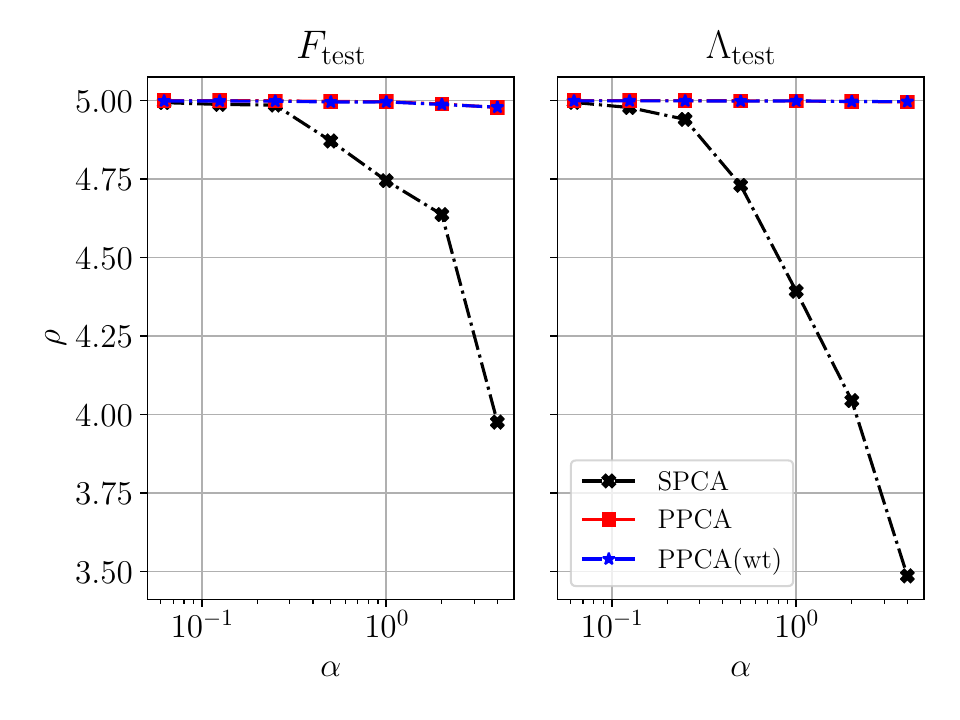}
					\subcaptab{$\rho$ with $\hat{\*F}$/$\hat{\*\Lambda}$}
				\end{subfigure}%
				\begin{subfigure}{.4\textwidth}
					\centering
					\includegraphics[width=1\linewidth]{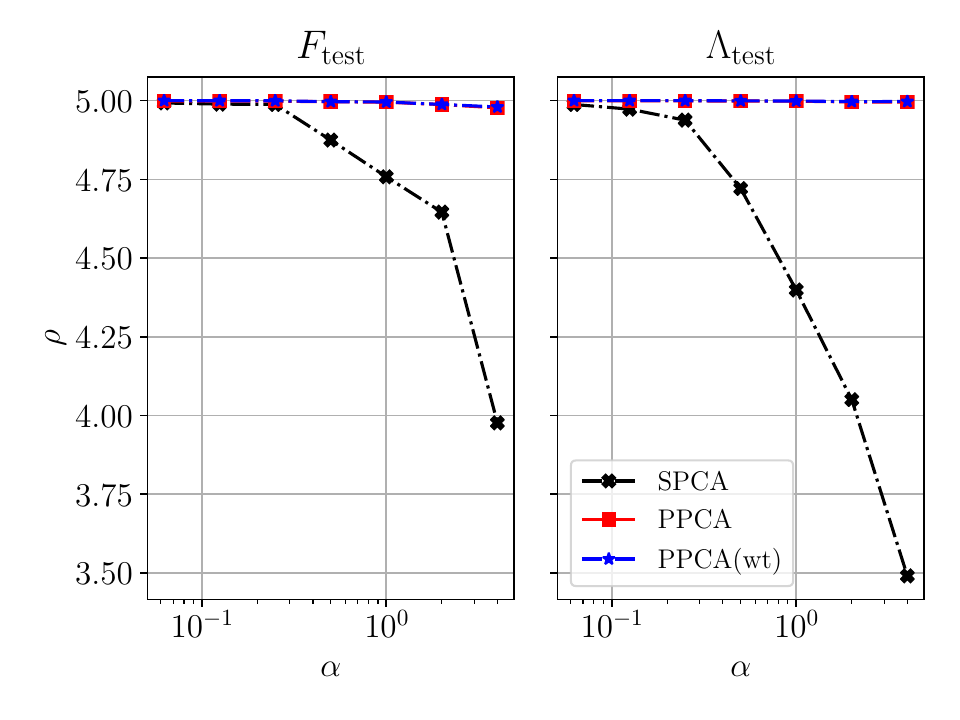}
					\subcaptab{$\rho$ with $\hat{\*F}^\twt$/$\hat{\*\Lambda}^\twt$}
				\end{subfigure}%
				\begin{subfigure}{.235\textwidth}
					\centering
					\includegraphics[width=1\linewidth]{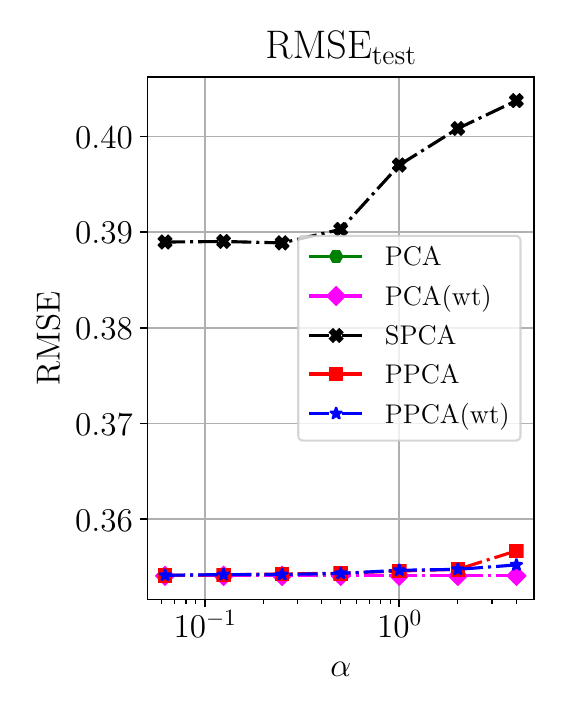}
					\subcaptab{RMSE}
				\end{subfigure}
			\end{adjustwidth}
			\label{fig:370-port-rho-rmse}
			\bnotefig{This figure compares the out-of-sample generalized correlations for factors and loadings and out-of-sample RMSE for proximate factors (PPCA), weighted proximate factors (PPCA (wt)), sparse PCA (SPCA), weighted PCA (PCA (wt)) and unweighted PCA. PPCA (wt) and PCA (wt) use the inverse standard errors as weights. In order to achieve the same sparsity level for various methods, we first choose $\alpha$, the $\ell_1$ penalty for SPCA, and set the number of nonzero weights $m_j$ for each factor in PPCA and PPCA (wt) to obtain the same number of nonzero elements in each factor as SPCA. The left figure shows the generalized correlation of the factors and loadings with the PCA estimates  $\hat{\*F}$/loadings $\hat{\*\Lambda}$. The middle figure show the corresponding generalized correlations with weighted PCA estimates $\hat{\*F}^\twt$/loadings $\hat{\*\Lambda}^\twt$. The right figure displays the RMSE for all five methods. PCA, PCA (wt), PPCA and PPCA (wt) achieve very similar performance and significantly outperform SPCA.} 
		\end{figure}
		
		Figure \ref{fig:370-port-rho-rmse} compares proximate factors with sparse PCA using the same metrics as in Section \ref{subsec:sim-comp-wtih-spca}. As the population factors and loadings are unknown, we compare the proximate and sparse factors and loadings with the non-sparse PCA factors and loadings. We also include the weighted PCA factors $\hat{\*F}^\twt$ with the inverse standard deviation of the residuals as weights.\footnote{We use the residuals from a 5-factor model to estimate the standard deviation of the errors. The results are robust to this choice.} Both, unweighted PCA factors $\hat{\*F}$ and weighted PCA factors $\hat{\*F}^\twt$ are consistent estimators of the population factors. However, in the presence of heteroskedastic errors, we expect the non-sparse weighted PCA estimates to be more efficient and hence to be the more appropriate comparison benchmark. The first half of the time-series serves as the training and the second half as the test data set. The number of nonzero elements $m_j$ for each proximate factor is set equal to the number of nonzero elements of the corresponding SPCA factor for a given $\ell_1$-lasso penalty $\alpha$.\footnote{We only report the out-of-sample results in the main text. The Internet Appendix collects the in-sample results with very similar findings. An alternative approach to implement sparse PCA is to directly impose the cardinality constraint $\norm{\*\Lambda_j}_0 \leq m$  rather than using the $\ell_1$ penalty term $\alpha \sum_{j = 1}^{K} \norm{\*\Lambda_j}_1$ in the optimization problem \eqref{eqn:spca-obj}, where $\norm{\*\Lambda_j}_0 $ equals the number of nonzero elements in $\*\Lambda_j$. \cite{sigg2008expectation} develop an expectation-maximization (EM) algorithm based on a probabilistic expression of PCA to solve this optimizaton problem with constraint  $\norm{\*\Lambda_j}_0 \leq m$.  This approach will in general return sparse loadings with $m$ nonzero elements in each sparse loading vector when $m$ is reasonably small. The Internet Appendix collects the results for this alternative implementation with very similar findings as in the main text.} 
		
		First, unweighted PCA factors and loadings are almost identical to the weighted PCA factors and loadings as the residuals have only a low degree of heteroskedasticity. As the two sets of proximate factors are very close approximations of the corresponding PCA factors, PPCA and PPCA (wt) are also almost identical with a slightly lower RMSE for PPCA (wt) out-of-sample. Second, sparse PCA performs substantially worse. For a higher degree of sparsity, the SPCA factors are neither close to the unweighted nor the weighted PCA factors. In particular, the sparse loadings are miss-specified, resulting in a higher out-of-sample RMSE for all levels of sparsity.

		Last but not least, we show how the theoretical lower bound in Theorem \ref{thm-evt-multi-factor} and Proposition \ref{prop:evt-multi-factor-simplified} can serve as a guidance to select the number of nonzero $m$ in factor weights. We consider the case where each proximate factor can have a different number of nonzero entries $m_k$. For analytical tractability we use the simple model where we assume that population loadings are i.i.d normally distributed, but it turns out that this simple model already provides a very accurate description capturing the relevant trade-offs. More specifically, we assume that $\Lambda_{i,k} \stackrel{\iid}{\sim} \mathcal{N}(\mu_{\Lambda, k}, \sigma_{\Lambda,k}^2)$ and estimate the mean $\hat{\mu}_{\Lambda,k} = \frac{1}{N} \sum_{i = 1}^{N} \hat\Lambda_{i,k}$ and sample variance $\hat{\sigma}^2_{\Lambda, k} -1 - \hat{\mu}_{\Lambda,k}^2$ following from $\hat{\Lambda}^\T\hat{\Lambda}/N = I_K$. The parameters of the GEV distribution are only a function of $\hat{\mu}_{\Lambda,k}$ and $\hat{\sigma}^2_{\Lambda, k}$ and known in closed-form as specified in Table \ref{tab:extreme-value-example}. For the given target probability $\underline{p}=0.95\%$ we solve for $y_{k,m_{k}}$ in $\underline{p} := \prod_{k = 1}^{K} \bar{G}_{k,m_k}(y_{k,m_{k}})$. We want each factor to contribute equally to the generalized correlation and set $\bar{G}_{k,m_k}(y_{k,m_{k}}) = \underline{p}^{\frac{1}{K} }$. Hence, we obtain 
		\begin{align*}
		y_{k,m_{k}} = \bar{G}^{-1}_{k,m_K}\big(\underline{p}^{\frac{1}{K} } \big), \quad \text{ where }  \bar{G}_{k,m_k}(y_{k,m_{k}}) = 1 - e^{-\tau} \cdot  \sum_{j = 0}^{m-1} \frac{\tau^j}{j!} \text{ and } \tau = \exp \Big(- \frac{y_{k,m_{k}} - b_{k,N} }{a_{k,N} }   \Big). 
		\end{align*}
		The theoretical lower bound for $\rho$ equals $K - \sum_{k=1}^K \frac{{\sigma}_{e}^2}{ m_k  \sigma_{\*F_k}^2 y_{k,m_{k}}^2 }$ which holds with probability $\underline{p}$. The noise variance is estimated by $\hat{\sigma}_{e}^2 = \frac{1}{NT} \sum_{i, t} \big(X_{it} - \hat{\Lambda}_i^\T \hat{F}_t \big)^2$ and factor variances $\sigma_{\*F_k}^2$ are estimated by the largest sample eigenvalues.
		
				\begin{figure}[t!]
			\centering
			\tcapfig{Financial Portfolio Data: Theoretical Lower Bound as A Guidance}
			\includegraphics[width=0.8\linewidth]{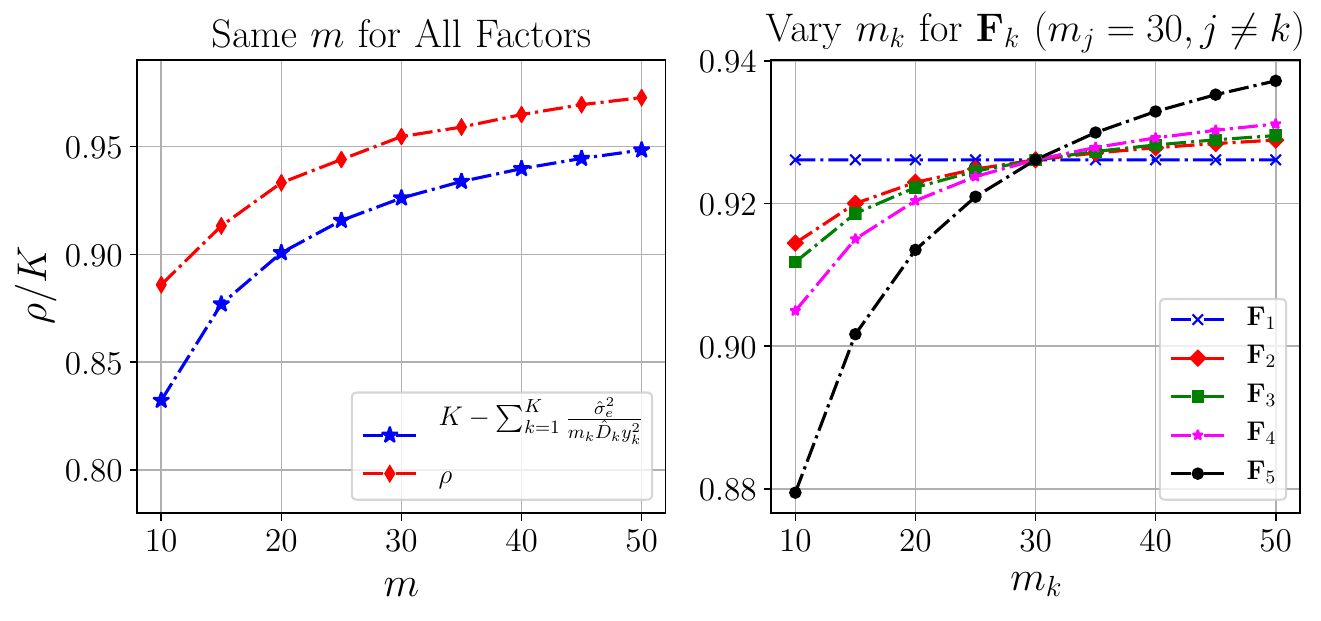}
			\bnotefig{
				The left figure shows the normalized generalized correlation $\rho/K$ of proximate factors with PCA factors and theoretical lower bound that is achieved with $\underline{p}=95\%$ probability. The normalized theoretical lower bound equals $1 - \frac{1}{K} \sum_{k=1}^K \frac{\hat{\sigma}_{e}^2}{ m_k \hat{\sigma}_{\*F_k}^2 y_{k,m_{k}}^2 }$. The left figures uses the same sparsity level $m$ for all 5 factors. The right figure displays the normalized theoretical lower bound as function of $m_k$ for factor $\*F_k$ while the sparsity level of the other factors is set to $m=30$. First, the theoretical lower bound captures the same trade-offs as the actual generalized correlation. Second, stronger factors can be well approximated with less nonzero weights $m_k$.}
			\label{fig:370-port-prob-bound}
		\end{figure}
		
		Figure \ref{fig:370-port-prob-bound} shows the theoretical bound as a function of $m_k$. First, the theoretical lower bound captures the same trade-offs as the generalized correlation with the PCA factors. In particular, the slopes of the curves in the left subfigure are almost identical and indicate that the largest gains can be achieved by adding the first 30 nonzero elements. In the right subfigure, we study the partial effect of changing the sparsity $m_k$ of the $k$-th factor while keeping the sparsity level of the other factors constant. The results illustrate that weaker factors which explain less variation require more nonzero weights while stronger factors can already be approximated well with very sparse proximate factors. These theoretical predictions support the empirical findings in Table \ref{tab:single-sorted-portfolios}.

		\subsection{Macroeconomic Data}
		
		We study 128 monthly U.S. macroeconomic indicators from 01/1959 to 02/2018 as in \citep{mccracken2016fred}. This data is from the Federal Reserve Economic Data (FRED) and is publicly available.\footnote{This dataset is updated in a timely manner and is available at \url{https://research.stlouisfed.org/econ/mccracken/fred-databases/}} The 128 macroeconomic indicators can be classified into 8 groups: 1. output and income; 2. labor market; 3. housing; 4. consumption, orders, and inventories; 5. money and credit; 6. interest and exchange rates; 7. prices; 8. stock market \citep{mccracken2016fred}. Macroeconomic datasets have been successfully analyzed with PCA factors for forecasting and estimation of factor augmented regressions \citep{stock2002macroeconomic, boivin2005understanding}. \cite{mccracken2016fred} estimate 8 factors in this macroeconomic data.

		Figure \ref{fig:macroeconomic} plots the generalized correlation between $\tilde{\*F}$ and $\hat{\*F}$ normalized by $K$ and the proportion of variance explained with different number of factors. When the number of nonzero entries $m$ in each sparse factor weight vector is around 10, which is less than 8\% of all macroeconomic variables, the ratio of $\rho$ to $K$ is close to 0.9, which corresponds to an average correlation of 0.95\%. The amount of variation explained by $\tilde{\*F}$ and $\hat{\*F}$ also indicate that $\tilde{\*F}$ consisting of around 15 macroeconomic variables can approximate $\hat{\*F}$ very well. 
		
		In the following, we study the 8-factor model in more detail. Table \ref{tab:macroeconomic} shows the generalized correlations for each factor which correspond to the $R^2$ in regressions of each PCA factor on all proximate factors. The first five factors can be captured well by proximate factors with 10 macroeconomic variables in each proximate factors. The sixth to eighth factors are weaker, but they can be reasonably captured by proximate factors with around 20 variables. 
		
		\begin{figure}[t!]
			\centering
			\tcapfig{Macroeconomic Data: Generalized Correlation and Explained Variance }
			\begin{subfigure}{.5\textwidth}
				\centering
				\includegraphics[width=0.8\linewidth]{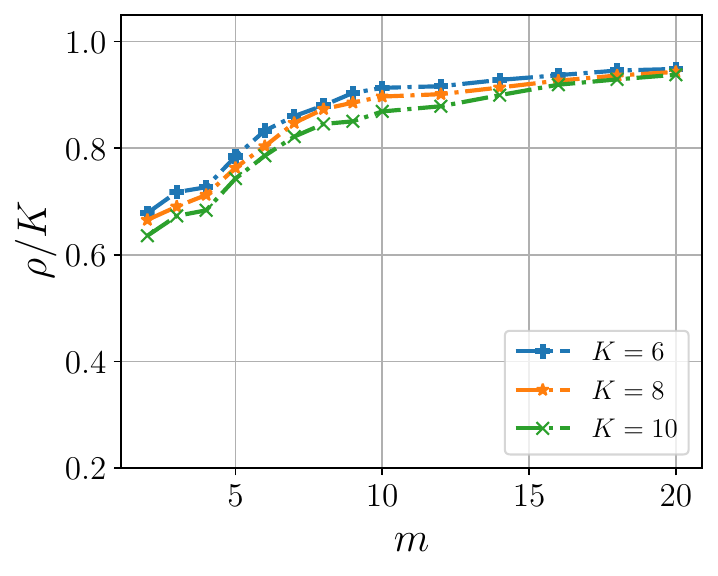}
				\subcaptab{Generalized Correlations}
			\end{subfigure}%
			\begin{subfigure}{.5\textwidth}
				\centering
				\includegraphics[width=0.82\linewidth]{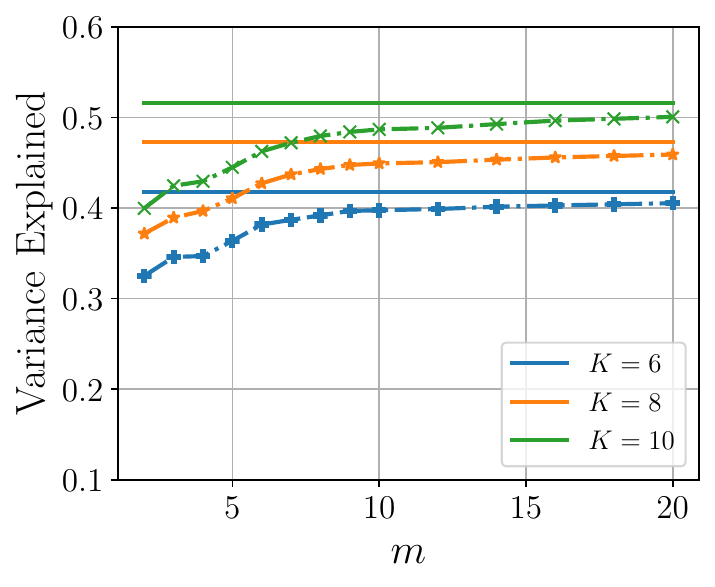}
				\subcaptab{Variance Explained}
			\end{subfigure}
			\bnotefig{The left figure shows the normalized generalized correlation $\rho/K$  between $\tilde{\*F}$ and $\hat{\*F}$ as a function of $m$. Th right figure shows the proportion of variance explained by $\tilde{\*F}$ and $\hat{\*F}$ as a function of  $m$. $N$ is 128, $T$ is 707 and $K$ varies from 6 to 8. Proximate factors constructed with only 10\% of the variables are very close to the non-sparse PCA factors.}
			\label{fig:macroeconomic}
		\end{figure}

		\begin{table}[t!]
			\centering
			\tcaptab{Macroeconomic Data: Generalized Correlation Decomposition}
			\begin{tabular}{ccccccccc}
				\toprule
				\diagbox{$m$}{} & $\hat{\*F}_1$ & $\hat{\*F}_2$ & $\hat{\*F}_3$ & $\hat{\*F}_4$ & $\hat{\*F}_5$ & $\hat{\*F}_6$  & $\hat{\*F}_7$ & $\hat{\*F}_8$ \\
				\midrule
				10 & 0.953 & 0.959 & 0.949 & 0.953 & 0.961 & 0.799 & 0.833 & 0.767\\
				15 & 0.967 & 0.970 & 0.958 & 0.956 & 0.964 & 0.857 & 0.867 & 0.837\\
				20 & 0.977 & 0.974 & 0.957 & 0.963 & 0.961 & 0.905 & 0.919 & 0.891\\
				25 & 0.983 & 0.980 & 0.961 & 0.979 & 0.973 & 0.937 & 0.943 & 0.929\\
				\bottomrule
			\end{tabular}
			\bnotetab{This table shows the generalized correlation between each $\hat{\*F}_j$ and all $\tilde{\*F}$ in a model with 8 factors. $N$ is 128 and $T$ is 707. These generalized correlations correspond to the $R^2$ from a regression of each $\hat{\*F}_j$ on all $\tilde{\*F}$. The sum of each row equals $\rho$. In general, stronger factors have larger signal-to-noise ratio and can be better approximated with a fixed $m$. }
			\label{tab:macroeconomic}
		\end{table}

		\begin{figure}[t!]
			\centering
			\tcapfig{Macroeconomic Data: Composition of Proximate Factors}
			\includegraphics[width=0.4\linewidth]{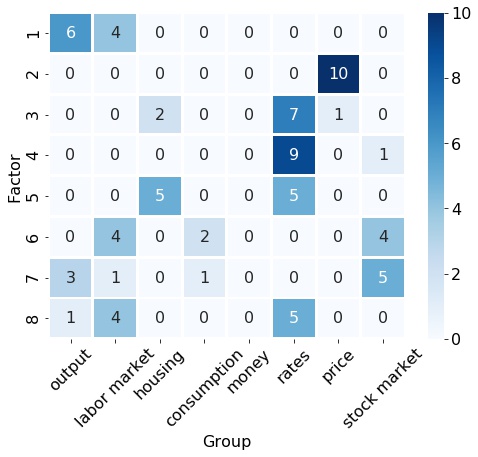}
			\bnotefig{This figure shows the composition of each of the 8 proximate factor in a particular group. Each proximate factor is constructed from 10 macroeconomic variables. The 8 groups are: 1. output and income; 2. labor market; 3. housing; 4. consumption, orders and inventories; 5. money and credit; 6. interest and exchange rates; 7. prices; 8. stock market.}
			\label{fig:macroeconomic-loading}
		\end{figure}
		
		Figure \ref{fig:macroeconomic-loading} shows a clear interpretation of the statistical factors. When each hard-thres\-holded loading has 10 nonzero entries, we see a strong pattern in the composition of each latent sparse factor that allows us to label each one. The first factor is a labor and productivity factor; the second factor is a price factor; both third and fourth factors are interest and exchange rate factors; the fifth factor is a housing and rate factor; the sixth factor is a labor and stock market factor; the seventh factor is a productivity and stock market factor; the eighth factor is a labor and rate factor. Macroeconomic variables in groups of output, income, labor market, and prices are the main compositions of the first and second factors. The third, fourth, fifth, and eighth latent factors are mainly composed of macroeconomic variables in the group of interest and exchange rates.  Variances explained by latent factors are in decreasing order. Thus, variables in groups of  output, income, labor market, prices, interest and exchange rates explain most of the variation in the whole dataset.

		\begin{figure}[t!]
			\tcapfig{Macroeconomic Data: Out-of-Sample Generalized Correlations and RMSE}
			\begin{adjustwidth}{-1cm}{}
				\centering
				\begin{subfigure}{.4\textwidth}
					\centering
					\includegraphics[width=1\linewidth]{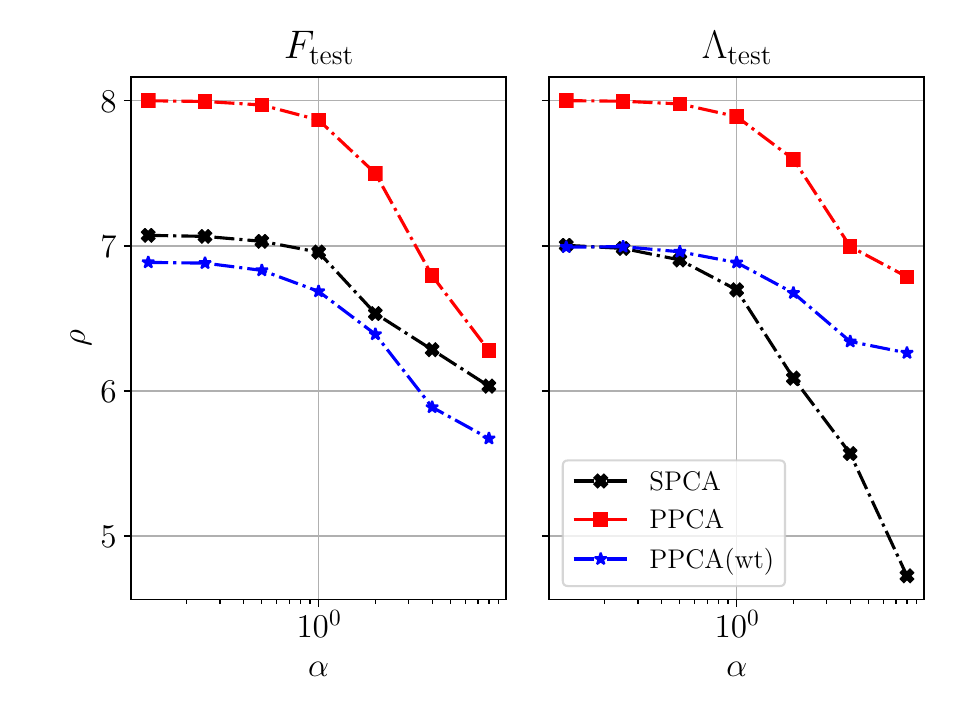}
					\caption{$\rho$ with $\hat{\*F}$/$\hat{\*\Lambda}$}
				\end{subfigure}%
				\begin{subfigure}{.4\textwidth}
					\centering
					\includegraphics[width=1\linewidth]{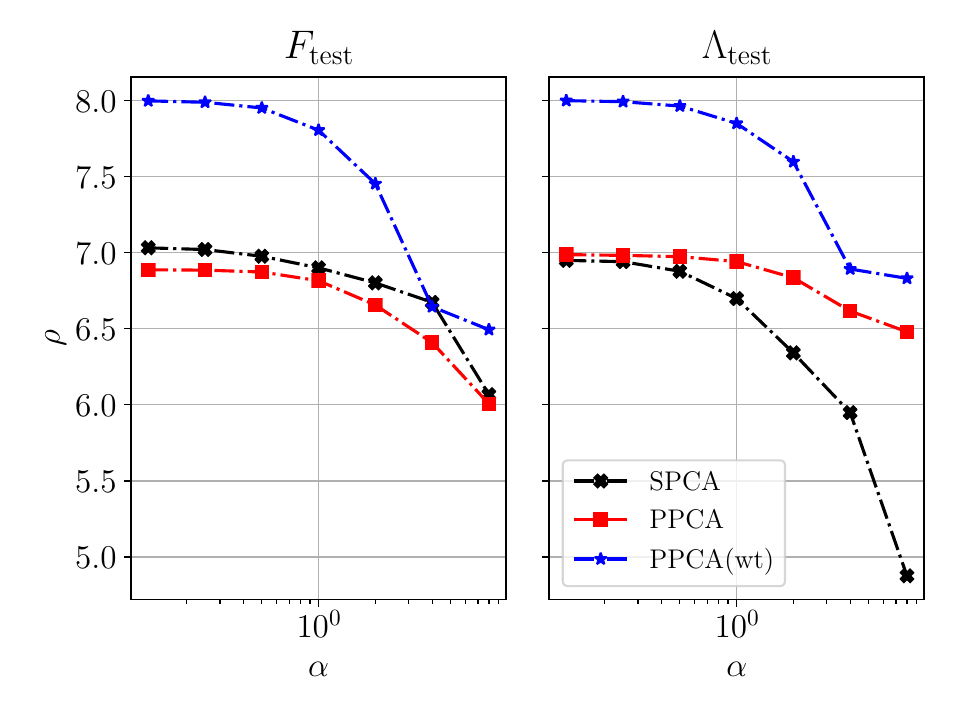}
					\caption{$\rho$ with $\hat{\*F}^\twt$/$\hat{\*\Lambda}^\twt$}
				\end{subfigure}%
				\begin{subfigure}{.235\textwidth}
					\centering
					\includegraphics[width=1\linewidth]{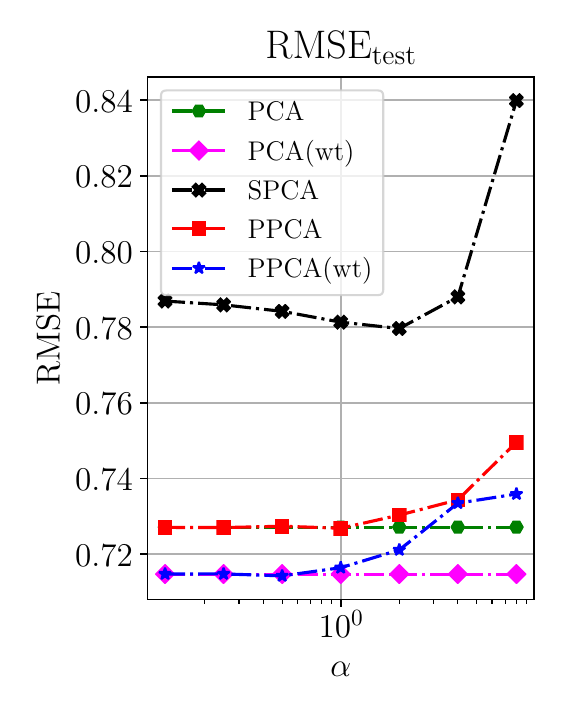}
					\caption{RMSE}
				\end{subfigure}
			\end{adjustwidth}
			\bnotefig{This figure compares the out-of-sample generalized correlations for factors and loadings and out-of-sample RMSE for proximate factors (PPCA), weighted proximate factors (PPCA (wt)), sparse PCA (SPCA), weighted PCA (PCA (wt)) and unweighted PCA. PPCA (wt) and PCA (wt) use the inverse standard errors as weights. In order to achieve the same sparsity level for various methods, we first choose $\alpha$, the $\ell_1$ penalty for SPCA, and set the number of nonzero weights $m_j$ for each factor in PPCA and PPCA (wt) to obtain the same number of nonzero elements in each factor as SPCA. The left figure shows the generalized correlation of the factors and loadings with the PCA estimates  $\hat{\*F}$/loadings $\hat{\*\Lambda}$. The middle figure show the corresponding generalized correlations with weighted PCA estimates $\hat{\*F}^\twt$/loadings $\hat{\*\Lambda}^\twt$. The right figure displays the RMSE for all five methods. PCA, PCA (wt), PPCA and PPCA (wt) achieve very similar performance and significantly outperform SPCA. }
			\label{fig:fred-md-rho-rmse}
		\end{figure}
		
		\begin{figure}[t!]
			\centering
			\tcaptab{Macroeconomic Data: Theoretical Lower Bound as A Guidance}
			\includegraphics[width=0.8\linewidth]{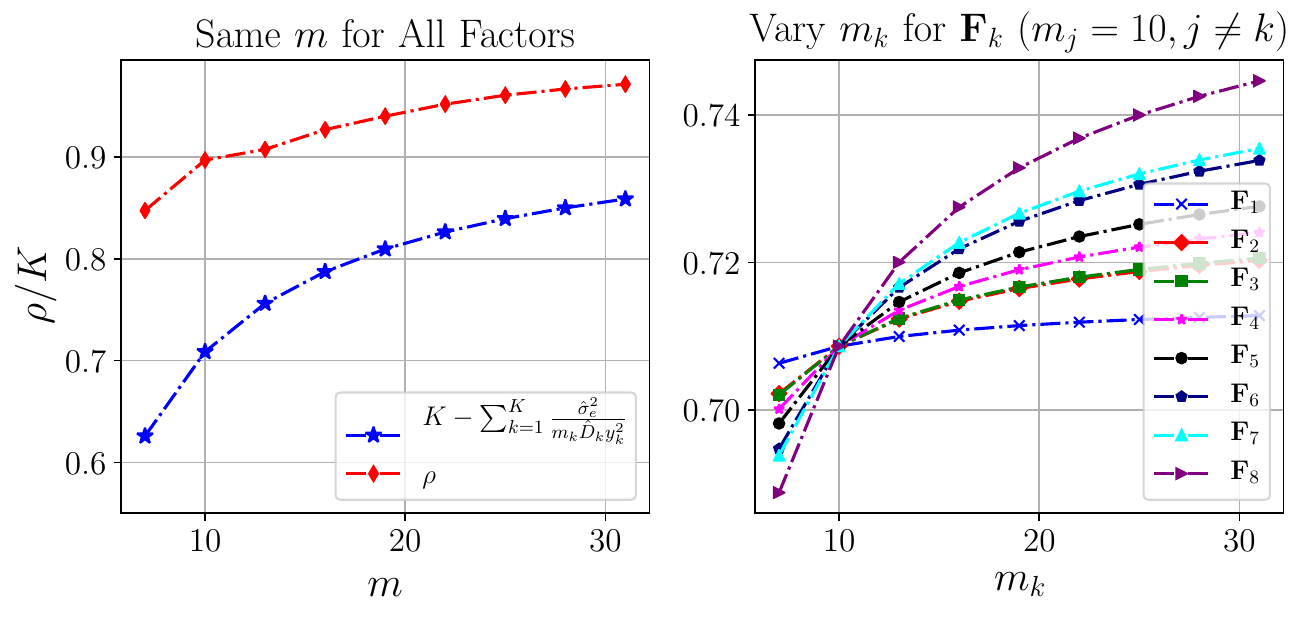}
			\bnotefig{The left figure shows the normalized generalized correlation $\rho/K$ of proximate factors with PCA factors and theoretical lower bound that is achieved with $\underline{p}=95\%$ probability. The normalized theoretical lower bound equals $1 - \frac{1}{K} \sum_{k=1}^K \frac{\hat{\sigma}_{e}^2}{ m_k \hat{\sigma}_{\*F_k}^2 y_{k,m_{k}}^2 }$. The left figures uses the same sparsity level $m$ for all 5 factors. The right figure displays the normalized theoretical lower bound as function of $m_k$ for factor $\*F_k$ while the sparsity level of the other factors is set to $m=10$. First, the theoretical lower bound captures the same trade-offs as the actual generalized correlation. Second, stronger factors can be well approximated with less nonzero weights $m_k$.}
			\label{fig:fred-md-prob-bound}
		\end{figure}
		
		Figures \ref{fig:fred-md-rho-rmse} and \ref{fig:fred-md-prob-bound} apply the same analysis as in Figures \ref{fig:370-port-rho-rmse} and \ref{fig:370-port-prob-bound} to the macroeconomic data with similar insights. The key difference is that the macroeconomic time-series are more heterogeneous. Hence, the weighted PCA achieves a lower RMSE out-of-sample than the unweighted PCA. Note that as expected the weighted proximate factors are more highly correlated with the weighted PCA factors while unweighted proximate factors mimic the unweighted PCA factors. This is also reflected in the RMSE. Sparse PCA overall performs significantly worse as expected. The theoretical lower bound provides similar guidance as in the portfolio case. Note that in the macroeconomic data the factors are generally weaker and explain less variation than in the portfolio data which leads to a larger gap between the empirical and theoretical generalized correlation bound. Overall the empirical results justify why we can use the proximate factors as a replacement for standard PCA factors.

		\section{Conclusion}\label{sec:conclusion}
		
		In this paper, we propose a method to construct proximate factors that consist of only a small number of cross-section units and can approximate latent factors well. These proximate factors are usually much easier to interpret than PCA factors. The closeness between proximate factors and latent factors is measured by the generalized correlation. We provide an asymptotic probabilistic lower bound for the generalized correlation based on extreme value theory. This lower bound explains why proximate factors are close to latent factors and provides guidance on how to chose the sparsity in the proximate factors. Simulations verify that the lower bound closely approximates the exceedance probability of the generalized correlation, especially in the most relevant case of an exceedance probability close to 1. The proximate factors have non-sparse loadings which are consistent estimates of the true population loadings. Empirical applications to two different datasets, financial portfolios, and macroeconomic data, show that proximate factors consisting of 5-10\% of the cross-section units can approximate latent factors well. Using these proximate factors, we provide a meaningful interpretation of the latent factors in our datasets. 
		
	\end{onehalfspacing}

	\singlespacing
	\bibliographystyle{econometrica}
	{\footnotesize
		\bibliography{reference}

@article{lettau2020theory,
	Author = {Lettau, Martin and Pelger, Markus},
	Journal = {Journal of Econometrics},
	Number = {1},
	Pages = {1-31},
	Title = {Estimating Latent Asset Pricing Factors},
	Volume = {218},
	Year = {2020}}

@article{xiong2020,
	Author = {Xiong, R. and Pelger, M.},
	Journal = {Working paper},
	Title = {Large Dimensional Latent Factor Modeling with Missing Observations and Applications to Causal Inference},
	Year = {2020}}

@article{pelger2019,
	Author = {Pelger, M.},
	Journal = {Journal of Finance},
	Number = {4},
	Pages = {2179-2220},
	Title = {Understanding Systematic Risk: A High-Frequency Approach},
	Volume = {75},
	Year = {2020}}

@article{jones2001extracting,
	Author = {Jones, Christopher S},
	Journal = {Journal of Financial economics},
	Number = {2},
	Pages = {293--325},
	Publisher = {Elsevier},
	Title = {Extracting factors from heteroskedastic asset returns},
	Volume = {62},
	Year = {2001}}

@article{boivin2006more,
	Author = {Boivin, Jean and Ng, Serena},
	Journal = {Journal of Econometrics},
	Number = {1},
	Pages = {169--194},
	Publisher = {Elsevier},
	Title = {Are more data always better for factor analysis?},
	Volume = {132},
	Year = {2006}}

@article{bai2017,
	Author = {Bai, Jushan and Ng, Serena},
	Journal = {Journal of econometrics},
	Number = {1},
	Pages = {78--96},
	Publisher = {Elsevier},
	Title = {Rank regularized estimation of approximate factor models},
	Volume = {212},
	Year = {2019}}

@article{zou2005regularization,
	Author = {Zou, Hui and Hastie, Trevor},
	Journal = {Journal of the Royal Statistical Society: Series B (Statistical Methodology)},
	Number = {2},
	Pages = {301--320},
	Publisher = {Wiley Online Library},
	Title = {Regularization and variable selection via the elastic net},
	Volume = {67},
	Year = {2005}}

@article{kelly2018characteristics,
	Author = {Kelly, Bryan T and Pruitt, Seth and Su, Yinan},
	Journal = {Journal of Financial Economics},
	Number = {3},
	Pages = {501--524},
	Publisher = {Elsevier},
	Title = {Characteristics are covariances: A unified model of risk and return},
	Volume = {134},
	Year = {2019}}

@article{jolliffe2003modified,
	Author = {Jolliffe, Ian T and Trendafilov, Nickolay T and Uddin, Mudassir},
	Journal = {Journal of computational and Graphical Statistics},
	Number = {3},
	Pages = {531--547},
	Publisher = {Taylor \& Francis},
	Title = {A modified principal component technique based on the LASSO},
	Volume = {12},
	Year = {2003}}

@article{hsing1988extreme,
	Author = {Hsing, Tailen},
	Journal = {Advances in Applied Probability},
	Number = {1},
	Pages = {11--11},
	Publisher = {Cambridge University Press},
	Title = {On the extreme order statistics for a stationary sequence},
	Volume = {20},
	Year = {1988}}

@article{kaiser1958varimax,
	Author = {Kaiser, Henry F},
	Journal = {Psychometrika},
	Number = {3},
	Pages = {187--200},
	Publisher = {Springer},
	Title = {The varimax criterion for analytic rotation in factor analysis},
	Volume = {23},
	Year = {1958}}

@article{efron2004least,
	Author = {Efron, Bradley and Hastie, Trevor and Johnstone, Iain and Tibshirani, Robert and others},
	Journal = {The Annals of statistics},
	Number = {2},
	Pages = {407--499},
	Publisher = {Institute of Mathematical Statistics},
	Title = {Least angle regression},
	Volume = {32},
	Year = {2004}}

@book{friedman2001elements,
	Author = {Friedman, Jerome and Hastie, Trevor and Tibshirani, Robert},
	Number = {10},
	Publisher = {Springer series in statistics New York, NY, USA:},
	Title = {The elements of statistical learning},
	Volume = {1},
	Year = {2001}}

@article{lettau2020factors,
	Author = {Lettau, Martin and Pelger, Markus},
	Journal = {The Review of Financial Studies},
	Number = {5},
	Pages = {2274--2325},
	Publisher = {Oxford University Press},
	Title = {Factors that fit the time series and cross-section of stock returns},
	Volume = {33},
	Year = {2020}}

@article{mccracken2016fred,
	Author = {McCracken, Michael W and Ng, Serena},
	Journal = {Journal of Business \& Economic Statistics},
	Number = {4},
	Pages = {574--589},
	Publisher = {Taylor \& Francis},
	Title = {FRED-MD: A monthly database for macroeconomic research},
	Volume = {34},
	Year = {2016}}

@techreport{boivin2005understanding,
	Author = {Boivin, Jean and Ng, Serena},
	Institution = {National Bureau of Economic Research},
	Title = {Understanding and comparing factor-based forecasts},
	Year = {2005}}

@article{kaufmann2013bayesian,
	Author = {Kaufmann, Sylvia and Schumacher, Christian},
	Journal = {Journal of Econometrics},
	Number = {1},
	Pages = {116--134},
	Publisher = {Elsevier},
	Title = {Bayesian estimation of sparse dynamic factor models with order-independent and ex-post mode identification},
	Volume = {210},
	Year = {2019}}

@article{pelger2018,
	Author = {Pelger, Markus},
	Journal = {Journal of Econometrics},
	Number = {1},
	Pages = {23-42},
	Title = {Large-Dimensional Factor Modeling Based on High-Frequency Observations},
	Volume = {208},
	Year = {2019}}

@article{xiu2018,
	Author = {A{\"\i}t-Sahalia, Yacine and Xiu, Dacheng},
	Journal = {Journal of the American Statistical Association},
	Number = {525},
	Pages = {287--303},
	Publisher = {Taylor \& Francis},
	Title = {Principal component analysis of high-frequency data},
	Volume = {114},
	Year = {2019}}

@article{tibshirani1996regression,
	Author = {Tibshirani, Robert},
	Journal = {Journal of the Royal Statistical Society. Series B (Methodological)},
	Pages = {267--288},
	Publisher = {JSTOR},
	Title = {Regression shrinkage and selection via the lasso},
	Year = {1996}}

@article{fan2016robust,
	Author = {Fan, Jianqing and Ke, Yuan and Liao, Yuan},
	Journal = {arXiv preprint arXiv:1603.07041},
	Title = {Robust Factor Models with Explanatory Proxies},
	Year = {2016}}

@article{connor2012efficient,
	Author = {Connor, Gregory and Hagmann, Matthias and Linton, Oliver},
	Journal = {Econometrica},
	Number = {2},
	Pages = {713--754},
	Publisher = {Wiley Online Library},
	Title = {Efficient semiparametric estimation of the Fama--French model and extensions},
	Volume = {80},
	Year = {2012}}

@article{bai2016efficient,
	Author = {Bai, Jushan and Liao, Yuan},
	Journal = {Journal of Econometrics},
	Number = {1},
	Pages = {1--18},
	Publisher = {Elsevier},
	Title = {Efficient estimation of approximate factor models via penalized maximum likelihood},
	Volume = {191},
	Year = {2016}}

@article{bai2008forecasting,
	Author = {Bai, Jushan and Ng, Serena},
	Journal = {Journal of Econometrics},
	Number = {2},
	Pages = {304--317},
	Publisher = {Elsevier},
	Title = {Forecasting economic time series using targeted predictors},
	Volume = {146},
	Year = {2008}}

@article{ludvigson2007,
	Author = {Ludvigson, Sydney C and Ng, Serena},
	Journal = {Journal of Financial Economics},
	Number = {1},
	Pages = {171--222},
	Title = {The empirical risk return relation: A factor analysis approach},
	Volume = {83},
	Year = {2007}}

@article{ludvigson2009,
	Author = {Ludvigson, Sydney C and Ng, Serena},
	Journal = {Review of Financial Studies},
	Number = {12},
	Pages = {5027--5067},
	Title = {Macro factors in bond risk premia},
	Volume = {22},
	Year = {2009}}

@article{fan2016projected,
	Author = {Fan, Jianqing and Liao, Yuan and Wang, Weichen},
	Journal = {Annals of statistics},
	Number = {1},
	Pages = {219},
	Publisher = {NIH Public Access},
	Title = {Projected principal component analysis in factor models},
	Volume = {44},
	Year = {2016}}

@article{connor2007semiparametric,
	Author = {Connor, Gregory and Linton, Oliver},
	Journal = {Journal of Empirical Finance},
	Number = {5},
	Pages = {694--717},
	Publisher = {Elsevier},
	Title = {Semiparametric estimation of a characteristic-based factor model of common stock returns},
	Volume = {14},
	Year = {2007}}

@article{bhattacharya2011sparse,
	Author = {Bhattacharya, Anirban and Dunson, David B},
	Journal = {Biometrika},
	Pages = {291--306},
	Title = {Sparse Bayesian infinite factor models},
	Year = {2011}}

@article{lucas2006sparse,
	Author = {Lucas, Joe and Carvalho, Carlos and Wang, Quanli and Bild, Andrea and Nevins, Joseph R and West, Mike},
	Journal = {Bayesian Inference for Gene Expression and Proteomics},
	Pages = {0--1},
	Publisher = {Cambridge University Press, Cambridge},
	Title = {Sparse statistical modelling in gene expression genomics},
	Volume = {1},
	Year = {2006}}

@article{pati2014posterior,
	Author = {Pati, Debdeep and Bhattacharya, Anirban and Pillai, Natesh S and Dunson, David and others},
	Journal = {The Annals of Statistics},
	Number = {3},
	Pages = {1102--1130},
	Publisher = {Institute of Mathematical Statistics},
	Title = {Posterior contraction in sparse Bayesian factor models for massive covariance matrices},
	Volume = {42},
	Year = {2014}}

@article{choi2010penalized,
	Author = {Choi, Jang and Oehlert, Gary and Zou, Hui},
	Journal = {Statistics and its Interface},
	Number = {4},
	Pages = {429--436},
	Publisher = {International Press of Boston},
	Title = {A penalized maximum likelihood approach to sparse factor analysis},
	Volume = {3},
	Year = {2010}}

@article{lan2014sparse,
	Author = {Lan, Andrew S and Waters, Andrew E and Studer, Christoph and Baraniuk, Richard G},
	Journal = {The Journal of Machine Learning Research},
	Number = {1},
	Pages = {1959--2008},
	Publisher = {JMLR. org},
	Title = {Sparse factor analysis for learning and content analytics},
	Volume = {15},
	Year = {2014}}

@article{kawano2015sparse1,
	Author = {Kawano, Shuichi and Fujisawa, Hironori and Takada, Toyoyuki and Shiroishi, Toshihiko},
	Journal = {Computational Statistics \& Data Analysis},
	Pages = {192--203},
	Publisher = {Elsevier},
	Title = {Sparse principal component regression with adaptive loading},
	Volume = {89},
	Year = {2015}}

@article{bai2003inferential,
	Author = {Bai, Jushan},
	Journal = {Econometrica},
	Number = {1},
	Pages = {135--171},
	Publisher = {Wiley Online Library},
	Title = {Inferential theory for factor models of large dimensions},
	Volume = {71},
	Year = {2003}}

@article{bai2002determining,
	Author = {Bai, Jushan and Ng, Serena},
	Journal = {Econometrica},
	Number = {1},
	Pages = {191--221},
	Publisher = {Wiley Online Library},
	Title = {Determining the number of factors in approximate factor models},
	Volume = {70},
	Year = {2002}}

@article{stock2002macroeconomic,
	Author = {Stock, James H and Watson, Mark W},
	Journal = {Journal of Business \& Economic Statistics},
	Number = {2},
	Pages = {147--162},
	Publisher = {Taylor \& Francis},
	Title = {Macroeconomic forecasting using diffusion indexes},
	Volume = {20},
	Year = {2002}}

@article{stock2002forecasting,
	Author = {Stock, James H and Watson, Mark W},
	Journal = {Journal of the American statistical association},
	Number = {460},
	Pages = {1167--1179},
	Publisher = {Taylor \& Francis},
	Title = {Forecasting using principal components from a large number of predictors},
	Volume = {97},
	Year = {2002}}

@article{candes2010power,
	Author = {Cand{\`e}s, Emmanuel J and Tao, Terence},
	Journal = {IEEE Transactions on Information Theory},
	Number = {5},
	Pages = {2053--2080},
	Publisher = {IEEE},
	Title = {The power of convex relaxation: Near-optimal matrix completion},
	Volume = {56},
	Year = {2010}}

@article{candes2011robust,
	Author = {Cand{\`e}s, Emmanuel J and Li, Xiaodong and Ma, Yi and Wright, John},
	Journal = {Journal of the ACM (JACM)},
	Number = {3},
	Pages = {11},
	Publisher = {ACM},
	Title = {Robust principal component analysis?},
	Volume = {58},
	Year = {2011}}

@article{fama1992cross,
	Author = {Fama, Eugene F and French, Kenneth R},
	Journal = {the Journal of Finance},
	Number = {2},
	Pages = {427--465},
	Publisher = {Wiley Online Library},
	Title = {The cross-section of expected stock returns},
	Volume = {47},
	Year = {1992}}

@article{diebold2006forecasting,
	Author = {Diebold, Francis X and Li, Canlin},
	Journal = {Journal of econometrics},
	Number = {2},
	Pages = {337--364},
	Publisher = {Elsevier},
	Title = {Forecasting the term structure of government bond yields},
	Volume = {130},
	Year = {2006}}

@article{mairal2010online,
	Author = {Mairal, Julien and Bach, Francis and Ponce, Jean and Sapiro, Guillermo},
	Journal = {Journal of Machine Learning Research},
	Number = {Jan},
	Pages = {19--60},
	Title = {Online learning for matrix factorization and sparse coding},
	Volume = {11},
	Year = {2010}}

@article{zou2006sparse,
	Author = {Zou, Hui and Hastie, Trevor and Tibshirani, Robert},
	Journal = {Journal of computational and graphical statistics},
	Number = {2},
	Pages = {265--286},
	Publisher = {Taylor \& Francis},
	Title = {Sparse principal component analysis},
	Volume = {15},
	Year = {2006}}

@article{pelger2018state,
	Author = {Pelger, Markus and Xiong, Ruoxuan},
	Journal = {Working paper},
	Title = {State-Varying Factor Models of Large Dimensions},
	Year = {2019}}

@book{anderson1958introduction,
	Author = {Anderson, Theodore Wilbur},
	Publisher = {Wiley New York},
	Title = {An introduction to multivariate statistical analysis},
	Volume = {2},
	Year = {1958}}

@book{coles2001introduction,
	Author = {Coles, Stuart and Bawa, Joanna and Trenner, Lesley and Dorazio, Pat},
	Publisher = {Springer},
	Title = {An introduction to statistical modeling of extreme values},
	Volume = {208},
	Year = {2001}}

@article{bai2006confidence,
	Author = {Bai, Jushan and Ng, Serena},
	Journal = {Econometrica},
	Number = {4},
	Pages = {1133--1150},
	Publisher = {Wiley Online Library},
	Title = {Confidence intervals for diffusion index forecasts and inference for factor-Augmented regressions},
	Volume = {74},
	Year = {2006}}

@article{andreou2019inference,
	Author = {Andreou, Elena and Gagliardini, Patrick and Ghysels, Eric and Rubin, Mirco},
	Journal = {Econometrica},
	Number = {4},
	Pages = {1267--1305},
	Publisher = {Wiley Online Library},
	Title = {Inference in Group Factor Models With an Application to Mixed-Frequency Data},
	Volume = {87},
	Year = {2019}}

@article{fan2013large,
	Author = {Fan, Jianqing and Liao, Yuan and Mincheva, Martina},
	Journal = {Journal of the Royal Statistical Society: Series B (Statistical Methodology)},
	Number = {4},
	Pages = {603--680},
	Publisher = {Wiley Online Library},
	Title = {Large covariance estimation by thresholding principal orthogonal complements},
	Volume = {75},
	Year = {2013}}

@article{ancona2000comparison,
	Author = {Ancona-Navarrete, Miguel A and Tawn, Jonathan A},
	Journal = {Extremes},
	Number = {1},
	Pages = {5--38},
	Publisher = {Springer},
	Title = {A comparison of methods for estimating the extremal index},
	Volume = {3},
	Year = {2000}}

@article{leadbetter1982extremes,
	Author = {Leadbetter, M Ross},
	Journal = {Working paper},
	Title = {Extremes and local dependence in stationary sequences.},
	Year = {1982}}

@inproceedings{sigg2008expectation,
  title={Expectation-maximization for sparse and non-negative PCA},
  author={Sigg, Christian D and Buhmann, Joachim M},
  booktitle={Proceedings of the 25th international conference on Machine learning},
  pages={960--967},
  year={2008}
}

@article{sigg2018package,
	title={Package ‘nsprcomp’},
	author={Sigg, Christian and Sigg, Maintainer Christian},
	year={2018}
}

@article{kozak2017shrinking,
	Author = {Kozak, Serhiy and Nagel, Stefan and Santosh, Shrihari},
	Journal = {Journal of Financial Economics},
	Number = {2},
	Pages = {271--292},
	Publisher = {Elsevier},
	Title = {Shrinking the cross-section},
	Volume = {135},
	Year = {2020}}
	}
	
	
	

	
	\appendix
	
	\renewcommand{\thetheorem}{A.\arabic{theorem}}%
	\renewcommand{\theproposition}{A.\arabic{proposition}}%
	\renewcommand{\thelemma}{A.\arabic{lemma}}%
	\renewcommand{\theassumption}{A.\arabic{assumption}}%
	
	\renewcommand{\theequation}{A.\arabic{equation}}\setcounter{equation}{0}
	\renewcommand{\thefigure}{A.\arabic{figure}} \setcounter{figure}{0}
	\renewcommand{\thetable}{A.\arabic{table}} \setcounter{table}{0}

	\setcounter{lemma}{0}
	\setcounter{proposition}{0}
	\setcounter{theorem}{0}
	\setcounter{assumption}{0}
	
	{\small

		\vspace{0.8cm}
		
		
		\begin{center}
			{\LARGE\bf{Appendix}}
		\end{center}

		\section{Extreme Value Theory for Dependent Data}
		In this section, we provide the asymptotic distribution of the $m$-th order statistic of the $k$-th factor loading $|\Lambda_{(m),k}|$ for cross-sectionally dependent loadings. We use these results to provide an asymptotic bound in Proposition \ref{prop:GEV} and \ref{prop:evt-multi-factor-simplified}.
		
		We assume $\{|\Lambda_{i,k}|\}$ is a strictly stationary sequence which satisfies the strong mixing condition (also known as $\alpha$-mixing). We divide the sample of size $N$ into blocks of length $r_N$, where $r_N = o(N)$. 
		For a given level $u_{k,N}(\tau)$ we count the number of $|\Lambda_{i,k}| $ exceeding this level in a block conditional on that there is at least one exceedance in this block. Formally, we define the cluster size distribution $\pi_{k,N}(i;\tau)$ as    
		\begin{align}
		\pi_{k,N}(i;\tau) = P\left[\sum_{j=1}^{r_N} \mathbbm{1}\left(|\Lambda_{j,k}| > u_{k,N}(\tau)\right) = i \Big| \sum_{j=1}^{r_N} \mathbbm{1}\left(|\Lambda_{j,k}| > u_{k,N}(\tau)\right) > 0  \right]. \label{eqn:cluster}
		\end{align}
		The following lemma provides the connection between the convergence of $\pi_{k,N}(i;\tau)$ and the convergence of the distribution of $|\Lambda_{(m),k}|$ which characterizes the asymptotic distribution of $|\Lambda_{(m),k}|$. This lemma is adapted from Theorem 3.3 in \cite{hsing1988extreme}.
		\begin{lemma}[Asymptotic Distribution of $|\Lambda_{(m),k}|$]\label{lemma:order-stats-dist}
			Define the cluster size distribution $\pi_{k,N}(i;\tau)$ as in equation \ref{eqn:cluster} for a given sequence $u_{k,N}(\tau)$ and sequence $r_N$ that satisfy $N/r_N \rightarrow \infty$, $ e^{N/r_N} \alpha(l_N) \rightarrow 0$, and $ e^{N/r_N} l_N/N \rightarrow 0$, where $l_N/N \rightarrow 0$ and $\alpha(\cdot)$ is the mixing function of the strong mixing condition that holds for $\{|\Lambda_{i,1}|\}$.
			\begin{enumerate}
				\item If for some $\tau > 0$, $\pi_{k,N}(i;\tau)$ converges to some $\pi_k(i)$ for $1 \leq i \leq m-1$ which is independent of $\tau$, then $P(|\Lambda_{(m),k}| \leq u_{k,N}(\tau))$ converges for each $\tau > 0$ and
				\begin{eqnarray}\label{eqn:G-1-m-limit}
				\lim_{N \rightarrow \infty} P(|\Lambda_{(m),k}| \leq u_{k,N}(\tau))  = e^{-\tau} \left[1 + \sum_{l = 1}^{m-1} c_{k,l} \cdot \frac{\tau^l}{l!}\right]
				\end{eqnarray}
				where 
				\begin{eqnarray}\label{eqn:pi-def}
				c_{k,l} =\sum_{i = l}^{m-1} \pi_k^{\ast^l} (i) \quad \text{and} \quad
				\pi_k^{\ast^l} (i) = \begin{cases}
				0, &  i < l \\
				\underset{ i_r \geq 1, 1 \leq r \leq l }{\underset{i_1 + \cdots i_l = i}{\sum \cdots \sum}} \pi_k(i_1)\cdots \pi_k(i_l) , &  i \geq l
				\end{cases}.
				\end{eqnarray}
				\item Conversely, if $P(|\Lambda_{(m),k}| \leq u_{k,N}(\tau)) $ converges for each $\tau > 0$, then for any $\tau > 0$ and $1 \leq i \leq m - 1$,  $\pi_{k,N}(i;\tau)$ converges to some $\pi_k(i)$ which is independent of $\tau$, and the limit of $P(|\Lambda_{(m),1}| \leq u_{1,N}(\tau)) $ is the same as (\ref{eqn:G-1-m-limit}). 
			\end{enumerate}
		\end{lemma}
		A special case of Lemma \ref{lemma:order-stats-dist} is that $|\Lambda_{i,k}|$ is independent, then $\pi_N(k;\tau) \rightarrow 1$ and
		\[ \lim_{N \rightarrow \infty} P(|\Lambda_{(m),k}| \leq u_{k,N}(\tau))  = e^{-\tau}  \sum_{l = 0}^{m-1} \frac{\tau^l}{l!} \]
		If the tail of $|\Lambda_{i,k}|$ follows the GEV distribution with parameters $(\mu, \sigma, \xi)$, then \\ $u_{k,N}(\tau) =a_{k,N}\left(\mu + \sigma \left( \frac{\tau^{-\xi} - 1}{\xi} \right) \right) + b_{k,N}$ for some normalizing sequences $\{a_{k,N} > 0\}$ and $\{b_{k,N}\}$ and $G_{k,m}(\tau)$ is the same as in Theorem 3.4 in \cite{coles2001introduction}. 
		


		\section{Uniform Convergence of PCA Loadings}
		
		\cite{bai2002determining} and \cite{bai2003inferential} show that factors and loadings can be estimated consistently with PCA under Assumptions \ref{ass_factor}-\ref{ass_f_e} when $N, T \rightarrow \infty$.  
		The derivations for our proximate factors require the stronger uniform consistency of loadings as stated in the following proposition. This result is of independent interest by itself.
		
		
		\begin{proposition} \label{thm:uniform-consistency-loading}
			Under Assumptions \ref{ass_factor}-\ref{ass_f_e} it holds that
			\[\max_{l \leq N} \norm{\hat{\Lambda}_l - H \Lambda_l} = O_p(\sqrt{1/N} + N^{1/4}/\sqrt{T}).\]
		\end{proposition}
		
		Proposition \ref{thm:uniform-consistency-loading} states that the maximum difference between the estimated loading and some rotation of the true loading for any cross-section unit converges to 0 at a specific rate. Relative to $\frac{1}{N} \sum_{l=1}^N \norm{\hat{\Lambda}_l - H \Lambda_l}^2 = O_p(1/N + 1/T)$, the uniform convergence rate of $\norm{\hat{\Lambda}_l - H \Lambda_l}$ is slower if $T/N^{3/2}=o(1)$.\footnote{\cite{bai2003inferential} show in their Proposition 2 that $\max_{1\leq t \leq T} \norm{\tilde{F}_t - (H^{-1})^\T F_t} =  O_p(\sqrt{1/T} + \sqrt{T}/\sqrt{N})$. \cite{fan2013large} show in their Theorem 4 that $\max_{l \leq N} \norm{\hat{\Lambda}_l - H \Lambda_l} = O_p(\sqrt{1/N} + \sqrt{\log N} /\sqrt{T})$ under the additional assumption of bounded loadings and $\max_{1\leq t \leq T} \norm{\tilde{F}_t - (H^{-1})^\T F_t} =  O_p(\sqrt{1/T} + T^{1/4}/\sqrt{N})$. Our results relax their assumption of bounded loadings which comes at the cost of rate that is slower than in \cite{fan2013large}.} Proposition \ref{thm:uniform-consistency-loading} states that a large estimated loadings imply large population loadings. As a result, the $m$-th largest values of $|\hat{\*\Lambda} _j|$ is close to the $m$-th largest $|(\*\Lambda H)_{\cdot, j}|$, which is formally stated in Lemma \ref{lemma_orderstat}. Hence, we can derive the distribution of  the correlation $\mathrm{corr}(\tilde{\*F}, \*F)$ based on the largest population instead of estimated loadings.
		
	}

	
	\newpage
	
	\renewcommand{\thetheorem}{IA.\arabic{theorem}}  \setcounter{theorem}{0}
	\renewcommand{\theproposition}{IA.\arabic{proposition}} \setcounter{proposition}{0}
	\renewcommand{\thelemma}{IA.\arabic{lemma}} \setcounter{lemma}{0}
	\renewcommand{\theassumption}{IA.\arabic{assumption}}  \setcounter{assumption}{0}
	
	\renewcommand{\theequation}{IA.\arabic{equation}}\setcounter{equation}{0}
	\renewcommand{\thefigure}{IA.\arabic{figure}} \setcounter{figure}{0}
	\renewcommand{\thetable}{IA.\arabic{table}} \setcounter{table}{0}
	
	\setcounter{page}{1}
	\setcounter{section}{0}
	\setcounter{subsection}{0}
	
	\renewcommand{\thesection}{IA.\Alph{section}}
	\renewcommand{\thesubsection}{\thesection.\arabic{subsection}}

	{\small

		\begin{center}
			{\LARGE\bf{Internet Appendix to \\Interpretable Proximate Factors for Large Dimensions}}
		\end{center}
		\vspace{0.1cm}

		\section{Uniform Convergence of PCA Loadings}

		\begin{proof}[Proof of Proposition \ref{thm:uniform-consistency-loading}]
			Let
			\begin{eqnarray}\label{eqn-decom-lambda}
			\hat{\Lambda}_l - H \Lambda_l = \hat{D}^{-1} \left( \frac{1}{N}\sum_{i=1}^N \hat{\Lambda}_i \+E(\*e_i^\T \*e_l)/T + \frac{1}{N} \sum_{i=1}^N \hat{\Lambda}_i \zeta_{il} + \frac{1}{N} \sum_{i=1}^N \hat{\Lambda}_i \eta_{il} + \frac{1}{N} \sum_{i=1}^N \hat{\Lambda}_i \xi_{il} \right)
			\end{eqnarray}
			where $\hat{D}$ are eigenvalues of $\frac{1}{NT} \*X \*X^\T$ corresponding to eigenvalues $\hat{\*\Lambda}$, $H$ is some rotation matrix, $\zeta_{il} = \*e_i^\T \*e_l/T -  E(\*e_i^\T \*e_l)/T$, $\eta_{il}=\Lambda_i^\T \sum_{t=1}^T F_t e_{lt}/T$ and $\xi_{il} = \Lambda_l^T \sum_{t=1}^T F_t e_{it}/T$. 
			
			\begin{lemma} \label{lemma-max-norm-lambda-error}
				Under Assumptions \ref{ass_loading}-\ref{ass_f_e},
				\begin{enumerate}
					\item $\max_{l \leq N} \norm{\frac{1}{NT} \sum_{i = 1}^N \hat{\Lambda}_i \+E(\*e_i^\T \*e_l) } = O_p(\sqrt{1/N})$
					\item $\max_{l \leq N} \norm{\frac{1}{N} \sum_{i = 1}^N \hat{\Lambda}_i \zeta_{il}} = O_p(N^{1/4}/\sqrt{T}) $
					\item $\max_{l \leq N} \norm{\frac{1}{N} \sum_{i = 1}^N \hat{\Lambda}_i \eta_{il} }= O_p(N^{1/4}/\sqrt{T}) $
					\item $\max_{l \leq N} \norm{\frac{1}{N} \sum_{i = 1}^N \hat{\Lambda}_i \xi_{il} }= O_p(N^{1/4}/\sqrt{T}) $
				\end{enumerate}
			\end{lemma}
			\begin{proof}[Proof of Lemma \ref{lemma-max-norm-lambda-error}]
				\begin{enumerate}
					\item
					By the Cauchy-Schwarz inequality and the fact that $\frac{1}{N}\sum_{i=1}^N \norm{\hat{\Lambda}_i}^2 = O_p(1)$, 
					\begin{eqnarray*}
						\max_{l \leq N} \norm{\frac{1}{NT} \sum_{i = 1}^N \hat{\Lambda}_i \+E[\*e_i^\T \*e_l] } &\leq& \max_{l \leq N} \left(\frac{1}{N} \sum_{i=1}^N \norm{\hat{\Lambda}_i}^2 \frac{1}{N} \sum_{i=1}^N (E[\*e_i^\T \*e_l]/T)^2 \right)^{1/2} \\
						&\leq& O_p(1) \max_{l \leq N} \left(\frac{1}{N} \sum_{i=1}^N (\+E[\*e_i^\T \*e_l]/T)^2 \right)^{1/2} \\
						&\leq&  O_p(1) \max_{i,l \leq N} \sqrt{|E[\*e_i^\T \*e_l]/T|} \max_{l \leq N}  \left(\frac{1}{N} \sum_{i=1}^N |\+E[\*e_i^\T \*e_l/T| \right)^{1/2} \\
						&=& O_p(\sqrt{1/N})
					\end{eqnarray*}
					by Assumption \ref{ass_error}.2.
					\item By Cauchy-Schwarz inequality, 
					\begin{eqnarray*}
						\max_{l \leq N} \norm{\frac{1}{N} \sum_{i = 1}^N\hat{\Lambda}_i \zeta_{il}}&\leq& \max_{l \leq N} \left(\frac{1}{N} \sum_{i=1}^N \norm{\hat{\Lambda}_i}^2 \frac{1}{N} \sum_{i = 1}^N  \zeta_{il}^2 \right)^{1/2} \leq O_p(1) \max_{l \leq N}  \left(\frac{1}{N} \sum_{i = 1}^N  \zeta_{il}^2\right)^{1/2} \\ 
						&=& O_p(N^{1/4}/\sqrt{T})
					\end{eqnarray*}
					It follows from Assumption \ref{ass_error}.4 that  $\+E(\frac{1}{N} \sum_{i=1}^N \zeta_{il}^2 )^2 \leq \max_{i,l} \+E\zeta_{il}^4 = O(T^{-2})$. From Markov inequality and Boole's inequality (the union bound), $\max_{l \leq N} \frac{1}{N} \sum_{i=1}^N \zeta_{il}^2 = O_p(N^{1/4}/\sqrt{T})$ \footnote{Denote $y_l =  \frac{1}{N} \sum_{i=1}^N \zeta_{il}^2 $. $\exists M_1, \+E y_l^2 \leq M_1/T^2 $. We have $\forall \epsilon$, $P\left( \max_{l \leq N} y_l^2 > NM_1/(T^2 \epsilon) \right) = P\left( \exists l,  y_l^2 > NM_1/(T^2 \epsilon) \right) \leq \sum_{l=1}^N P\left(  y_l^2 > NM_1/(T^2 \epsilon) \right)  \leq \sum_{l=1}^N \frac{\+Ey_l^2}{NM_1/(T^2 \epsilon)} \leq \epsilon$ by Markov Inequality and the union bound. Thus, $\max_{l \leq N} y_l=O_p(N^{1/2}/T)$. }.
					
					\item $\+E\norm{1/\sqrt{T}\sum_{t=1}^T F_t e_{it}}^4 \leq M$ by Assumption \ref{ass_f_e}. Using Markov inequality and Boole's inequality (the union bound), we have $\max_{i\leq N} \norm{\frac{1}{T} \sum_{t=1}^T F_t e_{it}} = O_p(N^{1/4}/\sqrt{T}) $. Thus, 
					\begin{eqnarray*}
						\max_{l \leq N} \norm{\frac{1}{N} \sum_{i = 1}^N \hat{\Lambda}_i \eta_{il} } \leq \norm{\frac{1}{N}\sum_{i=1}^N \hat{\Lambda}_i \Lambda_i^\T} \max_{l\leq N} \norm{\frac{1}{T} \sum_{t = 1}^T F_t e_{lt}} = O_p(N^{1/4}/\sqrt{T})
					\end{eqnarray*}
					follows from $\norm{\frac{1}{N}\sum_{i=1}^N \hat{\Lambda}_i \Lambda_i^\T} \leq \left(\frac{1}{N}\sum_{i=1}^N \norm{\hat{\Lambda}_i}^2 \right)^{1/2} \left(\frac{1}{N}\sum_{i=1}^N \norm{\Lambda_i}^2 \right)^{1/2} = O_p(1)$ by Assumption \ref{ass_loading}.
					\item By Assumption \ref{ass_f_e}, $\norm{\frac{1}{NT}\sum_{i=1}^N\sum_{t=1}^T F_t e_{it} \hat{\Lambda}_i} \leq \left(\frac{1}{N} \sum_{i=1}^N \norm{\hat{\Lambda}_i}^2 \max_{i \leq N} \norm{\frac{1}{T} \sum_{t = 1}^T F_t e_{it} }^2  \right)^{1/2} = O_p(\sqrt{1/T})$. In addition, since $\+E\norm{\Lambda_i}^4 < M$ from Assumption \ref{ass_loading}, $\max_{l \leq N}\norm{\Lambda_l} = O_p(N^{1/4})$.\footnote{$\forall \epsilon$, $P\left( \max_{l \leq N} \norm{\Lambda_l}^4 > NM/\epsilon \right) \leq \sum_{l=1}^N P\left(  \norm{\Lambda_l}^4 > NM/\epsilon \right) =\sum_{l=1}^N \frac{E\norm{\Lambda_l}^4}{NM/\epsilon} \leq \epsilon$ by Markov Inequality and the union bound. Thus, $\max_{l \leq N} \norm{\Lambda_l}^4=O_p(N)$ and $\max_{l \leq N} \norm{\Lambda_l} = O_p(N^{1/4})$.} Thus, 
					\begin{eqnarray*}
						\max_{l \leq N} \norm{\frac{1}{N} \sum_{i = 1}^N \hat{\Lambda}_i \xi_{il} } \leq \max_{i \leq N}\norm{\Lambda_i} \norm{\frac{1}{NT}\sum_{i=1}^N\sum_{t=1}^T F_t e_{it} \hat{\Lambda}_i} = O_p(N^{1/4}/\sqrt{T}).
					\end{eqnarray*}
				\end{enumerate}
			\end{proof}
			
			%
			%
			%
			%
			%
			Under Assumptions \ref{ass_factor}-\ref{ass_f_e}, from Lemma A.3 in \cite{bai2003inferential}, we have $\hat{D} \xrightarrow{P} D$, where $D = \diag(D_1, \cdots, D_K)$ are the eigenvalues of $\Sigma_F^{1/2} \Sigma_{\Lambda} \Sigma_F^{1/2}$. Since from Assumptions \ref{ass_factor} and \ref{ass_loading}, eigenvalues of $\Sigma_F \Sigma_{\Lambda}$ are bounded away from both zero and infinity, diagonal elements in $\hat{D}$ are bounded away from zero and infinity with probability 1, therefore, diagonal elements in $\hat{D}^{-1}$ are bounded away from zero and infinity with probability 1.
			\begin{eqnarray*}
				\max_{l \leq N}\norm{\hat{\Lambda}_l - H \Lambda_l} &\leq& \norm{\hat{D}^{-1}} \left\lbrace \max_{l \leq N} \norm{\frac{1}{NT}\sum_{i=1}^N \hat{\Lambda}_i \+E[\*e_i^\T \*e_l]} + \max_{l \leq N} \norm{\frac{1}{N}  \sum_{i=1}^N \hat{\Lambda}_i \zeta_{il}} \right. \\
				&&\left.\max_{l \leq N} \norm{\frac{1}{N} \sum_{i=1}^N \hat{\Lambda}_i \eta_{il}} + \max_{l \leq N} \norm{\frac{1}{N} \sum_{i=1}^N \hat{\Lambda}_i \xi_{il} } \right\rbrace  \\
				&=& O_p(1) (O_p(\sqrt{1/N}) + O_p(N^{1/4}/\sqrt{T})) = O_p(\sqrt{1/N}+N^{1/4}/\sqrt{T}) 
			\end{eqnarray*}
			from Lemma \ref{lemma-max-norm-lambda-error}.
		\end{proof}



		\section{Generalized Correlation Between Population Factors and Proximate Factors}
		We prove the results in Section \ref{sec:theory} under a general setup (each observation is multiplied by weight $\weight_i$)
		\[\underbrace{\Weight \*X}_{\*X^{\twt}} = \Weight \*\Lambda \*F^\T  + \Weight \*e\]
		where $\Weight = \diag(\weight_1, \weight_2, \cdots, \weight_N)$. Without loss of generality, we assume $\weight_i$ takes values bounded above from 1. {Suppose Assumption \ref{ass_loading} holds if $\Lambda_i$ is replaced by $\weight_i \Lambda_i$.}
		Denote the upper bound for weighted error's second moment as $\sigma_{e,\weight}^2 = \max_{i,j,t,s} |\weight_i \weight_j \+E[e_{it}e_{js}]|$.  When $\Weight_i = 1$ for all $i$, $\*X^{\twt} = \*X$ and the setting is the same as Section \ref{subsec:proxy-one-factor} and \ref{subsec:proxy-multi-factor}.  
		
		In the following, we use $\hat{\*\Lambda}$ and $\hat{\*F}$ to denote the PCA loadings and factors estimated from $\*X^{\twt} = \Weight \*X$. 
		
		
		%
		
		\subsection{Lemmas for the Proof of Theorems and Propositions in Section \ref{sec:theory}}
		We first show  in the following lemma that if a unit has large estimated loading, then with high probability, this unit's population loading is large as well.
		\begin{lemma} \label{lemma_orderstat}
			Under Assumptions \ref{ass_factor}-\ref{ass_f_e}, $\forall c, j, \forall 1 \leq i \leq m, \epsilon, \delta$, $\exists N_0, T_0$, when $N > N_0, T > T_0$,  with probability $\geq 1 - \delta$, if $|(\Weight \*\Lambda H)_{\fidx_j(m),j}| \geq c + 2\epsilon$, 
			$$|(\Weight  \*\Lambda H)_{\fidx_j(i),j}| > c,$$ where $|(\Weight  \*\Lambda H)_{\fidx_j(m),j}|$ is the $m$-th order statistic of $|(\Weight \*\Lambda H)_{\cdot, j}|$ and $\fidx_j(i)$ is the index of $i$-th order statistic of the $j$-th estimated loading $|\hat{\*\Lambda}_{j}|$.
		\end{lemma}

		\begin{proof}[Proof of Lemma \ref{lemma_orderstat}]
			
			Under Assumptions \ref{ass_factor}-\ref{ass_f_e}, from Theorem \ref{thm:uniform-consistency-loading}, 
			\[\max_{l\leq N}\norm{\hat{\Lambda}_l - H^\T \Lambda_l \weight_l } = O_p(1/\sqrt{N}+ N^{1/4}/\sqrt{T}) \]
			Thus, $\forall j \leq K$ $\max_{l\leq N}|\hat{\Lambda}_{l,j} -  (\Weight \*\Lambda H)_{l,j} | = O_p(1/\sqrt{N}+ N^{1/4}/\sqrt{T})$. In other words, $\forall \epsilon, \delta$, there exist $N_0$ and $ T_0$, such that $\forall N > N_0, T > T_0$,
			$$
			P(\max_{l}|\hat{\Lambda}_{l,j} - (\Weight \*\Lambda H)_{l,j}| > \epsilon) < \delta.
			$$
			Therefore, we have $\forall N > N_0, T > T_0$,
			$$
			P(\max_{i}||\hat{\Lambda}_{l,j} | - |(\Weight \*\Lambda H)_{l,j}|| > \epsilon) < \delta
			$$ 
			following $||\hat{\Lambda}_{l,j} | - |(\Weight\*\Lambda H)_{l,j}|| \leq |\hat{\Lambda}_{l,j} - (\Weight \*\Lambda H)_{l,j}|$.
			We have with probability at least $1-\delta$, $\forall 1\leq i \leq m$, $|(\Weight\*\Lambda H)_{\fidx_j(i),j}| - \epsilon < |\hat{\Lambda}_{\fidx_j(i), j}| < |(\Weight\*\Lambda H)_{\fidx_j(i),j}| + \epsilon$ and $|(\Weight\*\Lambda H)_{\fidx_j(m),j}| - \epsilon < |\hat{\Lambda}_{\fidx_j(m),j}| < |(\Weight\*\Lambda H)_{\fidx_j(m),j}| + \epsilon$. From $\forall 1\leq i \leq m$, $|\hat{\Lambda}_{\fidx_j(i),j}| \geq \min(|\hat{\Lambda}_{(1),j}|, \cdots, |\hat{\Lambda}_{(m),j}|)$ (recall the definition $|\hat{\Lambda}_{(i),j}|$ is the largest $i$-th entry in $\hat{\*\Lambda}_j$ ), with probability $\geq 1-\delta$, 
			\begin{eqnarray*}
				&&|(\Weight \*\Lambda H)_{\fidx_j(i),j}| + \epsilon > |\hat{\Lambda}_{\fidx_j(i),j}| \geq \min(|\hat{\Lambda}_{(1),j}|, \cdots, |\hat{\Lambda}_{(m),j}|) \\
				&>&\min(|(\Weight \*\Lambda H)_{(1),j}|-\epsilon, \cdots, |(\Weight \*\Lambda H)_{(m),j}-\epsilon| )= |(\Weight \*\Lambda H)_{(m),j}| - \epsilon
			\end{eqnarray*}
			Therefore, if $ |(\Weight \*\Lambda H)_{(m),j}| - 2\epsilon \geq c$, with probability $\geq 1-\delta$, $\forall 1\leq i \leq m$, $ |(\Weight \*\Lambda H)_{\fidx_j(i),j}| > c$.
		\end{proof}
		
		Next, in the following lemma, we show an asymptotic equivalent expression for $\rho$
		that will be used in the proof of Theorem \ref{thm-evt-one-factor} and \ref{thm-evt-multi-factor}.
		
		\begin{lemma} \label{lemma_rho}
			Under Assumptions \ref{ass_loading} and \ref{ass_f_e}, as $N, T \rightarrow \infty$, $$\rho =  tr \left(   \left(I + \left( \frac{\*F^\T  \*F}{T} \right)^{-1/2} \left(\wt{\*W}^\T \Weight \*\Lambda \right)^{-1} \frac{\wt{\*W}^\T \Weight \*e \*e^\T \Weight \wt{\*W}}{T}  \left(\*\Lambda^\T \Weight \wt{\*W}\right)^{-1} \left( \frac{\*F^\T  \*F}{T} \right)^{-1/2} \right)^{-1}  \right) +  o_p(1)$$
		\end{lemma}
		
		\begin{proof}[Proof of Lemma \ref{lemma_rho}]
			Denote $Q$ as $Q = \*\Lambda^\T \Weight \wt{\*W}$. Given $\tilde{\*F} =  \*X^\T \Weight \wt{\*W} (\wt{\*W}^\T  \wt{\*W})^{-1}$ and $\Weight X =\Weight \*\Lambda \*F^\T  + \Weight \*e$, we have 
			$$
			\*F^\T  \tilde{\*F}/T = (\*F^\T   \*F/T)(\*\Lambda^\T \Weight \wt{\*W}) (\wt{\*W}^\T  \wt{\*W})^{-1} + (\*F^\T  \*e^\T \Weight \wt{\*W}/T) (\wt{\*W}^\T  \wt{\*W})^{-1}  = (\*F^\T \*F/T) Q (\wt{\*W}^\T  \wt{\*W})^{-1} + o_p(1)
			$$
			follows from $\forall i$, $\*e_i^\T \*F/T = o_p(1)$ from Assumption \ref{ass_f_e} and $\forall i, j$, $\tilde w_{i,j}$ to be bounded because $\wt{\*W}_j^\T  \wt{\*W}_j = 1$ and each column in $\wt{\*W}$ has $m$ (fixed) nonzero entries, so $\*F^\T  \*e^\T \Weight \wt{\*W}/T = o_p(1)$. $\tilde{\*F}^\T \tilde{\*F}/T$ has
			\begin{eqnarray*}
				&& \tilde{\*F}^\T \tilde{\*F}/T \\
				&=& (\wt{\*W}^\T  \wt{\*W})^{-1} (\*F \*\Lambda^\T \Weight \wt{\*W}  + \*e^\T \Weight \wt{\*W})^\T (\*F \*\Lambda^\T \Weight \wt{\*W}  + \*e^\T \Weight \wt{\*W})(\wt{\*W}^\T  \wt{\*W})^{-1}/T \\
				&=& (\wt{\*W}^\T  \wt{\*W})^{-1} \left(Q^\T (\*F^\T  \*F/T) Q + ( \wt{\*W}^\T\Weight e \*F/T)Q + Q^\T (\*F^\T  \*e^\T \Weight \wt{\*W} /T)  +  \wt{\*W}^\T  \Weight \*e \*e^\T \Weight \wt{\*W}/T \right) (\wt{\*W}^\T  \wt{\*W})^{-1} \\
				&=&(\wt{\*W}^\T  \wt{\*W})^{-1} \left( Q^\T (\*F^\T  \*F/T) Q + \wt{\*W}^\T \Weight \*e \*e^\T \Weight \wt{\*W}/T \right) (\wt{\*W}^\T  \wt{\*W})^{-1}  + o_p(1) 
			\end{eqnarray*}
			follows from $Q = [q_{jl}] = \*\Lambda^\T \Weight \wt{\*W}$, where $q_{jl} = \sum_{i = 1}^m \Lambda_{\fidx_l(i),j} \wt W_{\fidx_l(i),l}$, where $\wt W_{\fidx_l(i),l}$ is nonzero for $\fidx_l(1), \fidx_l(2), \cdots, \fidx_l(m)$. $\Lambda_{\fidx_l(i),j} = O_p(1)$ from assumption \ref{ass_loading} and then $q_{jl} = O_p(1)$.  
			
			Plug the above into $\rho$ and 
			\begin{eqnarray*}
				\rho &=& \tr \left( (\*F^\T \*F/T)^{-1} (\*F^\T  \tilde{\*F}/T) (\tilde{\*F}^\T \tilde{\*F}/T)^{-1} (\tilde{\*F}^\T \*F/T) \right) \\
				&=& \tr \left( \left( \frac{\*F^\T  \*F}{T} \right)^{-1} \left( \frac{\*F^\T  \*FQ}{T} \right) \left( \frac{Q^\T\left(\*F^\T  \*F + (Q^\T)^{-1} \wt{\*W}^\T \Weight \*e \*e^\T \Weight \wt{\*W} Q^{-1}\right)Q}{T}  \right)^{-1} \left( \frac{Q^\T \*F^\T  \*F}{T} \right) \right) +  o_p(1)  \\
				&=& \tr \left(  \left( \frac{\*F^\T  \*F + (Q^\T)^{-1} \wt{\*W}^\T \Weight  \*e \*e^\T \Weight \wt{\*W} Q^{-1}}{T}  \right)^{-1} \left( \frac{\*F^\T  \*F}{T} \right) \right)  + o_p(1)  \\
				&=& \tr \left( \left( \frac{\*F^\T  \*F}{T} \right)^{1/2}   \left( \frac{\*F^\T  \*F + (Q^\T)^{-1}  \wt{\*W}^\T \Weight \*e \*e^\T \Weight \wt{\*W} Q^{-1}}{T}  \right)^{-1} \left( \frac{\*F^\T  \*F}{T} \right)^{1/2} \right) + o_p(1)  \\
				&=& \tr \left(   \left(I + \left( \frac{\*F^\T  \*F}{T} \right)^{-1/2}(Q^\T)^{-1} \frac{ \wt{\*W}^\T \Weight \*e \*e^\T \Weight \wt{\*W}  }{T}  Q^{-1} \left( \frac{\*F^\T  \*F}{T} \right)^{-1/2} \right)^{-1}  \right)  + o_p(1) 
			\end{eqnarray*}
			
		\end{proof}
		
		In the following lemma, we show that if $\wt{\*W}$ satisfies the non-overlapping assumption, then we can bound the term  $\frac{1}{T} \wt{\*W}^\T \Weight \*e \*e^\T \Weight \wt{\*W} $ that appears in the asymptotic equivalence form of $\rho$. This lemma will be used in the proof of Proposition \ref{prop-one-factor} and Theorem \ref{thm-evt-multi-factor}.
		\begin{lemma} \label{lemma_lam_error}
			If indices of nonzeros entries in columns of $\wt{\*W}$ do not overlap and under Assumption \ref{ass_error}.2, then
			$$
			\frac{1}{T} \wt{\*W}^\T \Weight \*e \*e^\T \Weight \wt{\*W} \leq (1+h(m)) \sigma_{e,\weight}^2 I_K  + o_p(1) 
			$$
		\end{lemma}
		
		\begin{proof}[Proof of Lemma \ref{lemma_lam_error}]
			Let the indices of nonzero entries in $\wt{\*W}_j$ be $\fidx_j(1), \cdots, \fidx_j(m)$. From Cauchy-Schwartz inequality, the $(j,j)$ element in $\frac{1}{T}\wt{\*W}^\T  \*e \*e^\T \wt{\*W}$ has 
			
			\begin{eqnarray}
			\nonumber && \frac{1}{T}\wt{\*W}_j^\T \Weight \*e \*e^\T \Weight \wt{\*W}_j \\
			\nonumber &=& \frac{1}{T} \sum_{i = 1}^{m} \sum_{k =1 }^{m} \weight_{\fidx_j(i)} \weight_{\fidx_j(k)} \tilde w_{\fidx_j(i),j} \tilde w_{\fidx_j(k),j} e_{\fidx_j(i)}^\T e_{\fidx_j(k)} \\
			\nonumber &=& \frac{1}{T} \sum_{i = 1}^{m} \weight_{\fidx_j(i)}^2 \tilde w_{\fidx_j(i),j}^2  \*e_{\fidx_j(i)}^\T  \*e_{\fidx_j(i)} + \frac{1}{T} \sum_{i \neq k} \weight_{\fidx_j(i)} \weight_{\fidx_j(k)}  \tilde w_{\fidx_j(i),j} \wt w_{\fidx_j(k),j} \*e_{\fidx_j(i)}^\T \*e_{\fidx_j(k)} \\
			&\leq& \sigma_{e,\weight}^2 \sum_{i = 1}^{m} \tilde w_{\fidx_j(i),j}^2 + \left( \sum_{i \neq k} \tilde w_{\fidx_j(i),j}^2 \tilde w_{\fidx_j(k),j}^2 \right)^{1/2} \left( \sum_{i \neq k} \left( \frac{1}{T} \weight_{\fidx_j(i)} \weight_{\fidx_j(k)} \*e_{\fidx_j(i)}^\T \*e_{\fidx_j(k)} \right)^2 \right)^{1/2} + o_p(1) \label{eqn:lameelam1} \\ 
			&\leq& \sigma_{e,\weight}^2 + \left( \sum_{i \neq k} \tilde w_{\fidx_j(i),j}^2 \tilde w_{\fidx_j(k),j}^2 \right)^{1/2} \left( \sum_{i \neq k}\weight_{\fidx_j(i)} \weight_{\fidx_j(k)} \tau_{\fidx_j(i) \fidx_j(k),tt}^2 \right)^{1/2} + o_p(1) \label{eqn:lameelam2} \\
			&\leq& (1 + h(m))\sigma_{e,\weight}^2 + o_p(1) \label{eqn:lameelam3}
			\end{eqnarray}
			where $e_{i} \in R^{T \times 1}$ is the $i$-th row of $e$. Inequality (\ref{eqn:lameelam1}) holds from the definition of $\sigma_{e,\weight}^2$ in Assumption \ref{ass_error}. Inequality (\ref{eqn:lameelam2}) holds from the stationarity of $\*e_t$ in Assumption \ref{ass_error}.2, $\frac{1}{T}\weight_{\fidx_j(i)}^2 \*e_{\fidx_j(i)}^\T \*e_{\fidx_j(i)} = \weight_{\fidx_j(i)}^2 \tau_{\fidx_j(i) \fidx_j(k), tt}+ o_p(1) \leq \sigma_{e,\weight}^2 + o_p(1)$.  Inequality (\ref{eqn:lameelam3}) holds from $\sum_{i \neq k} \tilde w_{\fidx_j(i),j}^2 \tilde w_{\fidx_j(k),j}^2 \leq (\sum_{i =1}^m \tilde w_{\fidx_j(i),j}^2 )^2 = 1$ and the definition of $h(m)$, $\sum_{i \neq k} \weight_{\fidx_j(i)} \weight_{\fidx_j(k)}  \tau_{\fidx_j(i) \fidx_j(k),tt}^2 \leq (h(m))^2 \sigma_{e,\weight}^4$. 
			
			The $(j,k)$ element in $\frac{\wt{\*W}^\T \Weight  \*e \*e^\T \Weight \wt{\*W}}{T}$ is 0 given there are no overlapping nonzero elements among loadings. Thus, we have $$
			\frac{1}{ T} \wt{\*W}^\T \Weight \*e \*e^\T \Weight \wt{\*W} \leq (1+h(m)) \sigma_{e,\weight}^2 I_K + o_p(1).
			$$
		\end{proof}
		
		The following lemma provides an asymptotic equivalent expression for $\hat{D}$. This lemma will be used in the proof of Lemma \ref{lemma_u}.
		\begin{lemma} \label{lemma_sbar}
			Suppose Assumptions \ref{ass_factor}-\ref{ass_f_e} hold. Let $H = \frac{\*F^\T  \*F}{T} \frac{\*\Lambda^\T \Weight \hat{\*\Lambda}}{N} (\hat{D})^{-1}$, then $H$ is invertible. Let $\bar{D} = H^{-1} \frac{\*F^\T  \*F}{T} (H^\T)^{-1}$, we have
			$$
			\hat{D} = \bar{D} + o_p(1)
			$$
		\end{lemma}
		\begin{proof}[Proof of Lemma \ref{lemma_sbar}]
			Recall $\hat{\*\Lambda}$ are the PCA loadings estimated from $\*X^{\twt} = \Weight \*X$. 
			
			Under Assumptions \ref{ass_factor}-\ref{ass_f_e} and Theorem 1 in \cite{bai2002determining}, $\mathrm{rank}(\frac{\*F^\T  \*F}{T} \frac{\*\Lambda^\T \Weight  \hat{\*\Lambda}}{N} ) = K$ and therefore $\frac{\*F^\T  \*F}{T} \frac{\*\Lambda^\T \Weight \hat{\*\Lambda}}{N} $ is invertible. $\hat{D}$ is invertible by Lemma A.3 in \cite{bai2003inferential} and Assumptions \ref{ass_factor}-\ref{ass_f_e}. Substituting $\*X = \Weight \*\Lambda F^\T  + \*e$ into $\frac{1}{N} \hat{\*\Lambda}^\T  \left( \frac{1}{NT}  \tilde{\*X} \tilde{\*X}^\T \right) \hat{\*\Lambda} = \hat{D}$, from Equation \eqref{eqn-decom-lambda} and Theorem 1 in \cite{bai2003inferential}, we have
			$$
			\hat{D} = \frac{\hat{\*\Lambda}^\T  \Weight \*\Lambda}{N} \frac{\*F^\T  \*F}{T} \frac{\*\Lambda^\T \Weight \hat{\*\Lambda}}{N} + o_p(1).
			$$
			Plug it into $H = \frac{\*F^\T  \*F}{T} \frac{\*\Lambda^\T \Weight \hat{\*\Lambda}}{N} (\hat{D})^{-1}$,  we have
			$$
			H = \left(\frac{\hat{\*\Lambda}^\T \Weight  \*\Lambda}{N} \right)^{-1} + o_p(1).
			$$
			Thus, 
			
			$$
			\bar{D} = H^{-1} \frac{\*F^\T  \*F}{T} (H^\T)^{-1} = \left( \left(\frac{\hat{\*\Lambda}^\T \Weight  \*\Lambda}{N} \right)^{-1} \right)^{-1} \frac{\*F^\T  \*F}{T}  \left(\left( \left(\frac{\hat{\Lambda}^\T \Weight  \*\Lambda}{N} \right)^{-1} \right)^\T \right)^{-1} + o_p(1) = \hat{D} + o_p(1)
			$$

		\end{proof}
		
		The following lemma provides an asymptotic equivalent expression for $\frac{1}{T} \wt{\*W}^\T  (\Weight \*\Lambda \*F^\T  \*F  \*\Lambda^\T \Weight) \wt{\*W} $, which is a critical term in the asymptotic expression for $\rho$. This lemma will be used in the proof of Theorem \ref{thm-evt-multi-factor}.
		\begin{lemma} \label{lemma_u}
			Under Assumption \ref{ass_factor}-\ref{ass_f_e}, let $D = \diag(D_1, D_2,\cdots, D_{K})$ be the diagonal matrix consisting the eigenvalues of $\Sigma_F \Sigma_{\Lambda,\Weight}$ in decreasing order, $H = \frac{\*F^\T  \*F}{T} \frac{\*\Lambda^\T \Weight \hat{\*\Lambda}}{N} (\hat{D})^{-1}$ and $\*U = \Weight \*\Lambda H D^{1/2}$, where $D^{1/2} = \diag(D_1^{1/2}, D_2^{1/2}, \cdots, D_{K}^{1/2})$, then for all $i$ and $l$, we have 
			$$
			\frac{1}{T} \weight_i \Lambda_i^\T \*F^\T  \*F \Lambda_l \weight_l = u_i^\T  u_l + o_p(1).
			$$
			Furthermore, we rescale each column of $\*U \odot \*M $ to get $\tilde{\*U}$, 
			$$
			\tilde{\*U}  =
			\begin{bmatrix}
			\frac{ \*U_1 \odot \*M_1 }{\norm{\*U_1 \odot \*M_1}} & \frac{ \*U_2 \odot \*M_2 }{\norm{\*U_2 \odot \*M_2 }} & \cdots &
			\frac{ \*U_K \odot \*M_K }{\norm{\*U_K \odot \*M_K}}
			\end{bmatrix} 
			$$
			where $\*U_j$ and $\*M_j$ are $j$-th column in $\*U$ and $\*M$ ($\*M$ is the mask matrix defined in (\ref{eqn-tilde-lambda-def})). We have $\tilde w_i = \tilde u_i + o_p(1)$ and 
			$$ 
			\frac{1}{T} \wt{\*W}^\T  (\Weight \*\Lambda \*F^\T  \*F  \*\Lambda^\T \Weight) \wt{\*W} =  \tilde{\*U}^\T   \*U \*U^\T  \tilde{\*U} + o_p(1)
			$$
		\end{lemma}
		
		\begin{proof}[Proof of Lemma \ref{lemma_u}]
			Recall the definition of $\bar{D} = H^{-1} \frac{\*F^\T  \*F}{T} (H^\T)^{-1}$, we have $\frac{1}{T} \Weight \*\Lambda \*F^\T  \*F \*\Lambda^\T \Weight = \Weight \*\Lambda H \bar{D} H^\T \*\Lambda^\T \Weight$. Under Assumptions \ref{ass_factor}-\ref{ass_f_e}, and from Lemma \ref{lemma_sbar}, $\hat{D} = \bar{D} + o_p(1)$. From Lemma A.3 in \cite{bai2003inferential}, $\hat{D} = \frac{1}{N} \hat{\Lambda}^\T  \left( \frac{1}{NT} \Weight \*X \*X^\T \Weight \right) \hat{\Lambda} = D + o_p(1)$. Thus, $D = \bar{D} + o_p(1)$. Since $U = \Weight \*\Lambda H D^{1/2}$, 
			$$
			\frac{1}{T} \weight_i \Lambda_i^\T \*F^\T  \*F \Lambda_l \weight_l = \weight_i \Lambda_i^\T H \bar{D} H^\T  \Lambda_l \weight_l  = \weight_i  \Lambda_i^\T H D H^\T  \Lambda_l \weight_l + o_p(1) =u_i^\T  u_l  + o_p(1),
			$$
			where $\Lambda_i \in R^{K \times 1}$ and $u_i \in R^{K \times 1}$ are the transposes of $i$-th rows in $\*\Lambda$ and $\*U$. From Theorem \ref{thm:uniform-consistency-loading}, $\hat{\Lambda}_i = H^\T \Lambda_i \weight_i  + o_p(1)$, together with $\hat{D} = D + o_p(1)$, we have $\hat{D}^{1/2} \hat{\Lambda}_i = u_i + o_p(1)$. Moreover, $\forall j$, $\frac{\hat{\*\Lambda}_j \hat{D}_j^{1/2} \odot \*M_j}{\norm{\hat{\*\Lambda}_j \hat{D}_j^{1/2} \odot \*M_j}} = \frac{\hat{\*\Lambda}_j \odot \*M_j}{\norm{\hat{\*\Lambda}_j \odot \*M_j}} $. Since $\wt{\*W} = \begin{bmatrix}
			\frac{ \hat{\*\Lambda}_1 \odot \*M_1}{\norm{\hat{\*\Lambda}_1 \odot \*M_1}} & \frac{\hat{\*\Lambda}_2 \odot \*M_2}{\norm{\hat{\*\Lambda}_2 \odot \*M_2}}  & \cdots &
			\frac{\hat{\Lambda}_K \odot \*M_K}{\norm{\hat{\Lambda}_K \odot \*M_K}} 
			\end{bmatrix} $ and $
			\tilde{\*U}  =
			\begin{bmatrix}
			\frac{ \*U_1 \odot \*M_1 }{\norm{\*U_1 \odot \*M_1}} & \frac{\*U_2 \odot \*M_2 }{\norm{\*U_2 \odot \*M_2 }} & \cdots &
			\frac{ \*U_K \odot \*M_K }{\norm{\*U_K \odot \*M_K}}
			\end{bmatrix} 
			$, we have for each cross-section unit $i$,
			$$
			\tilde w_i = \tilde u_i + o_p(1)
			$$
			and therefore
			$$ 
			\frac{1}{T} \wt{\*W}^\T  (\Weight \*\Lambda \*F^\T  \*F  \*\Lambda^\T \Weight) \wt{\*W} = \wt{\*W}^\T  ( \*U \*U^\T ) \wt{\*W} + o_p(1) = \tilde{\*U}^\T   \*U \*U^\T  \tilde{\*U} + o_p(1)
			$$
			
		\end{proof}
		
		\subsection{Proof of Proposition \ref{prop-one-factor} }
		Let us first prove Proposition \ref{prop-one-factor} and use the proof of this proposition as an intermediate step to prove Theorem \ref{thm-evt-one-factor}. In this section, for notation simplicity, we denote the CDF of $|\Lambda_{i,1}|$ as $F_{|\Lambda_{i,1}|}(y_{m})$ (the same as $\mathrm{CDF}_{|\Lambda_{i,1}|}(y_{m})$ in the statement of Proposition \ref{prop-one-factor}).
		\begin{proof}[Proof of Proposition \ref{prop-one-factor}] 
			\texttt{}\\
			Without loss of generality, suppose $\weight_i$ take values $\weight_{0,1}, \weight_{0,2}, \cdots, \weight_{0,n_\weight}$ and $\weight_{0,1} < \weight_{0,2} < \cdots < \weight_{0,n_\weight} = 1$. \\
			As a preparation for the proof, we first define a few notations. $\fidx_1(l)$ is the index of $l$-th largest entry in $|\Weight \*\Lambda_1|$ ($\*\Lambda = \*\Lambda_1$ in the one-factor case), and $\tfidx_1(l)$ is the index of $l$-th largest entry in $|\hat{\*\Lambda}|$. \\
			\textbf{Step 1: Provide an intermediate lower bound for $P(\rho > \rho_0)$.} \\
			Recall the definition of $Q$ from Lemma \ref{lemma_rho} ($Q = \*\Lambda^\T \Weight  \wt{\*W}$), we have $Q^\T\left( \frac{\*F^\T  \*F}{T} \right)Q = \wt{\*W}^\T  \left( \*U \*U^\T  \right) \wt{\*W}  $. From Lemma \ref{lemma_rho} and \ref{lemma_lam_error}, we have
			\begin{align}
			\nonumber \rho \geq& \tr \left(   \left(I + \left( \frac{\*F^\T  \*F}{T} \right)^{-1/2}(Q^\T)^{-1} \left((1+h(m))\sigma_{e,\weight}^2  \right)  Q^{-1} \left( \frac{\*F^\T  \*F}{T} \right)^{-1/2} \right)^{-1}  \right) + o_p(1)  \\
			=&   \left(1 + \frac{(1+h(m))\sigma_{e,\weight}^2   }{\frac{1}{T} \wt{\*W}^\T  \left( \Weight \*\Lambda \*F^\T  \*F  \*\Lambda^\T \Weight \right) \wt{\*W}  }  \right)^{-1}   + o_p(1). \label{eqn:proof-prop-one-factor-1} 
			\end{align}
			Since $\bar{D} = H^{-1} \frac{\*F^\T  \*F}{T} (H^\T)^{-1}$, and $\bar{D}$ and $H$ are scalars in the one-factor model, we have $\sigma_{\*F_1}^2 = \frac{\*F^\T  \*F}{T} + o_p(1) = \bar{D} H^2 + o_p(1)$, and $\hat{\Lambda}_{\tfidx_1(l),1} = \weight_{\tfidx_1(l)} \Lambda_{\tfidx_1(l),1} H + o_p(1)$, for all $l$. We can simplify $\frac{1}{T} \wt{\*W}^\T  \left(\Weight \*\Lambda \*F^\T  \*F  \*\Lambda^\T \Weight \right) \wt{\*W}$ to 
			\begin{align}
			\nonumber \frac{1}{T} \wt{\*W}^\T  \left(\Weight \*\Lambda \*F^\T  \*F  \*\Lambda^\T \Weight \right) \wt{\*W} =& (\wt{\*W}^\T  \hat{\*\Lambda}) \bar{D} (\hat{\*\Lambda}^\T  \wt{\*W}) + o_p(1) \\
			\nonumber =& \bar{D} (\sum_{i = 1}^{m} \hat{\Lambda}_{\tfidx_1(l),1}^2)^2/(\sum_{i = 1}^{m} \hat{\Lambda}_{\tfidx_1(l),1}^2) + o_p(1) \\
			=& \sigma_{\*F_1}^2 \sum_{i = 1}^{m} \weight_{\tfidx_1(l),1}^2 \Lambda_{\tfidx_1(l),1}^2 + o_p(1). \label{eqn:proof-prop-one-factor-2}  
			\end{align}
			We plug \eqref{eqn:proof-prop-one-factor-2} into \eqref{eqn:proof-prop-one-factor-1}
			\begin{align}
			\rho \geq  \left(1 + \frac{(1+h(m)) \sigma_{e,\weight}^2 }{\sigma_{\*F_1}^2 \sum_{i = 1}^{m} \weight_{\tfidx_1(l),1}^2 \Lambda_{\tfidx_1(l),1}^2  }  \right)^{-1} + o_p(1). \label{eqn:proof-prop-one-factor-3}    
			\end{align}
			The next step is to provide a lower bound for $P(\rho > \rho_0) $. Since $F_{|\weight_{\tfidx_1(l)} \Lambda_{\tfidx_1(l),1}|}(y) = P(|\weight_{\tfidx_1(l)} \Lambda_{\tfidx_1(l),1}| \leq y)$ is continuous in $y$ for all $y$,\footnote{$F_{|\weight_{\tfidx_1(l)} \Lambda_{i,1}|}(y)$ is continuous in $y$ implies $\lim_{N\rightarrow \infty } P \left(  \left(1 + \frac{(1+h(m))\sigma_{e,\weight}^2 }{\sigma_{\*F_1}^2 \sum_{i = 1}^{m} \weight_{\tfidx_1(l)} \Lambda_{\tfidx_1(l),1}^2  }  \right)^{-1} > \rho_0 \right)$ is continuous in $\rho_0$} and convergence in probability implies convergence distribution, we have 
			\begin{align}
			\lim_{N, T \rightarrow \infty} P(\rho > \rho_0) \geq \lim_{N\rightarrow \infty } P \left(  \left(1 + \frac{(1+h(m))\sigma_{e,\weight}^2 }{\sigma_{\*F_1}^2 \sum_{i = 1}^{m}\weight_{\tfidx_1(l)}  \Lambda_{\tfidx_1(l),1}^2  }  \right)^{-1} > \rho_0 \right).    \label{eqn:proof-prop-one-factor-4} 
			\end{align}
			We can further simplify \eqref{eqn:proof-prop-one-factor-4}. Note that 
			$$\left(1 + \frac{(1+h(m))\sigma_{e,\weight}^2 }{\sigma_{\*F_1}^2 \sum_{i = 1}^{m} \weight_{\tfidx_1(l)}^2  \Lambda_{\tfidx_1(l),1}^2  }  \right)^{-1} > \rho_0 \Longleftrightarrow \frac{\sigma_{\*F_1}^2    }{ \sigma_{e,\weight}^2} \sum_{i = 1}^{m} \weight_{\tfidx_1(l)}^2 \Lambda_{\tfidx_1(l),1}^2  > \frac{\rho_0}{ 1 - \rho_0}.$$
			From the definition of $\fidx_1(l)$, we have $|\weight_{\fidx_1(1)} \Lambda_{\fidx_1(1),1}| \geq |\weight_{\fidx_1(2)} \Lambda_{\fidx_1(2),1}|  \geq \cdots \geq |\weight_{\fidx_1(N)} \Lambda_{\fidx_1(N),1}|  $. From Lemma \ref{lemma_orderstat}, we can ignore term $H$ since it is a scalar, and we know if $|\weight_{\fidx_1(m)} \Lambda_{\fidx_1(m),1}| > c$, let $\epsilon = \frac{1}{4}(|\weight_{\fidx_1(m)} \Lambda_{\fidx_1(m),1}| - c)$, then $|\weight_{\fidx_1(m)} \Lambda_{\fidx_1(m),1}| > c + 2 \epsilon$ and $\forall 0 < \delta < 1$, $\exists N_0, T_0$, such that $N > N_0$, $T > T_0$, with probability at least $1 - \delta$, $|\weight_{\tfidx_1(l)} \Lambda_{\tfidx_1(l),1}| > c$. Thus, as $N, T \rightarrow \infty$, $\forall c$, if $\weight_{\fidx_1(m)} \Lambda_{\fidx_1(m),1} > c$, then $\weight_{\tfidx_1(l)} \Lambda_{\tfidx_1(l),1} > c$ with probability 1 for all $1 \leq l \leq m$.  Denote $y_{m,\weight} = \sqrt{ \frac{1+h(m)}{m} \frac{\sigma_{e,\weight}^2}{ \sigma_{\*F_1}^2} \frac{\rho_0}{ 1 - \rho_0}}$. As $N, T \rightarrow \infty$, with probability 1, 
			\begin{align*}
			& |\weight_{\fidx_1(m)} \Lambda_{\fidx_1(m),1}| > y_{m,\weight} \\ \Rightarrow& \frac{\sigma_f^2    }{(1+h(m)) \sigma_{e,\weight}^2} \sum_{i = 1}^{m}\weight_{\fidx_1(i)}^2  \Lambda_{\fidx_1(i),1}^2  > \frac{\rho_0}{ 1 - \rho_0} \\
			\Rightarrow&  \frac{\sigma_{\*F_1}^2    }{(1+h(m)) \sigma_{e,\weight}^2} \sum_{i = 1}^{m} \weight_{\tfidx_1(i)}^2 \Lambda_{\tfidx_1(i),1}^2  > \frac{\rho_0}{ 1 - \rho_0}.
			\end{align*}
			Then a simplified lower bound for \eqref{eqn:proof-prop-one-factor-4} is
			\begin{align}
			\nonumber \lim_{N, T \rightarrow \infty} P(\rho > \rho_0) \geq& \lim_{N\rightarrow \infty } P \left(  \left(1 + \frac{(1+h(m)) \sigma_{e,\weight}^2 }{\sigma_{\*F_1}^2 \sum_{i = 1}^{m} \weight_{\tfidx_1(l)}^2 \Lambda_{\tfidx_1(l),1}^2  }  \right)^{-1} > \rho_0 \right) \\ \geq& \lim_{N \rightarrow \infty} P(|\weight_{\fidx_1(m)} \Lambda_{ \fidx_1(m),1}| > y_{m,\weight}).    \label{eqn:proof-prop-one-factor-5}
			\end{align}
			\textbf{Step 2: Consider the special case where $\weight_i = 1$ for all $i$.} \\
			Denote $y_m$ as $y_m = y_{m,\weight}$ and recell the CDF of $|\Lambda_{i,1}|$ is denoted as $F_{|\Lambda_{i,1}|}(y_{m})$ (the same as $\mathrm{CDF}_{|\Lambda_{i,1}|}(y_{m})$ in the statement of Proposition \ref{prop-one-factor}). Since
			\begin{align*}
			&P(|\weight_{\fidx_1(m)} \Lambda_{\fidx_1(m),1}| > y_{m})  \\
			=& 1 - P(|\weight_{\fidx_1(1) } \Lambda_{\fidx_1(1),1} | \leq y_{m}) - \sum_{j = 1}^{m} P\left(|\weight_{\fidx_1(j) } \Lambda_{\fidx_1(j),1} | \geq y_{m}, |\weight_{\fidx_1(j+1) } \Lambda_{\fidx_1(j+1),1} | \leq y_{m}\right) \\
			=& 1 - \sum_{j = 0}^{m - 1} {{N}\choose{j}} (1 - F_{|\Lambda_{i,1}|}(y_{m}))^j  F_{|\Lambda_{i,1}|}(y_{m})^{N-j}
			\end{align*}
			and $F_{|\Lambda_{i,1}|}(y)$ is continuous in $y$ for all $y$, we have for a specific $m$ and a given level $\rho_0$ , as $N, T \rightarrow \infty$,
			$$
			\lim_{N, T \rightarrow \infty} P(\rho > \rho_0) \geq 1 - \lim_{N\rightarrow \infty} \sum_{j = 0}^{m - 1} {{N}\choose{j}} (1 - F_{|\Lambda_{1,i}|}(y_{m}))^j  F_{|\Lambda_{1,i}|}(y_{m})^{N-j} .
			$$
			\textbf{Step 3: Consider the general case where $\weight_i \neq 1$ for some $i$.} \\
			We can provide a lower bound for $P(\rho > \rho_0)$ that is strictly better than $1 - \lim_{N\rightarrow \infty} \sum_{j = 0}^{m - 1} {{N}\choose{j}} (1 - F_{|\Lambda_{1,i}|}(y_{m}))^j $ (such as $ 1 - p_{y_{m,\weight}}$ defined in \eqref{eqn:proof-prop-one-factor-6}) if the weight is inverse proportional to error's standard error, where $\rho_0$ is the same as the $\rho_0$ in Step 2. The proof is as follows.
			\begin{align}
			\nonumber &P(|\weight_{\fidx_1(m)} \Lambda_{\fidx_1(m),1}| > y_{m,\weight})  \\
			\nonumber =& 1 - P(|\weight_{\fidx_1(1) } \Lambda_{\fidx_1(1),1} | \leq y_{m,\weight}) \\
			\nonumber & - \sum_{j = 1}^{m} P\left(|\weight_{\fidx_1(j) } \Lambda_{\fidx_1(j),1} | \geq y_{m,\weight}, |\weight_{\fidx_1(j+1) } \Lambda_{\fidx_1(j+1),1} | \leq y_{m,\weight}\right) \\    
			\nonumber =& 1 - \sum_{j = 0}^{m - 1} \sum_{\stackrel{j_1, \cdots, j_{n_\weight}:}{j_1 + \cdots + j_{n_\weight}=j} } {{N_1}\choose{j_1}} {{N_2}\choose{j_2}} \cdots {{N_{n_\weight}}\choose{j_{n_\weight}}} (1 - F_{|\weight_{0,1} \Lambda_{i,1}|}(y_{m,\weight}))^{j_1}  F_{|\weight_{0,1} \Lambda_{i,1}|}(y_{m,\weight})^{N_1-j_1} \\
			\nonumber & \cdot (1 - F_{|\weight_{0,2} \Lambda_{i,1}|}(y_{m,\weight}))^{j_2}  F_{|\weight_{0,2} \Lambda_{i,1}|}(y_{m,\weight})^{N_2-j_2} \cdots (1 - F_{|\weight_{0,n_\weight} \Lambda_{i,1}|}(y_{m,\weight}))^{j_{n_\weight}}  F_{|\weight_{0,n_\weight} \Lambda_{i,1}|}(y_{m,\weight})^{N_{n_\weight}-j_{n_\weight} } \\
			\stackrel{\Delta}{=}& 1 - p_{y_{m,\weight}} \label{eqn:proof-prop-one-factor-6}
			\end{align}
			When the weight is inverse proportional to error's standard error, we have $\sigma_{e,\weight}^2 =  \weight_{0,1}^2 \sigma_e^2$ and $y_{m,\weight} = \weight_{0,1} \sqrt{ \frac{1+h(m)}{m} \frac{\sigma_{e}^2}{ \sigma_{\*F_1}^2} \frac{\rho_0}{ 1 - \rho_0}} = \weight_{0,1} y_m$. Note that for any $\ell$, we have 
			\[F_{|\weight_{0,\ell} \Lambda_{i,1}|}(y_{m,\weight}) =F_{|\weight_{0,\ell} \Lambda_{i,1}|}(\weight_{0,1} y_m ) \leq F_{|\weight_{0,\ell} \Lambda_{i,1}|}(\weight_{0,\ell} y_m ) =  F_{| \Lambda_{i,1}|}( y_m ) \]
			Note that 
			\begin{eqnarray*}
				\frac{\partial}{\partial x} \sum_{j = 0}^{m^\prime - 1} {{N^\prime}\choose{j}} (1-x)^j x^{N^\prime-j} &=& N^\prime x^{N^\prime-1} + \sum_{j = 1}^{m^\prime - 1} {{N^\prime}\choose{j}} (1-x)^{j-1} x^{N^\prime -j-1} (N^\prime -j-N^\prime x) \\
				&=& m^\prime  {{N^\prime}\choose{m^\prime}} x^{N^\prime - m^\prime} (1 - x)^{m^\prime - 1} > 0.
			\end{eqnarray*}
			$\frac{\partial}{\partial x} \sum_{j = 0}^{m^\prime - 1} {{N^\prime}\choose{j}} (1-x)^j x^{N^\prime-j}$ is increasing in $x$ for any $m^\prime$ and $N^\prime$. Then we can rearrange the terms in $p_{y_{m,\weight}}$ and provide an upper bound for $p_{y_{m,\weight}}$ (that is $ \sum_{j = 0}^{m - 1}  {{N}\choose{j}} (1 - F_{| \Lambda_{i,1}|}(y_{m}))^{j} F_{| \Lambda_{i,1}|}(y_{m})^{N -j}$),
			\begin{align*}
			&p_{y_{m,\weight}} \\
			=& \sum_{l = 0}^{m - 1} \sum_{\stackrel{j_2, \cdots, j_{n_\weight}:}{j_2 + \cdots + j_{n_\weight}=l} }  {{N_2}\choose{j_2}} \cdots {{N_{n_\weight}}\choose{j_{n_\weight}}} (1 - F_{|\weight_{0,2} \Lambda_{i,1}|}(y_{m,\weight}))^{j_2}  F_{|\weight_{0,2} \Lambda_{i,1}|}(y_{m,\weight})^{N_2-j_2}  \cdots (1 - F_{|\weight_{0,n_\weight} \Lambda_{i,1}|}(y_{m,\weight}))^{j_{n_\weight}}   \\
			& F_{|\weight_{0,n_\weight} \Lambda_{i,1}|}(y_{m,\weight})^{N_{n_\weight}-j_{n_\weight} } \cdot \underbrace{\Lp  \sum_{j_1 = 0}^{m - 1 - l} {{N_1}\choose{j_1}} (1 - F_{|\weight_{0,1} \Lambda_{i,1}|}(y_{m,\weight}))^{j_1}  F_{|\weight_{0,1} \Lambda_{i,1}|}(y_{m,\weight})^{N_1-j_1} \Rp }_{\leq \Lp  \sum_{j_1 = 0}^{m - 1 - l} {{N_1}\choose{j_1}} (1 - F_{|\Lambda_{i,1}|}(y_{m}))^{j_1}  F_{| \Lambda_{i,1}|}(y_{m})^{N_1-j_1} \Rp}  \\
			\leq& \sum_{j = 0}^{m - 1} \sum_{\stackrel{j_1, \cdots, j_{n_\weight}:}{j_1 + \cdots + j_{n_\weight}=j } } {{N_1}\choose{j_1}} {{N_2}\choose{j_2}} \cdots {{N_{n_\weight}}\choose{j_{n_\weight}}} (1 - F_{|\Lambda_{i,1}|}(y_{m}))^{j_1}  F_{| \Lambda_{i,1}|}(y_{m})^{N_1-j_1}   \\
			&  (1 - F_{|\weight_{0,2} \Lambda_{i,1}|}(y_{m,\weight}))^{j_2}  F_{|\weight_{0,2} \Lambda_{i,1}|}(y_{m,\weight})^{N_2-j_2}  \cdots   (1 - F_{|\weight_{0,n_\weight} \Lambda_{i,1}|}(y_{m,\weight}))^{j_{n_\weight}} F_{|\weight_{0,n_\weight} \Lambda_{i,1}|}(y_{m,\weight})^{N_1+N_{n_\weight} -j_{n_\weight} } \\
			=& \sum_{l = 0}^{m - 1} \sum_{\stackrel{j_1, j_3, \cdots, j_{n_\weight}:}{j_1 + j_3 + \cdots + j_{n_\weight}=l} }  {{N_1}\choose{j_1}} {{N_3}\choose{j_3}}  \cdots {{N_{n_\weight}}\choose{j_{n_\weight}}} (1 - F_{| \Lambda_{i,1}|}(y_{m}))^{j_1}  F_{| \Lambda_{i,1}|}(y_{m})^{N_1-j_1}  \cdots   \\
			& (1 - F_{|\weight_{0,n_\weight} \Lambda_{i,1}|}(y_{m,\weight}))^{j_{n_\weight}}  F_{|\weight_{0,n_\weight} \Lambda_{i,1}|}(y_{m,\weight})^{N_{n_\weight}-j_{n_\weight} } \cdot \underbrace{\Lp  \sum_{j_2 = 0}^{m - 1 - l} {{N_2}\choose{j_2}} (1 - F_{|\weight_{0,2} \Lambda_{i,1}|}(y_{m,\weight}))^{j_2}  F_{|\weight_{0,2} \Lambda_{i,1}|}(y_{m,\weight})^{N_2-j_2} \Rp }_{< \Lp  \sum_{j_2 = 0}^{m - 1 - l} {{N_2}\choose{j_2}} (1 - F_{|\Lambda_{i,1}|}(y_{m}))^{j_2}  F_{| \Lambda_{i,1}|}(y_{m})^{N_2-j_2} \Rp}  \\
			<& \sum_{l = 0}^{m - 1} \sum_{\stackrel{j_1, j_3, \cdots, j_{n_\weight}:}{j_1 + j_3 + \cdots + j_{n_\weight}=l} }  {{N_1}\choose{j_1}} {{N_2}\choose{j_2}} {{N_3}\choose{j_3}}  \cdots {{N_{n_\weight}}\choose{j_{n_\weight}}} (1 - F_{| \Lambda_{i,1}|}(y_{m}))^{j_1+j_2}  F_{| \Lambda_{i,1}|}(y_{m})^{N_1+N_2-j_1-j_2}  \cdots   \\
			& (1 - F_{|\weight_{0,n_\weight} \Lambda_{i,1}|}(y_{m,\weight}))^{j_{n_\weight}}  F_{|\weight_{0,n_\weight} \Lambda_{i,1}|}(y_{m,\weight})^{N_{n_\weight}-j_{n_\weight} }  \\
			<& \cdots \\
			<& \sum_{l = 0}^{m - 1} \sum_{\stackrel{j_1, j_3, \cdots, j_{n_\weight}:}{j_1 + j_3 + \cdots + j_{n_\weight}=l} }  {{N_1}\choose{j_1}} {{N_2}\choose{j_2}}\cdots {{N_{n_\weight}}\choose{j_{n_\weight}}} (1 - F_{| \Lambda_{i,1}|}(y_{m}))^{j_1+j_2+\cdots+j_{n_\weight}}  F_{| \Lambda_{i,1}|}(y_{m})^{N_1+N_2+\cdots+N_{n_\weight}-j_1-j_2-\cdots-j_{n_\weight}}  \\
			=& \sum_{j = 0}^{m - 1}  {{N}\choose{j}} (1 - F_{| \Lambda_{i,1}|}(y_{m}))^{j} F_{| \Lambda_{i,1}|}(y_{m})^{N -j}
			\end{align*}
			Since $F_{|\Lambda_{i,1}|}(y)$ is continuous in all $y$, we have for a specific $m$ and a given level $\rho_0$ , as $N, T \rightarrow \infty$,
			\begin{align*}
			\lim_{N, T \rightarrow \infty} P(\rho > \rho_0) \geq&  1 - p_{y_{m,\weight}} > 1 - \lim_{N\rightarrow \infty} \sum_{j = 0}^{m - 1} {{N}\choose{j}} (1 - F_{|\Lambda_{1,i}|}(y_{m}))^j  F_{|\Lambda_{1,i}|}(y_{m})^{N-j}
			\end{align*}
			\textbf{Step 4: Show if $\theta_i = 1$ for all $i$ and $F_{|\Lambda_{i,1}|}(y_{m}) < 1$, then $\lim_{N, T\rightarrow \infty} P(\rho > \rho_0) \rightarrow 1$} \\
			For a particular $0 < \rho_0 < 1$ and $m$, if $F_{|\Lambda_{i,1}|}(y_{m,\weight}) < 1$, $\forall 0 \leq j \leq m$,  $\lim_{N \rightarrow \infty} {{N}\choose{j}} (1 - F_{|\Lambda_{i,1}|}(y_{m,\weight}))^j  F_{|\Lambda_{i,1}|}(y_{m,\weight})^{N-j} \rightarrow 0$. Thus, as $N \rightarrow \infty$
			$$\sum_{j = 0}^{m - 1} {{N}\choose{j}} (1 - F_{|\Lambda_{i,1}|}(y_{m,\weight}))^j  F_{|\Lambda_{i,1}|}(y_{m,\weight})^{N-j} \rightarrow 0.$$
			Since $P(\rho > \rho_0) \leq 1$, we have as $N, T \rightarrow \infty$, 
			$$
			P(\rho > \rho_0) \rightarrow 1.
			$$
		\end{proof}

		\subsection{Proof of Theorem \ref{thm-evt-one-factor} and One-Factor Case in Proposition \ref{prop-one-factor}}
		The main of Theorem \ref{thm-evt-one-factor} is mainly based on the proof of Proposition \ref{prop-one-factor}. This proof also contains the proof of the one-factor case in Proposition \ref{prop:weighted}
		.	\begin{proof}[Proof of Theorem \ref{thm-evt-one-factor} and One-Factor Case of Proposition \ref{prop-one-factor}]
			\texttt{}\\
			In this proof, Step 1 and 2 prove Theorem  \ref{thm-evt-one-factor} . Step 3 proves the one-factor case in Proposition \ref{prop:weighted}.
			
			Without loss of generality, suppose $\weight_i$ take values $\weight_{0,1}, \weight_{0,2}, \cdots, \weight_{0,n_\weight}$ and $\weight_{0,1} < \weight_{0,2} < \cdots < \weight_{0,n_\weight} = 1$. \\
			\textbf{Step 1: Provide an intermediate lower bound for $P(\rho > \rho_0)$.} \\
			From the proof of Proposition \ref{prop-one-factor},
			\begin{align*}
			\lim_{N, T \rightarrow \infty} P(\rho > \rho_0) \geq& \lim_{N\rightarrow \infty } P \left(  \left(1 + \frac{(1+h(m)) \sigma_{e,\weight}^2 }{\sigma_{\*F_1}^2 \sum_{i = 1}^{m} \weight_{\tfidx_1(i)}^2 \Lambda_{\tfidx_1(i),1}^2  }  \right)^{-1} > \rho_0 \right) \\ \geq& \lim_{N \rightarrow \infty} P(|\weight_{\fidx_1(i)} \Lambda_{\fidx_1(i),1}| > y_{m,\weight}) 
			\end{align*}
			where $y_{m,\weight} = \sqrt{\frac{(1+h(m))}{m}\frac{\sigma_{e,\weight}^2}{\sigma_{\*F_1}^2}\frac{\rho_0}{1-\rho_0}} $, $\fidx_1(l)$ is the index of $l$-th largest entry in $|\Weight \*\Lambda_1|$, and $\tfidx_j(l)$ is the index of $l$-th largest entry in $|\hat{\*\Lambda}_1|$.  \\
			\textbf{Step 2: Consider the special case where $\weight_i = 1$ for all $i$.} \\
			Let $y_m \stackrel{\Delta}{=} \sqrt{\frac{(1+h(m))}{m}\frac{\sigma_{e}^2}{\sigma_{\*F_1}^2}\frac{\rho_0}{1-\rho_0}} = y_{m,\weight} $. Then probability lower bound is 
			\begin{align}
			\nonumber \lim_{N, T \rightarrow \infty} P\left( \rho > \rho_0 \right)  &\geq \lim_{N\rightarrow \infty}  P(|\Lambda_{(m),1}| > u_{1,N}(\tau)) \\ &= \bar{G}_{1,m}(y_m).
			\end{align}
			\textbf{Step 3: Consider the general case where $\weight_i \neq 1$ for some $i$. (Proof of the one-factor case in Proposition \ref{prop:weighted})} \\
			We can provide a lower bound for $P(\rho > \rho_0)$ that is better than $ \bar{G}_{1,m}(\tau)$ in \eqref{eqn:one-factor-lower-bound-evt} (such as $ \lim_{N\rightarrow \infty } P \left( |\weight_{\fidx_1(m)} \Lambda_{\fidx_1(m),1}|  > \weight_{0,1}  y_m   \right)$) if the weight is inverse proportional to error's standard error, where $\rho_0$ in defined in Step 2. The proof is as follows. 
			\begin{align*} 
			& \lim_{N, T \rightarrow \infty} P(\rho > \rho_0) \\
			\geq&\lim_{N, T \rightarrow \infty} P\left( \rho > \frac{m\sigma_{\*F_1}^2 u_{1,N}^2(\tau)}{(1+h(m))\sigma_{e}^2 + m \sigma_{\*F_1}^2 u_{1,N}^2(\tau)} \right) \\
			\geq& \lim_{N\rightarrow \infty } P \left(  \left(1 + \frac{(1+h(m)) \sigma_{e,\weight}^2 }{\sigma_{\*F_1}^2 \sum_{i = 1}^{m} \weight_{\tfidx_1(l)}^2 \Lambda_{\tfidx_1(l),1}^2  }  \right)^{-1} > \frac{m\sigma_{\*F_1}^2 u_{1,N}^2(\tau)}{(1+h(m))\sigma_{e}^2 + m \sigma_{\*F_1}^2 u_{1,N}^2(\tau)} \right)\\
			=&  \lim_{N\rightarrow \infty } P \left( \sum_{i = 1}^{m} \weight_{\tfidx_1(l)}^2 \Lambda_{\tfidx_1(l),1}^2  > \frac{\sigma_{e,\weight}^2}{\sigma_{e}^2}  u_{1,N}^2(\tau)   \right)\\
			\geq& \lim_{N\rightarrow \infty } P \left( |\weight_{\fidx_1(m)} \Lambda_{\fidx_1(m),1}|  > \frac{\sigma_{e,\weight}}{\sigma_{e}}  u_{1,N}(\tau)  \right) \\
			=& \lim_{N\rightarrow \infty } P \left( |\weight_{\fidx_1(m)} \Lambda_{\fidx_1(m),1}|  > \frac{\sigma_{e,\weight}}{\sigma_{e}}  y_m   \right) = \bar{G} (y_{m,\weight})\\
			=&  \lim_{N\rightarrow \infty } P \left( |\weight_{\fidx_1(m)} \Lambda_{\fidx_1(m),1}|  > \weight_{0,1}  y_m   \right)\\
			>&  \lim_{N\rightarrow \infty } P \left( |\weight_{0,1}  \Lambda_{(m),1}|  > \weight_{0,1}  y_m   \right) =  \lim_{N\rightarrow \infty } P \left( | \Lambda_{(m),1}|  >   y_m   \right) = \bar{G}_{1,m}(y_m),
			\end{align*}
			following that $\Lambda_{(m),1} = \Lambda_{\fidx_1(m),1}$.
		\end{proof}
		

		\subsection{Proof of Theorem \ref{thm-evt-multi-factor}, Multi-Factor Case in Proposition \ref{prop-one-factor} and Proposition \ref{prop:vary-m}}
		This proof also contains the proof of the multi-factor case in Proposition \ref{prop:weighted} and Proposition \ref{prop:vary-m}.
		\begin{proof}[Proof of Theorem \ref{thm-evt-multi-factor}, Multi-Factor Case in Proposition \ref{prop-one-factor} and Proposition \ref{prop:vary-m}]
			\texttt{}\\
			In this proof, Step 1-4 prove Theorem  \ref{thm-evt-multi-factor}. Step 5 proves the multi-factor case in Proposition \ref{prop:weighted}. Step 6 proves Proposition \ref{prop:vary-m}.

			Without loss of generality, suppose $\weight_i$ take values $\weight_{0,1}, \weight_{0,2}, \cdots, \weight_{0,n_\weight}$ and $\weight_{0,1} < \weight_{0,2} < \cdots < \weight_{0,n_\weight} = 1$. As a preparation for the proof, we define some notations. We use  $\fidx_j(l)$ to denote the index of $l$-th largest entry in $|\Weight \*V_j|$. We use $\tfidx_j(l)$ to denote the index of $l$-th largest entry in $|\hat{\*\Lambda}_j|$). \\
			\textbf{Step 1: Simplify $\rho$.} \\
			We start with the asymptotic equivalent expression for $\rho$ shown in Lemma \ref{lemma_rho}. 
			For any finite dimension symmetric matrix $M$, from Taylor expansion and mean value theorem, there exist a symmetric matrix $\tilde{M}$,
			$$\tr((I + M)^{-1}) = \tr(I - M + \frac{1}{2} \tilde{M}^2) \geq \tr(I - M),$$ 
			therefore, (recall $Q$ is defined as $Q = \*\Lambda^\T \Weight \wt{\*W}$)
			\begin{align}
			\rho =& \tr \left(   \left(I + \left( \frac{\*F^\T  \*F}{T} \right)^{-1/2}(Q^\T)^{-1} \frac{\wt{\*W}^\T \Weight \*e \*e^\T \Weight \wt{\*W}  }{T}  Q^{-1} \left( \frac{\*F^\T  \*F}{T} \right)^{-1/2} \right)^{-1}  \right)  + o_p(1) \label{eqn:prop3-1}\\
			\nonumber >&  \tr \left( I -  \left( \frac{\*F^\T  \*F}{T} \right)^{-1/2}(Q^\T)^{-1} \frac{ \wt{\*W}^\T \Weight \*e \*e^\T \Weight \wt{\*W} }{T}  Q^{-1} \left( \frac{\*F^\T  \*F}{T} \right)^{-1/2} \right) + o_p(1) \\
			\nonumber =& K - \tr \left( \frac{ \wt{\*W}^\T \Weight \*e \*e^\T \Weight \wt{\*W} }{T} \nonumber \left(Q^\T \frac{\*F^\T  \*F}{T} Q \right)^{-1} \right) + o_p(1) \\
			\nonumber =& K - \tr \left( \frac{ \wt{\*W}^\T \Weight \*e \*e^\T \Weight \wt{\*W} }{T} \left( \frac{1}{T} \wt{\*W}^\T  (\*\Lambda \*F^\T  \*F  \*\Lambda^\T) \wt{\*W} \right)^{-1} \right) + o_p(1) \\
			\geq& K - \norm{\frac{ \wt{\*W}^\T \Weight \*e \*e^\T \Weight \wt{\*W} }{T}}_2  \tr \left( \left( \frac{1}{T} \wt{\*W}^\T  (\*\Lambda \*F^\T  \*F  \*\Lambda^\T) \wt{\*W}\right)^{-1}  \right)  + o_p(1) \label{eqn:prop3-2}\\
			=& K - \norm{\wt{\*W}^\T \Sigma_{e, \weight} \wt{\*W} }_2  \tr \left( \left( \frac{1}{T} \wt{\*W}^\T  (\Lambda \*F^\T  \*F  \*\Lambda^\T) \wt{\*W}\right)^{-1}  \right)  + o_p(1) \label{eqn:prop3-2-2} \\
			=& K - \norm{(\*W^{(\Weight, \tilde{\*F})})^\T \Sigma_{e, \weight} \*W^{(\Weight, \tilde{\*F})} }_2  \tr \left(  \left( \tilde{\*U}^\T   \*U \*U^\T  \tilde{\*U} \right)^{-1}  \right) + o_p(1) \label{eqn:prop3-3} \\
			=& K - \norm{(\*W^{(\Weight, \tilde{\*F})})^\T \Sigma_{e, \weight} \*W^{(\Weight, \tilde{\*F})} }_2  \tr( A^{-1})+ o_p(1)  \label{eqn:prop3-4}
			\end{align}
			where $\Sigma_{e,\weight} := \Weight \Sigma_e \Weight = \Weight \Big(\lim_{T \rightarrow \infty}  \frac{1}{T} \*e \*e^\T \Big) \Weight$, $A =[a_{jl}] = \tilde{\*U}^\T   \*U \*U^\T  \tilde{\*U}$, $\*U=\Weight \*\Lambda H D^{1/2}$, $\*W^{(\Weight, \tilde{\*F})}$ in \eqref{eqn:prop3-3} is defined as 
			\begin{align*}
			\*W^{(\Weight, \tilde{\*F})} = \begin{bmatrix}
			\frac{ (\Weight \*V_1 ) \odot \*M_1^{(\tilde{\*F})} }{\norm{(\Weight \*V_1 ) \odot \*M_1^{(\tilde{\*F})}}} & \frac{(\Weight \*V_2 ) \odot \*M_2^{(\tilde{\*F})} }{\norm{(\Weight \*V_2 ) \odot \*M_2^{(\tilde{\*F})} }} & \cdots &
			\frac{ (\Weight \*V_K ) \odot \*M_K^{(\tilde{\*F})} }{\norm{(\Weight \*V_K ) \odot \*M_K^{(\tilde{\*F})}}}
			\end{bmatrix},
			\end{align*}
			and $\*M_j^{(\tilde{\*F})} $ is the mask vector applied to $\hat{\* \Lambda}_j$, i.e., 
			\begin{align*}
			\wt{\*W} =
			\begin{bmatrix}
			\frac{ \hat{\*\Lambda}_1 \odot \*M_1^{(\tilde{\*F})} }{\norm{\hat{\*\Lambda}_1 \odot \*M_1^{(\tilde{\*F})} }} & \frac{\hat{\*\Lambda}_2 \odot \*M_2^{(\tilde{\*F})} }{\norm{\hat{\*\Lambda}_2 \odot \*M_2^{(\tilde{\*F})} }}  & \cdots &
			\frac{\hat{\*\Lambda}_K \odot \*M_K^{(\tilde{\*F})}}{\norm{\hat{\*\Lambda}_K \odot \*M_K^{(\tilde{\*F})} }}
			\end{bmatrix}, 
			\end{align*}
			We use the superscript $(\tilde{\*F})$ in $\*M_j^{(\tilde{\*F})}$ to denote $\*M_{j,i}^{(\tilde{\*F})} = 1$ if $i$ is among the indicies of  indices of largest $m$ entries in $\hat{\* \Lambda}_j$ and 0 otherwise. Note that in this proof, we need to separate the indices of largest $m$ entries in $|\hat{\* \Lambda}_j|$ from the indicies of largest $m$ entries in $|(\*\Lambda H)_j|$ (in this proof, $\*M_j(\tilde{\*F})$  is the same as the $\*M_j$ in \eqref{eqn-tilde-lambda-def}).  
			
			We use the superscript $(\Weight, \tilde{\*F})$ in $\*W^{(\Weight, \tilde{\*F})}$ to denote that  the sparse weights based on the rotated and weighted population loadings $\Weight \*V =\Weight \*\Lambda H$, and the sparse weights are nonzero if the corresponding value in $\hat{\*\Lambda}$ is among the largest $m$ entries. As a comparison, entries in $\*W^{(\Weight)}$ used at the end of Step 3 are nonzero if the corresponding value in $\Weight \*\Lambda H$ is among the largest $m$ entries.
			
			Moreover, $\tilde{\*U}$ defined as follows is equal to $\*W^{(\Weight, \tilde{\*F})} $.
			\begin{align*}
			\tilde{\*U}  &=
			\begin{bmatrix}
			\frac{ \*U_1 \odot \*M_1^{(\tilde{\*F})} }{\norm{\*U_1 \odot \*M_1^{(\tilde{\*F})}}} & \frac{ \*U_2 \odot \*M_2^{(\tilde{\*F})} }{\norm{\*U_2 \odot \*M_2^{(\tilde{\*F})} }} & \cdots &
			\frac{ \*U_K \odot \*M_K^{(\tilde{\*F})} }{\norm{\*U_K \odot \*M_K^{(\tilde{\*F})}}}
			\end{bmatrix} \\
			&= \begin{bmatrix}
			\frac{ (\Weight \*V_1 D_1^{1/2}) \odot \*M_1^{(\tilde{\*F})} }{\norm{(\Weight \*V_1 D_1^{1/2}) \odot \*M_1^{(\tilde{\*F})}}} & \frac{(\Weight \*V_2 D_2^{1/2}) \odot \*M_2^{(\tilde{\*F})} }{\norm{(\Weight \*V_2 D_2^{1/2}) \odot \*M_2^{(\tilde{\*F})} }} & \cdots &
			\frac{ (\Weight \*V_K D_K^{1/2}) \odot \*M_K^{(\tilde{\*F})} }{\norm{(\Weight \*V_K D_K^{1/2}) \odot \*M_K^{(\tilde{\*F})}}}
			\end{bmatrix}  \\
			&=\begin{bmatrix}
			\frac{ (\Weight \*V_1 ) \odot \*M_1^{(\tilde{\*F})} }{\norm{(\Weight \*V_1 ) \odot \*M_1^{(\tilde{\*F})}}} & \frac{(\Weight \*V_2 ) \odot \*M_2^{(\tilde{\*F})} }{\norm{(\Weight \*V_2 ) \odot \*M_2^{(\tilde{\*F})} }} & \cdots &
			\frac{ (\Weight \*V_K ) \odot \*M_K^{(\tilde{\*F})} }{\norm{(\Weight \*V_K ) \odot \*M_K^{(\tilde{\*F})}}}
			\end{bmatrix} = \*W^{(\Weight, \tilde{\*F})}
			\end{align*}
			Let us explain more about how we obtain \eqref{eqn:prop3-4}.  Equation \eqref{eqn:prop3-1} follows from Lemma \ref{lemma_rho}. For Inequality \eqref{eqn:prop3-2}, we use Von Neumann's trace inequality (For matrix $A$ and $B$, we have $\tr (AB) \leq \sum_i \alpha_i \beta_i \leq \max_i \alpha_i \sum_i \beta_i = \max_i \alpha_i \cdot \tr(B)$, where $\alpha_i$ and $\beta_i$ are eigenvalues of $A$ and $B$). Equation \eqref{eqn:prop3-2-2} holds because there are only finitely many nonzeros in each column of $\wt{\*W}$, each column of $\wt{\*W}$ has norm 1 and $\tr \left( \left( \frac{1}{T} \wt{\*W}^\T  (\Lambda \*F^\T  \*F  \*\Lambda^\T) \wt{\*W}\right)^{-1}  \right) = O_p(1)$ (because $\rho \leq K$).  Equation \eqref{eqn:prop3-3} follows from Lemma \ref{lemma_u}.  
			
			Since the CDF of $|v_i|$, $F_{|v_i|}(v)= F_{(|v_{i,1}|, \cdots, |v_{i,K}|)}(v)$, is continuous in $v$, we have  $F_{|u_i|}(u)$ is continuous in $u$. From the fact that convergence in probability implies convergence in distribution, we can use \eqref{eqn:prop3-4} to provide a probabilistic bound for $P(\rho > \rho_0)$, 
			\begin{align*}
			\lim_{N, T \rightarrow \infty} P(\rho > \rho_0) &\geq \lim_{N \rightarrow \infty} P(K - \norm{(\*W^{(\Weight, \tilde{\*F})})^\T \Sigma_{e, \weight} \*W^{(\Weight, \tilde{\*F})} }_2  \tr( A^{-1}) > \rho_0) \\
			&= \lim_{N \rightarrow \infty} P\left(\tr( A^{-1}) < \frac{K - \rho_0}{\norm{(\*W^{(\Weight, \tilde{\*F})})^\T \Sigma_{e, \weight} \*W^{(\Weight, \tilde{\*F})} }_2  } \right)    
			\end{align*}
			\textbf{Step 2: Simplify $\tr( A^{-1})$ in terms of the largest rotated and scaled population loadings $\*U$. } \\
			First, we provide the expression for every entry in $A$. Since $\{\tfidx_j(1),\tfidx_j(2), \cdots, \tfidx_j(m)  \}$ are indices of largest $m$ entries in $|\hat{\*\Lambda}_j|$, then indices of nonzero entries in both  $\wt{\*W}_j$ and  $\*U_j$ are\\ $\{\tfidx_j(1),\tfidx_j(2), \cdots, \tfidx_j(m)  \}$ as well. The diagonal entries in $A$ can be written as 
			\[a_{jj}= \frac{\sum_{k=1}^K (\sum_{i = 1}^{m} u_{\tfidx_j(i),j} u_{\tfidx_j(i),k} )^2}{\sum_{i = 1}^{m} u_{\tfidx_j(i),j}^2}\]
			and the off-diagonal entries in $A$ can be written as
			$$
			a_{jl} = \frac{\sum_{k=1}^K  (\sum_{i = 1}^{m} u_{\tfidx_j(i),j} u_{\tfidx_j(i),k} ) (\sum_{i = 1}^{m} u_{\tfidx_l(i),k} u_{\tfidx_l(i),l} ) }{ (\sum_{i = 1}^{m} u_{\tfidx_j(i),j}^2)^{1/2}  (\sum_{i = 1}^{m} u_{\tfidx_l(i),l}^2)^{1/2} },  \quad \forall j \neq l
			$$
			where $\*U = \Weight \*\Lambda H D^{1/2} = \Weight \*V D^{1/2}$. We decompose $A$ as $A= (\tilde{\*U}^\T   \*U) (\*U^\T  \tilde{\*U}) = D^{(\tilde{\*F})} (B^{(\tilde{\*F})})^\T B^{(\tilde{\*F})} D^{(\tilde{\*F})}$, where $D^{(\tilde{\*F})}$ is a diagonal matrix with $d^{(\tilde{\*F})}_{jj} = (\sum_{i = 1}^{m} u_{\tfidx_j(i),j}^2)^{1/2}$ and the diagonal entries in $B^{(\tilde{\*F})}$ are 1, while the off-diagonal entries have
			\[b_{kl}^{(\tilde{\*F})} =  \frac{ \sum_{i = 1}^{m} u_{\tfidx_l(i),k} u_{\tfidx_l(i),l}  }{  \sum_{i = 1}^{m}  u_{\tfidx_l(i),l}^2  } = \frac{D_k^{1/2} \sum_{i = 1}^{m}\weight_{\tfidx_l(i)}^2 v_{\tfidx_l(i),k} v_{\tfidx_l(i),l}  }{ D_l^{1/2} \sum_{i = 1}^{m}\weight_{\tfidx_l(i)}^2 v_{\tfidx_l(i),l}^2  },  \quad \text{for all $k$ and $l$}. \]
			We use the superscript $(\tilde{\*F})$ in $D^{(\tilde{\*F})}$ and $B^{(\tilde{\*F})}$. Both $D^{(\tilde{\*F})}$ and $B^{(\tilde{\*F})}$ involve the indicies of  largest $m$ entries in $|\hat{\* \Lambda}_j|$ for all $j$. On the other hand, $B$ (without the superscript $(\tilde{\*F})$) involves the indicies of  largest $m$ entries in $|(\Weight \*\Lambda H)_j|$ for all $j$.
			
			Since $(B^{(\tilde{\*F})})^\T B^{(\tilde{\*F})}$ is positive semidefinite, we can provide an upper bound for $\tr(A^{-1})$ using the matrix inequality
			\begin{eqnarray*}
				\tr(A^{-1}) &=& \tr((D^{(\tilde{\*F})} (B^{(\tilde{\*F})})^\T B^{(\tilde{\*F})} D^{(\tilde{\*F})} )^{-1}) = \tr((D^{(\tilde{\*F})})^{-1} ((B^{(\tilde{\*F})})^\T B^{(\tilde{\*F})})^{-1} (D^{(\tilde{\*F})})^{-1}) \\
				&\leq& \lambda_{\max}(((B^{(\tilde{\*F})})^\T B^{(\tilde{\*F})})^{-1}) \tr((D^{(\tilde{\*F})})^{-2}) =\frac{1}{\lambda_{\min}((B^{(\tilde{\*F})})^\T B^{(\tilde{\*F})})}  \tr((D^{(\tilde{\*F})})^{-2}) \\
				&=&  \frac{1}{\sigma_{\min}^2(B^{(\tilde{\*F})})}  \tr((D^{(\tilde{\*F})})^{-2}) = \frac{1}{\sigma_{\min}^2(B^{(\tilde{\*F})})}  \sum_{j=1}^K \frac{1}{\sum_{i = 1}^{m} u_{\tfidx_j(i),j}^2},
			\end{eqnarray*}
			where $\lambda_{\max}(M)$ and $\lambda_{\min}(M)$ represent the maximum and minimum eigenvalue of matrix $M$,  $\sigma_{\max}(M)$ and $\sigma_{\min}(M)$ represent the maximum and minimum singular value of matrix $M$.  \\
			\texttt{} \\
			\textbf{Step 3: Provide a lower bound for $P\left(\tr( A^{-1}) < \frac{K - \rho_0}{\norm{(\*W^{(\Weight, \tilde{\*F})})^\T \Sigma_{e, \weight} \*W^{(\Weight, \tilde{\*F})} }_2  } \right) $. } \\
			Let us first define three events $E_{1,\weight} = \{\sigma_{\min}(B^{(\tilde{\*F})}  ) \geq \underline{\gamma} \}$ for some $\underline{\gamma}$, \\ $E_{2,\weight} = \left\lbrace \sum_{j=1}^K \frac{1}{D_j \weight^2_{\fidx_j(m)} v^2_{\fidx_j(m),j} } <  \frac{m(K-\rho_0)\underline{\gamma}^2}{\norm{(\*W^{(\Weight, \tilde{\*F})})^\T \Sigma_{e, \weight} \*W^{(\Weight, \tilde{\*F})} }_2 }  \right\rbrace$,\\ and  $E_{3,\weight} = \left\lbrace  \tr(A^{-1}) < \frac{K-\rho_0}{\norm{(\*W^{(\Weight, \tilde{\*F})})^\T \Sigma_{e, \weight} \*W^{(\Weight, \tilde{\*F})} }_2 } \right\rbrace$. Now, we would like to show
			\[P(E_{1,\weight} \cap E_{3,\weight}) \geq P(E_{1,\weight} \cap E_{2,\weight}).\]
			From Lemma \ref{lemma_orderstat} and the proof of Proposition \ref{prop-one-factor}, for all $i$ and $j$, and for any $c$, if the $m$-th largest population entry $v_{\fidx_j(m),j} > c$, then the unit $\fidx_j(i)$ selected from our method (with largest estimated loading) has $v_{\fidx_j(i), j} > c$ with probability 1 as $N, T \rightarrow \infty$. Recall the definition of $u_{i,j}: = \weight_i v_{i,j} s_{j}^{1/2} $. We have similar argument for $u_{i,j}$, that is, for all $i$ and $j$, and for any $c$,  if  the $m$-th largest population entry  $u_{\fidx_j(m),j} > c$,  then the unit $\fidx_j(i)$ selected from our method (with largest estimated loading) has $u_{\fidx_j(i), j} > c$ with probability 1 as $N, T \rightarrow \infty$. Using this property, we can provide a sufficient condition for event $E_{3,\weight}$. That is, if event $E_{1,\weight}$ happens, we have
			\begin{align*}
			& \sum_{j=1}^K \frac{1}{D_j \weight^2_{\fidx_j(m)} v^2_{\fidx_j(m),j} } <  \frac{m(K-\rho_0)\underline{\gamma}^2}{\norm{(\*W^{(\Weight, \tilde{\*F})})^\T \Sigma_{e, \weight} \*W^{(\Weight, \tilde{\*F})} }_2} \\
			\Rightarrow& \sum_{j=1}^K \frac{1}{m u^2_{\fidx_j(m),j}} <  \frac{(K-\rho_0) \underline{\gamma}^2 }{\norm{(\*W^{(\Weight, \tilde{\*F})})^\T \Sigma_{e, \weight} \*W^{(\Weight, \tilde{\*F})} }_2} \\
			\overset{\text{w.p.1}}{\Rightarrow} & \sum_{j=1}^K \frac{1}{\sum_{i = 1}^{m} u_{\tfidx_j(i),j}^2} <  \frac{(K-\rho_0) \underline{\gamma}^2}{\norm{(\*W^{(\Weight, \tilde{\*F})})^\T \Sigma_{e, \weight} \*W^{(\Weight, \tilde{\*F})} }_2} \\
			\Rightarrow& \frac{1}{\underline{\gamma}^2   }  \sum_{j=1}^K \frac{1}{\sum_{i = 1}^{m} u_{\tfidx_j(i),j}^2} <  \frac{(K-\rho_0) }{\norm{(\*W^{(\Weight, \tilde{\*F})})^\T \Sigma_{e, \weight} \*W^{(\Weight, \tilde{\*F})} }_2}  \\
			\overset{\text{under}\,\, E_{1,\weight}}{\Rightarrow} & \frac{1}{\sigma_{\min}^2(B^{(\tilde{\*F})} )}  \sum_{j=1}^K \frac{1}{\sum_{i = 1}^{m} u_{\tfidx_j(i),j}^2}  <  \frac{K-\rho_0}{\norm{(\*W^{(\Weight, \tilde{\*F})})^\T \Sigma_{e, \weight} \*W^{(\Weight, \tilde{\*F})} }_2} \\
			\Rightarrow& \tr(A^{-1}) < \frac{K-\rho_0}{\norm{(\*W^{(\Weight, \tilde{\*F})})^\T \Sigma_{e, \weight} \*W^{(\Weight, \tilde{\*F})} }_2}
			\end{align*}
			where $\weight_{\fidx_j(m)} $ is the weight for the unit that is the $m$-th order statistic $v_{\fidx_j(m),j} $ in $|\*V_j|$. ``$\Rightarrow$'' denotes if the event on the left-hand side happens, then the event on the right-hand side happens. Hence, if event $E_{1,\weight}$ and $E_{2,\weight}$ happens, then $E_{3,\weight}$ happens. We have 
			\begin{align*}
			P( E_{3,\weight})  &\geq P(E_{1,\weight} \cap E_{3,\weight}) \geq P(E_{1,\weight} \cap E_{2,\weight}) + o_p(1) \\
			&= P(E_{1,\weight} ) + P(E_{2,\weight}) - P(E_{1,\weight} \cup E_{2,\weight})  + o_p(1)    \\
			&\geq P(E_{1,\weight} ) + P(E_{2,\weight}) - 1 + o_p(1)    \\
			&= P(E_{2,\weight} ) - P(E_{1,\weight}^\complement) + o_p(1)
			\end{align*}
			Then we have 
			\begin{align*}
			&\lim_{N, T \rightarrow \infty} P(\rho > \rho_0) \geq  \lim_{N \rightarrow \infty} P\left(\tr( A^{-1}) <\frac{K - \rho_0}{\norm{(\*W^{(\Weight, \tilde{\*F})})^\T \Sigma_{e, \weight} \*W^{(\Weight, \tilde{\*F})} }_2} \right) = \lim_{N \rightarrow \infty} (P(E_{2,\weight} ) - P(E_{1,\weight}^\complement)) \\
			\geq&  \lim_{N \rightarrow \infty} P \left( \sum_{j=1}^K \frac{1}{D_j \weight^2_{\fidx_j(m)} v^2_{\fidx_j(m),j} } <  \frac{m(K-\rho_0)\underline{\gamma}^2}{\norm{(\*W^{(\Weight, \tilde{\*F})})^\T \Sigma_{e, \weight} \*W^{(\Weight, \tilde{\*F})} }_2}  \right) - \lim_{N \rightarrow \infty} P(\sigma_{\min}(B^{(\tilde{\*F})} ) < \underline{\gamma})
			\end{align*}
			Next, let us define two more events $E_{4,\weight} = \Big\lbrace \norm{(\*W^{(\Weight, \tilde{\*F})})^\T \Sigma_{e, \weight} \*W^{(\Weight, \tilde{\*F})} }_2 < \bar \gamma_{\*V,e} \sigma_{e,\weight}^2 \Big\rbrace$ and \\ $E_{5,\weight} = \Big\lbrace  \sum_{j=1}^K \frac{1}{D_j \weight^2_{\tfidx_j(m)} v^2_{\tfidx_j(m),j} } <  \frac{m(K-\rho_0)\underline{\gamma}^2}{\bar \gamma_{\*V,e} \sigma_{e,\weight}^2 }  \Big\rbrace$.  If event $E_{4,\weight}$ holds, we have 
			\begin{align*}
			&\sum_{j=1}^K \frac{1}{D_j \weight^2_{\tfidx_j(m)} v^2_{\tfidx_j(m),j} } <  \frac{m(K-\rho_0)\underline{\gamma}^2}{\bar \gamma_{\*V,e} \sigma_{e,\weight}^2 } \\
			\overset{\text{under}\,\, E_{4,\weight}}{\Rightarrow}& \sum_{j=1}^K \frac{1}{D_j \weight^2_{\tfidx_j(m)} v^2_{\tfidx_j(m),j} } <  \frac{m(K-\rho_0)\underline{\gamma}^2}{\norm{(\*W^{(\Weight, \tilde{\*F})})^\T \Sigma_{e, \weight} \*W^{(\Weight, \tilde{\*F})} }_2}  
			\end{align*}
			Then we have 
			\[P(E_{2,\weight} )  \geq P(E_{4,\weight} \cap E_{5,\weight}) \geq P(E_{5,\weight}) + E_{4,\weight} - 1 \geq P(E_{5,\weight}) - P(E_{4,\weight}^\complement). \]
			Hence the probability lower bound for $P(\rho > \rho_0 )$ is 
			\begin{align}
			\nonumber\lim_{N, T \rightarrow \infty} P(\rho > \rho_0) 
			\geq&  \lim_{N \rightarrow \infty} P \left( \sum_{j=1}^K \frac{1}{D_j \weight^2_{\tfidx_j(m)} v^2_{\tfidx_j(m),j} } <  \frac{m(K-\rho_0)\underline{\gamma}^2}{\bar \gamma_{\*V,e} \sigma_{e,\weight}^2  }  \right) \\
			&- \lim_{N \rightarrow \infty} P\Big(\sigma_{\min}(B^{(\tilde{\*F})} ) < \underline{\gamma}\Big) - \lim_{N \rightarrow\infty} P\Big(\norm{(\*W^{(\Weight, \tilde{\*F})})^\T \Sigma_{e, \weight} \*W^{(\Weight, \tilde{\*F})} }_2 > \bar \gamma_{\*V,e}  \sigma_{e,\weight}^2 \Big), 
			\end{align}
			where $ \*W^{(\Weight, \tilde{\*F})} =\begin{bmatrix}
			\frac{ (\Weight \*V_1 ) \odot \*M_1^{(\tilde{\*F})} }{\norm{(\Weight \*V_1 ) \odot \*M_1^{(\tilde{\*F})}}} & \frac{(\Weight \*V_2 ) \odot \*M_2^{(\tilde{\*F})} }{\norm{(\Weight \*V_2 ) \odot \*M_2^{(\tilde{\*F})} }} & \cdots &
			\frac{ (\Weight \*V_K ) \odot \*M_K^{(\tilde{\*F})} }{\norm{(\Weight \*V_K ) \odot \*M_K^{(\tilde{\*F})}}}
			\end{bmatrix}$ and \\ $B^{(\tilde{\*F})}   = \Big[ \frac{D_k^{1/2} \sum_{i = 1}^{m}\weight_{\tfidx_l(i)}^2 v_{\tfidx_l(i),k} v_{\tfidx_l(i),l}  }{D_l^{1/2}  \sum_{i = 1}^{m}\weight_{\tfidx_l(i)}^2 v_{\tfidx_l(i),l}^2  } \Big]$. 
			Since both $\Lambda_i$ and $v_i$ are identically distributed for all $i$, following a similar argument as the one at the beginning of Step 3, we can replace $\tfidx_l(i)$ by $\fidx_l(i)$ in asymptotics for $B^{(\tilde{\*F})} $ and we have
			\[\lim_{N \rightarrow \infty} P\Big(\sigma_{\min}(B^{(\tilde{\*F})} ) < \underline{\gamma}\Big)  = \lim_{N \rightarrow \infty} P\Big(\sigma_{\min}(B ) < \underline{\gamma}\Big),  \]
			where $B = \Big[ \frac{D_k^{1/2} \sum_{i = 1}^{m}\weight_{\fidx_l(i)}^2 v_{\fidx_l(i),k} v_{\fidx_l(i),l}  }{D_l^{1/2}  \sum_{i = 1}^{m}\weight_{\fidx_l(i)}^2 v_{\fidx_l(i),l}^2  } \Big]$. 
			
			Similarly, we can replace $\tfidx_l(i)$ by $\fidx_l(i)$ in asymptotics for $\*W^{(\Weight, \tilde{\*F})}$ 
			\[\lim_{N \rightarrow\infty} P\Big(\norm{(\*W^{(\Weight, \tilde{\*F})})^\T \Sigma_{e, \weight} \*W^{(\Weight, \tilde{\*F})} }_2 > \bar \gamma_{\*V,e} \sigma_{e,\weight}^2 \Big) = \lim_{N \rightarrow\infty} P\Big(\norm{({\*W}^{(\Weight)})^\T \Sigma_{e, \weight} {\*W}^{(\Weight)} }_2 > \bar \gamma_{\*V,e} \sigma_{e,\weight}^2 \Big), \]
			where ${\*W}^{(\Weight)} =\begin{bmatrix}
			\frac{ (\Weight \*V_1 ) \odot \*M_1 }{\norm{(\Weight \*V_1 ) \odot \*M_1 }} & \frac{(\Weight \*V_2 ) \odot \*M_2 }{\norm{(\Weight \*V_2 ) \odot \*M_2 }} & \cdots &
			\frac{ (\Weight \*V_K ) \odot \*M_K }{\norm{(\Weight \*V_K ) \odot \*M_K}}
			\end{bmatrix}$ (entries in $\*W^{(\Weight)}$  are nonzero if the corresponding value in $\Weight \*\Lambda H$ is among the largest $m$ entries.). Then the probability lower bound for $P(\rho > \rho_0 )$ is 
			\begin{align}
			\nonumber\lim_{N, T \rightarrow \infty} P(\rho > \rho_0) 
			\geq&  \lim_{N \rightarrow \infty} P \left( \sum_{j=1}^K \frac{1}{D_j \weight^2_{\fidx_j(m)} v^2_{\fidx_j(m),j} } <  \frac{m(K-\rho_0)\underline{\gamma}^2}{\bar \gamma_{\*V,e} \sigma_{e,\weight}^2  }  \right) \\
			&- \lim_{N \rightarrow \infty} P\Big(\sigma_{\min}(B ) < \underline{\gamma}\Big) - \lim_{N \rightarrow\infty} P\Big(\norm{({\*W}^{(\Weight)})^\T \Sigma_{e, \weight} {\*W}^{(\Weight)} }_2 > \bar \gamma_{\*V,e}  \sigma_{e,\weight}^2 \Big),  \label{eqn:bound-multi-factor-general} 
			\end{align}
			\\
			\textbf{Step 4: Consider the special case where $\weight_i = 1$ for all $i$.} \\
			When $\weight_i = 1$ for all $i$, we have $\Sigma_{e,\weight} = \Sigma_e$, $\sigma_{e,\weight}^2 = \sigma_{e}^2$, $ {\*W}^{(\Weight )}  =\begin{bmatrix}
			\frac{ \*V_1  \odot \*M_1 }{\norm{\*V_1  \odot \*M_1 }} & \frac{\*V_2  \odot \*M_2 }{\norm{\*V_2  \odot \*M_2 }} & \cdots &
			\frac{ \*V_K  \odot \*M_K }{\norm{ \*V_K  \odot \*M_K }}
			\end{bmatrix} := {\*W}$ and $B  = \Big[ \frac{D_k^{1/2} \sum_{i = 1}^{m}v_{ \fidx_l(i),k} v_{\fidx_l(i),l}  }{ D_l^{1/2} \sum_{i = 1}^{m} v_{\fidx_l(i),l}^2  } \Big]$. The probability lower bound for $P(\rho > \rho_0 )$ can be simplified to 
			\begin{align}
			\nonumber\lim_{N, T \rightarrow \infty} P(\rho > \rho_0) 
			\geq&  \lim_{N \rightarrow \infty} P \left( \sum_{j=1}^K \frac{1}{D_j v^2_{\fidx_j(m),j} } <  \frac{m(K-\rho_0)\underline{\gamma}^2}{\bar \gamma_{\*V,e} \sigma_{e}^2  }  \right) \\
			&- \lim_{N \rightarrow \infty} P\Big(\sigma_{\min}(B) < \underline{\gamma}\Big) - \lim_{N \rightarrow\infty} P\Big(\norm{{\*W}^\T \Sigma_e {\*W}}_2 > \bar \gamma_{\*V,e} \sigma_{e}^2 \Big) \label{eqn:multi-factor-lower-bound-evt}
			\end{align}
			We set $\rho_0$ as $\rho_0 = K - \frac{\bar \gamma_{\*V,e} \sigma_{e}^2}{m \underline{\gamma}^2 } \sum_{j=1}^K \frac{1}{D_j y_j^2}$. Note that 
			\begin{align*}
			&\forall j, \, |v_{\fidx_j(m),j}| > y_k \Rightarrow \sum_{j = 1}^K \frac{1}{D_j v_{\fidx_j(m),j}^2} <  \sum_{j = 1}^K \frac{1}{D_j y_j^2} 
			\Rightarrow  \sum_{j = 1}^K \frac{1}{D_j v_{\fidx_j(m),j}^2}  < \frac{m(K-\rho_0) \underline{\gamma}^2 }{\bar \gamma_{\*V,e} \sigma_{e}^2}   
			\end{align*}
			Then the probability lower bound has 
			\begin{align}
			\nonumber & \lim_{N, T\rightarrow \infty} P\left( \rho \geq K - \frac{\bar \gamma_{\*V,e} \sigma_{e}^2}{m\underline{\gamma}^2 } \sum_{j=1}^K \frac{1}{D_j y_j^2} \right)   \\
			\geq&   \bar{G}_m(y_1, \cdots, y_K)  - \lim_{N \rightarrow \infty} P\Big(\sigma_{\min}(B) < \underline{\gamma}\Big) - \lim_{N \rightarrow\infty} P\Big(\norm{{\*W}^\T \Sigma_e {\*W}}_2 > \bar \gamma_{\*V,e} \sigma_{e}^2 \Big).
			\end{align}
			\textbf{Step 5: Consider the general case where $\weight_i \neq 1$ for some $i$ (Proof of the multi-factor case in Proposition \ref{prop:weighted})} \\
			We still set $\rho_0$ as $\rho_0 = K - \frac{\bar \gamma_{\*V,e} \sigma_{e}^2}{m \underline{\gamma}^2 } \sum_{j=1}^K \frac{1}{D_j y_j^2}$. Since $\frac{m(K-\rho_0)\underline{\gamma}^2}{\bar \gamma_{\*V,e} \sigma_{e}^2  }   = \frac{\sigma_e^2}{\sigma_{e,\weight}^2 }   \sum_{j = 1}^K \frac{1}{D_j y_j^2} $, we have
			\begin{align*}
			&\forall j, \, |v_{\fidx_j(m),j}| > \frac{\sigma_{e,\weight}}{\sigma_e} y_k \Rightarrow \sum_{j = 1}^K \frac{1}{D_j v_{\fidx_j(m),j}^2} <  \frac{\sigma_e^2}{\sigma_{e,\weight}^2 }   \sum_{j = 1}^K \frac{1}{D_j y_j^2} 
			\Rightarrow  \sum_{j = 1}^K \frac{1}{D_j v_{\fidx_j(m),j}^2}  < \frac{m(K-\rho_0) \underline{\gamma}^2 }{\bar \gamma_{\*V,e} \sigma_{e,\weight}^2}   
			\end{align*}
			Then the probability lower bound has 
			\begin{align}
			\nonumber & \lim_{N, T\rightarrow \infty} P\left( \rho \geq K - \frac{\bar \gamma_{\*V,e} \sigma_{e}^2}{m\underline{\gamma}^2 } \sum_{j=1}^K \frac{1}{D_j y_j^2} \right)   \\
			\geq&  \bar{G}_{m,\weight}\left(\frac{ \sigma_{e,\weight}}{ \sigma_{e}} y_{1}, \cdots, \frac{ \sigma_{e,\weight}}{ \sigma_{e}} y_{K} \right)  - \lim_{N \rightarrow \infty} P\Big(\sigma_{\min}(B) < \underline{\gamma}\Big) -\lim_{N \rightarrow\infty} P\Big(\norm{({\*W}^{(\Weight)})^\T \Sigma_{e, \weight} {\*W}^{(\Weight)} }_2 > \bar \gamma_{\*V,e}  \sigma_{e,\weight}^2 \Big).
			\end{align}
			If the weight is inverse proportional to error's standard error, then  we have $\weight_i  |v_{ij}| \geq \theta_{0,1}   |v_{ij}| = \frac{\sigma_{e,\weight}}{ \sigma_{e}}  |v_{ij}|  $ for all $i$ and 
			\begin{align*}
			\bar{G}_{m,\weight}\left(\frac{ \sigma_{e,\weight}}{ \sigma_{e}} y_{1}, \cdots, \frac{ \sigma_{e,\weight}}{ \sigma_{e}} y_{K} \right)    &= \lim_{N \rightarrow \infty} P\left(|\weight v|_{(m),1} \geq \frac{ \sigma_{e,\weight}}{ \sigma_{e}} y_1, \cdots, |\weight v|_{(m),K} \geq \frac{ \sigma_{e,\weight}}{ \sigma_{e}} y_K\right) \\
			&\geq \lim_{N \rightarrow \infty} P\left(| v|_{(m),1} \geq y_1, \cdots, | v|_{(m),K} \geq y_K\right) = \bar{G}_m(y_1, \cdots, y_K)
			\end{align*}
			\texttt{}\\
			\textbf{Step 6: Consider the number of nonzeros vary for different factors(Proof of Proposition \ref{prop:vary-m})}  \\
			Suppose the number of nonzeros is $m_1, m_2, \cdots, m_K$ for the first, second, ..., $K$-th factors. We need to revisit Step 2 and 3 to simplify $\tr(A^\I)$ and provide a lower bound for\\ $P\left(\tr( A^{-1}) < \frac{K - \rho_0}{\norm{(\*W^{(\Weight, \tilde{\*F})})^\T \Sigma_{e, \weight} \*W^{(\Weight, \tilde{\*F})} }_2  } \right) $.  
			
			The diagonal entries in $A$ can be written as ($m$ is changed to $m_j$)
			\[a_{jj}= \frac{\sum_{k=1}^K (\sum_{i = 1}^{m_j} u_{\tfidx_j(i),j} u_{\tfidx_j(i),k} )^2}{\sum_{i = 1}^{m_j} u_{\tfidx_j(i),j}^2}\]
			and the off-diagonal entries in $A$ can be written as ($m$ is changed to $m_j$ and $m_l$)
			$$
			a_{jl} = \frac{\sum_{k=1}^K  (\sum_{i = 1}^{m_j} u_{\tfidx_j(i),j} u_{\tfidx_j(i),k} ) (\sum_{i = 1}^{m_l} u_{\tfidx_l(i),k} u_{\tfidx_l(i),l} ) }{ (\sum_{i = 1}^{m_j} u_{\tfidx_j(i),j}^2)^{1/2}  (\sum_{i = 1}^{m_l} u_{\tfidx_l(i),l}^2)^{1/2} }, \quad \forall j \neq l
			$$
			Similarly, we decompose $A$ as $A= (\tilde{\*U}^\T   \*U) (\*U^\T  \tilde{\*U})= D^{(\tilde{\*F})} (B^{(\tilde{\*F})})^\T B^{(\tilde{\*F})} D^{(\tilde{\*F})}$, where $D^{(\tilde{\*F})}$ is a diagonal matrix with $d^{(\tilde{\*F})}_{jj} = (\sum_{i = 1}^{m} u_{\tfidx_j(i),j}^2)^{1/2}$ and the diagonal entries in $B^{(\tilde{\*F})}$ are 1, while the off-diagonal entries are ($m$ is changed to $m_l$)
			\[b_{kl}^{(\tilde{\*F})} =  \frac{ \sum_{i = 1}^{m_l} u_{\tfidx_l(i),k} u_{\tfidx_l(i),l}  }{  \sum_{i = 1}^{m_l}  u_{\tfidx_l(i),l}^2  } = \frac{D_k^{1/2} \sum_{i = 1}^{m_l}\weight_{\tfidx_l(i)}^2 v_{\tfidx_l(i),k} v_{\tfidx_l(i),l}  }{ D_l^{1/2} \sum_{i = 1}^{m_l}\weight_{\tfidx_l(i)}^2 v_{\tfidx_l(i),l}^2  },  \quad \text{for all $k$ and $l$}. \]
			Since $(B^{(\tilde{\*F})})^\T B^{(\tilde{\*F})}$ is positive semidefinite, similar as Step 2, we have  ($m$ is changed to $m_j$)
			\begin{eqnarray*}
				tr(A^{-1})  \leq \frac{1}{\sigma_{\min}^2(B^{(\tilde{\*F})})}  \sum_{j=1}^K \frac{1}{\sum_{i = 1}^{m_j} u_{\tfidx_j(i),j}^2}.
			\end{eqnarray*}
			We modify events $E_{2,\weight}$ in Step 3 as  $E_{2,\weight} = \left\lbrace \sum_{j=1}^K \frac{1}{m_j D_j \weight^2_{\fidx_j(m_j)} v^2_{\fidx_j(m_j),j} } <  \frac{(K-\rho_0)\underline{\gamma}^2}{\norm{(\*W^{(\Weight, \tilde{\*F})})^\T \Sigma_{e, \weight} \*W^{(\Weight, \tilde{\*F})} }_2 }  \right\rbrace$. If event $E_{1,\weight}$ (defined in Step 3) happens, we have
			\begin{align*}
			& \sum_{j=1}^K \frac{1}{m_j D_j \weight^2_{\fidx_j(m_j)} v^2_{\fidx_j(m_j),j} } <  \frac{(K-\rho_0)\underline{\gamma}^2}{\norm{(\*W^{(\Weight, \tilde{\*F})})^\T \Sigma_{e, \weight} \*W^{(\Weight, \tilde{\*F})} }_2} \\
			\overset{\text{w.p.1}}{\Rightarrow} & \sum_{j=1}^K \frac{1}{\sum_{i = 1}^{m_j} u_{\tfidx_j(i),j}^2} <  \frac{(K-\rho_0) \underline{\gamma}^2}{\norm{(\*W^{(\Weight, \tilde{\*F})})^\T \Sigma_{e, \weight} \*W^{(\Weight, \tilde{\*F})} }_2} \\
			\overset{\text{under}\,\, E_{1,\weight}}{\Rightarrow} & \frac{1}{\sigma_{\min}^2(B^{(\tilde{\*F})} )}  \sum_{j=1}^K \frac{1}{\sum_{i = 1}^{m_j} u_{\tfidx_j(i),j}^2}  <  \frac{K-\rho_0}{\norm{(\*W^{(\Weight, \tilde{\*F})})^\T \Sigma_{e, \weight} \*W^{(\Weight, \tilde{\*F})} }_2} \\
			\Rightarrow& tr(A^{-1}) < \frac{K-\rho_0}{\norm{(\*W^{(\Weight, \tilde{\*F})})^\T \Sigma_{e, \weight} \*W^{(\Weight, \tilde{\*F})} }_2}
			\end{align*}
			We follow the remaining steps in Step 3 and have the probabilistic bound for $P(\rho > \rho_0)$
			\begin{align}
			\nonumber\lim_{N, T \rightarrow \infty} P(\rho > \rho_0) 
			\geq&  \lim_{N \rightarrow \infty} P \left( \sum_{j=1}^K \frac{1}{m_j D_j \weight^2_{\fidx_j(m_j)} v^2_{\fidx_j(m_j),j} } <  \frac{(K-\rho_0)\underline{\gamma}^2}{\bar \gamma_{\*V,e} \sigma_{e,\weight}^2  }  \right) \\
			&- \lim_{N \rightarrow \infty} P\Big(\sigma_{\min}(B ) < \underline{\gamma}\Big) - \lim_{N \rightarrow\infty} P\Big(\norm{({\*W}^{(\Weight)})^\T \Sigma_{e, \weight} {\*W}^{(\Weight)} }_2 > \bar \gamma_{\*V,e}  \sigma_{e,\weight}^2 \Big),  \label{eqn:bound-multi-factor-general-vary-m} 
			\end{align}
			If $\weight_i = 1$ for all $i$, then we have the probability lower bound has 
			\begin{align}
			\nonumber & \lim_{N, T\rightarrow \infty} P\left( \rho \geq K - \frac{\bar \gamma_{\*V,e} \sigma_{e}^2}{\underline{\gamma}^2 } \sum_{j=1}^K \frac{1}{m_j D_j u_{j,N}^2(\tau_j)} \right)   \\
			\geq&  \bar{G}_{m_1,...,m_K}(y_{1}, \cdots, y_{K})  - \lim_{N \rightarrow \infty} P\Big(\sigma_{\min}(B) < \underline{\gamma}\Big) - \lim_{N \rightarrow\infty} P\Big(\norm{{\*W}^\T \Sigma_e {\*W}}_2 > \bar \gamma_{\*V,e} \sigma_{e}^2 \Big)
			\end{align}
		\end{proof}

		\subsection{Proof of Proposition \ref{prop:GEV}-\ref{prop:vary-m}}
		\begin{proof}[Proof of Proposition \ref{prop:GEV}]
			If sequence $u_{1,N}(\tau)$ satisfy the assumptions in Lemma \ref{lemma:order-stats-dist}, then $c_{1,j} =  \sum_{i = j}^{m-1} \pi_1^{\ast^j} (i) $ defined in \eqref{eqn:G-1-m-limit} and \eqref{eqn:GEV} holds.
			
			If loadings are i.i.d. then from Theorem 3.4 in \cite{coles2001introduction}, $c_{1,j} = 1$ and \eqref{eqn:evt-one-factor-iid} holds.
		\end{proof}

		\begin{proof}[Proof of Proposition \ref{prop:evt-multi-factor-simplified}]
			When $\tilde{\*V}_k$ and $\tilde{\*V}_l$ are asymptotically independent for $k\neq l$, the nonzero entries in $\wt{\*W}$ are asymptotically non-overlapping.  From Lemma \ref{lemma_lam_error}, $\frac{1}{T} \wt{\*W}^\T \Weight \*e \*e^\T \Weight \wt{\*W} \leq (1+h(m)) \sigma_{e,\weight}^2 I_K  + o_p(1) $. Let $\bar \gamma_{\*V,e}=1+h(m)$. Then we have  $\lim_{N \rightarrow\infty} P\Big(\norm{{\*W}^\T \Sigma_e {\*W} }_2 > (1+h(m)) \sigma_{e}^2 \Big)=0$ and hence we can neglect this term and get Proposition \ref{prop:evt-multi-factor-simplified}.
		\end{proof}

		\begin{proof}[Proof of Proposition \ref{prop:weighted}]
			For the one-factor case, the proof is provided in Step 3 in the proof of Theorem \ref{thm-evt-one-factor} and Step 3 in the proof of Proposition \ref{prop-one-factor}. For the multi-factor case, the proof is provided in Step 5 in the proof of Theorem \ref{thm-evt-multi-factor}.
		\end{proof}
		
		\begin{proof}[Proof of Proposition \ref{prop:vary-m}]
			The proof is provided in Step 6 in the proof of Theorem  \ref{thm-evt-multi-factor}.
		\end{proof}

		\subsection{Proof of  Theorem \ref{thm:gen-lam}}
		
		\begin{proof}[Proof of Theorem \ref{thm:gen-lam}]
			Without loss of generality, we assume $\tilde{\*F}^\T  \tilde{\*F}/T = I_K$. Otherwise, we can multiply an invertible rotation matrix $M$ to $\tilde{\*F}$, denoted as $\breve{\*F} = \tilde{\*F} M$, such that $\breve{\*F}^\T  \breve{\*F}/T = I_K$. Furthermore, $ \rho_{\tilde{\*F}, \*F} = \rho_{\breve{\*F}, \*F}$. Let 
			\[\breve{\*\Lambda} = \*X \breve{\*F} (\breve{\*F}^\T  \breve{\*F})^{-1} = \*X \breve{\*F}/T. \]
			Since $\breve{\*F} = \tilde{\*F} M$ and $\tilde{\*\Lambda} = \*X \tilde{\*F} (\tilde{\*F}^\T  \tilde{\*F})^{-1} $, we have 
			\[\breve{\*\Lambda} = \*X \breve{\*F}/T = \*X \tilde{\*F} M/T = \tilde{\*\Lambda} (\tilde{\*F}^\T  \tilde{\*F}) (M/T)^{-1} = \tilde{\*\Lambda} M^\prime, \]
			where $M^\prime = (\tilde{\*F}^\T  \tilde{\*F}) (M/T)^{-1}$ is also an invertible rotation matrix. 
			
			The closeness between $\breve{\*\Lambda}$ and $\*\Lambda$, $\rho_{\breve{\*\Lambda}, \*\Lambda}$, is defined similarly as $\rho_{\tilde{\*\Lambda}, \*\Lambda}$. Since $M^\prime$ is invertible, we have 
			\begin{align}
			\rho_{\breve{\*\Lambda}, \*\Lambda} = \rho_{\tilde{\*\Lambda}, \*\Lambda}.
			\end{align}
			In the following, we focus on  $\rho_{\breve{\*\Lambda}, \*\Lambda}$.
			
			Denote $P$ as $P = \*F^\T  \breve{\*F}/T$. Given $\breve{\*\Lambda} = \*X \breve{\*F}/T$ and $\*X = \*\Lambda \*F^\T  + \*e$, we have 
			\[\*\Lambda^\T \breve{\*\Lambda}/N = (\*\Lambda^\T \*\Lambda/N )(\*F^\T  \breve{\*F}/T) + \*\Lambda^\T \*e \breve{\*F}/(NT) = (\*\Lambda^\T \*\Lambda/N )(\*F^\T  \breve{\*F}/T) +  o_p(1)\]
			following $\*\Lambda^\T \*e_t/N = o_p(1)$ for all $t$ from Assumption \ref{ass_f_e}. The term $\breve{\*\Lambda}^\T \breve{\*\Lambda}/N$ has
			\begin{eqnarray*}
				\breve{\*\Lambda}^\T \breve{\*\Lambda}/N &=& (\*X \breve{\*F}/T)^\T (\*X \breve{\*F}/T) \\
				&=& \frac{\breve{\*F}^\T  \*F}{T} \frac{\*\Lambda^\T \*\Lambda}{N} \frac{\*F^\T  \breve{\*F}}{T} + \frac{\breve{\*F}^\T  \*e^\T \*e \breve{\*F}}{N T^2} + o_p(1)
			\end{eqnarray*}
			
			Then we have 
			\begin{eqnarray*}
				\rho_{\tilde{\*\Lambda}, \*\Lambda} &=& \tr \Lp ( \*\Lambda^\T \*\Lambda/N )^{-1} ( \*\Lambda^\T \breve{\*\Lambda}/N ) (\breve{\*\Lambda}^\T \breve{\*\Lambda}/N )^{-1} ( \breve{\*\Lambda}^\T \*\Lambda/N ) \Rp \\
				&=& \tr \Lp \Lp \frac{\*\Lambda^\T \*\Lambda}{N} \Rp^{-1} \Lp \frac{\*\Lambda^\T \*\Lambda \*F^\T  \breve{\*F}}{NT} \Rp \Lp \frac{\breve{\*F}^\T  \*F \*\Lambda^\T \*\Lambda \*F^\T  \breve{\*F} + \breve{\*F}^\T  \*e^\T \*e \breve{\*F}}{NT^2} \Rp^{-1} \Lp \frac{\breve{\*F}^\T  \*F \*\Lambda^\T \*\Lambda}{NT} \Rp \Rp + o_p(1) \\
				&=& \tr \Lp \Lp I_K +  \Lp \frac{\*\Lambda^\T \*\Lambda}{N} \Rp^{-1/2} \Lp  \frac{\breve{\*F}^\T  \*F}{T} \Rp^{-1} \frac{\breve{\*F}^\T  \*e^\T \*e \breve{\*F}}{NT^2} \Lp \frac{\*F^\T  \breve{\*F}}{T} \Rp^{-1} \Lp \frac{\*\Lambda^\T \*\Lambda}{N} \Rp^{-1/2} \Rp^{-1} \Rp + o_p(1) \\
			\end{eqnarray*}
			For any small symmetric matrix $M$, from Taylor expansion and mean value theorem, there exists a symmetric matrix $\tilde{M}$,
			$$\tr((I_K + M)^{-1}) = \tr(I_K - M + \frac{1}{2} \tilde{M}^2) \geq \tr(I_K - M).$$ 
			Using this property, we have
			\begin{eqnarray}
			\nonumber \rho_{\breve{\*\Lambda}, \*\Lambda} &\geq& K - \tr \Lp \Lp \frac{\*\Lambda^\T \*\Lambda}{N} \Rp^{-1/2} \Lp  \frac{\breve{\*F}^\T  \*F}{T} \Rp^{-1} \frac{\breve{\*F}^\T  \*e^\T \*e \breve{\*F}}{NT^2} \Lp \frac{\*F^\T  \breve{\*F}}{T} \Rp^{-1} \Lp \frac{\*\Lambda^\T \*\Lambda}{N} \Rp^{-1/2} \Rp  + o_p(1) \\
			\nonumber &=& K - \tr \Lp  \Lp \frac{\breve{\*F}^\T  \*e^\T \*e \breve{\*F}}{NT^2} \Rp \Lp \frac{\breve{\*F}^\T  \*F \*\Lambda^\T \*\Lambda \*F^\T  \breve{\*F}}{NT^2} \Rp^{-1}    \Rp + o_p(1) \\
			\nonumber &\geq& K - \sqrt{\tr \Lp  \Lp \frac{\breve{\*F}^\T  \*e^\T \*e \breve{\*F}}{NT^2} \Rp^2 \Rp} \sqrt{\tr \Lp \Lp \frac{\breve{\*F}^\T  \*F \*\Lambda^\T \*\Lambda \*F^\T  \breve{\*F}}{NT^2} \Rp^{-2}    \Rp}  + o_p(1)  \\
			&\geq& K - \tr \Lp  \frac{\breve{\*F}^\T  \*e^\T \*e \breve{\*F}}{NT^2}  \Rp \tr \Lp \Lp \frac{\breve{\*F}^\T  \*F \*\Lambda^\T \*\Lambda \*F^\T  \breve{\*F}}{NT^2} \Rp^{-1}  \Rp  + o_p(1), \label{eqn:rho-lambda-breve-lambda}
			\end{eqnarray}
			followed from $tr(M^k) \leq (tr(M))^k$ for a symmetric positive semidefinite matrix $M$. For the term $\tr \Lp \Lp \frac{\breve{\*F}^\T  \*F \*\Lambda^\T \*\Lambda \*F^\T  \breve{\*F}}{NT^2} \Rp^{-1}    \Rp$, we have
			\begin{eqnarray}
			\nonumber \tr \Lp \Lp \frac{\breve{\*F}^\T  \*F \*\Lambda^\T \*\Lambda \*F^\T  \breve{\*F}}{NT^2} \Rp^{-1} \Rp &=& \tr  \Lp \Lp \frac{\*\Lambda^\T \*\Lambda \*F^\T  \*F}{N T} \Rp^{-1} \Lp \frac{(\*F^\T  \*F)^{-1} \*F^\T  \breve{\*F} \breve{\*F}^\T  \*F }{T} \Rp^{-1} \Rp \\
			\nonumber &\leq& \sqrt{\tr\Lp \Lp \frac{\*\Lambda^\T \*\Lambda \*F^\T  \*F}{N T} \Rp^{-2} \Rp } \sqrt{\tr \Lp \Lp  \Lp \frac{\*F^\T  \*F}{T} \Rp^{-1} \frac{\*F^\T  \breve{\*F} \breve{\*F}^\T  \*F}{T^2} \Rp^{-2}\Rp} \\
			&\leq& \tr\Lp \Lp \frac{\*\Lambda^\T \*\Lambda \*F^\T  \*F}{N T} \Rp^{-1} \Rp  \tr \Lp \Lp  \Lp \frac{\*F^\T  \*F}{T} \Rp^{-1} \frac{\*F^\T  \breve{\*F} \breve{\*F}^\T  \*F}{T^2} \Rp^{-1}\Rp \label{eqn:f-inequality}
			\end{eqnarray}
			Denote $A_{\tilde{\*F}, \*F} = ( \*F^\T  \*F/T )^{-1} ( \*F^\T  \tilde{\*F}/T ) (\tilde{\*F}^\T  \tilde{\*F}/T )^{-1} ( \tilde{\*F}^\T  \*F/T ) $ \\ and $A_{\breve{\*F}, \*F} = ( \*F^\T  \*F/T )^{-1} ( \*F^\T  \breve{\*F}/T ) (\breve{\*F}^\T  \breve{\*F}/T )^{-1} ( \breve{\*F}^\T  \*F/T ) $. We have
			\[A_{\breve{\*F}, \*F} = ( \*F^\T  \*F/T )^{-1} ( \*F^\T  \tilde{\*F} M/T ) (M^\T \tilde{\*F}^\T  \tilde{\*F} M/T )^{-1} (M^\T \tilde{\*F}^\T  \*F/T ) = A_{\tilde{\*F}, \*F}  \]
			and together with $\breve{\*F}^\T  \breve{\*F}/T = I_K$, we have 
			\[\tr \Lp \Lp  \Lp \frac{\*F^\T  \*F}{T} \Rp^{-1} \frac{\*F^\T  \breve{\*F} \breve{\*F}^\T  \*F}{T^2} \Rp^{-1} \Rp =  tr(A_{\breve{\*F}, \*F}^{-1}) = tr (A_{\tilde{\*F}, \*F}^{-1})\]
			From Lemma \ref{lemma_rho}, 
			\[\rho_{\tilde{\*F}, \*F} = \tr(A_{\tilde{\*F}, \*F}) =  \tr \left(   \Bigg(I_K + \underbrace{ \left( \frac{\*F^\T  \*F}{T} \right)^{-1/2}(Q^\T)^{-1} \frac{\wt{\*W}^\T  \*e \*e^\T \wt{\*W}}{T}  Q^{-1} \left( \frac{\*F^\T  \*F}{T} \right)^{-1/2} }_{C} \Bigg)^{-1}  \right) + o_p(1). \]
			where $Q = \*\Lambda \Weight \wt{\*W} = \*\Lambda \wt{\*W} $ under the setting of this theorem ($\Weight = I_N$). Therefore, $ \tr(A_{\tilde{\*F}, \*F}) = \tr(I_K + C)  + o_p(1)$ and
			\[\tr \Lp \Lp  \Lp \frac{\*F^\T  \*F}{T} \Rp^{-1} \frac{\*F^\T  \breve{\*F} \breve{\*F}^\T  \*F}{T^2} \Rp^{-1} \Rp = tr \Lp  I_K + C\Rp  + o_p(1). \]
			We plug it into \eqref{eqn:rho-lambda-breve-lambda}, use the inequality \eqref{eqn:f-inequality} and have 
			\[\rho_{\breve{\*\Lambda}, \*\Lambda} \geq K - \tr  \Lp \frac{\breve{\*F}^\T  \*e^\T \*e \breve{\*F}}{NT^2} \Rp \tr \Lp \Lp \frac{\*\Lambda^\T \*\Lambda \*F^\T  \*F}{N T} \Rp^{-1} \Rp tr \Lp I_K + C \Rp + o_p(1).  \]
			From \eqref{eqn:prop3-2-2} in the proof of Theorem \ref{thm-evt-multi-factor}, we have
			\begin{align}
			\nonumber \tr(I_K + C) \leq&  K +  \norm{\wt{\*W}^\T \Sigma_e \wt{\*W}}_2  \tr \left( \left( \frac{1}{T} \wt{\*W}^\T  (\Lambda \*F^\T  \*F  \*\Lambda^\T) \wt{\*W}\right)^{-1}  \right)  + o_p(1) \\
			\nonumber \leq&  K +  \norm{\wt{\*W}^\T \Sigma_e \wt{\*W}}_2  \tr \left( Q^{-1} \left( \frac{\*F^\T  \*F}{T} \right)^{-1}  (Q^\T)^{-1}  \right)  + o_p(1) \\
			\leq& K +  K  (1 + h(m)) \sigma_e^2   \tr \left( Q^{-1} \left( \frac{\*F^\T  \*F}{T} \right)^{-1}  (Q^\T)^{-1}  \right)  + o_p(1) \label{eqn:proof-thm-lambda-1} \\
			\leq& K +  K  (1 + h(m)) \sigma_e^2 \norm{(Q Q^\T)^{-1}}_2 \tr(\Sigma_{F}^{-1})  + o_p(1) \label{eqn:proof-thm-lambda-2} \\
			\leq& K +  K  (1 + h(m)) \sigma_e^2 \norm{Q^{-1}}^2_2 \tr(\Sigma_{F}^{-1})  + o_p(1) \label{eqn:proof-thm-lambda-3} \\
			=& O_p(1) \label{eqn:proof-thm-lambda-4} 
			\end{align}
			where \eqref{eqn:proof-thm-lambda-1} follows $\norm{\wt{\*W}^\T \Sigma_e \wt{\*W}}_2 \leq \norm{\wt{\*W}^\T \Sigma_e \wt{\*W}}_1 \leq K (1 + h(m)) \sigma_e^2 $ from Holder's inequality and symmetricity of this matrix; \eqref{eqn:proof-thm-lambda-2} follows trace inequality $\tr(AB) \leq \norm{A}_2 \tr(B)$; \eqref{eqn:proof-thm-lambda-3} follows matrix 2-norm is submultiplicative $\norm{AB}_2 \leq \norm{A}_2 \norm{B}_2$,  and \eqref{eqn:proof-thm-lambda-4} follows the assumption in Theorem \ref{thm:gen-lam} that $\norm{(W^\T \Lambda)^{-1}}_2 = O_p(1)$ and the proof of Lemma \ref{lemma_u}. 
			
			Next, we need to provide an upper bound for $\tr  \Lp \frac{\tilde{\*F}^\T  \*e^\T \*e \tilde{\*F}}{NT^2} \Rp $ and $ \tr \Lp \Lp \frac{\*\Lambda^\T \*\Lambda \*F^\T  \*F}{N T} \Rp^{-1} \Rp$. For the term $\tr  \Lp \frac{\tilde{\*F}^\T  \*e^\T \*e \tilde{\*F}}{NT^2} \Rp $,
			from Assumption \ref{ass:max-eigen-err-cov}, $\frac{1}{NT}  \*e^\T \*e = o_p(1)$ and  together with $\tilde{\*F}^\T  \tilde{\*F}/T = I_K$, we have 
			\[\frac{1}{NT^2} \tilde{\*F}_j^\T  \*e^\T \*e \tilde{\*F}_j  = o_p(1) \]
			and then 
			\[ \tr  \Lp \frac{\tilde{\*F}^\T  \*e^\T \*e \tilde{\*F}}{NT^2} \Rp = \sum_{j = 1}^K \frac{1}{NT^2} \tilde{\*F}_j^\T  \*e^\T \*e \tilde{\*F}_j  = o_p(1).  \]
			For the term $ \tr \Lp \Lp \frac{\*\Lambda^\T \*\Lambda \*F^\T  \*F}{N T} \Rp^{-1} \Rp$, note that
			\[\frac{1}{NT} \*\Lambda^\T \*\Lambda \*F^\T  \*F \xrightarrow{p}  \Sigma_{\Lambda} \Sigma_F. \]
			From Assumptions \ref{ass_factor} and \ref{ass_loading}, $\Sigma_F$ and $\Sigma_{\Lambda}$ are positive definite matrix. Then $\Sigma_F \Sigma_{\Lambda}$ is invertible and the maximum eigenvalue of $(\Sigma_F \Sigma_{\Lambda})^{-1}$ is bounded away from $\infty$. Thus, $ \tr \Lp \Lp \frac{\*\Lambda^\T \*\Lambda \*F^\T  \*F}{N T} \Rp^{-1} \Rp = O_p(1)$. We then have 
			\[\tr  \Lp \frac{\tilde{\*F}^\T  \*e^\T \*e \tilde{\*F}}{NT^2} \Rp \tr \Lp \Lp \frac{\*\Lambda^\T \*\Lambda \*F^\T  \*F}{N T} \Rp^{-1} \Rp tr \Lp I_K + C\Rp = o_p(1)\]
			followed from $ \tr  \Lp \frac{\tilde{\*F}^\T  \*e^\T \*e \tilde{\*F}}{NT^2} \Rp = o_p(1)$, $\tr \Lp \Lp \frac{\*\Lambda^\T \*\Lambda \*F^\T  \*F}{N T} \Rp^{-1} \Rp = O_p(1)$
			and $\tr \Lp I_K + C\Rp = O_p(1)$. Thus, $\rho_{\breve{\*\Lambda}, \*\Lambda} \xrightarrow{p} K$ and $\rho_{\tilde{\*\Lambda}, \*\Lambda} \xrightarrow{p} K$.

		\end{proof}
		
		\begin{proof}[Proof of Proposition \ref{prop:comp-lasso}]
			Define $\tilde{\rho}$ and $\bar{\rho}$ as 
			
			\begin{eqnarray*}
				\tilde{\rho} &=& tr \left( (\*F^\T \*F/T)^{-1} (\*F^\T  \tilde{\*F}/T) (\tilde{\*F}^\T \tilde{\*F}/T)^{-1} (\tilde{\*F}^\T \*F/T) \right) \\
				\bar{\rho} &=& tr \left( (\*F^\T \*F/T)^{-1} (\*F^\T \bar{\*F}/T) (\bar{\*F}^\T  \bar{\*F}/T)^{-1} (\bar{\*F}^\T  \*F/T) \right)
			\end{eqnarray*}
			
			Assume there are $m$ nonzero elements in both $\tilde{\*\Lambda}$ and $\bar{\*\Lambda}$. Since it is an one factor model, $\tilde{\*\Lambda} = \tilde{\*\Lambda}_1$ and $\bar{\*\Lambda} = \bar{\*\Lambda}_1$. Assume the indexes of nonzero values in $\tilde{\*\Lambda}_1$ and $\bar{\*\Lambda}_1$ are $\tfidx_1(l)$ and $\bar{f}^{(\idx)}(1,l)$ $i = 1, 2, \cdots m$ respectively.

			From lemma 4, 
			we have as $N, T \rightarrow \infty$, 
			$$\tilde{\rho} =  tr \left(   \left(I + \left( \frac{\*F^\T  \*F}{T} \right)^{-1/2} \left(\tilde{\*\Lambda}^\T \*\Lambda \right)^{-1} \frac{\tilde{\*\Lambda}^\T \*e \*e^\T \tilde{\*\Lambda}}{T}  \left(\*\Lambda^\T \tilde{\*\Lambda} \right)^{-1} \left( \frac{\*F^\T  \*F}{T} \right)^{-1/2} \right)^{-1}  \right) +  o_p(1)$$
			
			For $\frac{\tilde{\*\Lambda}^\T \*e \*e^\T \tilde{\*\Lambda}}{T}$, in a one-factor model,
			\begin{eqnarray}
			\nonumber \frac{1}{T}\tilde{\*\Lambda}^\T \*e \*e^\T \tilde{\*\Lambda} &=& \frac{1}{T} \sum_{i = 1}^{m} \sum_{k =1 }^{m}  \tilde{\Lambda}_{\tfidx_1(l),1} \tilde{\Lambda}_{\tfidx_1(k),1} \*e_{\tfidx_1(l)}^\T  \*e_{\tfidx_1(k)} \\
			\nonumber &=& \frac{1}{T} \sum_{i = 1}^{m} \tilde{\Lambda}_{\tfidx_1(l),1}^2  \*e_{\tfidx_1(l)}^\T  \*e_{\tfidx_1(l)} + \frac{1}{T} \sum_{i \neq k}   \tilde{\Lambda}_{\tfidx_1(l),1} \tilde{\Lambda}_{\tfidx_1(k),1} \*e_{\tfidx_1(l)}^\T  \*e_{\tfidx_1(k)},
			\end{eqnarray}
			where $\*e_{i}$ is the $i$-th row in $\*e$. If errors are cross-sectionally independent (weekly dependent), then for fixed $m$, $\frac{1}{T} \sum_{i \neq k}   \tilde{\Lambda}_{\tfidx_1(l),1} \tilde{\Lambda}_{\tfidx_1(k),1} \*e_{\tfidx_1(l)}^\T  \*e_{\tfidx_1(k)} = o_p(1)$. Furthermore, since $\sum_{i \neq k}   \tilde{\Lambda}^2_{\tfidx_1(l),1}  = 1$, if there is some $\sigma_e^2$, such that $\frac{1}{T} \*e_{i}^\T \*e_{i} \rightarrow \sigma_e^2$, then
			\[\frac{1}{T}\tilde{\*\Lambda}^\T \*e \*e^\T \tilde{\*\Lambda}  = \sigma_e^2 + o_p(1) \]
			
			Then we have 
			\begin{eqnarray*}
				\tilde{\rho} =   \left(1 + \frac{\sigma_e^2 }{\frac{1}{T} \tilde{\*\Lambda}^\T \left( \*\Lambda \*F^\T  \*F \*\Lambda^\T \right) \tilde{\*\Lambda}  }  \right)^{-1}   + o_p(1) \\
			\end{eqnarray*}
			
			Similarly, for sparse loadings calculated from Lasso, with the standardization, $\bar{\*\Lambda}^\T  \bar{\*\Lambda} = 1$, we also have 
			\begin{eqnarray*}
				\bar{\rho} &=&  \left(1 + \frac{\sigma_e^2 }{\frac{1}{T} \bar{\*\Lambda}^\T  \left( \*\Lambda \*F^\T  \*F \*\Lambda^\T \right) \bar{\*\Lambda}  }  \right)^{-1}   + o_p(1) \\
			\end{eqnarray*}
			
			From the proof of Proposition 1, we have 
			\begin{eqnarray*}
				\frac{1}{T} \tilde{\*\Lambda}^\T \left( \Lambda \*F^\T  \*F \*\Lambda^\T \right) \tilde{\*\Lambda} = \sigma_{\*F_1}^2 \sum_{i = 1}^{m} \Lambda_{\tfidx_1(l),1}^2 + o_p(1) + o_p(1) \\
			\end{eqnarray*}
			
			On the other hand, 
			\begin{eqnarray*}
				\frac{1}{T} \bar{\*\Lambda}^\T  \left( \*\Lambda \*F^\T  \*F \*\Lambda^\T \right) \bar{\*\Lambda} &=& \bar{D} (\sum_{i = 1}^{m} \bar{\Lambda}_{\bar{s}_1(i),1} \hat{\Lambda}_{\bar{s}_1(i),1})^2 + o_p(1),
			\end{eqnarray*}
			where $\bar{D} = H^{-1} \frac{\*F^\T  \*F}{T} (H^\T)^{-1}$. 
			
			From Cauchy-Schwarz inequality, 
			\begin{eqnarray*}
				\Lp \sum_{i = 1}^m \bar{\Lambda}_{\bar{s}_1(i),1} \hat{\Lambda}_{\bar{s}_1(i),1} \Rp^2 \leq \Lp\sum_{i = 1}^m \bar{\Lambda}_{\bar{s}_1(i),1}^2\Rp \Lp\sum_{i = 1}^m \hat{\Lambda}_{\bar{s}_1(i),1}^2\Rp = \sum_{i = 1}^m \hat{\Lambda}_{\bar{s}_1(i),1}^2.
			\end{eqnarray*}
			The equality holds only if $\forall i$, $\frac{\tilde{\Lambda}_{\tfidx_1(i),1} }{\hat{\Lambda}_{\tfidx_1(i),1}}$ is the same. However, the sparse loadings are obtained from Lasso Regression. The coefficients equal to those from Least Angle Regression \citep{efron2004least, friedman2001elements}, which in general are not proportional to the least square coefficients. Thus, in general, $\frac{\tilde{\Lambda}_{\tfidx_1(i),1} }{\hat{\Lambda}_{\tfidx_1(i),1}}$ is not the same $\forall i$. The above inequality is a strict inequality. 
			
			Thus,
			\begin{eqnarray*}
				\frac{1}{T} \bar{\*\Lambda}^\T  \left( \*\Lambda \*F^\T  \*F \*\Lambda^\T \right) \bar{\*\Lambda} < \frac{1}{T} \tilde{\*\Lambda}^\T \left( \Lambda \*F^\T  \*F \*\Lambda^\T \right) \tilde{\*\Lambda} + o_p(1) 
			\end{eqnarray*}
			and therefore, as $N, T \rightarrow 0$,
			$$
			\Delta \rho =  \tilde{\rho} - \bar{\rho} \geq 0
			$$
			with probability 1.
			
		\end{proof}

		\newpage

		\section{Supplementary Empirical Results}\label{sec:add-empirical}

		\begin{table}[H]
			\centering
			\begin{tabular}{llll}
				\toprule
				&Anomaly characteristics & & Anomaly characteristics\\ \midrule
				1&Accruals - accrual & 20& Momentum (12m) - mom12           \\
				2&Asset Turnover - aturnover& 21& Momentum-Reversals - momrev \\
				3&Cash Flows/Price - cfp & 22& Net Operating Assets - noa\\
				4&Composite Issuance - ciss & 23& Price - price \\
				5&Dividend/Price - divp & 24& Gross Protability - prof\\
				6&Earnings/Price - ep & 25& Return on Assets (A) - roaa\\
				7&Gross Margins - gmargins & 26& Return on Book Equity (A) - roea  \\
				8&Asset Growth - growth & 27& Seasonality - season  \\
				9&Investment Growth - igrowth & 28& Sales Growth - sgrowth \\
				10&Industry Momentum - indmom & 29& Share Volume - shvol \\
				11&Industry Mom. Reversals - indmomrev  & 30& Size - size \\
				12&Industry Rel. Reversals \footnote{Industry Rel. Reversals: Industry Relative Reversals} - indrrev & 31& Sales/Price - sp                  \\
				13&Industry Rel. Rev. (L.V.) \footnote{Industry Relative Reversal (Low Volatility): Industry Relative Reversals} - indrrevlv & 32& Short-Term Reversals - strev \\ 
				14&Investment/Assets - inv & 33& Value-Momentum - valmom \\
				15&Investment/Capital - invcap & 34& Value-Momentum-Prof. - valmomprof \\
				16&Idiosyncratic Volatility - ivol & 35& Value-Protability -valprof        \\
				17&Leverage - lev & 36& Value (A) - value                 \\
				18&Long Run Reversals - lrrev & 37& Value (M) - valuem                \\
				19&Momentum (6m) - mom &  &                                   \\ \bottomrule
			\end{tabular}
			\caption{List of 37 anomaly characteristics in 370 single-sorted portfolios}
			\label{tab:list-anomaly}
		\end{table}
		
		An alternative approach to implement sparse PCA is to directly impose the cardinality constraint $\norm{\*\Lambda_j}_0 \leq m$  rather than using the $\ell_1$ penalty term $\alpha \sum_{j = 1}^{K} \norm{\*\Lambda_j}_1$ in the optimization problem \eqref{eqn:spca-obj}, where $\norm{\*\Lambda_j}_0 $ equals the number of nonzero elements in $\*\Lambda_j$. \cite{sigg2008expectation} develop an expectation-maximization (EM) algorithm based on a probabilistic expression of PCA to solve this optimizaton problem with constraint  $\norm{\*\Lambda_j}_0 \leq m$.\footnote{This algorithm is implemented in the R package \textsf{nsprcomp} \citep{sigg2018package}. } This approach will in general return sparse loadings with $m$
		nonzero elements in each sparse loading vector when $m$ is reasonably small. Then the sparse loadings from this approach will have the same number of nonzero elements as PPCA.
		
		Figure \ref{fig:370-port-rho-rmse-by-m}-\ref{fig:fred-md-rho-rmse-by-m-train} compare PPCA and sparse PCA with the cardinality constraint using the same metrics as in Section \ref{sec:empirical}. Similar as Figure \ref{fig:370-port-rho-rmse} and \ref{fig:fred-md-rho-rmse} for out-of-sample results and Figure \ref{fig:370-port-rho-rmse-train} and \ref{fig:fred-md-rho-rmse-train} for in-sample results, sparse PCA performs significantly worse than PPCA on both training and test data for the financial portfolio and macroeconomic data. Factors and loadings from sparse PCA are neither close to unweighted nor weighted PCA factors and loadings. Also, the RMSE for sparse PCA is much higher than PPCA.

		\begin{figure}[H]
			\tcapfig{Financial Portfolio Data: Out-of-Sample Generalized Correlations and RMSE (SPCA with Cardinality Constraints)}
			\begin{adjustwidth}{-1cm}{}
				\centering
				\begin{subfigure}{.4\textwidth}
					\centering
					\includegraphics[width=1\linewidth]{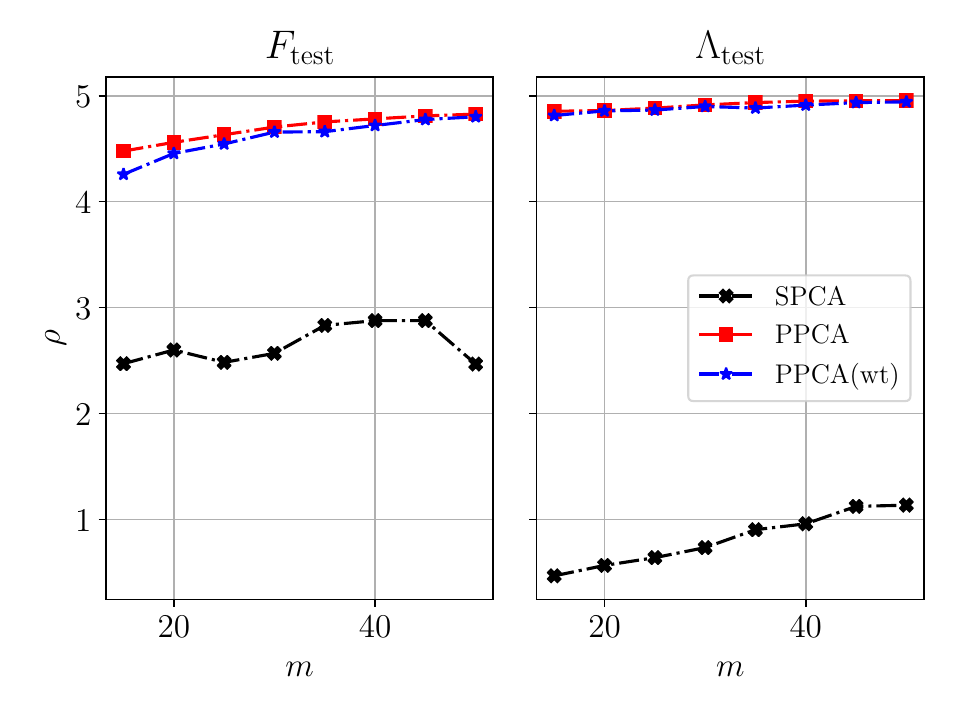}
					\subcaptab{$\rho$ with $\hat{\*F}$/$\hat{\*\Lambda}$}
				\end{subfigure}%
				\begin{subfigure}{.4\textwidth}
					\centering
					\includegraphics[width=1\linewidth]{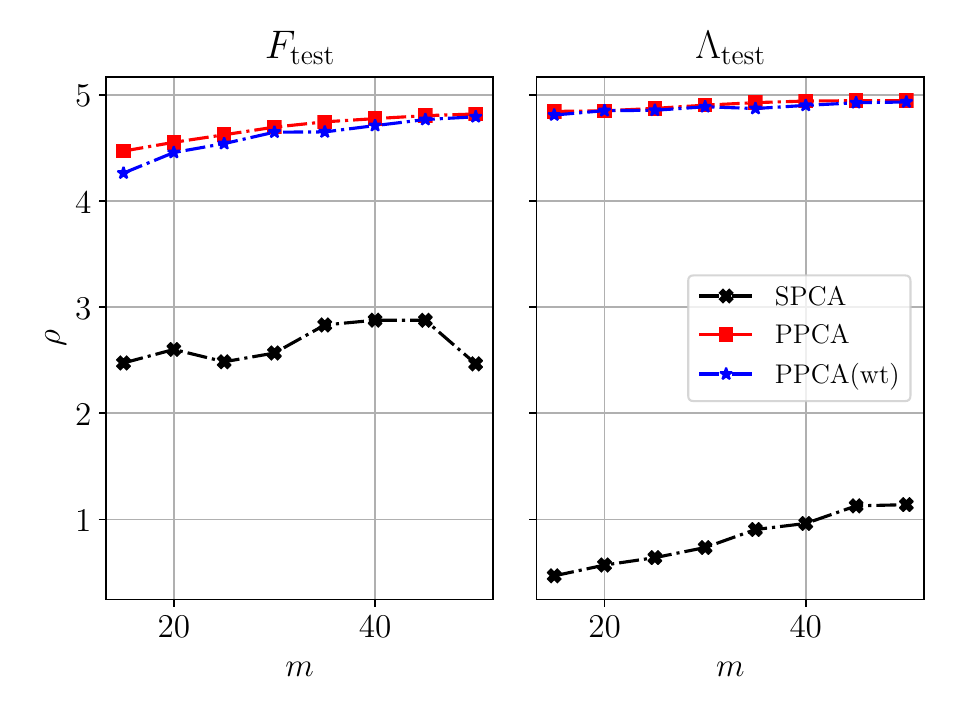}
					\subcaptab{$\rho$ with $\hat{\*F}^\twt$/$\hat{\*\Lambda}^\twt$}
				\end{subfigure}%
				\begin{subfigure}{.235\textwidth}
					\centering
					\includegraphics[width=1\linewidth]{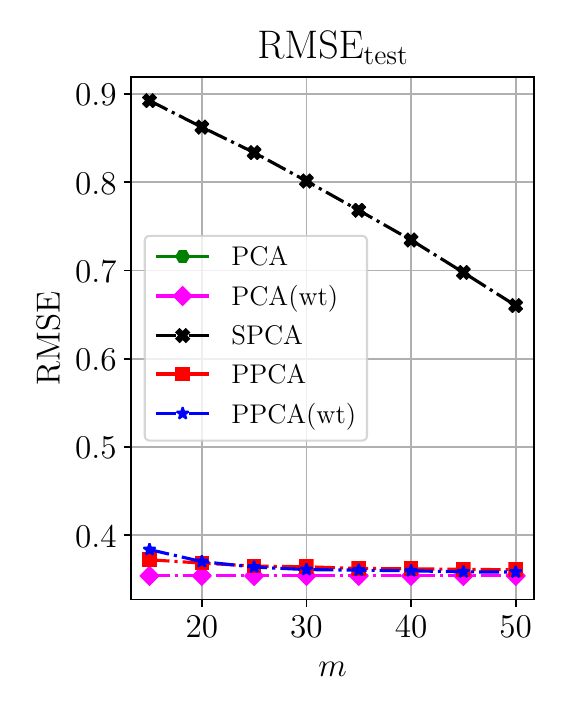}
					\subcaptab{RMSE}
				\end{subfigure}
			\end{adjustwidth}
			\label{fig:370-port-rho-rmse-by-m}
			\bnotefig{This figure compares the out-of-sample generalized correlations for factors and loadings and out-of-sample RMSE for proximate factors (PPCA), weighted proximate factors (PPCA (wt)), sparse PCA (SPCA), weighted PCA (PCA (wt)) and unweighted PCA. PPCA (wt) and PCA (wt) use the inverse standard errors as weights. In order to achieve the same sparsity level for various methods, we first choose $m$ for PPCA and set the cardinality constraint $\norm{\*\Lambda_j}_0 \leq m$ for SPCA. The left figure shows the generalized correlation of the factors and loadings with the PCA estimates  $\hat{\*F}$/loadings $\hat{\*\Lambda}$. The middle figure show the corresponding generalized correlations with weighted PCA estimates $\hat{\*F}^\twt$/loadings $\hat{\*\Lambda}^\twt$. The right figure displays the RMSE for all five methods. PCA, PCA (wt), PPCA and PPCA (wt) achieve very similar performance and significantly outperform SPCA.} 
		\end{figure}
		
		\begin{figure}[H]
			\tcapfig{Financial Portfolio Data: In-Sample Generalized Correlations and RMSE (SPCA with Cardinality Constraints)}
			\begin{adjustwidth}{-1cm}{}
				\centering
				\begin{subfigure}{.4\textwidth}
					\centering
					\includegraphics[width=1\linewidth]{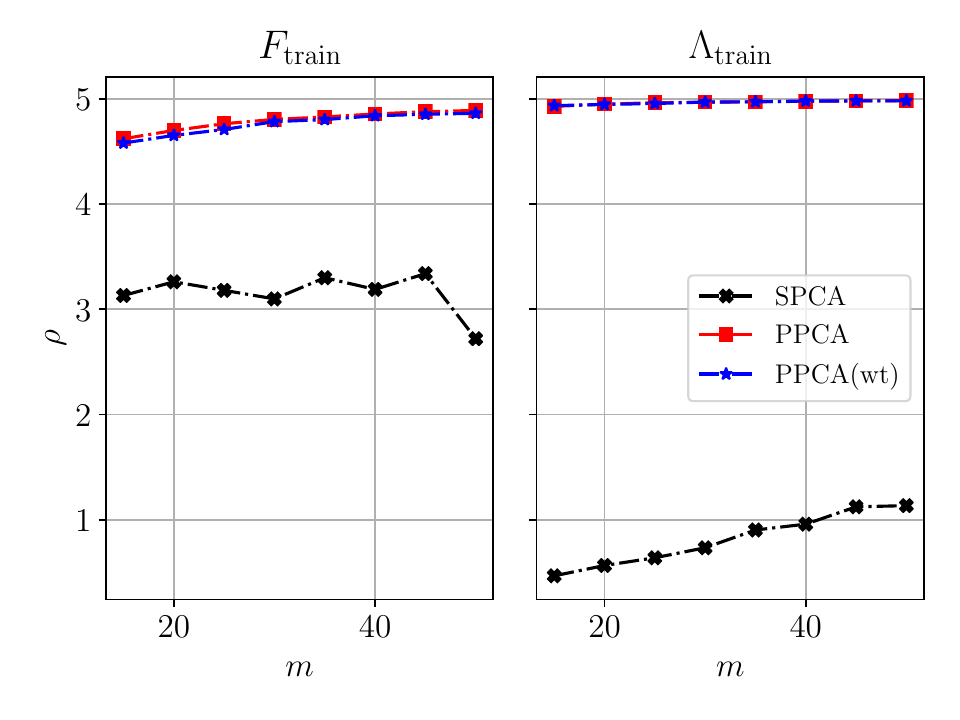}
					\subcaptab{$\rho$ with $\hat{\*F}$/$\hat{\*\Lambda}$}
				\end{subfigure}%
				\begin{subfigure}{.4\textwidth}
					\centering
					\includegraphics[width=1\linewidth]{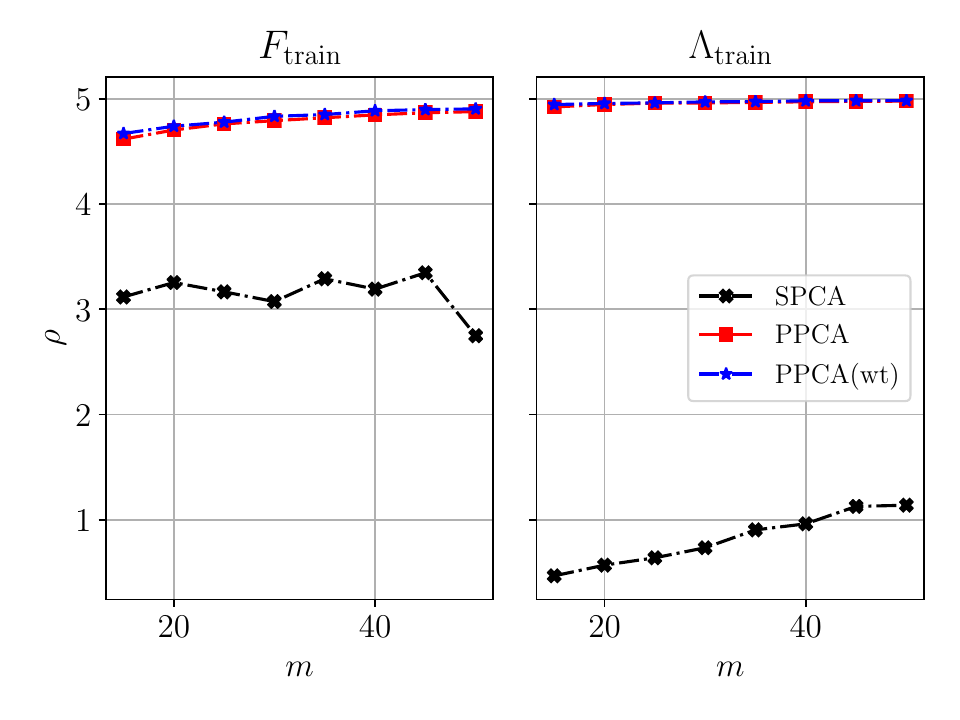}
					\subcaptab{$\rho$ with $\hat{\*F}^\twt$/$\hat{\*\Lambda}^\twt$}
				\end{subfigure}%
				\begin{subfigure}{.235\textwidth}
					\centering
					\includegraphics[width=1\linewidth]{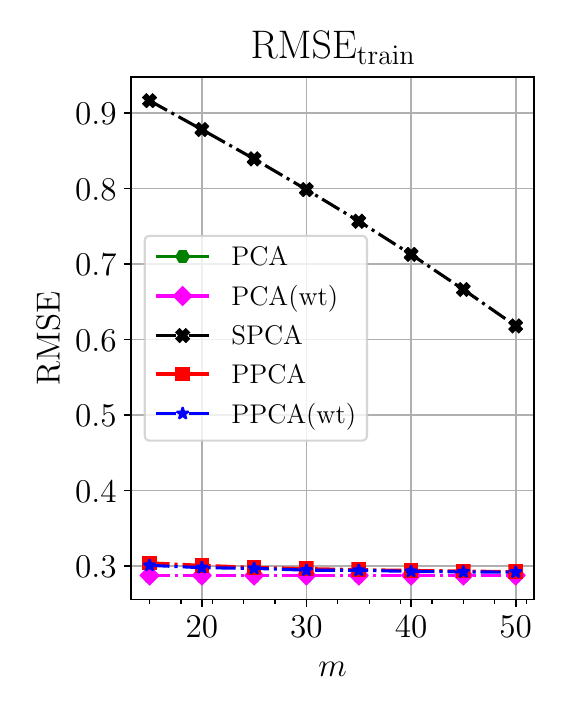}
					\subcaptab{RMSE}
				\end{subfigure}
			\end{adjustwidth}
			\label{fig:370-port-rho-rmse-by-m-train}
			\bnotefig{This figure compares the in-sample generalized correlations for factors and loadings and out-of-sample RMSE for proximate factors (PPCA), weighted proximate factors (PPCA (wt)), sparse PCA (SPCA), weighted PCA (PCA (wt)) and unweighted PCA. PPCA (wt) and PCA (wt) use the inverse standard errors as weights. In order to achieve the same sparsity level for various methods, we first choose $m$ for PPCA and set the cardinality constraint $\norm{\*\Lambda_j}_0 \leq m$ for SPCA.  The left figure shows the generalized correlation of the factors and loadings with the PCA estimates  $\hat{\*F}$/loadings $\hat{\*\Lambda}$. The middle figure show the corresponding generalized correlations with weighted PCA estimates $\hat{\*F}^\twt$/loadings $\hat{\*\Lambda}^\twt$. The right figure displays the RMSE for all five methods. PCA, PCA (wt), PPCA and PPCA (wt) achieve very similar performance and significantly outperform SPCA.} 
		\end{figure}

		\begin{figure}[H]
			\tcapfig{Macroeconomic Data: Out-of-Sample Generalized Correlations and RMSE (SPCA with Cardinality Constraints)}
			\begin{adjustwidth}{-1cm}{}
				\centering
				\begin{subfigure}{.4\textwidth}
					\centering
					\includegraphics[width=1\linewidth]{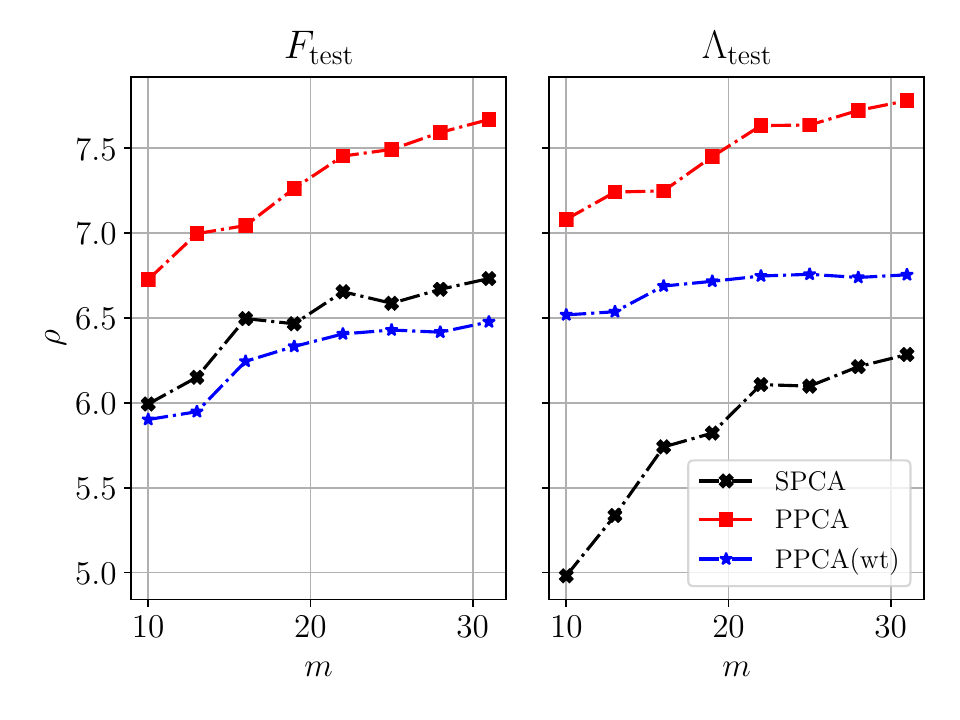}
					\caption{$\rho$ with $\hat{\*F}$/$\hat{\*\Lambda}$}
				\end{subfigure}%
				\begin{subfigure}{.4\textwidth}
					\centering
					\includegraphics[width=1\linewidth]{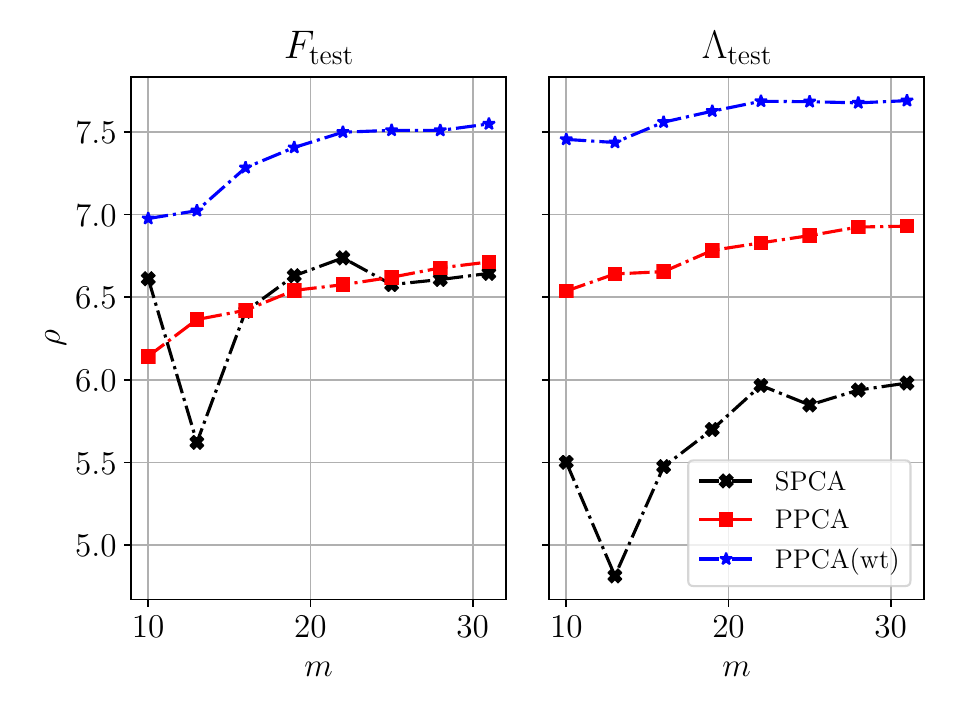}
					\caption{$\rho$ with $\hat{\*F}^\twt$/$\hat{\*\Lambda}^\twt$}
				\end{subfigure}%
				\begin{subfigure}{.235\textwidth}
					\centering
					\includegraphics[width=1\linewidth]{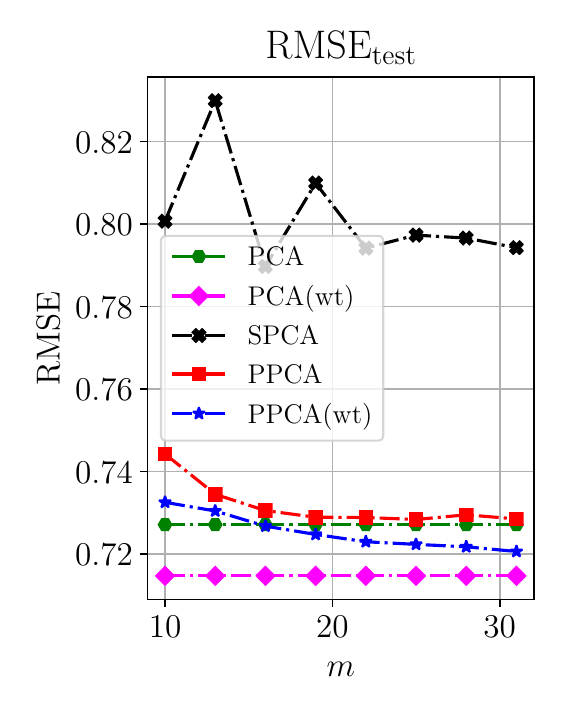}
					\caption{RMSE}
				\end{subfigure}
			\end{adjustwidth}
			\bnotefig{This figure compares the out-of-sample generalized correlations for factors and loadings and out-of-sample RMSE for proximate factors (PPCA), weighted proximate factors (PPCA (wt)), sparse PCA (SPCA), weighted PCA (PCA (wt)) and unweighted PCA. PPCA (wt) and PCA (wt) use the inverse standard errors as weights. In order to achieve the same sparsity level for various methods, we first choose $m$ for PPCA and set the cardinality constraint $\norm{\*\Lambda_j}_0 \leq m$ for SPCA.  The left figure shows the generalized correlation of the factors and loadings with the PCA estimates  $\hat{\*F}$/loadings $\hat{\*\Lambda}$. The middle figure show the corresponding generalized correlations with weighted PCA estimates $\hat{\*F}^\twt$/loadings $\hat{\*\Lambda}^\twt$. The right figure displays the RMSE for all five methods. PCA, PCA (wt), PPCA and PPCA (wt) achieve very similar performance and significantly outperform SPCA. }
			\label{fig:fred-md-rho-rmse-by-m}
		\end{figure}

		\begin{figure}[H]
			\tcapfig{Macroeconomic Data: In-Sample Generalized Correlations and RMSE (SPCA with Cardinality Constraints)}
			\begin{adjustwidth}{-1cm}{}
				\centering
				\begin{subfigure}{.4\textwidth}
					\centering
					\includegraphics[width=1\linewidth]{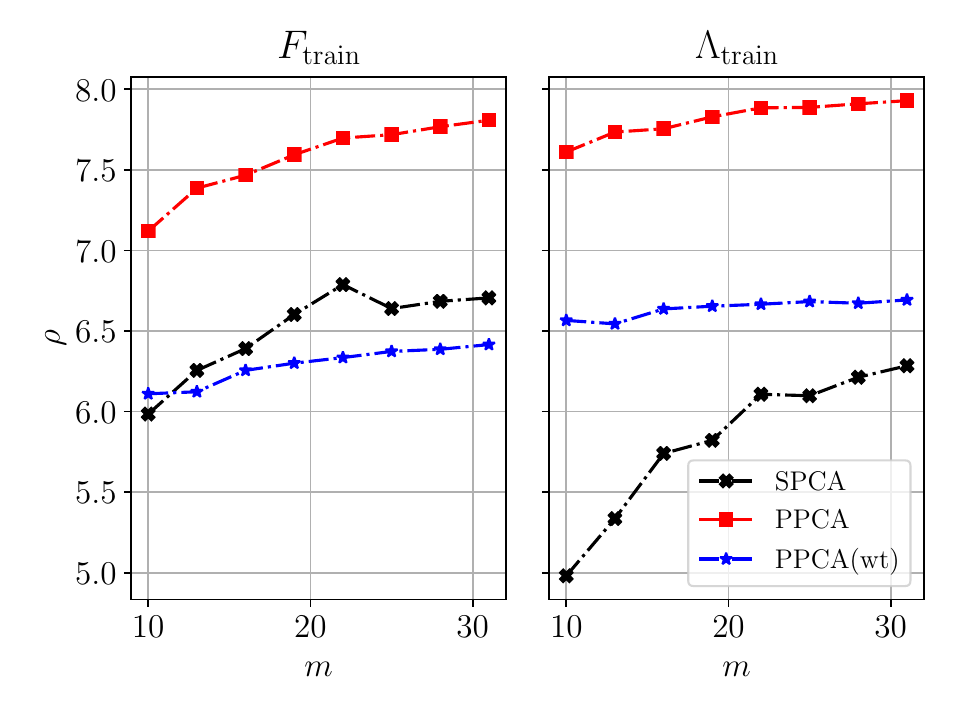}
					\caption{$\rho$ with $\hat{\*F}$/$\hat{\*\Lambda}$}
				\end{subfigure}%
				\begin{subfigure}{.4\textwidth}
					\centering
					\includegraphics[width=1\linewidth]{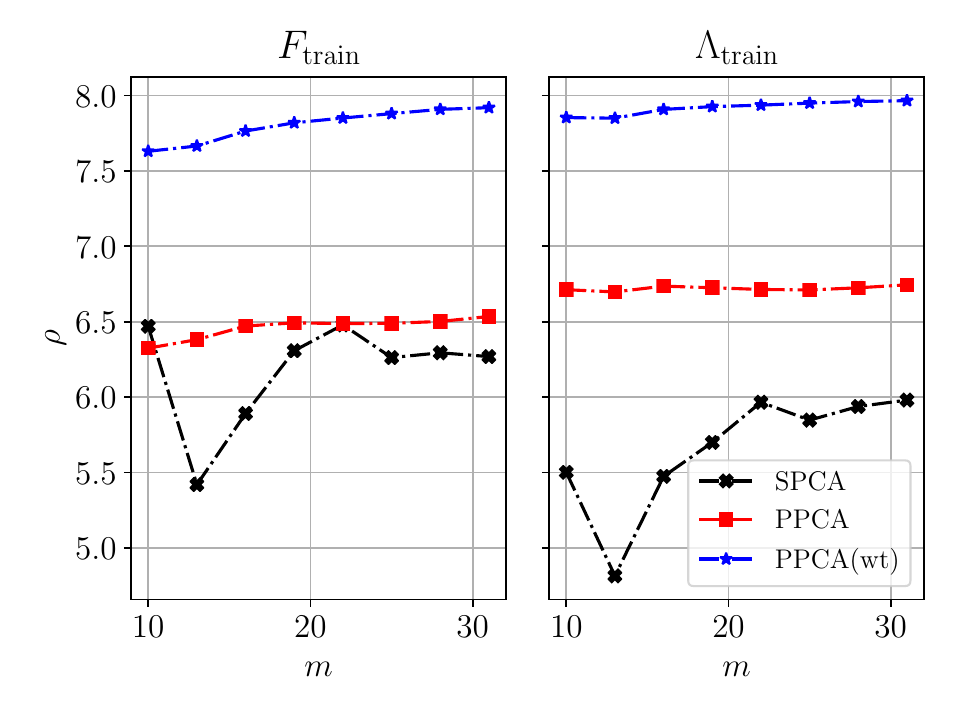}
					\caption{$\rho$ with $\hat{\*F}^\twt$/$\hat{\*\Lambda}^\twt$}
				\end{subfigure}%
				\begin{subfigure}{.235\textwidth}
					\centering
					\includegraphics[width=1\linewidth]{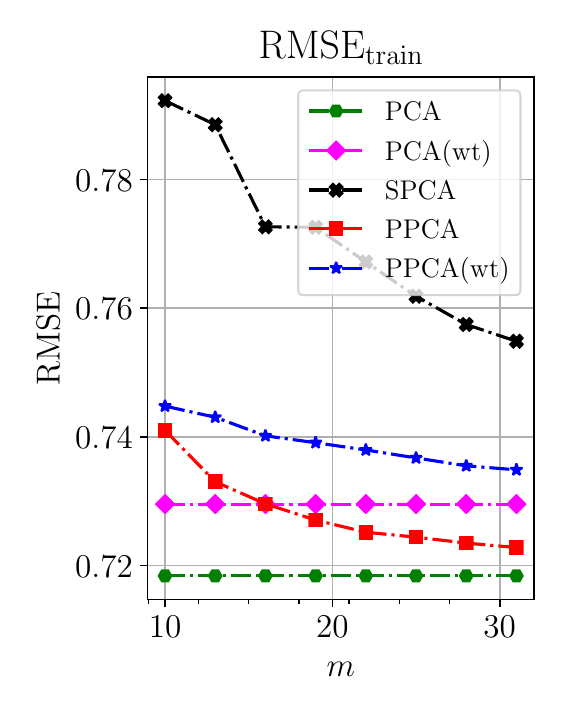}
					\caption{RMSE}
				\end{subfigure}
			\end{adjustwidth}
			\bnotefig{This figure compares the in-sample generalized correlations for factors and loadings and in-sample RMSE for proximate factors (PPCA), weighted proximate factors (PPCA (wt)), sparse PCA (SPCA), weighted PCA (PCA (wt)) and unweighted PCA. PPCA (wt) and PCA (wt) use the inverse standard errors as weights. In order to achieve the same sparsity level for various methods, we first choose $m$ for PPCA and set the cardinality constraint $\norm{\*\Lambda_j}_0 \leq m$ for SPCA.  The left figure shows the generalized correlation of the factors and loadings with the PCA estimates  $\hat{\*F}$/loadings $\hat{\*\Lambda}$. The middle figure show the corresponding generalized correlations with weighted PCA estimates $\hat{\*F}^\twt$/loadings $\hat{\*\Lambda}^\twt$. The right figure displays the RMSE for all five methods. PCA, PCA (wt), PPCA and PPCA (wt) achieve very similar performance and significantly outperform SPCA. }
			\label{fig:fred-md-rho-rmse-by-m-train}
		\end{figure}
		
		\begin{figure}[H]
			\tcapfig{Financial Portfolio Data: In-Sample Generalized Correlations and RMSE (SPCA with $\ell_1$ Penalty)}
			\begin{adjustwidth}{-1cm}{}
				\centering
				\begin{subfigure}{.4\textwidth}
					\centering
					\includegraphics[width=1\linewidth]{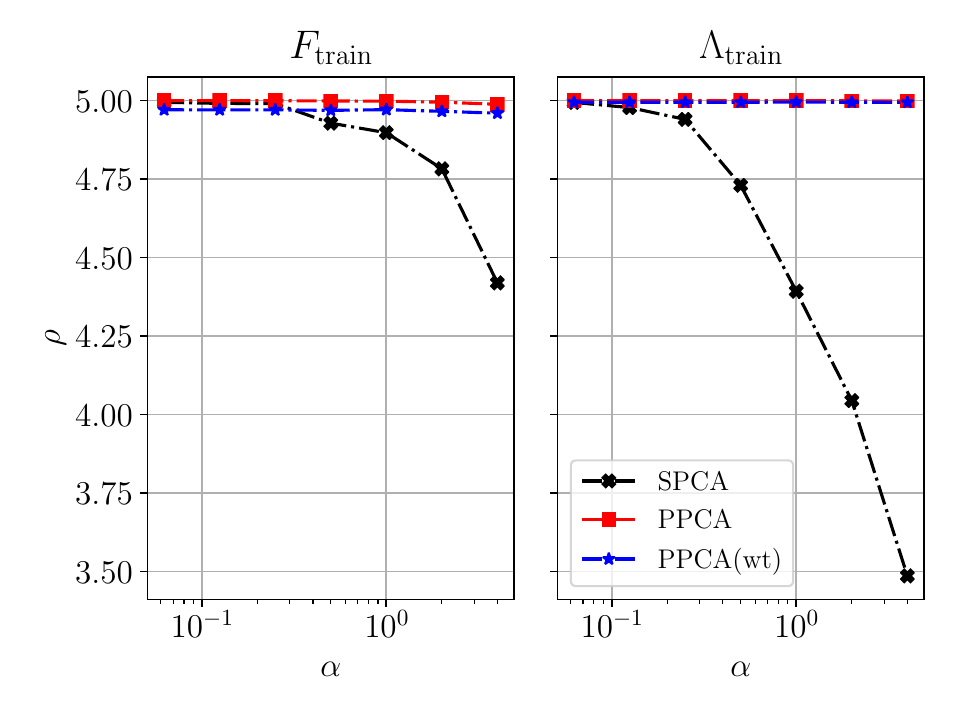}
					\subcaptab{$\rho$ with $\hat{\*F}$/$\hat{\*\Lambda}$}
				\end{subfigure}%
				\begin{subfigure}{.4\textwidth}
					\centering
					\includegraphics[width=1\linewidth]{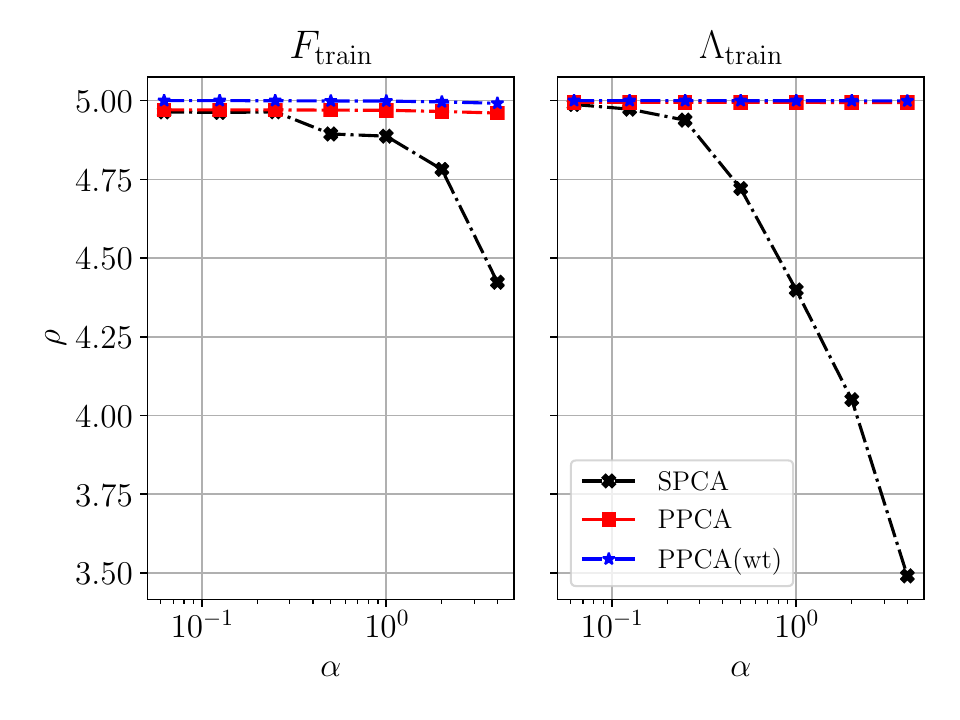}
					\subcaptab{$\rho$ with $\hat{\*F}^\twt$/$\hat{\*\Lambda}^\twt$}
				\end{subfigure}%
				\begin{subfigure}{.235\textwidth}
					\centering
					\includegraphics[width=1\linewidth]{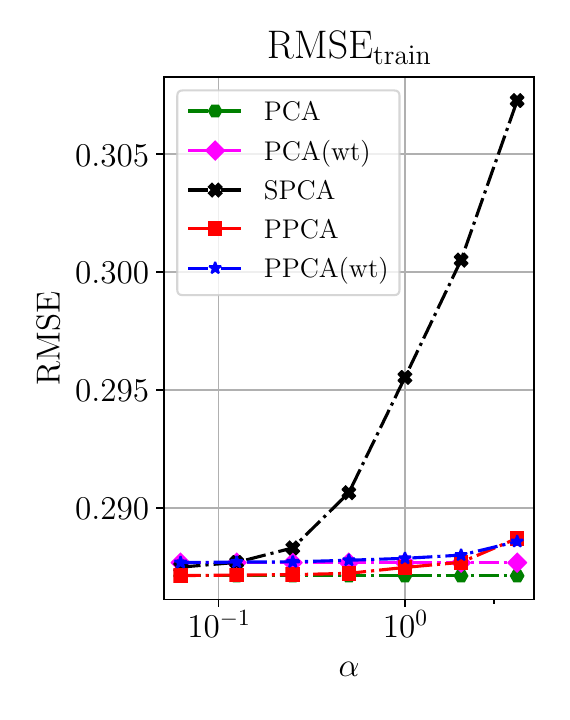}
					\subcaptab{RMSE}
				\end{subfigure}
			\end{adjustwidth}
			\label{fig:370-port-rho-rmse-train}
			\bnotefig{This figure compares the in-sample generalized correlations for factors and loadings and in-sample RMSE for proximate factors (PPCA), weighted proximate factors (PPCA (wt)), sparse PCA (SPCA), weighted PCA (PCA (wt)) and unweighted PCA. PPCA (wt) and PCA (wt) use the inverse standard errors as weights. In order to achieve the same sparsity level for various methods, we first choose $\alpha$, the $\ell_1$ penalty for SPCA, and set the number of nonzero weights $m_j$ for each factor in PPCA and PPCA (wt) to obtain the same number of nonzero elements in each factor as SPCA. The left figure shows the generalized correlation of the factors and loadings with the PCA estimates  $\hat{\*F}$/loadings $\hat{\*\Lambda}$. The middle figure show the corresponding generalized correlations with weighted PCA estimates $\hat{\*F}^\twt$/loadings $\hat{\*\Lambda}^\twt$. The right figure displays the RMSE for all five methods. PCA, PCA (wt), PPCA and PPCA (wt) achieve very similar performance and significantly outperform SPCA.} 
		\end{figure}
		
		\begin{figure}[H]
			\tcapfig{Macroeconomic Data: In-Sample Generalized Correlations and RMSE  (SPCA with $\ell_1$ Penalty)}
			\begin{adjustwidth}{-1cm}{}
				\centering
				\begin{subfigure}{.4\textwidth}
					\centering
					\includegraphics[width=1\linewidth]{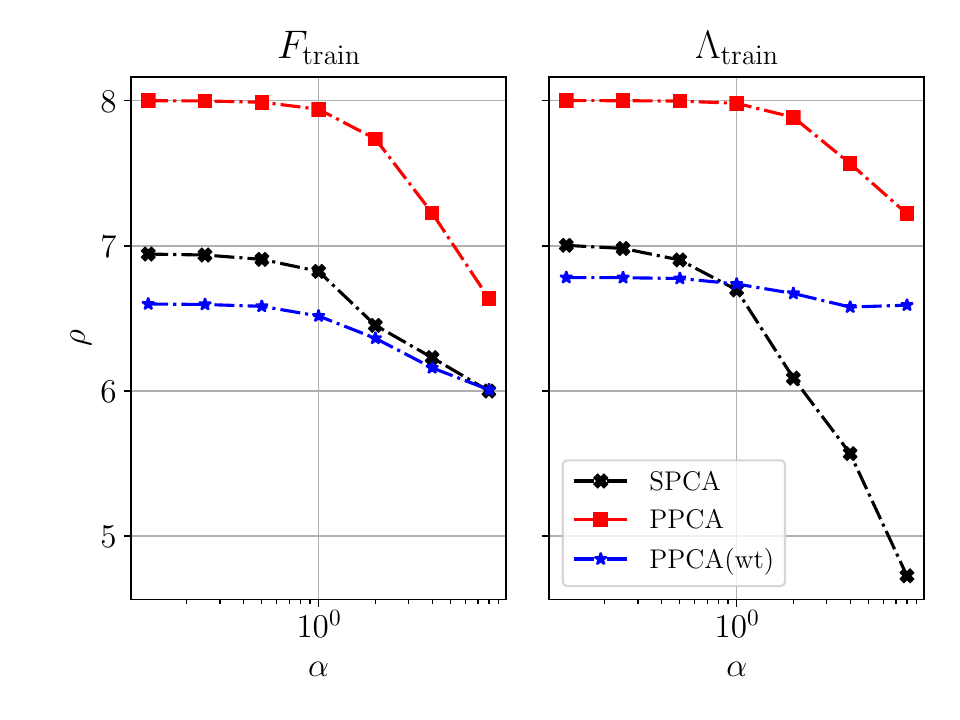}
					\caption{$\rho$ with $\hat{\*F}$/$\hat{\*\Lambda}$}
				\end{subfigure}%
				\begin{subfigure}{.4\textwidth}
					\centering
					\includegraphics[width=1\linewidth]{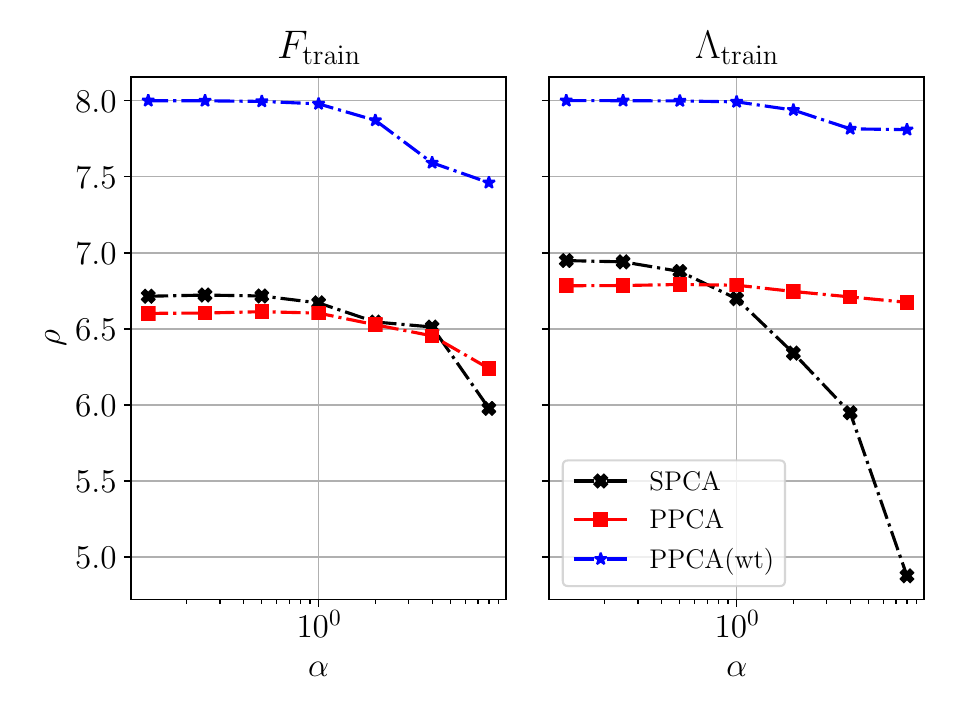}
					\caption{$\rho$ with $\hat{\*F}^\twt$/$\hat{\*\Lambda}^\twt$}
				\end{subfigure}%
				\begin{subfigure}{.235\textwidth}
					\centering
					\includegraphics[width=1\linewidth]{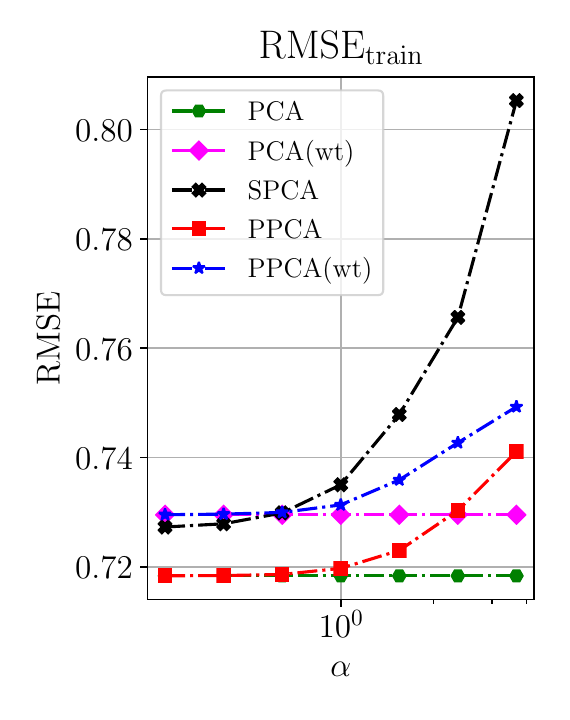}
					\caption{RMSE}
				\end{subfigure}
			\end{adjustwidth}
			\bnotefig{This figure compares the in-sample generalized correlations for factors and loadings and in-sample RMSE for proximate factors (PPCA), weighted proximate factors (PPCA (wt)), sparse PCA (SPCA), weighted PCA (PCA (wt)) and unweighted PCA. PPCA (wt) and PCA (wt) use the inverse standard errors as weights. In order to achieve the same sparsity level for various methods, we first choose $\alpha$, the $\ell_1$ penalty for SPCA, and set the number of nonzero weights $m_j$ for each factor in PPCA and PPCA (wt) to obtain the same number of nonzero elements in each factor as SPCA. The left figure shows the generalized correlation of the factors and loadings with the PCA estimates  $\hat{\*F}$/loadings $\hat{\*\Lambda}$. The middle figure show the corresponding generalized correlations with weighted PCA estimates $\hat{\*F}^\twt$/loadings $\hat{\*\Lambda}^\twt$. The right figure displays the RMSE for all five methods. PCA, PCA (wt), PPCA and PPCA (wt) achieve very similar performance and significantly outperform SPCA. }
			\label{fig:fred-md-rho-rmse-train}
		\end{figure}

		\begin{figure}[H]
			\centering
			\tcapfig{Financial Portfolio Data: Portfolio Weights of 1st Proximate Factor}
			\includegraphics[width=1\linewidth]{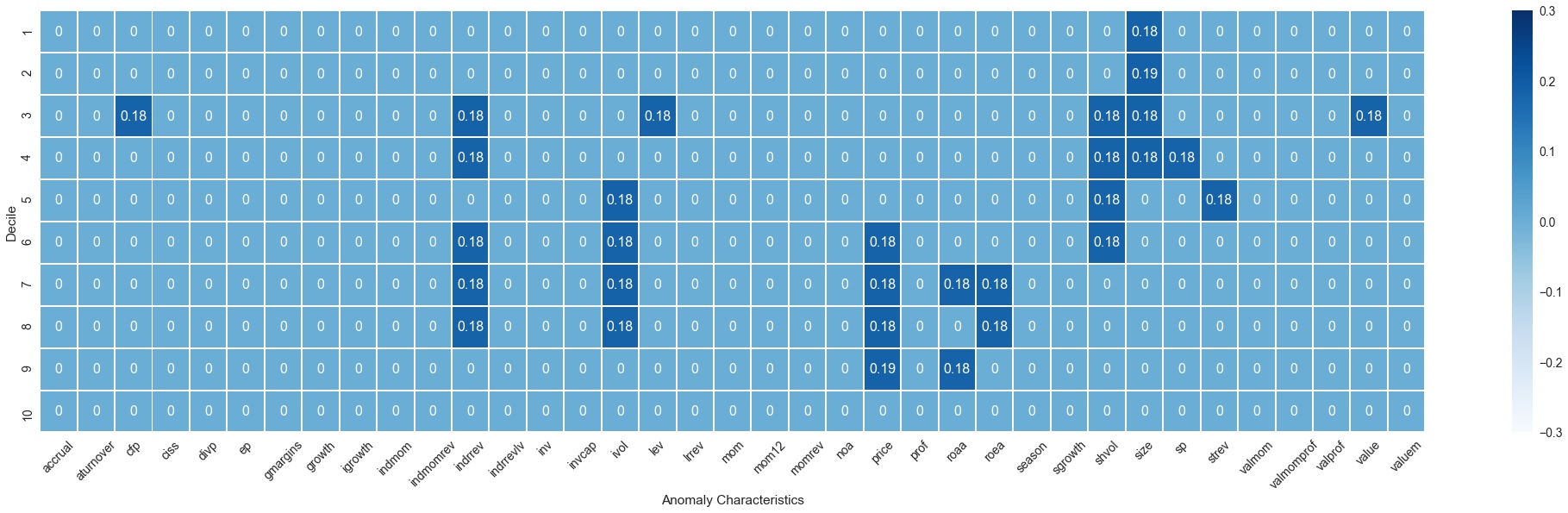}
			\bnotefig{This figure shows the portfolio weights of the 1st proximate factor. The 30 nonzero entries in the portfolio weights are composed of the extreme quantiles of momentum related characteristics. These factor weights are positive for the low quantiles and negative for the high quantiles. This proximate factor can be interpreted as a long-short momentum factor.  The names and description of the 37 anomaly characteristics are listed in Table \ref{tab:list-anomaly} in the Appendix.}
			\label{fig:single-sorted-portfolios-1st-loading}
		\end{figure}

		\begin{figure}[H]
			\centering
			\tcapfig{Financial Portfolio Data: Portfolio Weights of 2nd Proximate Factor}
			\includegraphics[width=1\linewidth]{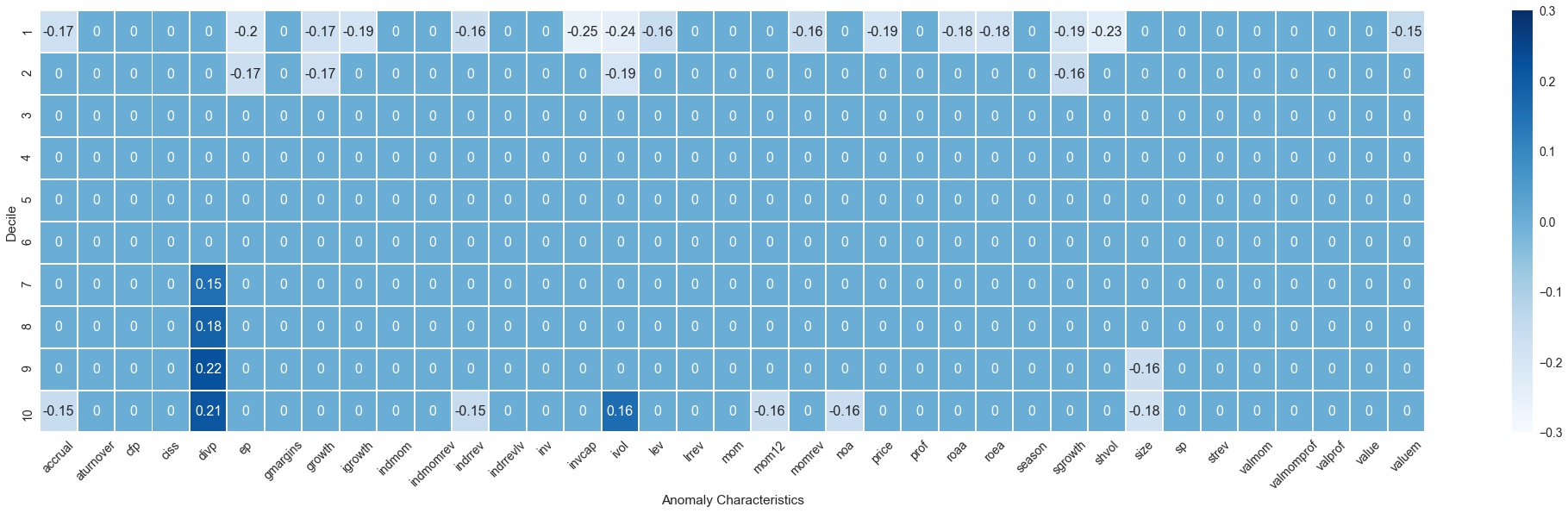}
			\bnotefig{This figure shows the portfolio weights of the 2nd proximate factor. The 30 nonzero entries in the portfolio weights are composed of the extreme quantiles of momentum related characteristics. These factor weights are positive for the low quantiles and negative for the high quantiles. This proximate factor can be interpreted as a long-short momentum factor.  The names and description of the 37 anomaly characteristics are listed in Table \ref{tab:list-anomaly} in the Appendix.}
			\label{fig:single-sorted-portfolios-2nd-loading}
		\end{figure}
		
		\begin{figure}[H]
			\centering
			\tcapfig{Financial Portfolio Data: Portfolio Weights of 3rd Proximate Factor}
			\includegraphics[width=1\linewidth]{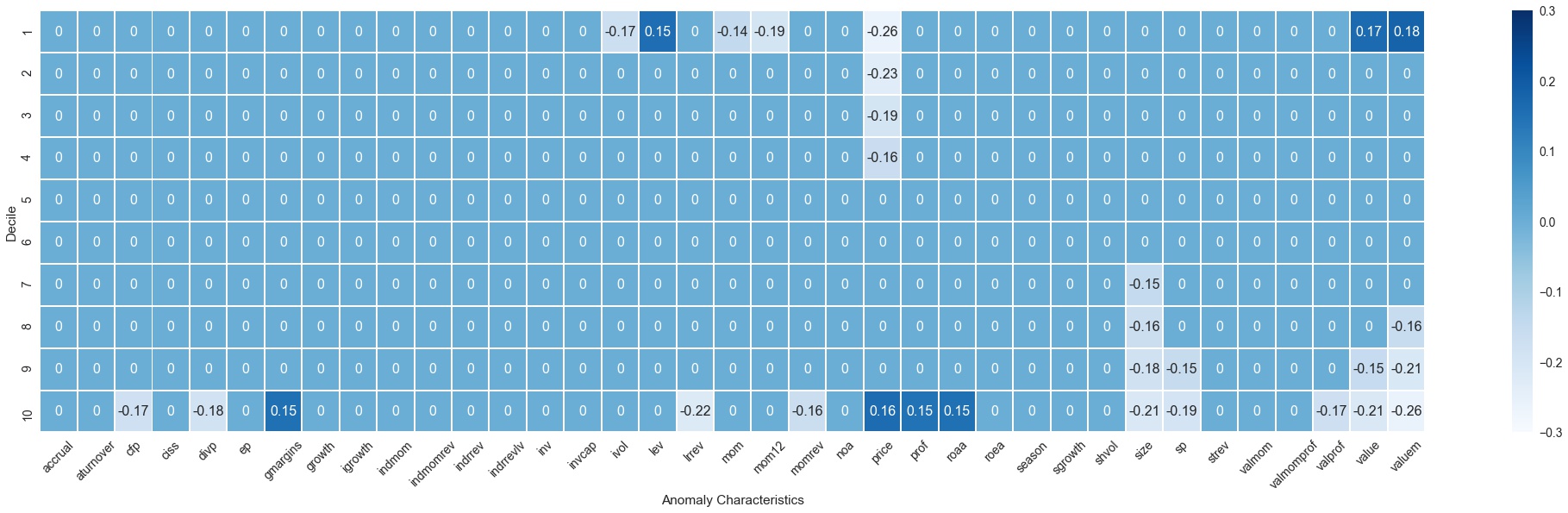}
			\bnotefig{This figure shows the portfolio weights of the 3rd proximate factor. The 30 nonzero entries in the portfolio weights are composed of the extreme quantiles of momentum related characteristics. These factor weights are positive for the low quantiles and negative for the high quantiles. This proximate factor can be interpreted as a long-short momentum factor.  The names and description of the 37 anomaly characteristics are listed in Table \ref{tab:list-anomaly} in the Appendix.}
			\label{fig:single-sorted-portfolios-3rd-loading}
		\end{figure}

		\begin{figure}[H]
			\centering
			\tcapfig{Financial Portfolio Data: Portfolio Weights of 5th Proximate Factor}
			\includegraphics[width=1\linewidth]{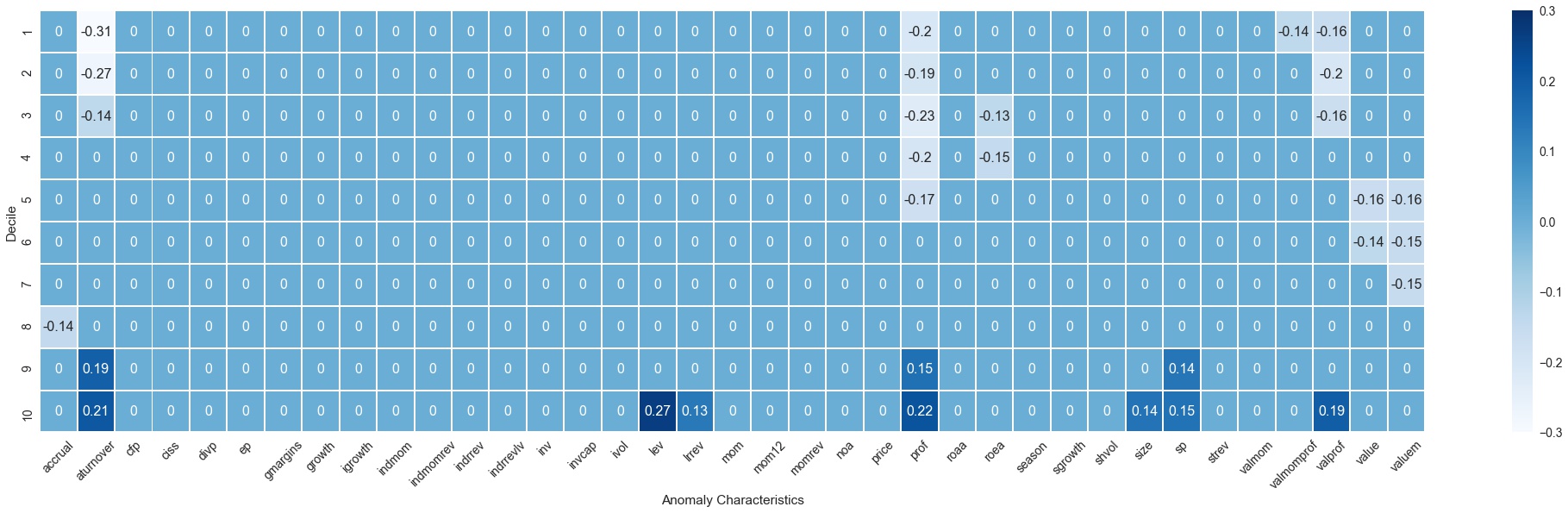}
			\bnotefig{This figure shows the portfolio weights of the 5th proximate factor. The 30 nonzero entries in the portfolio weights are composed of the extreme quantiles of momentum related characteristics. These factor weights are positive for the low quantiles and negative for the high quantiles. This proximate factor can be interpreted as a long-short momentum factor.  The names and description of the 37 anomaly characteristics are listed in Table \ref{tab:list-anomaly} in the Appendix.}
			\label{fig:single-sorted-portfolios-5th-loading}
		\end{figure}
		
		\newpage
		
		\section{Supplementary Simulation Results}\label{sec:add-simulation}
		In this section, we evaluate proximate factors (PPCA), weighted proximate factors PPCA (wt) based on inverse standard errors and sparse PCA (SPCA). 
		Our baseline model assumes
		\[\Lambda_i \overset{\iid}{\sim} \mathcal{N}(0,I_{K}), \quad F_t \overset{\iid}{\sim} \mathcal{N}(0,\Sigma_F),\quad   e_{it} \overset{\iid}{\sim} \mathcal{N}(0,\sigma_i^2), \quad \sigma_i  \overset{\iid}{\sim} \mathrm{U}(0.5,1). \]
		For the idiosyncratic components we consider two types of correlated errors
		\begin{enumerate}
			\item Cross sectional dependence:  $\*e_t \stackrel{\iid}{\sim} \mathcal{N}(0, \Sigma_e)$, where $\Sigma_e = (c_{ij}) \in \+R^{N \times N}$ with $c_{ij} = 0.5^{|i - j|}$
			\item Time series dependence: $\*e_t \stackrel{\iid}{\sim} \mathcal{N}(0, \Sigma_e)$, where $\Sigma_e = (c_{ts}) \in \+R^{T \times T}$ with $c_{ts} = 0.5^{|t - s|}$
		\end{enumerate}
		For the cross-sectional dependent errors, we present the results in Figure \ref{fig:thm-evt-one-factor_cross_dep}, \ref{fig:thm-evt-one-factor-N_cross_dep},  \ref{fig:thm-evt-multi-factor_cross_dep}  and \ref{fig:thm-evt-multi-factor-N_cross_dep}. For the time-series dependent errors, we present the results in Figure \ref{fig:thm-evt-one-factor_time_dep}, \ref{fig:thm-evt-one-factor-N_time_dep}, \ref{fig:thm-evt-multi-factor_time_dep} and \ref{fig:thm-evt-multi-factor-N_time_dep}.
		We confirm that even with correlated errors,  proximate factors are a very good approximation of the population factors and that the lower bounds are an accurate description of the exceedance probabilities for the generalized correlations. 
		\begin{figure}[h!]
			\centering
			\tcapfig{Correlations in One-Factor Model as a Function of $m$ with Cross-Dependent Errors }
			\begin{subfigure}{.5\textwidth}
				\centering
				\includegraphics[width=1\linewidth]{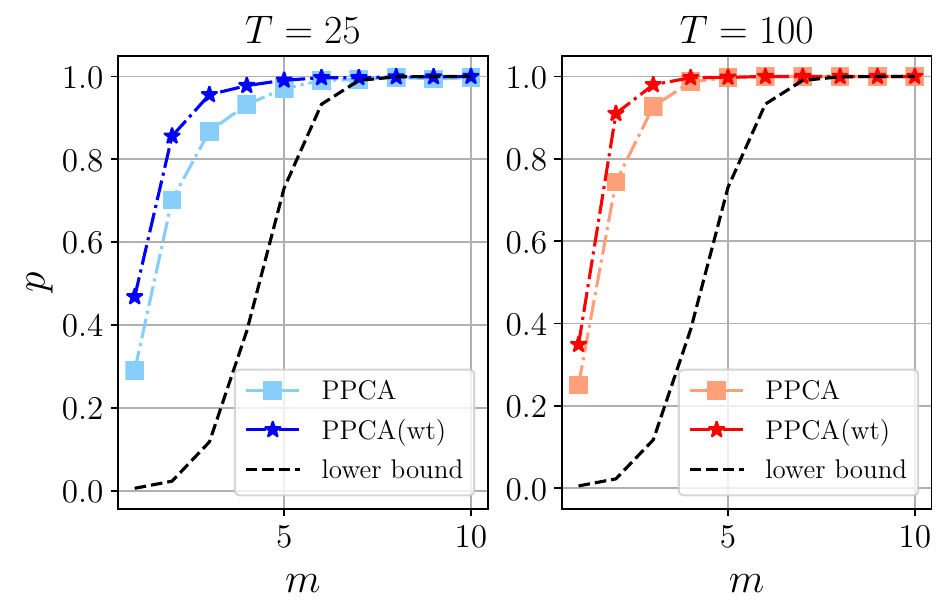}
				\subcaptab{$\sigma_{\*F_1} = 1.0$}
			\end{subfigure}%
			\begin{subfigure}{.5\textwidth}
				\centering
				\includegraphics[width=1\linewidth]{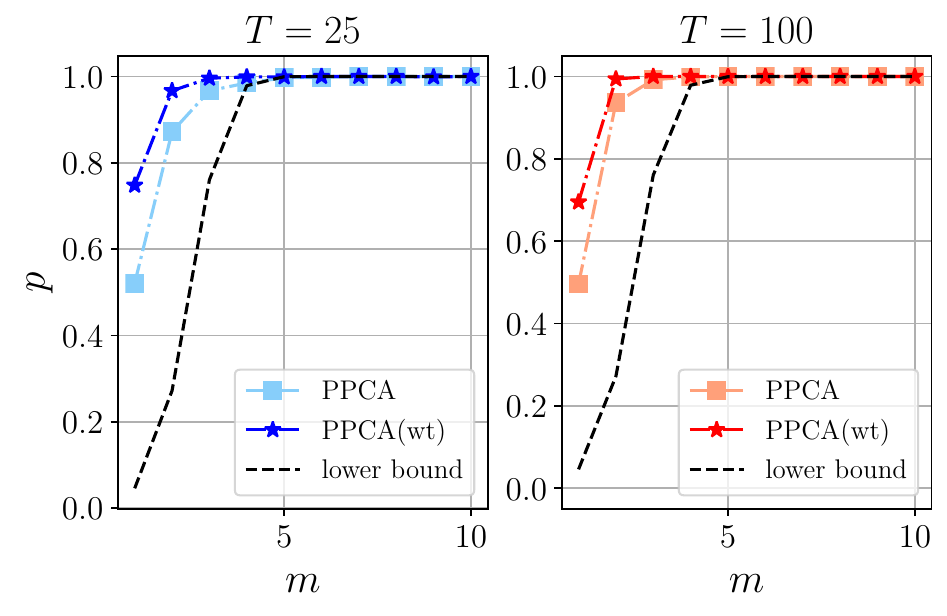}
				\subcaptab{$\sigma_{\*F_1} = 1.2$}
			\end{subfigure}
			\bnotefig{This figure compares $P(\rho \geq \rho_0)$ based on 1,000 Monte Carlo simulations and the probability lower bound $\underline{p} = \bar{G}_{1,m}(y_m)$ as a function of $m$. PPCA are the unweighted proximate factors and PPCA (wt) use the inverse standard errors as weights. We set $N=100$ and $\rho_0=0.95$. The left plots use $\sigma_{\*F_1} = 1.0$ while the right plots have the higher $\sigma_{\*F_1} = 1.2$. We calculate the probabilities for $T=25$ and $T=100$.  Both, $P(\rho \geq \rho_0)$ and $\underline{p}$, are very close to 1 with about 5-10\% of units $m$ to construct the proximate factors.}
			\label{fig:thm-evt-one-factor_cross_dep}
		\end{figure}
		
		\begin{figure}[h!]
			\centering
			\tcaptab{Generalized Correlations as a Function of $N$ with Cross-Dependent Errors}
			\begin{subfigure}{.5\textwidth}
				\centering
				\includegraphics[width=1\linewidth]{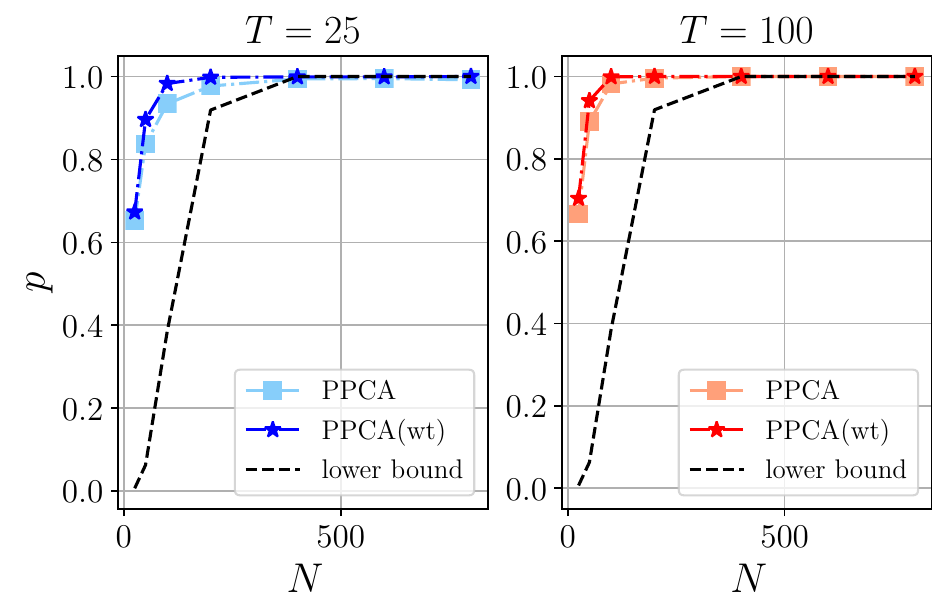}  \subcaptab{One-factor model}
				\label{fig:thm-evt-one-factor-N_cross_dep}
			\end{subfigure}%
			\begin{subfigure}{.5\textwidth}
				\centering
				\includegraphics[width=1\linewidth]{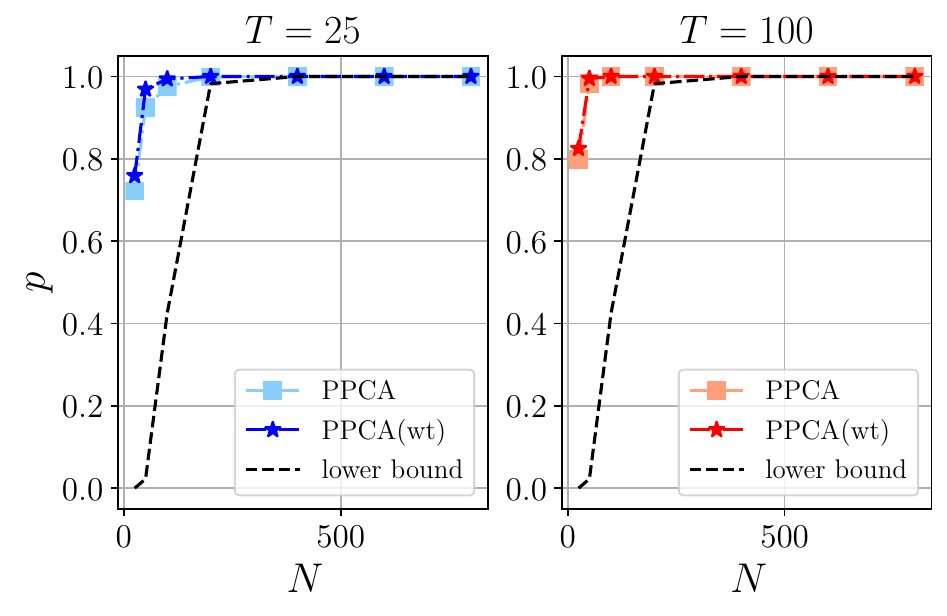}
				\subcaptab{Multi-factor model}
				\label{fig:thm-evt-multi-factor-N_cross_dep}
			\end{subfigure}
			\bnotetab{This figure compares $P(\rho \geq \rho_0)$ based on 1,000 Monte Carlo simulations and the probability lower bound $\underline{p}$ as a function of $N$. The nonzero elements are set to $m=4$ in all cases. PPCA are the unweighted proximate factors and PPCA (wt) use the inverse standard errors as weights. The one-factor model uses $\rho_0 = 0.95$ and $\sigma_{\*F_1}=1.0$, while the multi-factor model sets $\rho_0 = 1.9$ and $[\sigma_{\*F_1}, \sigma_{\*F_2}]=[1.2, 1.0]$. The lower bound is $\underline{p} = \bar{G}_{1,m}(y_m)$ for one factor and is defined in Equation \eqref{eqn:mod-lower-bound-multi-factor} for two factors. We calculate the probabilities for $T=25$ and $T=100$. Both, $P(\rho \geq \rho_0)$ and $\underline{p}$, are very close to 1 for $N > 250$.
			}
		\end{figure}

		\begin{figure}[h!]
			\centering
			\tcapfig{Generalized Correlations in Multi-Factor Model as a Function of $m$  with Cross-Dependent Errors}
			\begin{subfigure}{.5\textwidth}
				\centering
				\includegraphics[width=1\linewidth]{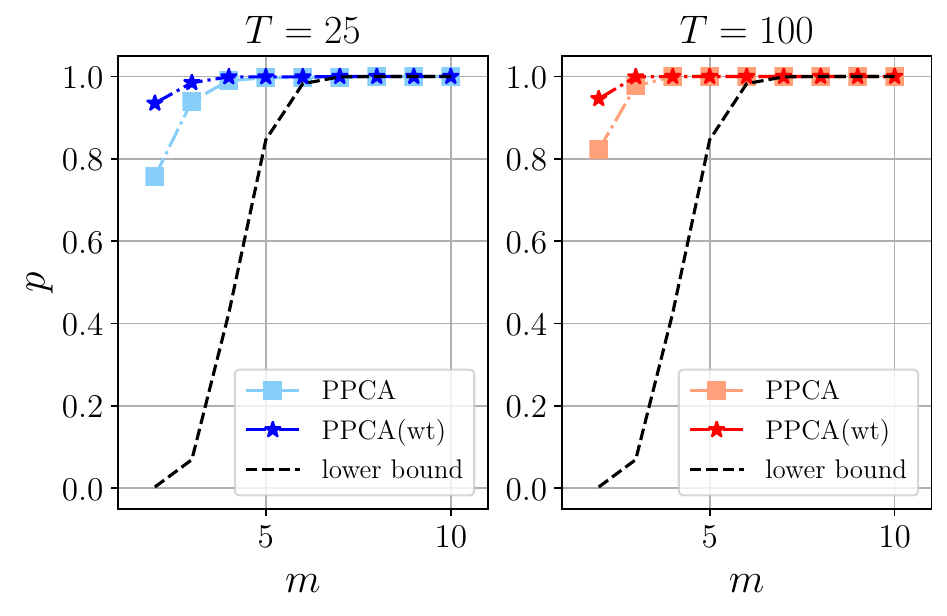}
				\subcaptab{$[\sigma_{\*F_1}, \sigma_{\*F_2}] = [1.2, 1.0]$}
			\end{subfigure}%
			\begin{subfigure}{.5\textwidth}
				\centering
				\includegraphics[width=1\linewidth]{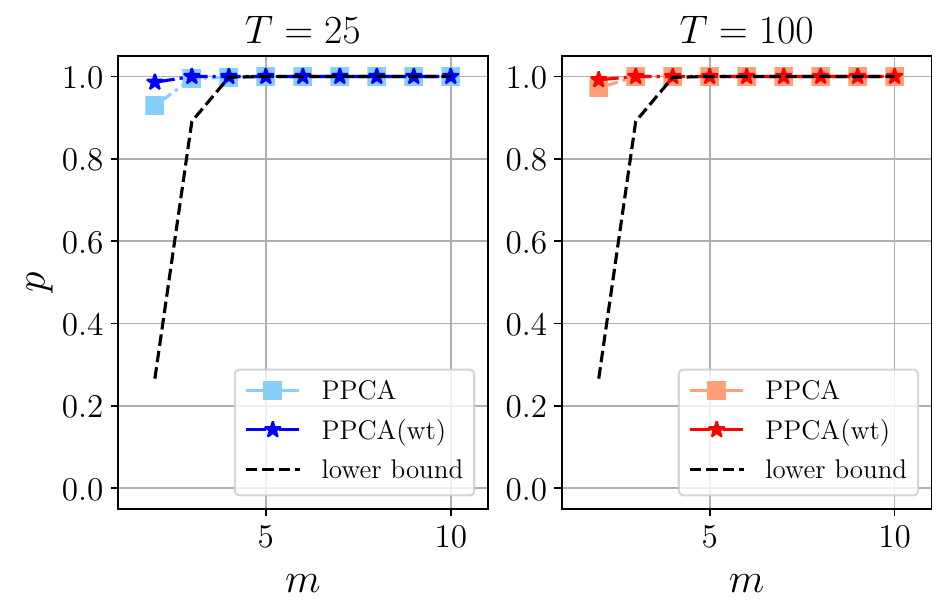}
				\subcaptab{$[\sigma_{\*F_1}, \sigma_{\*F_2}] = [1.5, 1.2]$}
			\end{subfigure}
			\bnotefig{This figure compares $P(\rho \geq \rho_0)$ based on 1,000 Monte Carlo simulations and the probability lower bound defined in Equation \eqref{eqn:mod-lower-bound-multi-factor} as a function of $m$. We have $K=2$ and set $\rho_0=1.9$ and $N=100$. PPCA are the unweighted proximate factors and PPCA (wt) use the inverse standard errors as weights. The left plots use $[\sigma_{\*F_1}, \sigma_{\*F_2}] = [1.2, 1.0]$ while the right plots have the higher $[\sigma_{\*F_1}, \sigma_{\*F_2}] = [1.5, 1.2]$. We calculate the probabilities for $T=25$ and $T=100$.  Both, $P(\rho \geq \rho_0)$ and $\underline{p}$, are very close to 1 with about 5-10\% of units $m$ to construct the proximate factors.}
			\label{fig:thm-evt-multi-factor_cross_dep}
		\end{figure}
		
		\begin{figure}[h!]
			\centering
			\tcapfig{Correlations in One-Factor Model as a Function of $m$ with Time-Dependent Errors }
			\begin{subfigure}{.5\textwidth}
				\centering
				\includegraphics[width=1\linewidth]{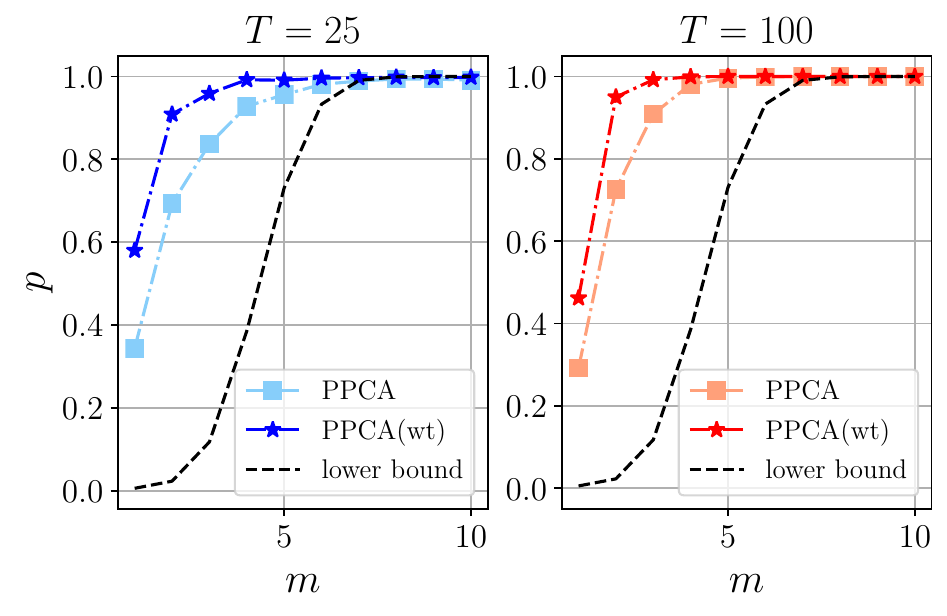}
				\subcaptab{$\sigma_{\*F_1} = 1.0$}
			\end{subfigure}%
			\begin{subfigure}{.5\textwidth}
				\centering
				\includegraphics[width=1\linewidth]{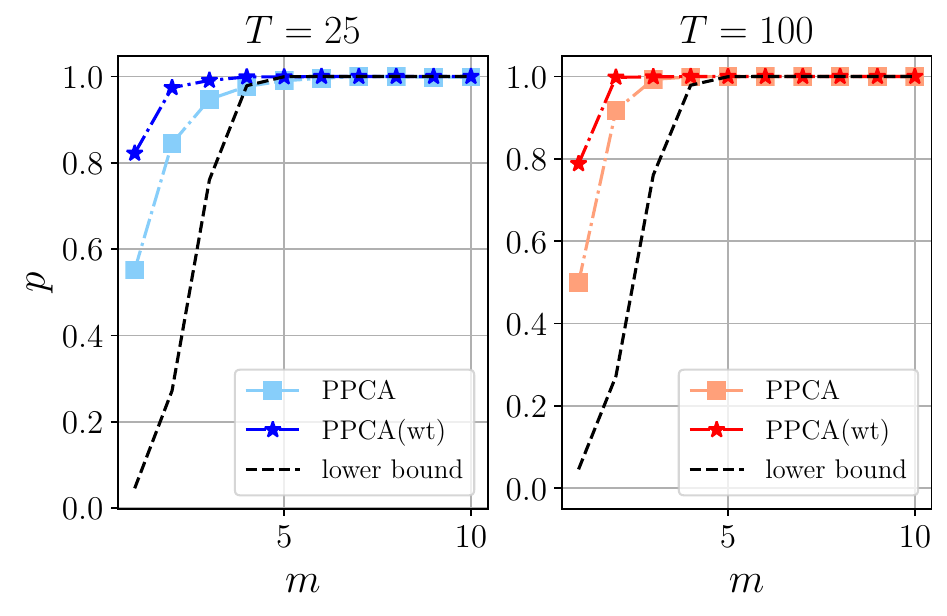}
				\subcaptab{$\sigma_{\*F_1} = 1.2$}
			\end{subfigure}
			\bnotefig{This figure compares $P(\rho \geq \rho_0)$ based on 1,000 Monte Carlo simulations and the probability lower bound $\underline{p} = \bar{G}_{1,m}(y_m)$ as a function of $m$. PPCA are the unweighted proximate factors and PPCA (wt) use the inverse standard errors as weights. We set $N=100$ and $\rho_0=0.95$. The left plots use $\sigma_{\*F_1} = 1.0$ while the right plots have the higher $\sigma_{\*F_1} = 1.2$. We calculate the probabilities for $T=25$ and $T=100$.  Both, $P(\rho \geq \rho_0)$ and $\underline{p}$, are very close to 1 with about 5-10\% of units $m$ to construct the proximate factors.}
			\label{fig:thm-evt-one-factor_time_dep}
		\end{figure}
		
		\begin{figure}[h!]
			\centering
			\tcaptab{Generalized Correlations as a Function of $N$ with Time-Dependent Errors}
			\begin{subfigure}{.5\textwidth}
				\centering
				\includegraphics[width=1\linewidth]{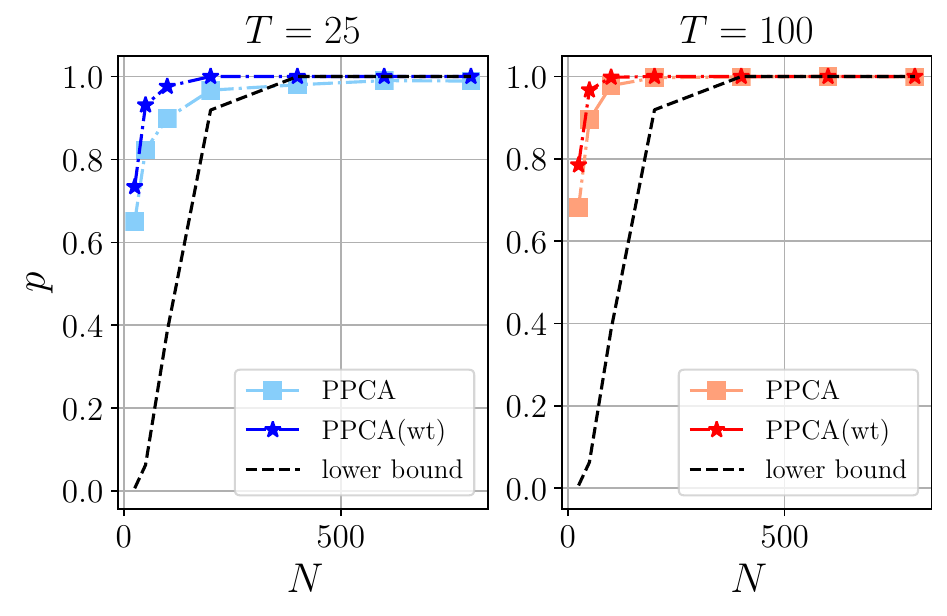}  \subcaptab{One-factor model}
				\label{fig:thm-evt-one-factor-N_time_dep}
			\end{subfigure}%
			\begin{subfigure}{.5\textwidth}
				\centering
				\includegraphics[width=1\linewidth]{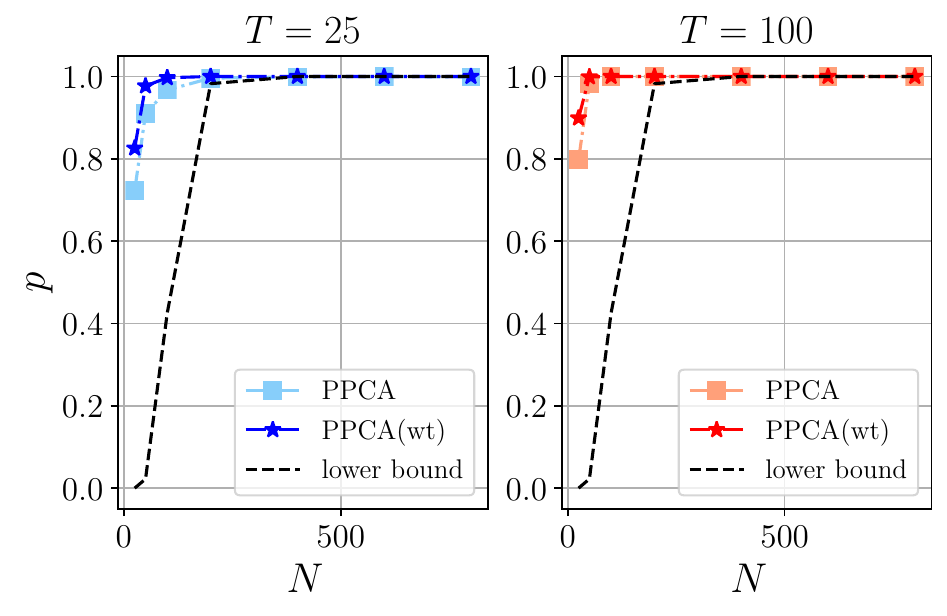}
				\subcaptab{Multi-factor model}
				\label{fig:thm-evt-multi-factor-N_time_dep}
			\end{subfigure}
			\bnotetab{This figure compares $P(\rho \geq \rho_0)$ based on 1,000 Monte Carlo simulations and the probability lower bound $\underline{p}$ as a function of $N$. The nonzero elements are set to $m=4$ in all cases. PPCA are the unweighted proximate factors and PPCA (wt) use the inverse standard errors as weights. The one-factor model uses $\rho_0 = 0.95$ and $\sigma_{\*F_1}=1.0$, while the multi-factor model sets $\rho_0 = 1.9$ and $[\sigma_{\*F_1}, \sigma_{\*F_2}]=[1.2, 1.0]$. The lower bound is $\underline{p} = \bar{G}_{1,m}(y_m)$ for one factor and is defined in Equation \eqref{eqn:mod-lower-bound-multi-factor} for two factors. We calculate the probabilities for $T=25$ and $T=100$. Both, $P(\rho \geq \rho_0)$ and $\underline{p}$, are very close to 1 for $N > 250$.
			}
		\end{figure}

		\begin{figure}[h!]
			\centering
			\tcapfig{Generalized Correlations in Multi-Factor Model as a Function of $m$  with Time-Dependent Errors}
			\begin{subfigure}{.5\textwidth}
				\centering
				\includegraphics[width=1\linewidth]{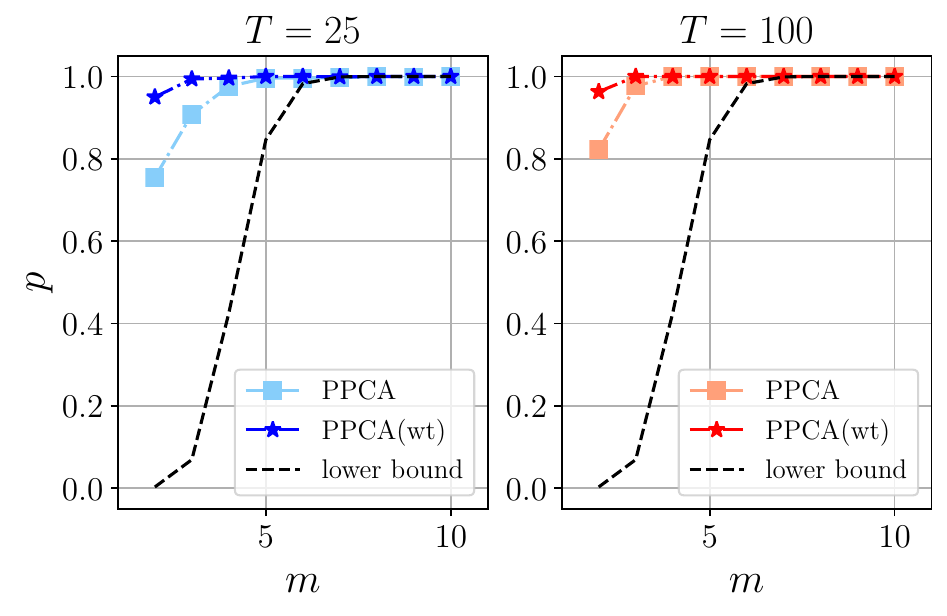}
				\subcaptab{$[\sigma_{\*F_1}, \sigma_{\*F_2}] = [1.2, 1.0]$}
			\end{subfigure}%
			\begin{subfigure}{.5\textwidth}
				\centering
				\includegraphics[width=1\linewidth]{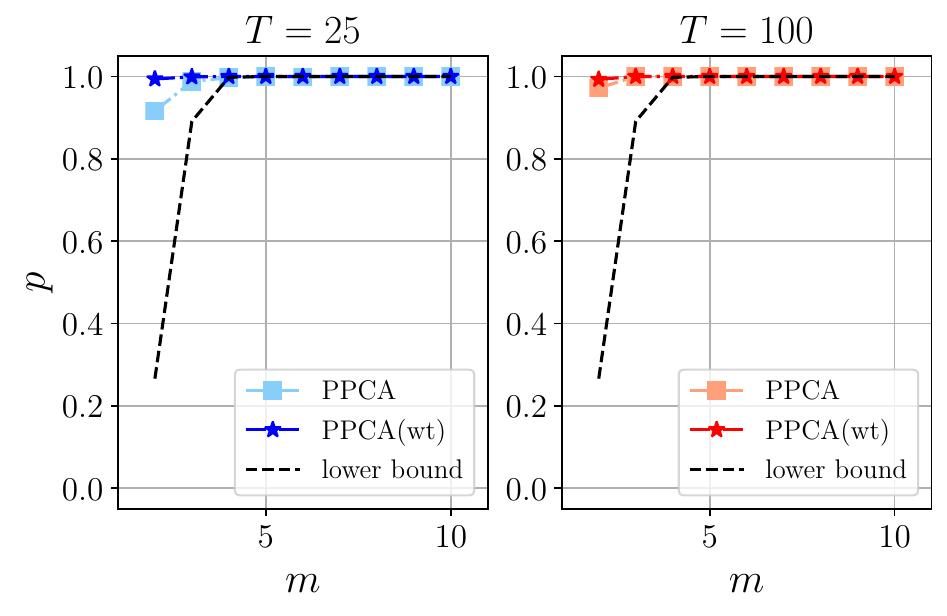}
				\subcaptab{$[\sigma_{\*F_1}, \sigma_{\*F_2}] = [1.5, 1.2]$}
			\end{subfigure}
			\bnotefig{This figure compares $P(\rho \geq \rho_0)$ based on 1,000 Monte Carlo simulations and the probability lower bound defined in Equation \eqref{eqn:mod-lower-bound-multi-factor} as a function of $m$. We have $K=2$ and set $\rho_0=1.9$ and $N=100$. PPCA are the unweighted proximate factors and PPCA (wt) use the inverse standard errors as weights. The left plots use $[\sigma_{\*F_1}, \sigma_{\*F_2}] = [1.2, 1.0]$ while the right plots have the higher $[\sigma_{\*F_1}, \sigma_{\*F_2}] = [1.5, 1.2]$. We calculate the probabilities for $T=25$ and $T=100$.  Both, $P(\rho \geq \rho_0)$ and $\underline{p}$, are very close to 1 with about 5-10\% of units $m$ to construct the proximate factors.}
			\label{fig:thm-evt-multi-factor_time_dep}
		\end{figure}

	\end{document}